\documentclass[conference]{IEEEtran}

\usepackage{stmaryrd}
\usepackage{amsfonts}
\usepackage{amsmath,amssymb}
\usepackage{amsthm}
\usepackage[]{bbm}
\usepackage{dsfont}
\usepackage{array}
\usepackage{graphicx}
\usepackage{scalerel}
\usepackage{mathtools}
\usepackage{fontawesome}
\usepackage{wasysym}
\usepackage{calc}
\usepackage{cleveref}

\usepackage{thmtools, thm-restate}
\declaretheorem{theorem}

\usepackage{enumitem}[shortlabels]

\usepackage{breakcites}
\usepackage{etoolbox}
\makeatletter
\patchcmd{\@citex}{,}{,$\!$}{}{}
\makeatother

\usepackage{url}

\newcommand{\probbox}[1]{\centerline{\framebox{\parbox{0.9\columnwidth}{#1}}}}

\newtheorem{proposition}[theorem]{Proposition}
\newtheorem{lemma}[theorem]{Lemma}
\newtheorem{corollary}[theorem]{Corollary}

\newtheorem{definition}[theorem]{Definition}

\newcommand{\homclosed}{homclosed}

\DeclareMathAlphabet{\mathbbold}{U}{bbold}{m}{n}

\newcommand{\bal}{\rotatebox[origin=c]{180}{$\mathbbold{A}$}}
\newcommand{\bex}{\rotatebox[origin=c]{180}{$\mathbbold{E}$}}
\newcommand{\SO}{\mbox{$\mathbb{SO}$}}

\newcommand{\ESO}{\mbox{$\Ex{}\SO$}}

\newcommand{\Logic}{\mathbb{L}}

\newcommand{\bA}{{\mathfrak A}}
\newcommand{\bB}{{\mathfrak B}}

\newcommand{\FO}{\mathbb{FO}}
\newcommand{\FOt}{\FO^\mathbbm{2}}
\newcommand{\FOe}{\FO_{=}}
\newcommand{\FOte}{\mbox{$\FO_{=}^\mathbbm{2}$}}

\newcommand{\GFO}{\mbox{$\mathbb{GFO}_{=}$}}
\newcommand{\GNFO}{\mbox{$\mathbb{GNFO}_{=}$}}
\newcommand{\TGF}{\mbox{$\mathbb{TGF}$}}
\newcommand{\TGD}{\mbox{$\mathbb{TGD}$}}
\newcommand{\MDTGD}{\mbox{$\mathbb{MDTGD}$}}
\newcommand{\DTGD}{\mbox{$\mathbb{DTGD}$}}
\newcommand{\CQ}{\mbox{$\mathbb{CQ}$}}
\newcommand{\UCQ}{\mbox{$\mathbb{UCQ}$}}

\newcommand{\Ex}{\mbox{$\bex$}}
\newcommand{\Al}{\mbox{$\bal$}}
\newcommand{\EpFO}{\mbox{$\Ex^*\FO^+_{=}$}}
\newcommand{\EFO}{\mbox{$\Ex{}\FO^+$}}
\newcommand{\EAAE}{\mbox{$\Ex^*\Al\Al\Ex^*\FO$}}
\newcommand{\EEAA}{\mbox{$\Ex\Ex\Al\Al\FO$}}
\newcommand{\EEEA}{\mbox{$\Ex\Ex\Ex\Al\FO$}}
\newcommand{\AAEE}{\mbox{$\Al\Al\Ex\Ex\FO$}}

\newcommand{\AAAE}{\mbox{$\Al\Al\Al\Ex\FO$}}
\newcommand{\AAEe}{\mbox{$\Al\Al\Ex\FOe$}}

\newcommand{\EEAe}{\mbox{$\Ex\Ex\Al\FOe$}}

\newcommand{\EAFO}{\mbox{$\Ex^*\Al^*\FOe$}}
\newcommand{\AsFO}{\mbox{$\Al^*\FOe$}}

\newcommand{\EsFO}{\mbox{$\Ex^*\FOe$}}
\newcommand{\mcl}[1]{\multicolumn{2}{l}{#1}}

\newcommand{\cq}[1]{\mathrm{cq}(#1)}

\DeclareMathOperator*{\smallbigwedge}{{\textstyle\bigwedge}}
\DeclareMathOperator*{\smallbigvee}{{\textstyle\bigvee}}

\newcommand{\spoil}{{\scaleobj{0.8}{\text{\faBolt}}}}

\renewcommand{\bold}[1]{{\boldsymbol{#1}}}

\newcommand{\fin}{{\scaleobj{0.8}{<}\omega}}
\newcommand{\ffin}{{\scaleobj{0.8}(\!\,\fin\,\!\scaleobj{0.8})}}

\newcommand{\strok}{\rotatebox[origin=c]{-45}{\text{--}}}

\newcommand{\mergehom}{\mathrel{{\raisebox{0.5ex}{\hstretch{.8}{\scaleobj{0.7}{\strok}}}\hspace{-1.3ex}\to}}}

\newcommand{\str}[2][0pt]{\mathrel{{#2}\llap{\hspace{0ex}\raisebox{-0.20ex}{\hstretch{0.6}{\scaleobj{1}{\raisebox{-0.20ex}{$-$}}\hspace{#1}\phantom{\scaleobj{0.75}{-}}}}}}}

\newcommand{\arbhom}{\mathrel{{\to}}}
\newcommand{\injhom}{\mathrel{{\hookrightarrow}}}
\newcommand{\surhom}{\mathrel{{\to\hspace{-0.8ex}\!\!\!\to}}}
\newcommand{\injsur}{\mathrel{{\hookrightarrow\hspace{-0.8ex}\!\!\!\to}}}
\newcommand{\mmhom}{\mathrel{\str{\mergehom}}}
\newcommand{\embed}{\mathrel{{\str{\injhom}}}}
\newcommand{\sshom}{\mathrel{{\str[0.6ex]{\surhom}}}}

\newcommand{\cl}[1]{{\scaleobj{1.1}{{\textnormal{${#1}$}}}}}

\newcommand{\homcl}[1]{#1^\cl{\arbhom}}

\newcommand{\homclfin}[1]{#1^{\cl\arbhom_\fin}}

\newcommand{\homclffin}[1]{#1^{\cl\arbhom_{(\fin)}}}

\newcommand{\getmodels}[1]{\llbracket #1 \rrbracket}
\newcommand{\getfinmodels}[1]{\getmodels{#1}_\fin}

\newcommand{\getffinmodels}[1]{\getmodels{#1}_\ffin}

\newcommand{\getsummary}[1]{\llparenthesis #1 \rrparenthesis}

\newcommand{\impl}{\Rightarrow}

\newcommand{\sig}[1]{\mathtt{#1}}

\newcommand{\sigU}{\sig{U}}

\newcommand{\sigP}{\sig{P}}

\newcommand{\sigH}{\sig{H}}
\newcommand{\sigV}{\sig{V}}
\newcommand{\sigT}{\sig{T}}
\newcommand{\sigc}{\sig{c}}
\newcommand{\Bit}{\sig{Bit}}

\newcommand{\tp}{\mathfrak{t}}
\newcommand{\gettp}{\mathfrak{t}^{?\!}}

\newcommand{\Tp}{\sig{Tp}}

\usepackage{scalerel,stackengine}
\stackMath
\newcommand\widecheck[1]{%
	\savestack{\tmpbox}{\stretchto{%
			\scaleto{%
				\scalerel*[\widthof{\ensuremath{#1}}]{\kern-.6pt\bigwedge\kern-.6pt}%
				{\rule[-\textheight/2]{1ex}{\textheight}}
			}{\textheight}%
		}{0.5ex}}%
	\stackon[1pt]{#1}{\scalebox{-1}{\tmpbox}}%
}
\usepackage{xcolor}
\definecolor{manuel}{rgb}{0.00,0.00,0.50}
\definecolor{simon}{rgb}{0.80,0.00,0.80}
\definecolor{thomas}{rgb}{0.00,0.50,0.00}
\definecolor{sebastian}{rgb}{0.50,0.00,0.00}
\definecolor{todo}{rgb}{0.99,0.00,0.00}

\newcommand{\neutral}{\color{black} }

\usepackage{changes} 
\newcommand{\stkout}[1]{\ifmmode\text{\sout{\ensuremath{#1}}}\else\sout{#1}\fi}
\renewcommand{\stkout}[1]{\ifmmode\text{\sout{\ensuremath{#1}}}\else\sout{#1}\fi}
\setdeletedmarkup{\stkout{#1}}

\definechangesauthor[ color=simon]{SK}

\begin{document}

\title{On Logics and Homomorphism Closure}

\author{\IEEEauthorblockN{Manuel Bodirsky\IEEEauthorrefmark{1},
Thomas Feller\IEEEauthorrefmark{2},
Simon Knäuer\IEEEauthorrefmark{1}, and
Sebastian Rudolph\IEEEauthorrefmark{2}}
\IEEEauthorblockA{\IEEEauthorrefmark{1} Institut f\"{u}r Algebra, TU Dresden}
\IEEEauthorblockA{\IEEEauthorrefmark{2} Computational Logic Group, TU Dresden}
}

\IEEEoverridecommandlockouts
\IEEEpubid{\makebox[\columnwidth]{\,Extended preprint version of paper accepted at LICS 2021.\hfill} \hspace{\columnsep}\makebox[\columnwidth]{ }}

\maketitle

\begin{abstract}
Predicate logic is the premier choice for specifying classes of relational structures.
Homomorphisms are key to describing correspondences between relational structures. Questions concerning the interdependencies between these two means of characterizing (classes of) structures are of fundamental interest and can be highly non-trivial to answer. 
We investigate several problems regarding the homomorphism closure (homclosure) of the class of all (finite or arbitrary) models of logical sentences: membership of structures in a sentence's homclosure; sentence homclosedness; homclosure characterizability in a logic; normal forms for homclosed sentences in certain logics. For a wide variety of fragments of first- and second-order predicate logic, we clarify these problems' computational properties.
\end{abstract}


%
\IEEEpeerreviewmaketitle


\section{Introduction}
The
\emph{homomorphism preservation theorem} states that, for any first-order sentence $\Phi$, the property of being a model of~$\Phi$~is preserved under homomorphisms\footnote{Henceforth, we will simply call such a $\Phi$ (and its model class)  \emph{\homclosed{}}.} if and only if $\Phi$ is equivalent to an existential positive sentence.
For arbitrary models, this correspondence goes back to \L{}os, Tarski, and Lyndon in the 1950s, whereas Rossman showed only rather recently that it also holds when considering \emph{finite} models only \cite{Rossman08}.

Classes of homclosed first- or second-order sentences are ubiq\-ui\-tous in computer science. In databases, they often serve as intuitive \emph{query languages} and check for the presence of a certain ``pattern'' in a given database \cite{DBLP:conf/pods/RudolphK13}. In particular, existential positive sentences, which can be expressed as so-called \emph{unions of conjunctive queries}, are encountered very frequently in practice. Homclosed classes of finite structures  also appear naturally in other areas of computer science; for example, the complement of any constraint satisfaction problem, viewed as a class of finite structures, is homclosed. 


This paper poses and comprehensively answers the following four fundamental questions concerning homclosures:  

\begin{enumerate}[itemindent=0ex, leftmargin=3ex]
\item \textbf{Homclosure membership.}
Given a sentence $\Phi$ from some logic and a finite structure $\mathfrak{A}$, 
we want to decide whether $\mathfrak{A}$ belongs to the homclosure of $\Phi$. How difficult is this problem computationally? 
We are interested both in the \emph{combined complexity} (where both $\mathfrak{A}$ and $\Phi$ are given) and in the \emph{data complexity} (where $\mathfrak{A}$ is given and $\Phi$ is fixed). 
\item \textbf{Homclosedness.}
Given a sentence $\Phi$ from some logic, we want to decide whether $\Phi$ is homclosed. How difficult is this problem computationally? 
\item \textbf{Homclosure characterizability.} 
Given a sentence $\Phi$ from some logic, 
does there exist a sentence $\Psi$ (from the same or a different logic) such that
the class of $\Psi$'s models is precisely the homclosure of the class of $\Phi$'s models? 
\item \textbf{Homclosed normal forms.}
For which logics 
exists a \emph{`syntactic normal form'}, i.e., a subset of easily recognizable sentences such that an arbitrary sentence from that logic is equivalent to a sentence in normal form if and only if the former is homclosed?
\end{enumerate}
\vspace{-2ex}
This paper provides answers to these questions for (fragments of) first- and se\-cond-or\-der logic of varying expressivity, with a particular emphasis on very expressive but decidable formalisms (including the two-variable fragment, the guarded (negation) fragment, and the triguarded fragment), popular formalisms in database theory (such as existential positive sentences and tuple-generating dependencies), and diverse prefix classes. 
We next summarize our findings and methods. Our results hold both for the finite- and the arbitrary-model case, so we will not make this distinction here, even if distinct proofs are sometimes required.   

As for Question 1), \Cref{sec:InHomCl} presents results regarding decidability and combined complexity. We show that the homclosure membership problem is generically interreducible with the considered logic's satisfiability problem in a wide variety of cases, thus tightly tying decidability and complexity of homclosure membership to satisfiability; the reduction from the former to the latter is shown by logically encoding model \emph{colorings}, inspired by similar ideas from constraint satisfaction. As notable exception, we find that for tuple-generating dependencies, the problem is undecidable despite satisfiability being trivial. Regarding the data complexity of homclosure membership, we find that question intimately related to Question 3), since in the presence of a logical characterization of a sentence's homclosure, checking homclosure membership boils down to model checking against the characterizing sentence. Along those lines, \Cref{sec:characterizability} provides a variety of results, the most striking of which is probably \Cref{sorpresa}, asserting polynomial time data complexity of checking membership in the homclosure of sentences in the guarded (negation) fragment of first-order logic.  

Regarding Question 2) addressed in \Cref{sec:homclosedness}, we find that under fairly weak assumptions, the satisfiability problem for the logic under consideration can be reduced to its homclosedness problem. Yet, a similar reduction can be found from satisfiability in the ``dual'' logic, identifying several prefix classes where homclosedness is undecidable, despite satisfiability being decidable. For the remaining cases, we propose a generic reduction from homclosedness to unsatisfiability,
based on the notion of \emph{spoiler} -- a homomorphism from a model to a non-model, i.e., a witness for non-homclosedness -- and a way of creating spoiler-detecting sentences. Again, tuple-generating dependencies need to be coped with differently; we find that for the plain version, homclosedness is \textsc{NP}-complete, while adding disjunction turns the problem undecidable.

For Question 3), which we focus on in \Cref{sec:characterizability}, it is important to note that already for very basic first-order sentences (including tuple-generating dependencies), their homclosure is not expressible in first-order logic, and even the question whether a given first-order sentence is expressible in any logic with a decidable model checking problem is undecidable. On the other hand, we show that the homclosure of any sentence from the Bernays-Schönfinkel class can even be characterized in existential positive first-order logic. Using elaborate model-theoretic characterizations based on \emph{types}, we are able to show that the two-variable fragment, the triguarded fragment and the Gödel class allow for homclosure-characterization in second-order logic (in fact, even in decidable fragments thereof), while sentences in the guarded and the guarded negation fragment can be homclosure-characterized even in first-order logic with least fixed points. Employing arguments from descriptive complexity we show that these characterizability results are optimal in a certain sense. 

Question 4) is dealt with in \Cref{sec:normalforms}. For fragments of first-order logics, we exploit the homomorphism-preservation theorem to settle the case for all fragments fully encompassing existential positive first-order logic. We can also give a syntactic normal form for Bernays-Schönfinkel subclasses with a bounded number of quantifiers. For fragments with a decidable homclosedness problem, we can resort to this (usually costly) check to identify normal form sentences. A few cases, however, remain open.
Finally, turning to full second-order logic, we are indeed able to establish a normal form for sentences whose models are closed under homomorphisms, too. As opposed to first-order logic, however, this normal form can be even computed from the given second-order sentence in polynomial time.

Detailed proofs can be found in the appendix.

%


\section{Preliminaries}

We assume the reader to be familiar with standard notions regarding model theory, first- and second-order logic, decidability, and complexity theory. 

\noindent\textbf{Structures and Homomorphisms.}\quad
We let capital Fraktur letters denote structures and the corresponding capital Roman letters their domains; i.e., $A$ denotes the domain of $\bA$. 
We assume that structures have non-empty domains. A structure is called finite 
if its domain is finite
. We work with signatures $\tau$ that consist of finitely many constant and predicate symbols (written in typewriter font, often simply called \textsl{constants} and \textsl{predicates} for brevity), and let $\text{ar}(\sigP)$ denote the arity of a predicate $\sigP \in \tau$. 
For convenience, we assume that our signatures contain at least one constant symbol; this is not a severe restriction since we may always add a new constant symbol to the signature without affecting our results. The \emph{size} $|\bA|$ of a finite $\tau$-structure $\bA$ is the number of bits needed to represent $\bA$ up to isomorphism, and bounded by
$(|\tau|+1)(|A|+1)^{m+1}$ where $m$ is the maximal arity of the predicates in $\tau$. 

A \emph{homomorphism} $h \colon \bA \to \bB$ from a \mbox{$\tau$-structure} $\bA$ to a \mbox{$\tau$-structure} $\bB$ is a function $h \colon A \to B$
where 

\vspace{-1ex}
\begin{itemize}[itemindent=0ex, leftmargin=2ex, partopsep=0ex]
\item for every predicate symbol $\sigP \in \tau$, 
$(a_1,\dots,a_{\text{ar}(\sigP)}) \in \sigP^{\bA}$ implies $\big (h(a_1),\dots,h(a_{\text{ar}(\sigP)}) \big) \in \sigP^{\bB}$, and 
\item for every constant symbol $\sigc \in \tau$, we have
$h(\sigc^{\bA}) = \sigc^{\bB}$. 
\end{itemize}
\vspace{-1ex}

A homomorphism $h \colon \bA \to \bB$ is called \emph{strong} if for every predicate symbol $\sigP \in \tau$ we have  
$(a_1,\dots,a_{\text{ar}(\sigP)}) \in \sigP^{\bA}$ \emph{exactly if}  
$\big (h(a_1),\dots,h(a_{\text{ar}(\sigP)}) \big) \in \sigP^{\bB}$. 
We use the common symbols $\injhom$ and $\surhom$ for injective and surjective homomorphisms, respectively, and add a horizontal stroke ($\str\arbhom$, $\embed$, $\sshom$) when they are strong.
Injective strong homomorphisms are also called \emph{embeddings}. 

Given a class $\mathcal{C}$ of $\tau$-structures, let $\homcl{\mathcal{C}}$ denote the \emph{homomorphism closure of $\mathcal{C}$}, i.e., the class of all $\tau$-structures into which some structure $\mathfrak{A} \in \mathcal{C}$ has a homomorphism. We use $\mathcal{C}^{\cl{\str{\injhom}}}$ and $\mathcal{C}^{\cl{\surhom}}$ analogously.
%
%
%
%
We write $\homclfin{\mathcal{C}}$ 
for the class obtained by removing all infinite structures from $\homcl{\mathcal{C}}$.

Given a $\tau$-structure $\mathfrak{A}$ and some $\sigma \subseteq \tau$, let  
$\mathfrak{A}|_\sigma$ denote the $\sigma$-reduct of $\mathfrak{A}$. Moreover, for some constant-free signature $\rho$ disjoint from $\tau$, we let $\mathfrak{A}\cdot{\widehat\rho}$ denote the $\tau{\uplus}\rho$-structure expanding~$\mathfrak{A}$ in which every element $\sigP$ of $\rho$ is mapped to the full relation $A^{\text{ar}(\sigP)}$. Likewise, we let $\mathfrak{A}\cdot{\widecheck\rho}$ denote the \mbox{$\tau{\uplus}\rho$-structure} expanding $\mathfrak{A}$ in which every element of $\rho$ is mapped to the empty relation.  

For a signature $\rho$ of only constants, let $\hat{\mathfrak{O}}_\rho$
denote the one-element $\rho$-structure (i.e., the up-to-isomorphism unique structure with all constants mapped to the only element), while $\check{\mathfrak{O}}_\rho$ denotes the $|\rho|$-element $\rho$-structure where each constant is interpreted separately. 
For a signature $\tau=\rho \uplus \sigma$ of constants $\rho$ and predicates $\sigma$, we let $\mathfrak{I}_\tau = \check{\mathfrak{O}}_\rho \cdot \widecheck{\sigma}$ and $\mathfrak{F}_\tau = \hat{\mathfrak{O}}_\rho \cdot \widehat{\sigma}$.\footnote{Note that $\mathfrak{I}_\tau$ 
has a homomorphism into every \mbox{$\tau$-structure}, and every \mbox{$\tau$-structure} has a homomorphism into $\mathfrak{F}_\tau$, i.e., $\homcl{\{\mathfrak{I}_\tau\}}$ is the class of all \mbox{$\tau$-structures}, whereas $\mathfrak{F}_\tau \in \homcl{\mathcal{C}}$ for every non-empty class $\mathcal{C}$ of \mbox{$\tau$-structures}.} We omit signature subscripts if they are clear from the context.

\noindent\textbf{Logics.}\quad
We give a brief description of the logics we study in this paper. By default, equality is not allowed in a logic ${\mathbb L}$; 
the variant of ${\mathbb L}$ where equality $=$ is allowed is denoted by ${\mathbb L}_=$. \emph{First-order logic $(\FOe)$} formulae are built from atomic formulae using 
existential quantification $\exists$, universal quantification $\forall$, negation $\neg$, conjunction $\wedge$, and disjunction $\vee$ in the usual way.
We use the symbol $\Rightarrow$ for implication and $\Leftrightarrow$ for equivalence as the usual shortcuts for better readability. 
All the considered logics are fragments of \emph{second-order logic} (\SO),\footnote{As $=$ can be axiomatized in \SO, we do not distinguish \SO{} from $\SO_=$.} which extends first-order logic by \emph{second-order quantifiers}, of the form $\exists \sigP$ or $\forall \sigP$ for a predicate symbol $\sigP$. A formula's \emph{quantifier rank} is the maximal nesting depth of its quantifiers.  



\begin{itemize}[itemindent=2ex, leftmargin=0ex]
\item A  \emph{term} is either a variable or a constant symbol (as per our assumption that the signature may contain constant symbols but no function symbols of arities greater than zero). 
\item 
We use bold variable/term/element symbols (e.g., $\bold{x}$, $\bold{y}$, or $\bold{t}$, or $\bold{a}$) to denote finite sequences of such entities (e.g., $x_1,\ldots, x_n$) of appropriate length.
Whenever order and multiplicity are irrelevant, we take the liberty to consider such sequences as sets, justifying expressions like $\bold{x} \cup \bold{y}$.
 \item
For a (sub)formula $\varphi$, we write $\varphi[\bold{x}]$ to indicate that all free first-order variables in $\varphi$ are from $\bold{x}$. We write $\varphi(\bold{x})$ to indicate a formula with free occurrences of $\bold{x}$ and use $\varphi(\bold{t})$ to denote $\varphi$ with all occurrences of free variables from $\bold{x}$ replaced by the corresponding element from $\bold{t}$.     \end{itemize}

We study the following 
fragments of first-order logic:
\begin{itemize}[itemindent=2ex, leftmargin=0ex, partopsep=0ex]
\item
\emph{Prefix classes} of $\FO$ / $\FOe$. 
For these logics, we assume that the formulae are in prenex normal form and that the quantifier prefix is restricted by a regular expression, following 
standard nomenclature as used by Börger, Grädel and Gurevich \cite{BorgerGG1997}. 
\item 
\emph{Existential positive $\FO$} ($\EpFO$). In this fragment, $\forall$ and $\neg$
are disallowed; however, the special sentence $\bot$  is allowed in order to express falsity (it could be read as empty disjunction). We obtain (Boolean) conjunctive queries ($\CQ$) from $\EpFO$ by disallowing $\vee$. The \emph{canonical query} for a finite structure $\mathfrak{A}$, denoted by $\cq{\mathfrak{A}}$,  is the \CQ{} sentence obtained by existentially quantifying over each variable $x_a$, for all $a \in A$, in 
\begin{equation}
\bigwedge_{\sigc\in\tau} \sigc \,{=}\, x_{\sigc^\mathfrak{A}} \wedge \hspace{-3ex} \bigwedge_{{\sigP \in \tau,\, k= \text{ar}(\sigP)}\atop{{(a_1,\ldots,a_k) \in \sigP^\mathfrak{A}}}} \hspace{-3ex} 
\sigP(x_{a_1},\ldots,x_{a_k}). 
\end{equation}
\item 
\emph{The guarded fragment of $\FOe\!$} ($\GFO$). 
All occurrences of $\forall$ have the form 
$\forall \bold{x}.(\sigP(\bold{y}) {\,\Rightarrow\,} \varphi[\bold{y}])$, 
and all of $\exists$ the form 
$\exists \bold{x}.(\sigP(\bold{y}) \wedge \varphi[\bold{y}])$
for $\bold{x} \subseteq \bold{y}$ and predicates $\sigP{}$.
This shape is called \emph{guarded quantification} and $\sigP(\bold{y})$ the respective \emph{guard}.\footnote{It is innocuous to
conjunctively add further atoms to the guard, as long as they use only variables from $\bold{x}$. We will consider such variants part of the  language; they can be \textsc{LogSpace}-transformed into the plain form. The same holds for unguarded quantification over formulae with just one free variable.}
\item
\emph{The triguarded fragment of} $\FO$ ($\TGF$). Every quantification over formulae with at least two free variables must be guarded.
\item
\emph{The guarded negation fragment of} $\FOe$ ($\GNFO$). The symbol $\forall$ is disallowed 
and every occurrence of $\neg$ must be in the form of
$\sigP(\bold{x}) \wedge \neg \varphi[\bold{x}]$.
\item
\emph{The $n$-variable fragment of} $\FOe$ / $\FO$. All variables must be from $\{x_1,\ldots,x_n\}$.  As special cases we obtain the 2-variable fragment with equality ($\FOte$) and without it ($\FOt$); here we normally use the variable names $x$, $y$ instead of $x_1$, $x_2$.
\item
\emph{(Disjunctive) tuple-generating dependencies} (TGDs). A \emph{disjunctive TGD} is a $\FO$ sentence of the form 
\begin{equation}
\forall \bold{x}.\varphi(\bold{x}) \impl \exists \bold{y}.\psi(\bold{x},\bold{y}),
\end{equation}
where
$\varphi$ is a (possibly empty) conjunction over atoms while $\psi$ is a non-empty disjunction over conjunctions of atoms.\footnote{Non-emptiness of disjunction is common in database theory and a deliberate choice: it ensures desirable properties like guaranteed satisfiability. Empty disjunction would make disjunctive TGDs as expressive as $\FO$.}
Note that for finite structures $\mathfrak{A}, \mathfrak{B}$, the formula $\cq{\mathfrak{A}} \Rightarrow \cq{\mathfrak{B}}$ can be easily transformed into an equivalent TGD (by pulling out the existential quantifiers of $\cq{\mathfrak{A}}$ thus turning them into universal ones, using that the variable sets of  $\cq{\mathfrak{A}}$ and $\cq{\mathfrak{B}}$ are disjoint), so we will consider it as such right away.
 
We obtain \DTGD{} (sentences) as finite conjunctions over disjunctive TGDs. Likewise, we obtain \TGD{} (sentences) as finite conjunctions over (non-disjunctive) TGDs, wherein the $\psi$ are plain conjunctions over atoms. As a middle ground between \TGD{} and \DTGD{}, we introduce \emph{mildly disjunctive TGD sentences} (\MDTGD) as sentences of the form $\Phi \vee \Psi$ where $\Phi \in \TGD$ and $\Psi \in \CQ$. Note that any \MDTGD{} sentence can be equivalently rewritten into \DTGD{} in polynomial time.
\end{itemize}

\emph{Existential \SO{}} ($\ESO$) is the fragment of \SO\ where all second-order quantifiers are existential and prenex.

We apply the standard model-theoretic semantics.
For a sentence $\Phi$, 
we let 
	$\getmodels{\Phi}$ denote the class of all models of $\Phi$ and
	$\getfinmodels{\Phi}$ the class of all finite models of $\Phi$. 	
A sentence in some logic is (finitely) satisfiable if it has a (finite) model. 
A logic has the \emph{finite model property (FMP)} if satisfiability and finite satisfiability coincide. \Cref{fig:main} gives results on the FMP for logics considered by us. 
A sentence $\Phi$ \emph{(finitely) characterizes} a class $\mathcal{C}$ of structures, if 
$\getffinmodels{\Phi} = \mathcal{C}$.
For $\tau' \subseteq \tau$, we say that a $\tau$-sentence $\Phi$ \emph{(finitely) projectively characterizes} a class $\mathcal{C}$ of $\tau'$-structures if 
$\getffinmodels{\Phi}|_{\tau'}  = \mathcal{C}$. 
A logic is \emph{non-degenerate} if there exists an unsatisfiable sentence $\Phi_\textbf{false}$.

%
%
%
%

\noindent\textbf{Tilings.}\quad
We introduce a variant of the $\mathbb{N} \times \mathbb{N}$ tiling problem, which serves as a versatile and useful vehicle for various subsequent undecidability arguments.

\begin{definition}[margin-constrained tiling problem]\label{def:tiling}
A \emph{(margin-constrained) domino system} $\mathcal{D} = (D,B,L,H,V)$ consists of a finite set $D$ of \emph{dominoes}, some of which are 
\emph{bottom-margin}  dominoes $B \subseteq D$
and some \emph{left-mar\-gin}  ones $L \subseteq D$. We refer to $L \cap B$ as \emph{seed dominoes}.
The two binary relations $H,V \subseteq D \times D$ specify \emph{horizontal} and \emph{vertical compatibility}.\footnote{In the context of domino systems, it will be convenient to let $R\in \{H,V\}$ also denote {the} map $R \colon D \rightarrow 2^D$ defined by $R(x):= \{d\in D\mid (x,d)\in R\}$. The intended meaning will always be clear from the context.}

We call  $\mathcal{D} = (D,B,L,H,V)$ \emph{deterministic} if\\[-3ex] 
\begin{enumerate}[itemindent=0ex, leftmargin=3ex,  itemsep=-0.8ex, topsep=0ex]
	\item	
	there is at most one seed domino (i.e., $|B \cap L|\leq 1$), 
	\item for every $d \in B$ there is at most one $e$ such that $(d,e) \in H$, and in this case $e \in B$, 
	\item	
	for every $d \in L$ there is at most one $e$ such that $(d,e) \in V$, and in this case $e \in L$, and
	\item
	for any $d_1,d_2 \in D$ there is at most one $e$ with $(d_1,e) \in H$ and $(d_2,e)\in V$. 
\end{enumerate}
A $\mathcal{D}$-\emph{tiling} (of $\mathbb{N} \times \mathbb{N}$) is a function $t\colon \mathbb{N} \times \mathbb{N} \to D$ such that for any $k, \ell \in \mathbb{N}$ hold
$t(0,\ell) \in L$ and $t(k,0) \in B$ as well as 
$\big(t(k,\ell),t(k+1,\ell)\big) \in H$ and $\big(t(k,\ell),t(k,\ell+1)\big) \in V$.  

A tiling $t$ is \emph{ultimately periodic}, if there are natural numbers $k_\mathrm{init},k_\mathrm{period},\ell_\mathrm{init},\ell_\mathrm{period}$ such that, for every $k, \ell \in \mathbb{N}$,
\begin{itemize}[itemindent=0ex, leftmargin=1.8ex, itemsep=-0.5ex, labelsep=0.7ex, topsep=0ex]
	\item	
	$t(k,\ell) = t(k+k_\mathrm{period},\ell)$ whenever $k\geq k_\mathrm{init}$ and  
	\item	
    $t(k,\ell) = t(k,\ell+\ell_\mathrm{period})$ whenever $\ell\geq \ell_\mathrm{init}$. 
\end{itemize}
\vspace{-1ex}
\end{definition}

It is straightforward to see that for deterministic domino systems, tilings are unique if they exist.
Also note that the ``classical'' $\mathbb{N} \times \mathbb{N}$ tiling problem without margin constraints is a special case of this variant (choose $D = L = B$).

\begin{lemma}[restate=gridlemma, label=gridlemma, name=]
The following problems are both undecidable even for deterministic margin-constrained domino systems $\mathcal{D}$:
\begin{enumerate}[itemindent=0ex, leftmargin=3ex, itemsep=-0.5ex]
\item
Given $\mathcal{D}$, does there exist a $\mathcal{D}$-\emph{tiling}?	
\item
Given $\mathcal{D}$, does there exist an ultimately periodic $\mathcal{D}$-\emph{tiling}?		
\end{enumerate}
\end{lemma}


\section{Some Tools}\label{sec:tools}

\newcommand{\indss}[2]{[#1]_{#2}}
In this section, we present basic tools that will be used in the subsequent sections; we start 
with the folklore technique of \emph{relativization} to make statements about a model's induced substructures.
In the following, let $\sigU{}$ denote a unary predicate. 
Given a $\tau {\uplus} \{\sigU{}\}$-structure $\mathfrak{A}$ 
we write 
$\indss{\mathfrak{A}}{\sigU{}}$ for the $\tau$-reduct of the substructure of $\mathfrak{A}$ induced by $\sigU{}^\mathfrak{A}$. 

Given an $\SO$ $\rho{\uplus}\tau$-sentence $\Phi$ with constants $\rho$ and predicates $\tau \mathrel{\rotatebox[y=0.57ex]{180}{$\notin$}}\sigU$, let $\Phi^{\mathrm{rel}(\sigU{})}$ denote 
the \emph{$\sigU{}$-relativization} of $\Phi$, defined as
\begin{equation}
\Phi^{\mathrm{guard}(\sigU{})} \wedge \bigwedge_{c \in \rho} \sigU{}(c),
\end{equation}
where $\Phi^{\mathrm{guard}(\sigU{})}$ is obtained from $\Phi$ by (inside-out) replacing all subformulae $\forall x.\varphi$ by $\forall x.(\sigU{}(x) \impl \varphi)$ and all subformulae $\exists x.\varphi$ by $\exists x.(\sigU{}(x) \wedge \varphi)$.
%
%
Relativization preserves many aspects of the sentence, such as the number of variables, the quantifier rank, and membership in the fragments $\mathbb{GFO}_=$, $\GNFO$, $\TGF$, $\FOte$, $\TGD$, $\DTGD$, $\MDTGD$, $\FOe$ (modulo some trivial syntactic equivalence-preserving transformations) as well as all prefix classes. Obviously, $\Phi^{\mathrm{rel}(\sigU{})}$ is \textsc{LogSpace}-computable from $\Phi$.
Moreover, the models of $\Phi^{\mathrm{rel}(\sigU{})}$ are precisely those $\rho{\uplus}\tau {\uplus} \{\sigU{}\}$-structures $\mathfrak{A}$ for which $\indss{\mathfrak{A}}{\sigU{}} \in \getmodels{\Phi}$.

Next, we introduce some machinery for labeling structures with numbers and using such endowed structures as compact representations of other structures.

\newcommand{\trans}{\mathrm{tr}}
\newcommand{\ass}{\mathit{as}}

Given a $\tau$-structure $\mathfrak{A}$ and a number $n\in \mathbb{N}$, an \emph{$n$-labeling of $\mathfrak{A}$} is a function $\lambda \colon  A \to \{1,\ldots,n\}$. We will then call $(\mathfrak{A},\lambda)$ an \emph{$n$-labeled structure}. If 
$(\mathfrak{A},\lambda)$ satisfies $\lambda(\sigc^\mathfrak{A})=1$ for every constant 
$\sigc \in \tau$, we call it \emph{constant-sole}. 
Given an $n$-labeled $\tau$-structure $(\mathfrak{A},\lambda)$, its \emph{implicit representation} $\mathfrak{A}_\lambda$ is the $\tau {\uplus} \{\Bit^\lambda_0,\ldots,\Bit^\lambda_{\ell}\}$-structure with $\ell = \lceil\log_2 n \rceil$ expanding $\mathfrak{A}$ by $(\Bit^\lambda_i)^{\mathfrak{A}_\lambda} = \{a \mid \lfloor\lambda(a) / 2^i\rfloor \mbox{ is odd}\}$. 

We next define $\FOe$ formulae allowing to compare some natural number against the label of a domain element associated with some term $t$ in $\mathfrak{A}_\lambda$. Given some $m \in \{1,\ldots,n\}$, assume $b_1\cdots b_\ell \in \{0,1\}^\ell$ is the fixed-length binary representation of $m-1$ (i.e., $m-1 = \sum_{i=1}^\ell b_i\cdot 2^{\ell-i}$). 
%
Then we use the following abbreviations for $\FOe$ sentences:
\begin{eqnarray}
{}[\lambda(t) = m] & \!\!\!\!:=\!\!\!\! &\!\!\bigg(\!\!\!\smallbigwedge_{{\ 1\leq j \leq \ell}\atop{b_j = 0}}\!\!\!\! \neg \Bit^\lambda_j(t) \!\bigg) \!\wedge\! \bigg(\!\!\!\smallbigwedge_{{\ 1\leq j \leq \ell}\atop{b_j = 1}}\!\!\! \Bit^\lambda_j(t) \!\bigg),\\
{}[\lambda(t) \geq m] & \!\!\!\!:=\!\!\!\! &\!\!\!\bigwedge_{{1\leq j \leq \ell}\atop{b_j = 1}}\! \bigg( \Bit^\lambda_j(t)   \vee \! \smallbigvee_{{1 \leq i \leq j-1}\atop{b_i=0}} \Bit^\lambda_i(t) \bigg) .  
\end{eqnarray}
These formulae's 
size is polynomial in $\ell$, and, for some $\mathfrak{A}$ and a variable assignment $v\colon \!V\!\! \to\! A$, we get $\mathfrak{A}_\lambda,\hspace{-1pt}v \models [\lambda(t) \hspace{1pt}{=}\hspace{1pt} m]$ whenever $\lambda(t^{\mathfrak{A},v}) = m$; analogously for  $[\lambda(t) \geq m]$. 

For a constant-sole $(\mathfrak{A},\lambda)$, let its \emph{unfolding} $\mathfrak{A}^\lambda$ be the $\tau$-structure with $A^\lambda = \{(a,i) \mid a\in A,\ 1\leq i\leq \lambda(a)\}$ 
where, for any $\sigP \in \tau$ of arity $k$, holds $((a_1,m_1), \ldots ,(a_k,m_k)) \in \sigP{}^{\mathfrak{A}^\lambda}$ iff
$(a_1, \ldots ,a_k) \in \sigP{}^{\mathfrak{A}^\lambda}$. We define $\pi \colon A^\lambda \to A$ via $(a,m) \mapsto a$. 
Obviously, $\pi$ is a strong surjective homomorphism from $\mathfrak{A}^\lambda$ to $\mathfrak{A}$. Conversely, for any strong surjective homomorphism $h \colon \mathfrak{B} \to \mathfrak{A}$ where every $a\in A$ satisfies $|h^{-1}(a)|\leq n$ (i.e., the size of each element's pre-image is bounded by $n$) there exists some $\lambda \colon A \to \{1,\ldots,n\}$ with $\mathfrak{B} \cong \mathfrak{A}^\lambda$ (witnessed by an isomorphism $g \colon \mathfrak{B} \to \mathfrak{A}^\lambda$ such that $h = \pi \circ g$).

For variable labelings $\ass \colon V\to \{1,\ldots,n\}$, let $\ass^+(t)$ denote $\ass(t)$ if $t$ is a variable and $1$ if $t$ is a constant. We recursively define the functions $\trans^n_\ass$ mapping $\FOe$ $\tau$-formulae in negation normal form to $\FOe$ $\tau{\uplus} \{\Bit^\lambda_0,\ldots,\Bit^\lambda_{\lceil\log_2 n \rceil}\}$-formulae as follows:
\begin{eqnarray}
\trans^n_\ass (\sig\sigP(\bold{t})) &\!\!\!\!=\!\!\!\! & \sig\sigP(\bold{t}) \\
\trans^n_\ass (\neg \sig\sigP(\bold{t})) & \!\!\!\!=\!\!\!\! & \neg \sig\sigP(\bold{t}) \\
\label{eq: 5}\trans^n_\ass (t_1=t_2) & \!\!\!\!=\!\!\!\! & \left\{\begin{array}{ll} 
	t_1=t_2 & \mbox{ if } \ass^+(t_1)=\ass^+(t_2) \\
	\bot & \mbox{ otherwise } \\
\end{array}\right.\\
\label{eq: 6}\trans^n_\ass (t_1\not=t_2) & \!\!\!\!=\!\!\!\! & \left\{\begin{array}{ll} 
	t_1\not=t_2 & \mbox{ if } \ass^+(t_1)=\ass^+(t_2)  \\
	\top & \mbox{ otherwise } \\
\end{array}\right.\\
\trans^n_\ass (\varphi \wedge \psi) &\!\!\!\!=\!\!\!\! & \trans^n_\ass (\varphi) \wedge \trans^n_\ass(\psi) \\ 
\trans^n_\ass (\varphi \vee \psi) &\!\!\!\!=\!\!\!\! & \trans^n_\ass (\varphi) \vee \trans^n_\ass(\psi) \\ 
\label{eq: 11}\trans^n_\ass (\exists x.\varphi) &\!\!\!\!=\!\!\!\! & \!\!\!\bigvee_{1 \leq i \leq n}\!\!\! \exists x. [\lambda(x) \geq i] \wedge \trans^n_{\ass \cup \{x \mapsto i\}} (\varphi) \\ 
\label{eq: 12}\trans^n_\ass (\forall x.\varphi) &\!\!\!\!=\!\!\!\! & \!\!\!\bigwedge_{1 \leq i \leq n}\!\!\! \forall x. [\lambda(x) \geq i] \impl \trans^n_{\ass \cup \{x \mapsto i\}} (\varphi)\ \ \ \ \ 
\end{eqnarray}
\vspace{-8pt}

Defining $\trans^n$ by letting 
$\trans^n(\Phi) = \trans^n_\emptyset (\mathrm{NNF}(\Phi))$, we obtain a lemma stating that $\trans^n$ truthfully rewrites sentences for their evaluation over implicit representations of unfoldings.  

\begin{lemma}[restate=translemma, label=translemma, name=]
Let $\Phi$ be any $\FOe$ sentence. Then $\mathfrak{A}^\lambda \models \Phi$ exactly if
$\mathfrak{A}_\lambda \models \trans^n(\Phi)$.  
\end{lemma}
\vspace{-6pt}

Note that 
the size of $\trans^n(\Phi)$ 
may be exponential in the size of $\Phi$. 
On the positive side, $\trans^n$ preserves (negation-)guardedness, number of variables, quantifier prefixes, and quantifier rank.


\section{Checking Homclosure Membership}\label{sec:InHomCl}

The first problem we address is determining whether a given finite structure is contained in the homclosure of the set of models (or finite models) of a given sentence $\Phi$.  

\bigskip

\probbox{
\textbf{Problem:} $\mathrm{InHomCl}_\mathrm{(fin)}$\\
\textbf{Input:} $\tau$, $\tau$-sentence $\Phi$, finite $\tau$-structure $\mathfrak{A}$.\\
\textbf{Output:} \textsc{yes}, if $\mathfrak{A} \in \homcl{\getffinmodels{\Phi}}$, \textsc{no} otherwise.  
}

\bigskip

The fact that $\Phi$ is (finitely) satisfiable iff $\mathfrak{F}_\tau \in \homcl{\getffinmodels{\Phi}}$, already yields some lower bounds for this problem's difficulty.

\begin{proposition}\label{prop:SatToInHomCl}
(Finite) satisfiability can be \textsc{LogSpace}-reduced to $\mathrm{InHomCl}_\mathrm{(fin)}$.	
\end{proposition}

This shows that $\mathrm{InHomCl}_\mathrm{(fin)}$ is undecidable for any logic with an undecidable (finite) satisfiability problem, such as $\FO$ and every logic containing it, but also all $\FO$ reduction classes. 

Toward some decidable cases and corresponding optimal upper bounds, we employ and formalize the straightforward idea of ``coloring'' a model with elements of the target structure. 

\begin{definition}[extrinsic/intrinsic $\mathfrak{A}$-coloring]\label{def:exincol}
Let $\mathfrak{A}$ be a finite $\tau$-structure over the domain $A=\{1,\ldots,n\}$. 
As before, we use a signature $\rho =\{ \Bit^\lambda_1,\ldots,\Bit^\lambda_\ell \}$ of unary predicates with $\ell = \lceil \log_2 n \rceil$ to represent $n$-labeled structures. 

Given a predicate $\sigP \in \tau$ of arity $k$ and a $k$-tuple of terms 
$\bold{t}=(t_1,\ldots,t_k)$, we let $\chi^{\sigP}_\mathfrak{A}(\bold{t})$ denote the formula
\begin{equation}
	\sig\sigP(\boldsymbol{t}) \wedge \bigvee_{{(m_1,\ldots,m_k)}\in \sigP{}^\mathfrak{A}\ } \bigwedge_{i=1}^k [\lambda(t_i)=m_i].
\end{equation}
We denote by $\Omega_\mathfrak{A}$ the following $\tau{\uplus}\rho$-sentence: 
\begin{equation}
\bigwedge_{{\sigP\in \tau}\atop{k=\text{ar}(\sigP)}} \forall x_1\ldots x_k. \sig\sigP(x_1,\ldots,x_k)
\impl \chi^{\sigP}_\mathfrak{A}(x_1,\ldots, x_k). 
\end{equation}
Given a $\tau$-sentence $\Phi$ and a finite $\tau$-structure $\mathfrak{A}$, we define the \emph{extrinsic \mbox{$\mathfrak{A}$-coloring}} of $\Phi$ by
$\Phi^\mathrm{ext}_\mathfrak{A} = \Phi \wedge \Omega_\mathfrak{A}$. The \emph{intrinsic \mbox{$\mathfrak{A}$-coloring}} of $\Phi$, denoted $\Phi^\mathrm{int}_\mathfrak{A}$ is obtained from $\Phi$ by replacing, for all predicates $\sigP \in \tau$, any occurrence of $\sig\sigP(\bold{t})$ by $\chi^{\sigP}_\mathfrak{A}(\bold{t})$. 
\end{definition}

Note that both $\Phi^\mathrm{ext}_\mathfrak{A}$ and $\Phi^\mathrm{int}_\mathfrak{A}$ are of polynomial size wrt.\ $|\mathfrak{A}|$ and the size of $\Phi$. 
The sentence $\Omega_\mathfrak{A}$ projectively characterizes the class of $\tau$-structures having a homomorphism into $\mathfrak{A}$, thus $\Phi^\mathrm{ext}_\mathfrak{A}$ projectively characterizes the class of (finite) models of $\Phi$  having a homomorphism into $\mathfrak{A}$.
Also, every model of $\Phi$ having a homomorphism into $\mathfrak{A}$ can be extended to a model of $\Phi^\mathrm{int}_\mathfrak{A}$, and conversely, for every model $\mathfrak{B}$ of $\Phi^\mathrm{int}_\mathfrak{A}$, the structure $\mathfrak{C}$ obtained by setting $C:=B$, $\sigc{}^\mathfrak{C}:=\sigc{}^\mathfrak{B}$ for all constants, and $\sigP{}^\mathfrak{C} := (\chi^{\sigP}_\mathfrak{A})^\mathfrak{B}$ is a model of $\Phi$ with a homomorphism into $\mathfrak{A}$. Hence we obtain the following lemma.

\begin{lemma}[restate=lemmacoloring, label=, name=]\label{lemma:coloring}
	Let $\Phi$ be an arbitrary $\tau$-sentence and let $\mathfrak{A}$ be a finite $\tau$-struc\-ture.
The following are equivalent:
\begin{enumerate}[itemindent=0ex, leftmargin=1.8ex, itemsep=-0.5ex, labelsep=0.7ex, topsep=0ex]
	\item
	$\Phi$ has a (finite) model admitting a homomorphism into $\mathfrak{A}$.
	\item
	$\Phi^\mathrm{int}_\mathfrak{A}$ is (finitely) satisfiable.
	\item
	$\Phi^\mathrm{ext}_\mathfrak{A}$ is (finitely) satisfiable.
\end{enumerate}
\end{lemma}

As a direct consequence, arbitrary and finite homclosure membership coincide for a wide range of fragments.

\begin{theorem}[restate=thearbsameasfin, label=thm:first-main, name=]\label{the:arb-sameas-fin}
Let $\Logic \,{\subseteq}\, \SO$ and $\Logic^\mathrm{ext}:=\{\Phi^\mathrm{ext}_\mathfrak{A}\mid \Phi \in \Logic, \mathfrak{A} \mbox{ finite}\}$
as well as $\Logic^\mathrm{int}:=\{\Phi^\mathrm{int}_\mathfrak{A}\mid \Phi \in \Logic, \mathfrak{A} \mbox{ finite}\}$.
If $\Logic^\mathrm{ext}$ or $\Logic^\mathrm{int}$ has the FMP, then $\homclfin{\getmodels{\Psi}}=\homclfin{\getfinmodels{\Psi}}$ for every $\Psi \in \Logic$.
\end{theorem}
This theorem applies to various fragments, including any (sublogic) of \DTGD{}, \GNFO{}, \TGF{}, \FOte{}, \EAAE{}, or \EAFO{}. We obtain the following complexity results.

\begin{theorem}[restate=firstmain, label=thm:first-main, name=]
$\mathrm{InHomCl}_\mathrm{(fin)}$ is complete for 
\begin{enumerate}[itemindent=0ex, leftmargin=3ex, itemsep=-0.5ex]
\item 
\textsc{2ExpTime} for \GNFO{},
\item 
\textsc{N2ExpTime} for \TGF{},
\item
\textsc{NExpTime} for
\FOte{}, 
\EAAE{},
\EAFO{}, 
\AsFO{},
\item
\textsc{NP} for \EsFO{} and 
\EpFO{}. 
\end{enumerate}

\end{theorem}

Next to the combined complexity addressed in the previous theorem, one might also ask for the complexity if the sentence is fixed and only the structure of interest varies. We will come back to this question in the course of \Cref{sec:characterizability}. 

We conclude with a case where the difficulty of $\mathrm{InHomCl}_\mathrm{(fin)}$ does not coincide with that of satisfiability (in the most egregious way). We recap that \TGD{} (just like \MDTGD{} and \DTGD{}) has a trivial satisfiability problem (in that every \TGD{} sentence is satisfiable and finitely satisfiable, with $\mathfrak{F}$ uniformly serving as a model), yet we present a fixed \TGD{} sentence $\Phi_\mathrm{grid}$ for which $\mathrm{InHomCl}_\mathrm{(fin)}$ is undecidable. 

\begin{definition}[restate=defgridformula, name=]
We let $\mathfrak{A}_{\mathbb{N} \times \mathbb{N}}$ denote the $\mathbb{N} \times \mathbb{N}$ grid, i.e., the $\{\sigH{},\sigV{}\}$-structure with domain $A_{\mathbb{N} \times \mathbb{N}} = \mathbb{N} \times \mathbb{N}$ that satisfies
$\sigH{}^{\mathfrak{A}_{\mathbb{N} \times \mathbb{N}}} = \{((k,\ell),(k+1,\ell)) \mid k,\ell \in \mathbb{N}\}$ as well as  
$\sigV{}^{\mathfrak{A}_{\mathbb{N} \times \mathbb{N}}} = \{((k,\ell),(k,\ell+1)) \mid k,\ell \in \mathbb{N}\}$.

Let $\Phi_\mathrm{grid}$ denote the \TGD{} sentence 
\begin{equation}
\begin{array}{l}
\Big( \forall x \exists y. \sigH{}(x,y) \Big) \wedge 
\Big( \forall x \exists y. \sigV{}(x,y) \Big) \wedge \\
\hspace{4ex}\Big( \forall xyzv. \sigH{}(x,y) \wedge \sigV{}(x,z) \wedge \sigV{}(y,v) \impl \sigH{}(z,v) \Big).
\end{array}
\end{equation}
%
%
%
For a given domino system $\mathcal{D}=(D,D,D,H,V)$, we define the finite $\{\sigH{},\sigV{}\}$-structure $\mathfrak{A}_\mathcal{D}$ as follows: $A_\mathcal{D} = D$ as well as $\sigH{}^{\mathfrak{A}_\mathcal{D}}=H$ and $\sigV{}^{\mathfrak{A}_\mathcal{D}}=V$.  
\end{definition}

Noting that every model $\mathfrak{B}$ of $\Phi_\mathrm{grid}$ allows for a homomorphism $h \colon \mathfrak{A}_{\mathbb{N} \times \mathbb{N}} \to \mathfrak{B}$,
it follows that every homomorphism from some (finite) model of $\Phi_\mathrm{grid}$ into $\mathfrak{A}_\mathcal{D}$ gives rise to a(n ultimately periodic) $\mathcal{D}$-tiling.
Thus, by \Cref{gridlemma}, both membership in $\homcl{\getmodels{\Phi_\mathrm{grid}}}$ and $\homcl{\getfinmodels{\Phi_\mathrm{grid}}}$ are undecidable problems.

\begin{theorem}[restate=undecTGD, label=corr:undecTGD, name=]
$\mathrm{InHomCl}_\mathrm{(fin)}$ is undecidable for \TGD{} (and hence for \MDTGD{} and \DTGD{}).
\end{theorem}


\section{Checking Homclosedness}\label{sec:homclosedness}

The next problem we investigate is the 
problem of 
 determining whether a given sentence is homclosed (that is, the class of its models is closed under homomorphisms).

\bigskip

\probbox{
	\textbf{Problem:} $\mathrm{HomClosed}$\\
	\textbf{Input:} $\tau$, $\tau$-sentence $\Phi$.\\
	\textbf{Output:} \textsc{yes}, if $\getmodels{\Phi} = \homcl{\getmodels{\Phi}}$, \textsc{no} otherwise.
}

\bigskip

For the finite-model version of this problem, note that any structure has a homomorphism into some infinite structure. Thus, any non-empty class $\mathcal{C}$ of finite structures cannot be homomorphism-closed (whence the question if $\getfinmodels{\Phi} = \homcl{\getfinmodels{\Phi}}$ trivializes to the question if $\getfinmodels{\Phi} = \emptyset$). This motivates the subsequent definition, focusing on finite target structures. 

\bigskip

\probbox{
	\textbf{Problem:} $\mathrm{HomClosed}_\mathrm{fin}$\\
	\textbf{Input:} $\tau$, $\tau$-sentence $\Phi$.\\
	\textbf{Output:} \textsc{yes}, if $\getfinmodels{\Phi} = \homclfin{\getfinmodels{\Phi}}$, \textsc{no} otherwise.
}

\medskip

To begin with, obviously, both problems are trivial (i.e., the answer is always \textsc{yes}) for fragments where every sentence's model class is homomorphism-closed, such as $\EpFO$.

Next, we will see that under very mild assumptions, (finite) satisfiability within the considered logic can be reduced to $\mathrm{HomClosed}_\mathrm{(fin)}$, yielding lower bounds for the latter in many cases. To this end, we state the following technical result.

\newcommand{\minus}{\vspace{-4pt}}

\begin{lemma}[restate=SattToHomClosed, label=prop:SattToHomClosed, name=]
Let $\Logic,\Logic' \subseteq \SO$ and $f \colon \Logic'\to \Logic$ such that
\minus
\begin{itemize}[itemindent=0ex, leftmargin=1.8ex, itemsep=-0.6ex, labelsep=0.7ex, topsep=0ex]
\item (finite) unsatisfiability in $\Logic'$ is $\textsc{C}$-hard and
\item $f$ is $\textsc{C}'$-computable  and maps $\tau$-sen\-ten\-ces to $\tau\cup\{\sigU\}\cup \tau'$-sen\-ten\-ces with  {$\getffinmodels{f(\Phi)}|_{\tau \cup \{\sigU{}\}}  = \{(\mathfrak{A}, \emptyset) \mid \mathfrak{A}\in \llbracket \Phi\rrbracket_{(<\omega)}    \}$.}
\end{itemize}
Then $\mathrm{HomClosed}_\mathrm{(fin)}$ for $\Logic$ is $\textsc{C}$-hard under $\textsc{C}'$-reductions.
\end{lemma}
\minus

In particular, \Cref{prop:SattToHomClosed} shows that for any logic allowing for extending any sentence $\Phi$ into $\Phi \wedge \forall x.\neg \sigU(x)$, (finite) hom\-closed\-ness
must be as hard as (finite) unsatisfiability. As a direct consequence of this lemma, we also obtain un\-de\-cida\-bi\-li\-ty for all logics containing $\FO$, but also for all undecidable $\FOe$ prefix classes featuring a universal quantifier (such as $\AAAE$ or $\AAEe$), even if the signature does not contain constants. As in the previous section, this justifies focusing our attention to logics with a decidable satisfiability problem. However, it turns out that even for many of those,  $\mathrm{HomClosed}_\mathrm{(fin)}$ is undecidable, as we will show now. To this end, we establish a companion result to \Cref{prop:SattToHomClosed}.

\begin{lemma}[restate=SattToHomClosedCompanion, label=prop:SattToHomClosedCompanion, name=]
	Let $\Logic,\Logic' \subseteq \SO$ and $f \colon \Logic'\to \Logic$ such that
\minus
	\begin{itemize}[itemindent=0ex, leftmargin=1.8ex, itemsep=-0.6ex, labelsep=0.7ex, topsep=0ex]
		\item (finite) unsatisfiability in $\Logic'$ is $\textsc{C}$-hard and
		\item $f$ is $\textsc{C}'$-computable and maps $\tau$-sen\-ten\-ces to $\tau\cup\{\sigU\}$-sen\-ten\-ces with  $\getffinmodels{f(\Phi)} = \{(\mathfrak{A},\tilde{A}) \mid \mathfrak{A} \,{\in}\, \getffinmodels{\neg\Phi} \mbox{ or } \tilde{A} \,{\not=}\, A\}$.
	\end{itemize}
	Then $\mathrm{HomClosed}_\mathrm{(fin)}$ for $\Logic$ is $\textsc{C}$-hard under $\textsc{C}'$-reductions.
\end{lemma}
\minus


In particular, this result shows that for any logics $\Logic$ and $\Logic'$ where one can easily find for any $\Logic'$ sentence $\Phi$ an $\Logic$ sentence equivalent to $\neg\Phi \vee \exists x.\neg \sigU{}(x)$, 
$\mathrm{HomClosed}_\mathrm{(fin)}$ for $\Logic$ is at least as hard as unsatisfiability for $\Logic'$. Picking $\AAAE$ and $\AAEe$ for $\Logic'$, this immediately entails undecidability of $\mathrm{HomClosed}_\mathrm{(fin)}$ for the prefix classes $\EEEA$ and $\EEAe$, despite decidability of satisfiability. Again, this even holds when constants are absent from the signature.
Of course, this also entails undecidability of $\mathrm{HomClosed}_\mathrm{(fin)}$ for the decidable prefix classes \EAAE{} (the Gödel class)
and \EAFO{} (the Bernays-Schönfinkel class). 


After these negative results, we will now introduce a generic method for establishing decidability and even tight complexity bounds of $\mathrm{HomClosed}_\mathrm{(fin)}$ for a wide variety of logics.

\begin{definition}[spoiler, FSP]
Let $\Phi$ be a $\tau$-sentence. A \emph{(homclosedness) spoiler} for $\Phi$ is a triple $(\mathfrak{A},\mathfrak{B},s)$, written as $\mathfrak{A}\stackrel{s}{\to}\mathfrak{B}$, where $\mathfrak{A} \in \getmodels{\Phi}$,  $\mathfrak{B} \not\in \getmodels{\Phi}$, and $s \colon \mathfrak{A} \to \mathfrak{B}$.
We call a spoiler \emph{injective}, \emph{surjective}, \emph{strong}, etc.~if $s$ is (and use the appropriate arrow symbol).
We call it \emph{finite} if both $\mathfrak{A}$ and $\mathfrak{B}$ are. A logic $\Logic$ has the \emph{finite spoiler property} (FSP) if each $\Logic$ sentence that has a spoiler also has a finite one.
\end{definition}
\minus

Obviously, $\homcl{\getmodels{\Phi}} \not= \getmodels{\Phi}$ iff $\Phi$ has a spoiler, 
and $\homclfin{\getfinmodels{\Phi}} \not= \getfinmodels{\Phi}$ iff $\Phi$ has a finite spoiler. 
Hence, homclosedness is verified or refuted by the absence or presence, respectively, of such a spoiler, the search for which can be confined using a straightforward decomposition result.

\begin{proposition}[restate=HomSeparation, label=prop:HomSeparation, name=]
For any two $\tau$-structures $\mathfrak{A}$ and $\mathfrak{B}$, any homo\-morphism $h \colon \mathfrak{A} \to \mathfrak{B}$ can be represented as the com\-po\-sition $h_2 \circ h_1$ of two homomorphisms where
\begin{itemize}[itemindent=0ex, leftmargin=1.8ex, itemsep=-0.8ex, labelsep=0.7ex, topsep=-0.6ex]
\item 
$h_1 \colon \mathfrak{A} \injhom \mathfrak{C}$ is injective and
\item 
$h_2 \colon \mathfrak{C} \str{\surhom} \mathfrak{B}$ is surjective and strong.
\end{itemize}
\vspace{-0.2ex}
The same holds when consideration is restricted exclusively to (homomorphisms between) finite structures.
\end{proposition}
\minus

\begin{corollary}\label{lem:inj-or-ss}
Let $\Phi$ be an arbitrary $\tau$-sentence. Then $\getffinmodels{\Phi} \not= \homclffin{\getffinmodels{\Phi}}$ exactly if $\Phi$ has a (finite) injective spoiler or a (finite) strong surjective spoiler.
\end{corollary}
\minus

We will first tend to the case of injective spoilers.

\begin{definition}
Let $\Phi$ be a $\tau$-sentence. Let $\tau'$ consist of copies $\sigP'$ for all predicates $\sigP \in \tau$.
Let $\Phi^-$ be the $\tau\cup\tau'$-sentence obtained from $\Phi$ by replacing every occurrence of some atom $\sigP(\bold{t})$ by $\sigP(\bold{t}) \wedge \sigP{}'(\bold{t})$. 
Then, let $\Phi^\spoil_\cl{\injhom}$ denote the $\tau {\cup} \tau' {\cup} \{\sigU{}\}$-sentence
$\neg\Phi \wedge (\Phi^-)^{\mathrm{rel}(\sigU{})}.$
\minus

A \emph{witness} of an injective spoiler $\mathfrak{A}\stackrel{s}{\injhom}\mathfrak{B}$ is a $\tau {\uplus} \tau' {\uplus} \{\sigU{}\}$-struc\-ture $\mathfrak{D}$ satisfying 
$\mathfrak{D}|_\tau = \mathfrak{B}$, 
$\sigU{}^\mathfrak{D} = s(A)$, and
$\sigP{}^\mathfrak{D} {\cap}\,{\sigP{}'}^\mathfrak{D}\! = \{(s(a_1),\ldots,s(a_k)) \mid   (a_1,\ldots,a_k) \in \sigP{}^\mathfrak{A}\}$ for every predicate $\sigP{}$ from $\tau$.
\end{definition} 
\minus

Note that any (finite) injective spoiler of some $\Phi$ has a (finite) witness but also any such witness is a (finite) model of $\Phi^\spoil_\cl{\injhom}$. Conversely, any (finite) model of $\Phi^\spoil_\cl{\injhom}$ is clearly a witness of some (finite) injective spoiler of $\Phi$. So we obtain the following.

\begin{lemma}[restate=InjFormula, label=InjFormula, name=]
$\Phi^\spoil_\cl{\injhom}$ is (finitely) satisfiable exactly if $\Phi$ has a (finite) injective spoiler. 
\end{lemma}
\minus

%
%
%

We next tend to the case of strong surjective spoilers, focusing on the case where  $\Phi$ is from $\FOe$. Before, note that for any equality-free logic (in particular any fragment of $\FO$), sentences are naturally preserved under strong surjective homomorphisms. Thus, for any such logic,
$\homclffin{\getffinmodels{\Phi}} = \getffinmodels{\Phi}$ holds precisely if $\Phi^\spoil_\cl{\injhom}$ is (finitely) unsatisfiable.

\renewcommand{\hslash}{{}h'}

\begin{lemma}[restate=squeeze, name=]
Let $\mathfrak{B}$ and $\mathfrak{A}$ be $\tau$-structures, let $h\colon\mathfrak{B} \sshom \mathfrak{A}$ be a strong surjective homomorphism. Let $n$ be a natural number. Then there exists a function $\lambda\colon A \to \{1,\ldots,n\}$ such that
\minus
\begin{itemize}[itemindent=0ex, leftmargin=1.8ex, itemsep=-0.9ex, labelsep=0.7ex, topsep=0ex]
\item
$\!$there is a strong surjective homomorphism $\hslash\colon\mathfrak{B} \sshom \mathfrak{A}^\lambda$ satisfying $h = \pi \circ \hslash$, and
\item
$\!\mathfrak{B}\!$ and $\!\mathfrak{A}^\lambda\!$ satisfy$\!$ the$\!$ same $\!\FO_\text{=}\!$ sentences$\!$ of$\!$ quantifier$\!$ rank $\!n$.
\end{itemize}
\end{lemma}

\begin{corollary}
If an $\FOe$ sentence $\Phi$ of quantifier rank $n$ has a strong sur\-jec\-tive spoiler, then it has a spoiler of the form $\mathfrak{A}^\lambda \stackrel{\pi}{\sshom} \mathfrak{A}$ for some $\lambda \colon A \to \{1,\ldots,n\}$ (which is automatically strong and surjective). 
\end{corollary}

Thus, in our quest to detect the existence of strong sur\-jec\-tive spoilers, we can focus on spoilers of this specific type.

\newcommand{\squeezeabit}{\vspace{-4.5pt}}

\begin{definition}
For a $\FOe$ sentence $\Phi$, let $\Phi^{\spoil}_\cl{\sshom}$ denote $\bot$ if $\Phi$ is equality-free, otherwise it denotes the sentence 
$\neg\Phi \wedge \trans^n_\emptyset(\Phi)$ where $n$ is the quantifier rank of $\Phi$.
\squeezeabit

Given a strong surjective spoiler of $\Phi$ of the form $\mathfrak{A}^\lambda \stackrel{\pi}{\sshom} \mathfrak{A}$, its \emph{witness} is the  $\tau{\uplus}\{\Bit^\lambda_0,\ldots,\Bit^\lambda_{\lceil\log_2 n \rceil}\}$-structure $\mathfrak{A}_\lambda$.
\end{definition}
\squeezeabit

We note that $\Phi^\spoil_\cl{\sshom}$ may be of exponential size in $\Phi$.
With help of \Cref{translemma} we obtain that, on the one hand, the witness of a (finite) $\Phi$-spoiler of the form $\mathfrak{A}^\lambda \stackrel{\smash{\pi}}{\sshom} \mathfrak{A}$ is a (finite) model of $\Phi^\spoil_\cl{\sshom}$, and, on the other hand, that any (finite) model of $\Phi^\spoil_\cl{\sshom}$ is isomorphic to a witness of some $\Phi$-spoiler of the form $\mathfrak{A}^\lambda \stackrel{\smash{\pi}}{\sshom} \mathfrak{A}$.
This gives rise to the following characterization.

\begin{lemma}\label{lem:bigspoiler}
$\Phi^\spoil_\cl{\sshom}$ is (finitely) satisfiable exactly if $\Phi$ has a (finite) strong surjective spoiler. 
\end{lemma}
\squeezeabit

As a next step, we combine the obtained characterizations.


\begin{definition}
Given a $\FOe$ sentence $\Phi$, let 
$\Phi^\spoil = \Phi^\spoil_\cl{\injhom} \vee \Phi^\spoil_\cl{\sshom}$.
\end{definition}
\squeezeabit

We find that $\Phi^\spoil$ is in \GNFO{}, \TGF{}, and $\FOte$ 
whenever $\Phi$ is; it can be easily rewritten into \EAAE{} whenever $\Phi$ is in \EEAA{} or \AAEE;
and into \EAFO{} whenever $\Phi$ is in \EsFO{} or \AsFO. Thanks to \Cref{InjFormula}, \Cref{lem:bigspoiler}, and \Cref{prop:HomSeparation}, 
$\Phi^\spoil$ is (finitely) satisfiable whenever a (finite) spoiler of $\Phi$ exists. So we obtain the following.

\begin{theorem}\label{lem:arbspoiler}
$\Phi$ is (finitely) homclosed if and only if $\Phi^\spoil$ is (finitely) unsatisfiable.
\end{theorem}
\squeezeabit

Unfortunately, the size of $\Phi^\spoil$ is exponential in $\Phi$ and hence too big to directly yield tight complexity results. However, whenever we know that $\Phi^\spoil$ belongs to a fragment with the finite model property, \Cref{lem:arbspoiler} allows us to establish that the existence of arbitrary spoilers and finite spoilers coincide.   

\begin{theorem}\label{thm:FSP}
$\GNFO$, $\TGF$, $\FOte$, $\AsFO$, $\EsFO$, \EEAA{}, and \AAEE{} have the finite spoiler property.
Thus, $\mathrm{HomClosed}$ and $\mathrm{HomClosed}_\mathrm{fin}$ coincide for these fragments. 
\end{theorem}
\squeezeabit
%
%
This justifies, for all the mentioned fragments, to focus on the case of finite spoilers. To this end, we introduce a very specific kind of strong surjective homomorphism, which only merges two ``indistinguishable'' domain elements. 

\begin{definition}[monomerge]
An \emph{monomerge} is a surjective strong homomorphism $h\colon\mathfrak{A} \sshom \mathfrak{B}$ satisfying $B = A \setminus \{a\}$ for some $a \in A$ and $h(a')=a'$ for every $a' \not= a$. We call $a$ the \emph{merge source} and $h(a)$ the \emph{merge target} of $h$. We will use $\mmhom$ to specifically refer to homomorphisms that are monomerges.	
\end{definition}
\squeezeabit

It is not hard to see that in the finite case, every strong surjective homomorphism can be realized, up to isomorphism, by iteratively merging pairs of domain elements in this way. 

\begin{lemma}
	Let $\mathfrak{A}_0$, $\mathfrak{B}$ be finite $\tau$-structures and $h\colon\mathfrak{A}_0 \sshom \mathfrak{B}$ a strong surjective homomorphism. 
	Then there is a finite se\-quence $h_1\colon\mathfrak{A}_0\! \mmhom \!\mathfrak{A}_1$, $\ldots$\ , $h_m\colon\mathfrak{A}_{m-1}\! \mmhom \!\mathfrak{A}_m$ of mono\-mer\-ges  and an isomorphism $g\colon\mathfrak{A}_m\! \to \!\mathfrak{B}$ such that
	$h = g \circ h_m \circ \ldots \circ h_1$.
\end{lemma}

Thus, in the finite, our search for strong surjective spoilers can be confined to the search for monomerge spoilers. 

\begin{corollary}[label=FinHomSeparation, name=]
	If a sentence $\Phi$ has a finite strong surjective spoiler, then it has a finite monomerge spoiler.
	Consequently, whenever $\Phi$ has an arbitrary finite spoiler, it must have an injective spoiler or a monomerge spoiler.   
\end{corollary}


\begin{definition}
For a $\tau$-sentence $\Phi$, we let $\Phi^\spoil_\cl{\mmhom}$ denote $\bot$ if  $\Phi$ is equality-free, otherwise it denotes the $\tau {\uplus} \{\sigc{}_s,\sigc{}_t,\dot{\sigU{}}\}$-sentence 
\begin{equation}
\Phi \wedge (\neg\Phi)^{\mathrm{rel}(\dot{\sigU{}})} \wedge \Omega_{\sigc{}_s\simeq \sigc{}_t} \wedge \forall x. \big( x{=}\sigc{}_s \Leftrightarrow \neg\dot{\sigU{}}(x) \big),
\end{equation}
where $\Omega_{\sigc{}_s\simeq \sigc{}_t}$ stands for
\begin{equation}
\bigwedge_{{\sigP \in \tau}\atop{k=\text{ar}(\sigP)}} \hspace{-6.5ex} 
\bigwedge_{{1\leq i \leq k}\atop{{\hspace{6.5ex} \bold{x}=(x_1,\ldots,x_{i-1})}\atop{\hspace{6.5ex}\bold{y}=(x_{i+1},\ldots,x_k)}}}
\hspace{-6.0ex}
\forall \bold{x}\bold{y}. \sigP(\bold{x},\sigc{}_s,\bold{y}) \Leftrightarrow \sigP(\bold{x},\sigc{}_t,\bold{y}).
\end{equation}
~\\[-1ex]
Let $\mathfrak{A} \stackrel{\raisebox{-0.5ex}{\scriptsize \smash{h}}}{\mmhom} \mathfrak{B}$ be a monomerge spoiler of $\Phi$
with merge source $a_s$ and merge target $a_t$.
Then, we call the $\tau {\uplus} \{\sigc{}_s,\sigc{}_t,\dot{\sigU{}}\}$-structure $(\mathfrak{A},a_s,a_t,B)$ 
the \emph{witness} of $\mathfrak{A} \stackrel{\raisebox{-0.5ex}{\scriptsize \smash{h}}}{\mmhom} 
\mathfrak{B}$.
\end{definition}

We note that the witness of a finite monomerge spoiler of $\Phi$ is a finite model of $\Phi^\spoil_\cl{\mmhom}$. Conversely, 
any finite model of $\Phi^\spoil_\cl{\mmhom}$ is isomorphic to the witness of a finite monomerge spoiler of $\Phi$. Thus we obtain the following characterization.
 
\begin{lemma}[label=MonomergeFormula, name=]
$\Phi^\spoil_\cl{\mmhom}$ is finitely satisfiable exactly if $\Phi$ has a finite monomerge spoiler. 
\end{lemma}

\begin{definition}
	Given a sentence $\Phi$ we let 
	$\Phi_{\fin}^\spoil = \Phi^\spoil_\cl{\injhom} \vee \Phi^\spoil_\cl{\mmhom}$.
\end{definition}

We note a few properties of $\Phi_{\fin}^\spoil$: It is polytime computable and of polynomial size with respect to $\Phi$. It is in \GNFO{}, \TGF{}, and \FOte{}  
whenever $\Phi$ is. It is in \EAAE{} whenever $\Phi$ is in \EEAA{} or \AAEE{} and it is in \EAFO{} whenever $\Phi$ is in \EsFO{} or \AsFO.
By \Cref{FinHomSeparation} as well as \Cref{InjFormula} and \Cref{MonomergeFormula}, we also obtain that $\Phi_{\fin}^\spoil$ is finitely satisfiable whenever a finite spoiler of $\Phi$ exists, which gives rise to the following characterization.

\begin{theorem}[label=FinHomClFormula, name=]
$\Phi$ is finitely homclosed if and only if $\Phi_{\fin}^\spoil$ is finitely unsatisfiable.
\end{theorem}


Now we have all the bits and pieces together to establish a bunch of tight complexity results -- both for the finite and the arbitrary case.

\begin{theorem}[restate=mainsecond, name=]
	$\mathrm{HomClosed}_\mathrm{(fin)}$ for
    \vspace{-4pt}
	\begin{enumerate}[itemindent=0ex, leftmargin=3ex, itemsep=-0.5ex]
		\item 
		\GNFO{} is \textsc{2ExpTime}-complete,
		\item 
		\TGF{} is \textsc{coN2ExpTime}-complete,
		\item 
		any of \FOte{},
		\AsFO{},
		\EsFO{},				
        \AAEE{}, and 
        \EEAA{} is \textsc{coNExpTime}-complete.				
	\end{enumerate}
\end{theorem}

Finally, we turn to \TGD, for which the above machinery fails to work (despite the trivial (un)satisfiability problem) as \TGD{} is not closed under negation. We start with a positive result.

\begin{theorem}[restate=TGDNPtheorem, name=]
$\mathrm{HomClosed}_\mathrm{(fin)}$ for \TGD{} is \textsc{NP}-complete.
\end{theorem}

Intuitively, this result is facilitated by the fact that homclosed \TGD{} sentences must have very small universal models, against which  homclosedness violations can be checked on a per-rule basis. However, all is lost as soon as just a tad of disjunction is introduced.

\newcommand{\tight}[1]{\hspace{1pt}{#1}\hspace{1pt}}
\newcommand{\tightimp}{\hspace{-1.7ex} & \impl & \hspace{-1.7ex}}

\begin{definition}\label{def:TGDdettiling}
	Given a deterministic domino system
	$\mathcal{D}=(D,B,L,H,V)$, we define the \TGD{} sentence 
	\begin{equation}
	\Phi_{\mathcal{D}\text{-}\mathrm{tiling}} = \Phi_\mathrm{grid} \wedge \Phi_\mathcal{D},
	\end{equation} 
	where $\Phi_\mathcal{D}$ is the conjunction over the following rules (with unary predicates $\sigT{}_\emptyset$ and $\sigT{}_{\{d\}}$ for all the dominoes $d \in D_\mathcal{D}$):
	\begin{eqnarray}
		  &  & \hspace{-7ex}\exists v.\sigT{}_{B \cap L}(v) \\[3pt]
		%
		\forall xy.\big(\sigT{}_{\{d\}}(x) \tight{\wedge} \sigV{}(x,y) \tightimp \sigT{}_{V(d)}(y) \big) 
		\hspace{10.7ex}{d \tight{\in} L}\hspace{4.5ex}\\[3pt]
		%
		\forall xy.\big(\sigT{}_{\{d\}}(x) \tight{\wedge} \sigH{}(x,y) \tightimp \sigT{}_{H(d)}(y) \big) 
		\hspace{10.5ex}{d \tight{\in} B}\\[3pt]
		%
		\forall xyz.\big( \sigT{}_{\{d\}}(x) \tight{\wedge} \sigH{}(x,z) \tight{\wedge} \sigT{}_{\{d'\}}(y) \tight{\wedge} \sigV{}(y,z) \hspace{-18.7ex} \nonumber \\[-1pt]
		 \tightimp 
		\sigT{}_{H(d)\cap V(d')}(z) \big) 
		\hspace{2.5ex} {d,\!d' \tight{\in} D}\\[3pt]
		%
		\forall x.\big(\sigT{}_{\{d\}}(x) \tight{\wedge} \sigT{}_{\{d'\}}(x) \tightimp
		\sigT{}_{\emptyset}(x) \big) 
		\hspace{5ex} {d,\!d' \tight{\in} D, d \tight{\neq} d'}
	\end{eqnarray}	
\end{definition}
Note that a(n ultimately periodic) $\mathcal{D}$-tiling exists exactly if 
$\Phi_{\mathcal{D}\text{-}\mathrm{tiling}}$ has a (finite) model $\mathfrak{A}$ with $\sigT{}_\emptyset{}^\mathfrak{A} = \emptyset$.
This is the case exactly if	$\getffinmodels{\exists x.\sigT{}_\emptyset{}(x)} \neq\getffinmodels{\Phi_{\mathcal{D}\text{-}\mathrm{tiling}} \vee \exists x.\sigT{}_\emptyset{}(x)}$. Yet, note that any -- finite or infinite -- model $\mathfrak{A}$ of $\Phi_{\mathcal{D}\text{-}\mathrm{tiling}}$ with $\sigT{}_\emptyset{}^\mathfrak{A} = \emptyset$ admits a homomorphism into $\mathfrak{A} \uplus \mathfrak{I}$, for which $\sigT{}_\emptyset{}^{\mathfrak{A}\uplus \mathfrak{I}} = \emptyset$ and it is not a model of $\Phi_{\mathcal{D}\text{-}\mathrm{tiling}}$, thus $\mathfrak{A}\uplus \mathfrak{I} \notin \getffinmodels{\Phi_{\mathcal{D}\text{-}\mathrm{tiling}} \vee \exists x.\sigT{}_\emptyset{}(x)}$. Therefore, the class of (finite) models of the \MDTGD{} sentence $\Psi = \Phi_{\mathcal{D}\text{-}\mathrm{tiling}} \vee \exists x.\sigT{}_\emptyset{}(x)$ is closed under (finite-target) homomorphisms exactly if there exists no (ultimately periodic) $\mathcal{D}$-tiling. We get the following.


\begin{theorem}
	$\mathrm{HomClosed}_\mathrm{(fin)}$ is undecidable for \MDTGD{} (and hence for \DTGD{}).
	This even holds for \DTGD{} sentences with $\leq\! 4$ universally and $\leq\! 1$ existentially quantified variables and predicates of arity $\leq\! 2$. 
\end{theorem}



\neutral

\section{Characterizability of Homclosures}\label{sec:characterizability}


In case a sentence $\Phi$ under consideration is not homclosed, one might still be interested in the (finite) homclosure $\homclffin{\getffinmodels{\Phi}}$ of its (finite) models. In particular one might want to (finitely) characterize it logically, that is, find some sentence $\Psi$
whose set of (finite) models is precisely $\Phi$'s (finite) homclosure.  

We start with the observation that even for very basic $\FOe$ sentences contained in very restricted fragments, it might be impossible to characterize their homclosure within $\FOe$.

\newcommand{\minuss}{\vspace{-6.2pt}}

\begin{definition}\label{infpath}
Let $\Phi_{\infty}$ denote the $\{\sigP\}$-sentence
$
\forall x \exists y. \sig\sigP(x,y).
$
\end{definition}
\minuss


We note that $\Phi_{\infty}$ is simultaneously contained in $\GFO$, $\FOt$, $\bal\bex\FO$ and $\TGD$.
$\homcl{\getmodels{\Phi_{\infty}}}$ consists of all structures containing an infinite $\sigP{}$-path.
In the finite, this condition is equivalent to containing a directed $\sigP{}$-cycle. Identifying such structures is \textsc{NLogSpace}-complete, hence not expressible in $\FOe$~\cite{FurstSS84}.

\begin{proposition}[restate=notFO, name=]\label{notfo}
$\homcl{\getmodels{\Phi_{\infty}}}$ cannot be characterized within $\FOe$.
\end{proposition}
\minuss

Given that homclosures are not guaranteed to be characterizable in $\FOe$ even for very inexpressive $\FOe$ fragments, we might ask whether one can at least tell sentences that have a $\FOe$-characterizable homclosure from those that do not.
%

$\left.\right.$\\[-1.5ex]
\probbox{
	\textbf{Problem:} $\mathrm{ExCharHomCl}_\ffin^\Logic$\\
	\textbf{Input:} $\tau$, $\tau$-sentence $\Phi$.\\
	\textbf{Output:} \textsc{yes}, if $\getffinmodels{\Psi} = \homclffin{\getffinmodels{\Phi}}$ for some $\Logic$-sen\-tence $\Psi$, \textsc{no} otherwise.  
}
\\

It turns out the answer is no. Also, there is no stronger logic capable of characterizing the homclosures of all $\FOe$ sentences, nor one that would facilitate $\mathrm{ExCharHomCl}_\ffin$.

\begin{theorem}[restate=ExCharHomA, name=]
$\mathrm{ExCharHomCl}_\ffin^\Logic$ is undecidable for $\FOe$ sentences whenever $\Logic$ is a non-degenerate logic with a decidable model checking problem.
\end{theorem}

\begin{theorem}[restate=ExCharHomB, name=]
	$\mathrm{ExCharHomCl}_\ffin^\Logic$ is undecidable for $\TGD$ sentences whenever $\Logic$ is a logic able to express $\cq{\mathfrak{F}}$ and having a decidable model checking problem.
\end{theorem}
\minuss

After establishing these negative results, we turn to the other end of the spectrum to identify cases where the homomorphism closure is guaranteed to be characterizable in $\FOte$ and thus -- as a consequence of the homomorphism preservation theorems -- even in $\EpFO$. For $\EAFO$, the Bernays-Schönfinkel class, this is an easy consequence of the existence of ``small'' induced submodels.

\begin{proposition}[restate=EAFOinEPO, name=]
	Every $\EAFO$ sentence $\Phi$ satisfies $\homcl{\getmodels{\Phi}} = \homcl{\getfinmodels{\Phi}}$ and one can effectively compute a $\EpFO$ sentence $\Psi$ with $\homcl{\getmodels{\Phi}}=\getmodels{\Psi}$.
\end{proposition}
\minuss

Using descriptive complexity theory \cite{DBLP:books/daglib/0095988}, this implies a very low data complexity of checking hom\-closure membership.  

\begin{corollary}
	For any fixed $\EAFO$ sentence $\Phi$, checking membership in $\homcl{\getffinmodels{\Phi}}$ is in \textsc{AC}$^0$ in the size of the structure.
\end{corollary}
\minuss

Now, after establishing the extremal cases (of characterizability failing altogether vs. it succeeding already within $\EpFO$) we investigate the middle ground in between: the fragments \EAAE{}, \FOte{}, \TGF{}, \GFO{}, and \GNFO{}. By \Cref{the:arb-sameas-fin}, we know that $\homclfin{\getmodels{\Phi}}=\homclfin{\getfinmodels{\Phi}}$ for every sentence $\Phi$ in any of these logics. Therefore, any sentence characterizing $\homcl{\getmodels{\Phi}}$ will also finitely characterize $\homclfin{\getfinmodels{\Phi}}$. Hence, we will next describe how to characterize $\homcl{\getmodels{\Phi}}$ for the considered fragments. The task of characterizing $\homcl{\getfinmodels{\Psi}}$ is slightly different and will be discussed thereafter.

It is helpful to study projective homclosure characterizations.

\begin{definition}
	For logics $\Logic, \Logic' \subseteq \SO$, we say $\Logic'$ \emph{homcaptures} $\Logic$ if for every $\tau$ and $\Logic$ $\tau$-sentence $\Phi$ there exist some $\sigma$ and a $\Logic'$ $\tau{\uplus}\sigma$-sentence $\Psi$ such that  $\homcl{\getmodels{\Phi}} = \getmodels{\Psi}|_{\tau}$.   
\end{definition}
\minuss

Note that whenever a logic $\Logic'$ homcaptures a logic $\Logic$, it follows that for any $\Logic$ sentence $\Phi$, we find an $\ESO(\Logic')$ sentence characterizing $\Phi$'s homclosure, where $\ESO(\Logic')$ describes the fragment of $\SO$ sentences led by second-order existential quantifiers followed by an $\Logic'$ sentence.   

We will next give a description of a generic type-based approach toward homcapturing results, which will be applicable to all mentioned logics with slight variations in the details. 

Fixing a finite signature $\tau$, we let \emph{$n$-types} denote maximal consistent sets $\tp$ of literals (i.e., negated and un\-ne\-ga\-ted atoms) using all constants from $\tau$ and a finite collection of variables $v_1,\ldots,v_n$ (denoted $\bold{v}$) disjoint from the variables used in the considered sentence. 
For convenience, we identify the (necessarily finite) set $\tp$ with $\bigwedge \! \tp$. 
We say an $n$-type $\tp$ is \emph{realized} by an $n$-tuple $\bold{a} = (a_1,\ldots,a_n) \in A^n$ of domain elements of a $\tau$-structure $\mathfrak{A}$, if $(\mathfrak{A},\{v_i \mapsto a_i\}) \models \tp$. We say $\tp$ is \emph{realized} ($k$ times) in $\mathfrak{A}$ if it is realized by ($k$ distinct) tuples of $\mathfrak{A}$. Note that every tuple $\bold{a}$ realizes one unique type, which entirely characterizes the substructure of $\mathfrak{A}$ induced by all the elements of $\bold{a}$ together with all elements denoted by constants.   
Our approach makes use of the fact that for several logics, many questions regarding a structure can be answered, if one knows what types of a certain kind are realized therein. On a generic level, we assume there is a function that, given a $\tau$-sentence $\Phi$ yields a finite set $\mathcal{E}$ of \emph{eligible types} (possibly containing $n$-types of varying $n$), which are sufficient for our purposes. Then, given a $\tau$-structure $\mathfrak{A}$, we let its \emph{type summary} $\mathcal{E}(\mathfrak{A})$ denote the pair $(\mathcal{E}_\text{+},\mathcal{E}_!)$ where $\mathcal{E}_\text{+}$ ($\mathcal{E}_!$) consists of those types from $\mathcal{E}$ realized (exactly once) in $\mathfrak{A}$.\footnote{In fact, $\mathcal{E}_!$ is only needed for the $\FOte$ case in order to deal with the phenomenon of so-called ``kings'', i.e., $1$-types forced to be unique in models. For the other logics considered, $\mathcal{E}_\text{+}$ suffices.} 
Furthermore, the \emph{model summary} $\getsummary{\Phi}$ of a sentence $\Phi$ is defined by $\{\mathcal{E}(\mathfrak{A}) \mid \mathfrak{A} \in \getmodels{\Phi}\}$, which is a finite set by assumption. 

We will extend our signature by a set $\sigma$ of \emph{type predicates}, containing an $n$-ary predicate symbol $\Tp_\tp$ for every eligible $n$-type $\tp \in \mathcal{E}$.
Given a $\tau$-structure $\mathfrak{A}$ and a set $\mathcal{E}$ of eligible types over $\tau$, the \emph{$\mathcal{E}$-adornment} of $\mathfrak{A}$, denoted $\mathfrak{A}\cdot\mathcal{E}$ is the $\tau\uplus\sigma$-structure expanding $\mathfrak{A}$ such that
$\Tp_\tp^{\mathfrak{A}\cdot\mathcal{E}} = \{ \bold{a} \mid \bold{a} \text{ realizes } \tp \text{ in } \mathfrak{A} \}$. In a way, $\mathcal{E}$-adornments materialize the information regarding relevant types in a model. While this information is obviously redundant in the model itself, we will make good use of it when characterizing a model's homomorphic targets.

Given a $\tau$-sentence $\Phi$ (which we presume to be in some normal form with bounded quantifier alternation), we will construct a $\tau\uplus\sigma$-sentence $\Psi$ that (projectively) identifies those $\mathfrak{B}$ into which some model $\mathfrak{A}$ of $\Phi$ has a homomorphism $h\colon\mathfrak{A} \arbhom \mathfrak{B}$, using the following idea: assuming such an $h$ exists, then it will also serve as homomorphism from $\mathfrak{A}\cdot\mathcal{E}$ to $\mathfrak{B}\cdot h(\mathfrak{A}\cdot\mathcal{E}|_\sigma)$ denoting the structure $\mathfrak{B}'$ expanding $\mathfrak{B}$ by letting $\Tp_\tp^{\mathfrak{B}'} = \{ (h(a_1),\ldots,h(a_n)) \mid (a_1,\ldots,a_n) \in \Tp_\tp^{\mathfrak{A}\cdot\mathcal{E}} \}$, that is, where all type information materialized in $\mathfrak{A}$ is transferred to $\mathfrak{B}$ along $h$. In fact, $\Psi$ is then built in a way that its models are precisely such $\mathfrak{B}'$.
Thereby, one integral part of $\Psi$ is a conjunct $\Psi_\mathrm{hom}$ ensuring that types are mapped to $\mathfrak{B}$ in a way compatible with the homomorphism condition, i.e., that at least all positive literals postulated by the type are indeed present: 
\begin{equation}
\Psi_\mathrm{hom} = \!\bigwedge_{\tp \in \mathcal{E}} \!\forall \bold{v}.\big( \Tp_\tp(\bold{v}) \impl \smallbigwedge \{\alpha \tight{\in} \tp \mid \alpha \text{ positive lit.}\} \big)
\end{equation}
Moreover, we require that the collection of ``transferred type atoms'' found in $\mathfrak{B}'$ corresponds to the type summary of some model. To this end, we include in $\Psi$ the case distinction 
\begin{equation}
\bigvee_{(\mathcal{E}_\text{+},\mathcal{E}_!)\in \getsummary{\Phi}} \Psi_{(\mathcal{E}_\text{+},\mathcal{E}_!)},
\end{equation}
where each $\Psi_{(\mathcal{E}_\text{+},\mathcal{E}_!)}$ details existence and relationships between the mapped types. Some of the conditions specified in $\Psi_{(\mathcal{E}_\text{+},\mathcal{E}_!)}$ are generically derived from $(\mathcal{E}_\text{+},\mathcal{E}_!)$, others are obtained by translating constraints from $\Phi$ into ``type language'', expressed by type atoms.
By construction, the sentence $\Psi$ thus obtained holds in some type-adorned structure $\mathfrak{A}\cdot\mathcal{E}$ if and only if $\mathfrak{A} \in \getmodels{\Phi}$. Also, its satisfaction is preserved under ``type-faithful'' homomorphisms  $h\colon\mathfrak{A}\cdot\mathcal{E} \arbhom \mathfrak{B}'$ as described above, thus ensuring $\homcl{\getmodels{\Phi}} \subseteq \getmodels{\Psi}|_{\tau}$.
The other direction $\homcl{\getmodels{\Phi}} \supseteq \getmodels{\Psi}|_{\tau}$ is shown by taking a model $\mathfrak{B}' \in \getmodels{\Psi}$, removing all information about $\tau$, and ``unfolding'' or  ``unraveling'' it such that, in the resulting structure $\mathfrak{A}'$, the diverse $\Tp_\tp$ instances do not provide contradicting information about the $\tau$-connections between their arguments. This allows for reconstructing the interpretations of the original signature elements $\tau$ from the $\Tp_\tp$ instances. After forgetting the latter, we have arrived at our wanted $\mathfrak{A} \in \getmodels{\Phi}$.

Applying the described strategy yields the following theorem.

\begin{theorem}[restate=allhomcaptures, label=thm:allhomcaptures, name=]\label{thm:homcapture}
	\begin{enumerate}[itemindent=0ex, leftmargin=3ex, itemsep=-0.3ex]
		\item $\FOte$ homcaptures itself.
		\item $\TGF$ homcaptures itself and $\EAAE$.
		\item $\GFO$ homcaptures itself and $\GNFO$.
	\end{enumerate}
\end{theorem}

As an example for all these logics, note that the hom\-clo\-sure of $\Phi_\infty$ from \Cref{infpath},
is projectively characterized by 
\begin{equation}
\exists x.\big(\hspace{1pt}\sig{Lv}(x)\hspace{1pt}\big) \wedge \forall x.\big(\, \sig{Lv}(x) \impl \exists y. (\hspace{1pt}\sigP(x,y) \wedge \sig{Lv}(y)\hspace{1pt})\,\big).
\end{equation}


In view of our earlier discussion, this ensures that all these logics admit homclosure characterization in $\ESO$, and that the data complexity of homclosure membership is in \textsc{NP} \cite{fagingeneralized}.

\begin{corollary}
	For any fixed \FOte, \TGF{}, \EAAE, \GFO{} or \GNFO{} sentence $\Phi$, checking $\mathfrak{A} \in \homcl{\getmodels{\Phi}}$ 
	is in \textsc{NP}.   
\end{corollary}

We next complement these $\textsc{NP}$-membership results by providing matching lower bounds even for much weaker logics.

\newcommand{\narrowimp}{ \!\! & {\!\!\!\!\!\!{\impl}\!\!\!\!\!\!} &  \!\! }
\newcommand{\narrowwedge}{\,{\wedge}\,}
\newcommand{\narrowvee}{\,{\vee}\,}

\begin{definition}[restate=phithreesat, name=]\label{def: 3SAT}
	Given an instance of 3SAT $\mathcal{S} = \{C_1, \ldots, C_n\}$ consisting of $n$ clauses each containing 3 literals over propositional variables $p_1,\ldots,p_m$, we define the structure $\mathfrak{A}_{\mathcal{S}}$ by
	\begin{eqnarray}
		A_\mathcal{S} & {\!\!\!\!=\!\!\!\!} & \{a_i \mid 1\leq i\leq n\tight{+}1\} \cup \{b_i, b'_i \mid 1\leq i\leq m\}\ \ \ \ \ \ \ \\
		\sig{First}^{\mathfrak{A}_{\mathcal{S}}} & {\!\!\!\!=\!\!\!\!} & \{a_1\}\\
		\sig{CLst}^{\mathfrak{A}_{\mathcal{S}}} & {\!\!\!\!=\!\!\!\!} & \{a_i \mid 1\leq i\leq n\}\\
		\sig{Nil}^{\mathfrak{A}_{\mathcal{S}}} & {\!\!\!\!=\!\!\!\!} & \{a_{n+1}\}\\
		\sig{Next}^{\mathfrak{A}_{\mathcal{S}}} & {\!\!\!\!=\!\!\!\!} & \{(a_i,a_{i+1}) \mid 1\leq i\leq n\}\\
		\sig{Lit}^{\mathfrak{A}_{\mathcal{S}}} & {\!\!\!\!=\!\!\!\!} & \{(a_i,b_j) \mid p_j \in C_i\} \cup \{(a_i,b'_j) \mid \neg p_j \in C_i\}\\
		\sig{Sel}^{\mathfrak{A}_{\mathcal{S}}} & {\!\!\!\!=\!\!\!\!} & \{b_i, b'_i \mid 1\leq i\leq m\}\\
		\sig{Cmp}^{\mathfrak{A}_{\mathcal{S}}} & {\!\!\!\!=\!\!\!\!} & \{(b_i,b_j),(b'_i,b'_j),(b_i,b'_k),(b'_i,b_k) \mid i\not= k\}
	\end{eqnarray}
	We let $\Phi_\mathrm{3SAT}$ be the sentence containing the conjunction of the following sentences (leading universal quantifier omitted):
\begin{eqnarray}
	& & \!\!  \exists x. \sig{First}(x)  	\label{Def14: 1}\\
	\sig{First}(x)
	\narrowimp
	\sig{CLst}(x)\ \ \ \ \ \label{Def14: 2} \\
	\sig{CLst}(x)
	\narrowimp
	\exists y. \sig{Next}(x,y) \narrowwedge \big(\sig{CLst}(y) \narrowvee \sig{Nil}(y)\big)\ \ \ \ \   \label{Def14: 3} \\
	\sig{CLst}(x) 
	\narrowimp 
	\exists y. \sig{Lit}(x,y)  \wedge \sig{Sel}(y)   \label{Def14: 4}
	\\
	\sig{Sel}(x) \narrowwedge \sig{Sel}(y) 
	\narrowimp 
	\sig{Cmp}(x,y).   \label{Def14: 5}
\end{eqnarray}
\end{definition}
Obviously, $\Phi_\mathrm{3SAT}$ is contained in constant-free, equality-free $\FOt$ (and therefore in $\FOte$ and \TGF). It can be 
rewritten into constant-free, equality-free $\Al\Al\Ex\Ex\FO$ (and therefore into 
$\EAAE$). Moreover, for any set $\mathcal{S}$ of $3$-clauses, satisfiability of the 3SAT instance $\bigwedge_{\{\ell_1,\ell_2,\ell_3\}\in \mathcal{S}} \ell_1 \vee \ell_2 \vee \ell_3$  coincides with $\mathfrak{A}_\mathcal{S} \in \homcl{\getffinmodels{\Phi_\mathrm{3SAT}}}$. Hence we obtain the following statement, providing matching lower bounds.

\begin{proposition}[restate=NPhardFOprop, name=]
	Checking homclosure membership is \textsc{NP}-hard for constant-free, equality-free $\FOt$ and $\Al\Al\Ex\Ex\FO$ with predicate arity $\leq 2$.
\end{proposition}

Note that this settles the case of the data complexity of homclosure membership for any logic between $\FOt$ / $\Al\Al\Ex\Ex\FO$ on the lower and  $\FOte$ / $\TGF{}$ / $\EAAE$ on the upper end, showing that one cannot hope for characterizability in a logic having a lower data complexity of model checking. It does not provide a lower bound for $\GFO$ or $\GNFO$, though. As it turns out, we can in fact do better for these\textellipsis

Given a $\GNFO$ sentence $\Phi$, let $\Psi$ be the $\GFO$ sentence projectively characterizing $\homcl{\getmodels{\Phi}}$ constructed as described before.
Then we obtain $\homcl{\getmodels{\Phi}} = \getmodels{\exists \sigma.\Psi}$ and also observe
\begin{align}
	\exists \sigma.\Psi & = \exists \sigma . \Big(\Psi_\mathrm{hom} \wedge \hspace{-1ex} \bigvee_{\scriptscriptstyle(\mathcal{E}_\text{+},\mathcal{E}_!)\in \getsummary{\Phi}} \hspace{-1ex} \Psi_{\!(\mathcal{E}_\text{+},\mathcal{E}_!)}\Big) \\
	 & \equiv \hspace{-1ex} \bigvee_{\scriptscriptstyle(\mathcal{E}_\text{+},\mathcal{E}_!)\in \getsummary{\Phi}} \hspace{-1ex} \exists \sigma . \big(\Psi_{\!(\mathcal{E}_\text{+},\mathcal{E}_!)} \wedge \Psi_\mathrm{hom} \big).
\end{align}
Thoroughly inspecting $\exists \sigma.(\Psi_{(\mathcal{E}_\text{+},\mathcal{E}_!)} \wedge \Psi_\mathrm{hom})$, we find that it can be equi\-valently rewritten into the following $\FOe$ sentence with least fixed points over simultaneous inductive definitions for ``complement predicates'' $\overline{\Tp}_\tp$ from a signature $\overline{\sigma}$:
\begin{equation} 
\bigwedge_{\tp\, \in \,\mathcal{E}_\text{+}} \!\exists \bold{x}.\neg\Big[\textbf{lfp}_{\overline{\Tp}_\tp} \big\{\, \overline{\Tp}_{\tp'}(\bold{y}) \leftarrow \xi_{\overline{\Tp}_{\tp'}}[\bold{y}] \,\, \big|\,  \tp' \tight{\in} \mathcal{E}\,\big\}\Big](\bold{x}),
\end{equation} 
where the $\xi_{\overline{\Tp}_\tp'}$ are $\FOe$ $\tau\uplus\overline{\sigma}$-formulae with only positive occurrences of type atoms, obtained from aggregating the bodies of the contrapositions of the implications in $\Psi_{(\mathcal{E}_\text{+},\mathcal{E}_!)}$.\footnote{Conceptually, this equivalent transformation can be explained as follows: for any two models $\mathfrak{B}_1$, $\mathfrak{B}_2$ of $\Psi_{(\mathcal{E}_\text{+},\mathcal{E}_!)}$ with $\mathfrak{B}_1|_\tau = \mathfrak{B}_2|_\tau$, the struc\-ture $\mathfrak{B}_1 {\cdot} \mathfrak{B}_2$ obtained as the $\sigma$-expansion of $\mathfrak{B}_1|_\tau$ by $\Tp_\tp^{\mathfrak{B}_1 {\cdot} \mathfrak{B}_2} = \Tp_\tp^{\mathfrak{B}_1} {\cup} \Tp_\tp^{\mathfrak{B}_2}$ for all eligible $\tp$ is also a model. Hence, for every $(\mathcal{E}_\text{+},\mathcal{E}_!)$-compatible $\tau$-structure $\mathfrak{B}$, there is a ``maximal $\sigma$-expansion'' $\mathfrak{B}^*$ satisfying $\Psi_{(\mathcal{E}_\text{+},\mathcal{E}_!)}$, which can be deterministically obtained in a type-elimination like process, via a least fixed-point computation over predicates $\overline{\Tp}_{\tp}$ that express exactly the absence of ${\Tp}_{\tp}$. This spares us the ``guessing'' of type relations via $\exists \sigma$.}
Simultaneous induction can in turn be expressed using plain least-fixed-point logic \cite{GurevichS86}, therefore $\Psi$ as a whole can be expressed in first-order least-fixed-point logic denoted $\FOe^\mathbf{lfp}$ (sometimes also referred to as LFP or FO+LFP). Using the fact that evaluating $\FOe^\mathbf{lfp}$ sentences over structures can be done in polynomial time \cite{Immerman82,Vardi82}, we obtain an improved result regarding the data complexity of homclosure membership.

\begin{theorem}[restate=GNFOinLFP, name=]\label{sorpresa}
Homclosures of $\GNFO$ and $\GFO$ sentences can be characterized in $\FOe^\mathbf{lfp}$. Thus, for a fixed $\GNFO$ or $\GFO$ sentence $\Phi$, checking $\mathfrak{A} \in \homcl{\getmodels{\Phi}}$ is \textsc{P}-complete.
\end{theorem}


As a minimalistic example, consider the $\GFO{}$ sentence $\Phi_\infty$ from \Cref{infpath}, whose homclosure cannot be characterized in $\FOe$ (see \Cref{notfo}),
but by the $\FOe^\mathbf{lfp}$ sentence
\begin{equation}
\exists x.\neg\Big[\textbf{lfp}_{\sig{Cds}} \big\{ \sig{Cds}(y) \leftarrow  \forall z. \big(\sigP (y,z) \impl \sig{Cds}(z)\big) \big\}\Big](x),
\end{equation}
where the fixed point predicate $\sig{Cds}$ identifies ``cul-de-sac'' elements in the given $\{\sig{P}\}$-structure.


We finally turn to the case of finding a characterization of $\homcl{\getfinmodels{\Phi}}$, i.e., of those -- finite as well as infinite -- structures receiving a homomorphism from a \emph{finite} model of $\Phi$. We start with a negative result for a very simple sentence, following  from compactness of $\FOe$.

\begin{proposition}[restate=cornercaseone, name=]
$\homcl{\getfinmodels{\Phi_\infty}}$ cannot be characterized in \ESO.
\end{proposition}

\newcommand{\ESOplus}{\ESO$^{\#}$}
However, it is easy to see that an immaterial extension of \ESO{} suffices to remedy this situation: let \ESOplus{} denote \ESO{} extended by a special second-order quantifier $\exists^\fin$ which requires that the quantified-over predicate corresponds to a finite relation -- an idea borrowed from weak monadic second-order logic. We note that $\exists^\fin \sig{U}$ is also expressible in plain $\SO$.

\begin{theorem}[restate=cornercasetwo, name=]
Let $\Phi$ be some $\SO$ sentence and let $\Psi$ be such that $\getfinmodels{\Psi} = \homclfin{\getfinmodels{\Phi}}$.
Then $\homcl{\getfinmodels{\Phi}} = \getmodels{\exists^\fin \sig{U}. \Psi^{\mathrm{rel}(\sigU)}}$.
\end{theorem}

Together with \Cref{the:arb-sameas-fin} and the previously established characterization results, this allows us to close this case.

\begin{corollary}[restate=cornercasethree, name=]
For all $\Phi$ in any of \EAAE{}, \TGF{}, \FOte{}, and \GNFO{}, $\homcl{\getfinmodels{\Phi}}$ can be characterized in \ESOplus{}.
\end{corollary}



\section{Normal Forms for Homclosed Fragments}\label{sec:normalforms}

This section is devoted to identifying ``homclosed normal forms'' of certain logics.
Given a logic $\Logic$, a homclosed normal form $\mathbb{H}\Logic$ of $\Logic$ is a fragment which satisfies
\begin{itemize}[itemindent=0ex, leftmargin=3ex, topsep=0ex]
\item $\mathbb{H}\Logic \subseteq \Logic$,
\item every $\Phi \in \mathbb{H}\Logic$ is homclosed,
\item for each homclosed $\Phi \in \Logic$ exists a $\Phi' \in \mathbb{H}\Logic$ with $\Phi \equiv \Phi'$,  
\item membership in $\mathbb{H}\Logic$ is decidable.  
\end{itemize}

It is, of course, desirable to strengthen the last requirement to say that membership can be decided ``easily'', e.g., in polytime.
Likewise, we might require the existence of a (preferably computationally inexpensive) algorithm to compute $\Phi'$ from a given $\Phi$, or of a (not too high) upper bound regarding the size of $\Phi'$ with respect to $\Phi$.

\noindent\textbf{Normal Forms for $\FOe$ and Fragments.}\quad
%
%
The homomorphism preservation theorem provides a normal form for full $\FOe$, namely $\EpFO$; the equivalent $\EpFO$ sentences may even be chosen to have equal quantifier rank. 
Rossman's theorem shows that the same normal form can be used if we consider finite models only;
more recently, Rossman even established that any homclosed first-order formula of quantifier rank $k$ is equivalent to an existential positive formula of quantifier rank $k^{O(1)}$~\cite{Rossman17}. 
Note that computing the normal form can be quite costly: there is a potentially non-elementary blow-up for the length of a shortest equivalent existential positive formula \cite[Theorem 6.1]{Rossman08}. 

Obviously, this characterization of a normal form for $\FOe$ carries over to any $\FOe$ fragment that fully encompasses $\EpFO$. For Bernays-Schönfinkel fragments with bounded number of existential quantifiers, we state the following observation.

\begin{proposition}[restate=boundedex, name=]
For every homclosed $\bex^n\bal^*\FOe$ sentence there exists an equivalent $\bex^n\FOe^+$ sentence. 
\end{proposition}

For other fragments, the case is less immediate. Whenever $\mathrm{HomClosed}$ is decidable for $\Logic$, we can use this (usually expensive) test to directly define $\mathbb{H}\Logic$, thus obtaining $\mathbb{H}\TGF$, $\mathbb{H}\FOte$, and $\mathbb{H}\AAEE$. For $\FOe$ fragments not covered by these cases (such as \AAAE), the question remains open.

\newcommand{\supclffin}[1]{#1^{\cl{\str{\injhom}}_\ffin}}
\newcommand{\homsucl}[1]{#1^{\mathrm{sH}_\ffin}}

%

\newcommand{\FOgfp}{$\FOe^\textbf{lfp}$}
\newcommand{\vfpFO}{$\varnothing\FO^+_=$}

\newcommand{\cFOe}{\hfill\cite{FOundecChurch,FOundecTuring}}
\newcommand{\cDTGD}{}
\newcommand{\cMDTGD}{}
\newcommand{\cTGD}{}
\newcommand{\cGFO}{\hfill\cite{Gra99,CateF05}}

\newcommand{\cGNFO}{\hfill\cite{GuardedNegation}}
\newcommand{\cTGF}{\hfill\cite{RS18}}
\newcommand{\cFOte}{\hfill\cite{GKV97}}
\newcommand{\cAAAE}{\hfill\cite{AAAE}}
\newcommand{\cEAAE}{\hfill\cite{Lewis}} 
\newcommand{\fEAAE}{\hfill\cite{Schuette}}
\newcommand{\fTGF}{\hfill\cite{kier2021finite}}
\newcommand{\fGFO}{\hfill\cite{Gra99,BaranyGO13}}
\newcommand{\cAAEE}{}
\newcommand{\cEAFO}{}
\newcommand{\cAsFO}{}
\newcommand{\cEEEA}{}
\newcommand{\cEEAA}{}
\newcommand{\cEsFO}{}
\newcommand{\cEpFO}{}
\newcommand{\cSO}{}
\newcommand{\cESO}{}
\newcommand{\cEGSO}{}

\renewcommand{\mcl}[1]{\multicolumn{2}{@{}l@{}}{#1}} 

\begin{table*}[t!]
	\caption{Overview of newly established results 
	(5 last columns, except normal form for $\FOte$) 
	as well as references used (easy insights are left without reference). 
	All complexities are time complexities. Closure of $\GNFO$ under $\neg$ only holds for sentences.
	\label{fig:main}}
\small~\hfill
\begin{tabular}{@{}l@{\ \ \ }l@{\ \ \ }l@{\ \ \ \ }l@{\ \ }l@{\ \ \ \ \ }l@{\ }l@{\ \ \ \ }l@{\ \ }l@{\ \ }l@{}}
	logic    & SAT     & finite model    & \mcl{closure} & \mcl{InHomCl}  & HomClosed & homclosure charac- & normal form   \\
	name     & fin/arb & property (size) & $\neg$ & $\wedge$ & comb. & data & fin/arb & terizable in logic & fragment  \\ \hline\hline \\[-2ex]
	$\FOe$   & und.    \cFOe   & no              & yes    & yes      & und.  & und.    & und.    & none               & \EpFO     \\
	$\DTGD$  & trivial \cDTGD  & yes (1)      & no     & yes      & und.  & und.    & und.    & none               & \UCQ     \\
	$\MDTGD$ & trivial \cMDTGD & yes (1)      & no     & no       & und.  & und.    & und.    & none               & \CQ{}$\vee$\CQ \\
	$\TGD$   & trivial \cTGD   & yes (1)      & no     & yes      & und.  & und.    & NP      & none               & \CQ \\
	$\TGF$   & N2Exp   \cTGF   & yes (2Exp) \fTGF  & yes    & yes      & N2Exp & NP      & coN2Exp & \ESO(\TGF)             & $\mathbb{HTGF}$ \\
    $\FOte$  & NExp    \cFOte  & yes (Exp) \cFOte  & yes    & yes      & NExp  & NP      & coNExp  & \ESO(\FOte)             & $\mathbb{H}\FOte$ \\ 
	$\GNFO$  & 2Exp    \cGNFO  & yes (2Exp) \cGNFO & yes    & yes      & 2Exp  & P  & 2Exp    & \FOgfp \text{ /} \ESO(\GFO)         & \EpFO \\
	$\GFO$  & 2Exp \   \cGFO  & yes (2Exp)\  \fGFO & yes    & yes      & 2Exp  & P  & 2Exp    & \FOgfp \text{ /} \ESO(\GFO)         & \EpFO \\ \hline \\[-2ex]
	\AAAE    & und.    \cAAAE  & no              & no     & no       & und.  & und.    & und.    & none               & ? \\ 
	\EAAE    & NExp    \cEAAE  & yes (2Exp) \fEAAE & no     & yes      & NExp  & NP      & und.    & \ESO(\TGF)             & \EpFO \\
	\AAEE    & NExp    \cAAEE  & yes (2Exp) & no     & no       & NExp  & NP      & coNExp  & \ESO(\TGF)             & $\mathbb{H}$\AAEE \\
	\EAFO    & NExp    \cEAAE  & yes (C+Ex)      & no     & yes      & NExp  & AC$^0$  & und.    & \EpFO     & \EpFO  \\
	\AsFO    & NExp    \cAsFO  & yes max(C,1)    & no     & yes      & NExp  & AC$^0$  & coNExp  & \EFO             & \vfpFO \\
	\EEEA    & NP      \cEEEA  & yes (C+3)       & no     & no       & NP    & AC$^0$  & und.    & $\Ex\Ex\Ex\FO^+$ & $\Ex\Ex\Ex\FO^+$ \\
	\EEAA    & NP      \cEEAA  & yes (C+2)       & no     & no       & NP    & AC$^0$  & coNExp  & $\Ex\Ex\FO^+$    & $\Ex\Ex\FO^+$ \\
	\EsFO    & NP      \cEsFO  & yes (C+Ex)      & no     & yes      & NP    & AC$^0$  & coNExp  & \EpFO            & \EpFO \\
	\EpFO    & const.  \cEpFO  & yes (C+Ex)      & no     & yes      & NP    & AC$^0$  & trivial & \EpFO            & \EpFO \\ \hline \\[-2ex]
	\SO      & und.    \cSO    & no              & yes    & yes      & und.  & und.    & und.    & none               & $\mathbb{HSO}$ \\
\end{tabular}
\hfill~
\end{table*}

\noindent\textbf{Normal Form for Homclosed $\SO$.}\quad
We next provide a syntactic normal form for $\SO$ sentences whose class of models is homclosed. This normal form arises as a combination of a normal form for superstructure-closed classes and a normal form for classes closed under surjective homomorphisms. Our proof relies  on the interplay of syntactic manipulations of $\SO$ sentences on the one hand and ``operators'' on classes of structures on the other. We focus on the following two operations that can be applied to classes $\mathcal{C}$ of $\tau$-structures:

\vspace{-2.5ex}
\begin{eqnarray}
\supclffin{\mathcal{C}} \hspace{-1.5ex} & = & \hspace{-1.5ex} \{\mathfrak{A} \text{~(finite)}\mid \mathfrak{B}\, {\str{\injhom}}\, \mathfrak{A} \text{ for some }\mathfrak{B}\in \mathcal{C}\}, \\
\homsucl{\mathcal{C}} \hspace{-1.5ex} & = & \hspace{-1.5ex} \{\mathfrak{A} \text{~(finite)}\mid 
\{\mathfrak{A}\}^\cl{\surhom} \subseteq \mathcal{C})
\}.
\end{eqnarray}
\vspace{-2.5ex}

With these notions, we can establish the following characterizations of classes of $\tau$-structures.

\begin{lemma}[restate=classnormalform, label=lem:classnormalform, name=]
	For any class $\mathcal{C}$  of (finite) $\tau$-structures holds:	
	\begin{enumerate}[itemindent=0ex, leftmargin=3ex, itemsep=-0.5ex]
		\item $\supclffin{\mathcal{C}}$ is closed under (finite) superstructures. Moreover, if $\mathcal{C}$  is closed under  (finite) superstructures, then $\supclffin{\mathcal{C}}=\mathcal{C}$.
		\item $\homsucl{\mathcal{C}}$ is closed under (finite) surjective homomorphisms. Moreover, if $\mathcal{C}$  is closed  under (finite) surjective homomorphisms, then $\homsucl{\mathcal{C}}=\mathcal{C}$.
		\item $\supclffin{(\homsucl{\mathcal{C}})}$ is (finitely) homclosed. Moreover, every (finitely) homclosed $\mathcal{C}$ satisfies $\supclffin{(\homsucl{\mathcal{C}})}=\mathcal{C}$.
	\end{enumerate}
\end{lemma}

Based on this, we can introduce a normal form for super\-struc\-ture-closed $\SO$.

\begin{definition}
	Let $\tau$ be a signature and let $\Psi$ be an~$\SO$ $\tau$-sentence. 
	Then the~$\SO$ $\tau$-sentence $\Psi^{\mathrm{sup}}$ is defined as 
	$\exists \dot{\sigU}. \Psi^{\mathrm{rel}(\dot{\sigU})}$,
	where $\dot{\sigU}$ is a fresh unary predicate.
\end{definition}

\begin{proposition}[restate=supstructure, label=prop:supstructure, name=]
	Let $\Psi$ be an~$\SO$ $\tau$-sentence. Then
	\begin{equation}
	\supclffin{\getffinmodels{\Psi}}= \getffinmodels{\Psi^{\mathrm{sup}}}.
	\end{equation}
\end{proposition}

\Cref{prop:supstructure} and \Cref{lem:classnormalform} give rise to the normal form for superstructure-closed $\SO$.

\begin{corollary}\label{cor:supstr}For any~$\SO$ $\tau$-sentence $\Psi$, $\getffinmodels{\Psi^{\mathrm{sup}}} $ is closed under (finite)  superstructures.
	Moreover, if $\getffinmodels{\Psi}$ is closed under (finite) superstructures, then $\getffinmodels{\Psi^{\mathrm{sup}}}=\getffinmodels{\Psi}$.
\end{corollary}


We proceed by establishing a normal form for $\SO$ sentences closed under surjective homomorphisms. 
For a signature $\tau$ and variables $z,z'$, we let $\eta^{\tau}(z,z')$ denote the formula
\begin{equation}
	\bigwedge_{{\sigP \in \tau}\atop{k=\text{ar}(\sigP)}} \hspace{-6.5ex} 
	\bigwedge_{{1\leq i \leq k}\atop{{\hspace{6.5ex} \bold{x}=(x_1,\ldots,x_{i-1})}\atop{\hspace{6.5ex}\bold{y}=(x_{i+1},\ldots,x_k)}}}
	\hspace{-6.0ex}
	\forall \bold{x}\bold{y}. \sigP(\bold{x},z,\bold{y}) \Leftrightarrow \sigP(\bold{x},z',\bold{y}),
\end{equation}
%
%
%
%
and, for a unary $\sigU \notin \tau$, we define the $\tau{\uplus}\{\sigU\}$-sentence $\Theta_\sigU^\tau$ by
\begin{equation}
\forall x \exists y. (\sigU(y)\wedge \eta^{\tau}(x,y)).
\end{equation}

\begin{definition}
	Let $\tau$  and $\tau'$ be  disjoint signatures such that $\tau'$ consists of copies $\sigP'$ of all predicates $\sigP \in \tau$. Let $\Psi$ be an $\SO$ $\tau$-sentence. 
	Then the $\SO$ $\tau$-sentence  $\Psi^{\mathrm{shom}}$ is defined by

\vspace{-3ex}
\begin{equation}
\forall \sigU \, \forall \tau'\! . \bigg(\Big( \Theta_\sigU^{\tau'} \wedge \! \bigwedge_ {\sigP'\in \tau'} \! \forall \bold{x}.\big(\sigP(\bold{x}) \,{\impl}\, \sigP'(\bold{x})\big)\Big) \Rightarrow \Psi_{\tau \mapsto \tau'}^{\mathrm{rel}(\sigU) } \bigg),
\end{equation}
\vspace{-2.5ex}

\noindent where $\sigU$ is a fresh unary predicate and $\Psi_{\tau \mapsto \tau'}^{\mathrm{rel}(\sigU)}$ denotes $\Psi^{\mathrm{rel}(\sigU)}$ with every $\sigP \in \tau$ replaced by its copy $\sigP' \in \tau'$.
\end{definition}

\begin{proposition}[restate=surjhom, label=prop:surjhom, name=]
Let $\Psi$ be an \SO{} $\tau$-sentence. Then
\begin{equation}
\homsucl{\getffinmodels{\Psi}}= \getffinmodels{\Psi^{\mathrm{shom}}}.
\end{equation}
\end{proposition}

This result, together with \Cref{prop:supstructure} and \Cref{lem:classnormalform}, yields the normal form for surjective-homomorphism-closed $\SO$.

\begin{corollary}\label{cor:suhom}For any $\SO$ $\tau$-sentence $\Psi$, $\getffinmodels{\Psi^{\mathrm{shom}}}$ is closed under  surjective (finite) homomorphisms.
	Moreover, if $\getffinmodels{\Psi}$ is closed under  surjective (finite) homomorphisms, then $\getffinmodels{\Psi^{\mathrm{shom}}}=\getffinmodels{\Psi}$.
\end{corollary}


We are now ready to combine the established results toward the desired normal form for homclosed $\SO$.
Invoking \Cref{cor:supstr}, \Cref{cor:suhom}, and \Cref{lem:classnormalform}, we obtain the following statement, motivating the subsequent definition.

\begin{theorem}\label{thm:normalformSO}An $\SO$-definable class of (finite) $\tau$-structures is of the form
\vspace{-1.5ex}
\begin{equation}
\getffinmodels{(\Psi^{\mathrm{shom} }) ^{\mathrm{sup}}}
\end{equation}	
for an $\SO$ $\tau$-sentence $\Psi$  if and only if it is closed under (finite) homomorphisms.
\end{theorem}

\begin{definition}
	We define  $\mathbb{HSO}$ as the set of  $\SO$ sentences of the form  $(\Psi^{\mathrm{shom} }) ^{\mathrm{sup}}$ for any $\SO$ sentence $\Psi$.
\end{definition}

It is clear that membership in $\mathbb{HSO}$ is decidable. We can now formulate the main result of this section, as a direct consequence from \Cref{thm:normalformSO} and the fact that $(\Psi^{\mathrm{shom} }) ^{\mathrm{sup}}$ is polytime-computable from $\Psi$.

\begin{corollary} Every (finitely) homclosed $\SO$ sentence $\Phi$ is equivalent (on finite structures) to a $\mathbb{HSO}$ sentence $\Psi$ which can be computed from $\Phi$ in polynomial time.
\end{corollary}



\section{Conclusion}

Inspired by the homomorphism preservation theorem and motivated by routinely encountering -- clearly fundamental yet seemingly largely neglected -- questions 
regarding homclosures of logically characterized model classes, we undertook a principled analysis of four basic questions related to that matter and clarified them for a wide range of theoretically and practically relevant logical formalisms, both in the finite-model setting and the arbitrary-model one. \Cref{fig:main} summarizes the results obtained.

Next to several newly introduced generic techniques, which might prove useful well beyond the specific formalisms considered here, our most noteworthy achievements are probably the establishment of a computationally well-behaved normal form for homclosed $\SO$ sentences, and the insight that checking homclosure membership for the very expressive guarded negation fragment of first-order logic ($\GNFO$) has only polynomial-time data complexity.   

Plenty of open problems remain.
Obviously other, different logical formalisms could be investigated regarding the considered questions.
Not all of the established techniques immediately lend themselves to coping with ``counting logics'', such as two-variable logic with counting quantifiers.
More generally, logics without the finite model property might turn out to be harder to handle.

The investigation into the prefix classes has not yet resulted in a complete characterization of decidability and complexity of the homclosedness problem. A first inspection seems to indicate that corresponding results would hinge on a more fine-grained analysis of fragments defined by the conjunction over sentences from different prefix classes.

As far as normal forms are concerned, for first-order prefix classes with bounded number of existential quantifiers, preceded by universal quantification, the question remains generically unsolved. For \TGF{} and \FOte{}, normal forms have been shown to exist, but recognizing them is costly and more syntactic solutions would be preferred. 

An interesting open problem is whether there exists a syntactic fragment of \emph{existential second-order logic} such that the classes of finite structures that are expressible in the fragment are precisely the homomorphism-closed classes in \textsc{NP}. Similarly, we ask whether there exists a logic (in the sense of Gurevich~\cite{GroheLogicforP}) that captures precisely the homclosed classes of finite structures that are in the complexity class \textsc{P}. 

Another interesting avenue arises from the area of constraint satisfaction problems (CSPs). It seems advisable to systematically explore which (complements of) 
CSPs can be described as the homclosure of sentences in 
one of the well-behaved logics $\GFO$ and $\GNFO$ -- guaranteeing solvability in \textsc{P}, by Theorem~\ref{sorpresa}.
For a start, note that \Cref{thm:allhomcaptures} implies that the homclosures of $\GNFO{}$ and $\GFO{}$
sentences are expressible in guarded existential second-order logic;
this even holds projectively. Yet, by a recent result of a subset of
the authors~\cite{BKR}, if the complement of a CSP can be expressed in
this form, then the CSP is of the form CSP$({\mathfrak B})$ for an
$\omega$-categorical structure ${\mathfrak B}$. This implies that the
so-called universal-algebraic approach can be used to study complexity
and expressivity questions CSPs of said type \cite{Book}.

\section*{Acknowledgments}
We are indebted to the as meticulous as benevolent anonymous reviewers for their appreciation and numerous valuable suggestions for improvement. We are particularly grateful to Bartosz Bednarczyk for his diligent help in polishing the final version (pun intended).
 
Manuel Bodirsky has received funding from the European Research Council through the ERC Consolidator Grant 681988 (CSP-Infinity).
Thomas Feller and Sebastian Rudolph are supported by the European Research Council through the ERC Consolidator Grant 771779 (DeciGUT). 
Simon Knäuer is supported by DFG Graduiertenkolleg 1763 (QuantLA).

\bibliographystyle{abbrv}
\bibliography{references}

\clearpage


\clearpage
\appendix

\subsection{On the Different Notions of Homclosure}

At the first glance, the different notions $\homcl{\getmodels{\Phi}}$, $\homclfin{\getmodels{\Phi}}$, $\homcl{\getfinmodels{\Phi}}$, and $\homclfin{\getfinmodels{\Phi}}$ might seem abundant and give rise to the question, if some of these may coincide. We show that this is not the case.

From $\getfinmodels{\Phi} \subseteq \getmodels{\Phi}$ and $\homclfin{\mathcal{C}} \subseteq \homcl{\mathcal{C}}$, one can directly infer the following inclusions:
\begin{itemize}
	\item
	$\homclfin{\getfinmodels{\Phi}} \subseteq \homcl{\getfinmodels{\Phi}}$
	\item
	$\homclfin{\getfinmodels{\Phi}} \subseteq \homclfin{\getmodels{\Phi}}$
	\item
	$\homclfin{\getmodels{\Phi}} \subseteq \homcl{\getmodels{\Phi}}$
	\item
	$\homcl{\getfinmodels{\Phi}} \subseteq \homcl{\getmodels{\Phi}}$
	\item
	$\homclfin{\getfinmodels{\Phi}} \subseteq \homcl{\getmodels{\Phi}}$ (by transitivity)
\end{itemize}

We proceed to show that all these inclusions are strict and that
$\homcl{\getfinmodels{\Phi}}$ and $\homclfin{\getmodels{\Phi}}$ are indeed incomparable, even for one fixed $\Phi$.

Let $\Phi$ be the conjunction over the following sentences:

\begin{itemize}
	\item $\forall x.\big( x \not= \sig{b} \Rightarrow \exists y. \sigP(x,y) \big)$
	\item $\forall x.\big( x \not= \sig{a} \Leftrightarrow \exists y. \sigP(y,x) \big)$
	\item $\forall xyz.\big( \sigP(x,y) \wedge \sigP(x,z) \Rightarrow y = z \big)$
	\item $\forall xyz.\big( \sigP(x,z) \wedge \sigP(y,z) \Rightarrow x = y \big)$
\end{itemize}

We now define the following $\{\sig{a},\sig{b},\sigP\}$-structures:

\begin{itemize}
	\item
	$\mathfrak{A}_\text{$k$-path} = (\{0,\ldots,k\},0,k,succ)$
	\item
	$\mathfrak{A}_\text{$k$-path$^+$} = (\mathbb{N},0,k,succ)$
	\item
	$\mathfrak{A}_\text{$2$loops} = (\{0,1\},0,1,\{(0,0),(1,1)\})$
	\item
	$\mathfrak{A}_\text{$\infty$-gap} = (\mathbb{Z} \setminus \{0\},1,-1,succ)$
\end{itemize}

We obtain: 
$$
\begin{array}{l|llll}
	& \homclfin{\getfinmodels{\Phi}} & \homcl{\getfinmodels{\Phi}} &
	\homclfin{\getmodels{\Phi}} & \homcl{\getmodels{\Phi}} \\\hline
	\mathfrak{A}_\text{$k$-path}     & \in & \in & \in & \in \\ 
	\mathfrak{A}_\text{$k$-path$^+$} & \not\in & \in & \not\in & \in \\
	\mathfrak{A}_\text{$2$loops}     & \not\in & \not\in & \in & \in \\
	\mathfrak{A}_\text{$\infty$-gap} & \not\in & \not\in & \not\in & \in \\

\end{array}
$$

\subsection{Tilings}

We prove in this section  the undecidability  of the deterministic margin-constraint tiling problems. The proof is a slight modification of the first undecidability proof for origin-constraint tiling problems in \cite{wang1963proceedings} (see also \cite{Wang1990}). We include it here for the convenience of the reader. { The words ``domino'' and ``tile'' are used fully  synonymously in the following, i.e., a tile system  is the same as a domino system and a set of tiles  is the same as a set of dominoes.

}

\gridlemma*

%
%

\newcommand{\red}[1]{{\color{red}#1}}

\begin{proof}In order to prove Item 1 we follow the construction of Wang in \cite{Wang1990}.
	
	 He showed undecidability of the origin-constrained tiling problem by a reduction from the halting problem of deterministic Turing machines. For a given Turing machine, Wang defined a set of tiles such that the $i$th computation step of the machine is completely given by a tiling of the $i$th row of the $\mathbb{N}\times \mathbb{N}$ grid. Therefore, a tiling of the $\mathbb{N}\times \mathbb{N}$ grid can only be found if the Turing machine runs forever. 
	 
	 We show in the following how this can be done with a margin-constrained system of dominoes. Our construction relies, in addition to Wang's idea, on two modifications of the classical  set of tiles. First we shift the Turing machine tape content in each new computation step one position to the right. To this end, we define a set of dominoes $\mathcal{D}_1$ of new tiles which occur only in the beginning of each row; their number increases by one from row to row.
	 The second idea is that we store the tape content not once but three times. Each placed tile corresponds to a position on the tape and contains the tape content at this position, but furthermore it also contains the tape content of its left and right neighbors. This idea leads to the definition of the tile system $\mathcal{D}_2$.

	 For a binary relation $R$ on a set $D$  we denote by $R$ also {the} map $R \colon D \rightarrow 2^D$ defined by $R(x)= \{d\in D\mid (x,d)\in R\}$. It will always be clear from the context which of the two objects we mean. Furthermore, if the  set $R(x)$ contains only one element we drop the set brackets, i.e., we write $R(x)=d$ for $d\in D$.

	 Let $\mathcal{M}=\langle\{q_0,\ldots, q_n, q_f\},\{0,1\}, \{0,1\},\delta, q_0, q_f \rangle$ be a deterministic Turing machine where  $\{q_0,\ldots, q_n, q_f\}$ is the set of states, $\{0,1\}$ is the input and tape alphabet, and $\delta$ is a partial transition function  $\{q_0,\ldots, q_n\}\times \{0,1\} \rightarrow \{q_0,\ldots, q_n,q_f\}\times \{0,1\} \times \{L,R\}$.
	 
	 Let $D_1=\{o, l, b ,c,c', d\}$ be a set of dominoes such that $o$ is the origin tile. 
	 We define the relations $H_1$ and $V_1$ on  $D_1$ as follows:
	 \begin{itemize}[itemindent=0ex, leftmargin=3ex, itemsep=-0.7ex, topsep=0ex]
	 	\item $H_1(o)= b$, $V_1(o)= l$,
	 	\item $H_1(b)=b$, $V_1(b)=\{c',d\},$
	 \item $V_1(l)=l$, $H_1(l)=\{c,c'\},$
	 \item $H_1(c)=\{c,c'\}$, $V_1(c)=c,$
	 \item $H_1(c')=d$, $V_1(c')=c,$
	 \item $V_1(d)=c'.$
	
	 \end{itemize}
	
	 In the next step we define tiles that contain the information of a Turing machine tape position. Consider the set $T=\{ T^p_i \mid  p\in \{q_0,\ldots, q_n, q_f, e\}, i\in \{0,1\}  \}$.
	 We define the set of tiles $D_2$ as $ (T \uplus \{d\}) \times T \times T$. We denote the three coordinates of a tile $x \in D_2$ with $x_l, x_m$ and $x_r$. 
	We want to ensure in the following, that a tile $x\in D_2$ that is placed at the grid corresponds to a tape position of the Turing machine after a certain number of computation steps. The tape content is encoded in the main tape information $x_m=T^p_i$.  The index $i$ tells us which symbol is written on the tape at that certain position. If the head of the machine is at this position we get the state information by the superscript $p$. Otherwise $p=e$, which means the head is somewhere else. The coordinates $x_l$ and $x_r$ give exactly this kind of information for the left and right neighbors of the tape position.
	 
We give now the relation $H_2$ and $V_2$ for a tile $x\in D_2$.

\begin{equation}
H_{2}(x)=\{y\in D_2 \mid x_m=y_l \wedge x_r=y_m   \}
\end{equation}

	 For $V_2$  we make the following case distinction on the different types of the form $x=(T^{u}_{r},T^{v}_{s},T^{w}_{t})\in D_2$.
%
%
%
%
%
%
%

\newcommand{\eqn}[1]{
\begin{equation}
#1
\end{equation}
}
	 
\begin{enumerate}
	\item $u=q$ for $q\in \{q_1,\ldots,q_n\}$ and $\delta(q,r)=(q',i,R)$:
	\eqn{V_2(x)=\{ (y_l,T^{e}_{i},T^{q'}_{s})\in D_2 \mid y_l  \in T\cup \{d\}\}.}
	
	\item $v=q$ for $q\in \{q_1,\ldots,q_n\}$ and $\delta(q,s)=(q',i,R)$:
	\eqn{V_2(x)=\{ (y_l,T^{u}_{r}, T^{e}_{i})\in D_2 \mid y_l \in T\cup \{d\}\}.}
	
	\item $v=q$ for $q\in \{q_1,\ldots,q_n\}$ and $\delta(q,s)=(q',i,L)$:
	\eqn{V_2(x)=\{ (y_l,T^{q'}_{s},T^{e}_{i})\in D_2 \mid  y_l \in T\cup \{d\}\}.}
	
	\item $w=q$ for $q\in \{q_1,\ldots,q_n\}$ and $\delta(q,s)=(q',i,L)$:
	\eqn{V_{2}(x)=\{ (y_l,x_l,T^{q'}_{s})\in D_2 \mid y_l \in T\cup \{d\}\}.}
	\item In all the remaining cases we define
	\eqn{V_2(x)=\{y\in D_2 \mid  x_m=y_r  \}.}
\end{enumerate}

	 Now consider the set of tiles $D=D_1 \cup D_2$. We define the relations $H_3$ and $V_3$ as follows. For the tiles in $D_2$ that are of the form $ x=(d,T^{v}_{s},T^{w}_{t})$, where $v,w\in \{q_0,\ldots,q_n,e\}$ holds, we set
	 $V_3(x)=\{d\}$. Furthermore we define $H_3(d)=( d\times T^2)$.

	 In the last step of our construction we have to specify the second row of the grid. This row should correspond to the initial tape configuration of the Turing machine. Therefore we define a relation $V_4$ such that
\begin{equation}
V(b)=\{(d, T^{q_0}_0, T^e_0) , ( T^{q_0}_0, T^e_0, T^e_0), ( T^{e}_0, T^e_0, T^e_0) \}.
\end{equation}
 	 
	 We claim that the domino system $\mathcal{D}$ on the set $D$ with the relations $B=\{b,o\}$, $L=\{l,o\}$, $H=H_1\cup H_2\cup H_3$ and $V=V_1\cup V_2\cup V_3\cup V_4$ is a deterministic margin constraint domino system which tiles the $\mathbb{N}\times \mathbb{N}$ grid if and only if the Turing machine $\mathcal{M}$ does not halt.
	 
	 In order to prove this we give the following description of the tile system $\mathcal{D}$.
	 
	\underline{Claim 1:} $\mathcal{M}$ runs at least $k$ steps if and only if $\mathcal{D}$ tiles at least the first $k{+1}$ rows. In this case the $i{+1}$th row, with $1\leq i\leq k$ has the following description from left to right:
	\begin{itemize}[itemindent=0ex, leftmargin=3ex, itemsep=-0.5ex, topsep=0ex]
		\item the margin symbol $l$,
		\item $i-1$ times the symbol $c$,
		\item one time the symbol $c'$,
			\item one time the symbol $d$,
			\item an infinite sequence of tiles from $D_2$ that store the Turing machine configuration after computation step $i-1$.
	\end{itemize}

It is clear that this claim holds for $k=1$, by our encoding of the initial machine configuration in the second row (see definition of $V_4$). Now assume the statement holds for $k-1$. It is easy to see from the definition of $H_1$ and $V_1$ that the $c$-block is enlarged by one $c$ from one row to the next.
Note that our definition  of $V_2$ ensures the last item in the claim. The main tape information of a tile gets the left neighbor information of its upper neighbor. This is due to the shifting of the tape. Furthermore one can check  that the state transition and head position change in each computation step are also encoded
in the definition of $V_2$ (see items 1 to 4). Note that this implies  also that a tile where the main information consists of a final state does not have an upper neighbor. This is since there exists no transition in a deterministic Turing machine from a final state to another state.
	 This proves Claim 1.

	 The domino system $\mathcal{D}$ satisfies clearly the items 1 to 3 in \Cref{def:tiling}. Let $d_1$ and $d_2$ be tiles. It is clear that item 4 is satisfied whenever $d_1\in D_1\setminus \{d\}$ or $d_2\in D_1$ holds. Suppose therefore $d_1\in D_2\cup \{d\}$ and $d_2\in D_2$. Note that for a tile $x\in H(d_1)$ the coordinates $x_l$ and $x_m$ are uniquely determined by $d_1$.
	 The coordinate $x_l$ corresponds as a tape position to $(d_2)_m$. The tape information at this position after the next computation step is  determined by its two neighbors $(d_2)_l$ and $(d_2)_r$. The uniqueness of this transition can be checked in the definition of $V_2$.
	 
	 This proves that $\mathcal{D}$ is deterministic. It follows from Claim 1 that this system tiles the grid whenever the machine $\mathcal{M}$ does not halt on the empty input. On the other hand, if the machine halts, then there exists a maximal $k$ such that the first $k$ rows can be tiled. This proves the first item of the lemma.

	\bigskip
	
Item 2 is also proved by a reduction from the halting problem for Turing machines. We give a reduction such that a Turing machine halts on the empty tape whenever the corresponding tile system admits an ultimately periodic tiling of the grid. On the other hand if the Turing machine runs forever, the defined  set of tiles does not allow for an ultimately periodic tiling.

Note that our reduction from Item 1 satisfies already the second implication. If  the Turing machine $\mathcal{M}$ runs forever the tile system $\mathcal{D}$ tiles the grid by a tiling $t$.  By the definition of $\mathcal{D}$ and Claim 1 above the $i$th row starts with exactly $i-1$ times the symbol $c$ after one $l$. This means that for every $\ell_\mathrm{init},\ell_\mathrm{period}\in \mathbb{N}$ there exists $k\in \mathbb{N}$ such that   $t(k,\ell_\mathrm{init})\not= c$ and $ t(k,\ell_\mathrm{init}+\ell_\mathrm{period})=c$. Therefore the tiling is not ultimately periodic. Moreover, since the used margin-constrained tile systems is deterministic, this is the unique tiling of the grid by this  set of tiles, which proves that $\mathcal{D}$ does not admit an ultimately periodic tiling.

We modify the tiling from Item 1 in the following such that $\mathcal{D}$ tiles the grid ultimately periodically whenever the Turing machine $\mathcal{M}$ halts. If the Turing machine halts it is in the final state. Also, we can assume without loss of generality that the reading head will then be at the leftmost tape position. 
Let $f$ and $e$ be new tiles and let $D'$ be the union of the set of dominoes from Item 1 with $\{f,e\}$. We define new relations $V_5$ and $H_4$ as follows.

\begin{enumerate}
	\item Let $x\in D_2$. If  $x$ is of the form $ x=(d,T^{q_f}_{s},T^{w}_{t})$  where $q_f$ is the final state we define	$V_5(x)=f$. 
\\Otherwise $V_5(x)=e$ .
	\item $V_5(f)=f$ and  $V_5(e)=e$.
	\item $H_4(c)=f $, $H_4(c')=f$, $H_4(f)=e$ and $H_4(e)=e$.
\end{enumerate}
One can  check that the new tile system $\mathcal{D}'$ on the  set of tiles $D'$ with the relations $B'=B$, $L'=L$, $H'=H\cup H_4$ and $V'=V\cup V_5$ is still a deterministic margin-constrained domino system: Consider tiles $d_1$ and $d_2$ and assume first that $d_1\in D$ and $d_2\in D$ hold, then the only possibilities for $d_1$ that have to be checked are $d_1=c$ and $d_1=c'$. The definition of $V_5$ involves only  old elements from $D_2$ (Item 1 above). Therefore $d_2$ can be assumed to be from $D_2$. It is easy to see that in this cases $|H'(d_1)\cap  V'(d_2)|\leq 1$ holds.
The remaining cases are $d_1\in \{e,f\}$ or $d_2\in \{e,f\}$. In these cases $|H'(d_1)\cap  V'(d_2)|\leq 1$ holds trivially since $e$ and $f$ have unique upper and right neighbors.

We characterize now how a tiling with the system $\mathcal{D}'$ looks like if the Turing machine halts.

		\underline{Claim 2:} If $\mathcal{M}$ runs  $k$ steps and then halts, then $\mathcal{D}$ tiles the grid such that for every $i\geq k+3$ the $i$th row has the following description from left to right
	\begin{itemize}[itemindent=0ex, leftmargin=3ex, itemsep=-0.5ex, topsep=0ex]
		\item The margin symbol $l$,
		\item $k$ times the symbol $c$,
		\item one time the symbol $f$,
		\item an infinite sequence of tiles $e$.
	\end{itemize}
	
Note that we have by Claim 1 a description of the  $k$th row. Since the Turing machine halts after $k$ steps the first symbol in this row after the symbol $d$ is of the kind $ x=(d,T^{q_f}_{s},T^{w}_{t})$  where $q_f$ is the final state. Therefore the unique tiling of the $k+2$th row is as follows:

		\begin{itemize}[itemindent=0ex, leftmargin=3ex, itemsep=-0.5ex, topsep=0ex]
		\item the margin symbol $l$,
		\item $k$ times the tile $c$,
		\item one time the tile $c'$,
			\item one time the tile $f$, and
		\item an infinite sequence of tiles $e$.
	\end{itemize}
	From this description it is easy to compute the unique tilings of the $k+3$th and $k+4$th row. Since both are identical and by Item 4 in \Cref{def:tiling} it follows that all further rows admit also tilings that are identical to the $k+3$th one.	This proves Claim 2.	
	
{
The tiling described in Claim 2 is ultimately periodic. For the vertical period this follows directly from Claim 2.  For the horizontal period node that a Turing machine can in $k$ computation step reach at most $k+1$ different tape positions. This implies that in the rows $2$ to $k+2$ the tiles right of the $2(k+1)$th column encode the empty tape. Since Claim 2 ensures the horizontal period for the rows greater than $k+1$ the tiling is ultimately periodic.}

  Note also that our modifications to achieve this do not affect the property of a (or the) tiling being not ultimately periodic if the Turing machine runs forever. 
	This finishes the proof of the lemma. \end{proof}

\subsection{Proofs of Tools Section}

\translemma*

\begin{proof}

	Note that the negation normal form for a formula $\Phi$ is equivalent to $\Phi$ itself.
	Then, by the definition of  $\trans^n$ it is enough to prove the statement for sentences $\Phi$ in $\mathrm{NNF}$. In order to do this, we prove the
	 following statement. For every $\FOe$ formula $\Phi(\bold{x})$ in $\mathrm{NNF}$, every $n$-labelled constant-sole structure 	$(\mathfrak{A},\lambda)$, and every variable assignment $\beta \colon V\rightarrow \mathfrak{A}^\lambda$ it holds that 
\begin{equation}
( \mathfrak{A}^\lambda, \beta) \models \Phi(\bold{x})  \text{~ iff~} 
(\mathfrak{A}_\lambda, \beta_1) \models \trans^n_{\beta_2}( \Phi(\bold{x})),
\end{equation}
	where $\beta_1 \colon V \to A$ and $\beta_2 \colon V \to \{1,\ldots,n\}$ are defined such that $\beta(x)=(\beta_1(x),\beta_2(x))$ for all $x \in V$.
	We denote by $\beta_2^+$ the extensions of the assignment $\beta_2$ to all constants $c$ such that $\beta_2^+(c)=1$  (cf. definition of  $\trans^n_\ass$). Furthermore $\beta_1^+$ is the  extension of $\beta_1$ to all constants $c$ such that $\beta_1^+(c)=c^{\mathfrak{A}_\lambda} $. The map $\beta^+$ is now defined  by $\beta^+ (c)=(\beta^+_1(c),\beta^+_2(c))$. Note that this is compatible with the definition of the constants in the  constant-sole structure $ \mathfrak{A}^\lambda$ since $c^{\mathfrak{A}^\lambda}= (c^{\mathfrak{A}},1)$ holds.	 
	
	We prove the statement by structural induction over the set of formulae.
For  literals of the form $\sig{P}(\bold{t})$ or $\neg\sig{P}(\bold{t})$ the statement follows immediately from the definition of $\trans^n_\ass$ since the projection map from $\mathfrak{A}^\lambda$ to   ${\mathfrak{A}_\lambda}|_\tau$  is a strong surjective homomorphism. Therefore the only cases that we have to check are formulae of the form $t_1=t_2$ and $t_1\not =t_2$ for terms $t_1$ and $t_2$. 
	
	Suppose that $\beta$ is any variable assignment such that  $ (\mathfrak{A}^\lambda, \beta) \models (t_1=t_2)$. Then $\beta_1^+(t_1)=\beta^+_1(t_2)$ and $\beta^+_2(t_1)=\beta^+_2(t_2)$. This implies together with (\ref{eq: 5}) that $(\mathfrak{A}_\lambda, \beta_1) \models \trans^n_{\beta_2}( t_1=t_2)$ holds. For the other direction observe that if $(\mathfrak{A}_\lambda, \beta_1) \models \trans^n_{\beta_2}( t_1=t_2)$  holds, then $\trans^n_{\beta_2}( t_1=t_2)$ is exactly the formula  $t_1=t_2$ and for the assignment $\beta$ it holds that  $\beta^+_2(t_1)=\beta^+_2(t_2)$ by (5).  Furthermore we get in this case that  $\beta^+_1(t_1)=\beta^+_1(t_2)$ must hold and therefore we get $( \mathfrak{A}^\lambda, \beta) \models (t_1=t_2)$.
	This concludes the second direction.
	
	For the case of  a literal of the form $t_1\not = t_2$ suppose that $\beta$ is a variable assignment such that  $ ( \mathfrak{A}^\lambda, \beta) \models (t_1\not =t_2)$ holds. This implies that $\beta_1^+(t_1)\not=\beta_1^+(t_2)$ or $\beta_2^+(t_1) \not=\beta_2^+(t_2)$ holds.
	For the first possibility it follows by (\ref{eq: 6}) that $(\mathfrak{A}_\lambda, \beta_1) \models \trans^n_{\beta_2}( t_1\not =t_2)$  holds. In the case that  $\beta^+_2(t_1) \not=\beta^+_2(t_2)$ holds, we get by (\ref{eq: 6}) that $\trans^n_{\beta_2}( t_1\not =t_2)$ is precisely $\top$ and therefore again the implication is valid. 
	For the other direction assume that  $(\mathfrak{A}_\lambda, \beta_1) \models \trans^n_{\beta_2}( t_1\not =t_2)$  holds. Then either $\trans^n_{\beta_2}( t_1\not =t_2)$ is equal to $\top$ or it is equal to $t_1\not = t_2$ and $\beta^+_1(t_1)\not=\beta^+_1(t_2)$. In both cases at least one of $\beta_i$ satisfies $\beta^+_i(t_1)\not =\beta^+_i(t_2)$ and therefore $\beta^+(t_1)\not =\beta^+(t_2) $ holds. This proves the claim.
	
	For the inductive step it is straightforward to see that statement holds for all formulae of the form $\Phi(\bold{x}) \wedge \Psi(\bold{x})$ and $\Phi(\bold{x}) \vee \Psi(\bold{x})$ 
	whenever it holds for $\Phi(\bold{x})$ and $\Psi(\bold{x}) $.
	
	Therefore it remains to prove the statement for formulae of the form $\exists x .\Phi(x,\bold{y})$ and   $\forall x .\Phi(x,\bold{y})$. Let $\beta$ be an arbitrary variable assignment such that $ ( \mathfrak{A}^\lambda, \beta) \models \exists x .\Phi(x,\bold{y})$ holds. This is true if and only if there exists a variable assignment $\beta'= \beta \cup \{x\mapsto (a,j )\} $, where $(a,i ) $ is a domain element of $\mathfrak{A}^\lambda$ such that $ ( \mathfrak{A}^\lambda, \beta') \models  \Phi(x,\bold{y})$ holds.  By the inductive assumption we get that this is the case if and only if   $(\mathfrak{A}_\lambda, \beta'_1) \models \trans^n_{\beta'_2}(\Phi(x,\bold{y}))$ holds. 
	Since $ [\lambda(a) \geq j]$ holds clearly 
	in $\mathfrak{A}_\lambda$ this is equivalent to
	the validity of
	$ (\mathfrak{A}_\lambda, \beta'_1) \models  [\lambda(x) \geq j] \wedge \trans^n_{\beta'_2} (\Phi(x,\bold{y}))$ for $\beta'= \beta \cup \{x\mapsto (a,j )\} $.
	Note that such a $\beta$ can be found  if and only if $ (\mathfrak{A}_\lambda, \beta_1) \models  \exists x.[\lambda(x) \geq j] \wedge \trans^n_{\beta'_2} (\Phi(x,\bold{y}))$ holds for some $j$.  By Definition (\ref{eq: 11})  this means that	 $ (\mathfrak{A}_\lambda, \beta_1) \models  \trans^n_{\beta_2}(\exists x .\Phi(x,\bold{y}))$  holds.
	
	The statement for formulae of the form $\forall x .\Phi(x,\bold{y})$ can be shown by an analogous argument.\end{proof}

\subsection{Proofs of Homclosure Membership Section}

\begin{lemma}\label{lemma:projchar PHIintext}
	Let $\Phi$ be a $\tau$-sentence and let $\mathfrak{A}$ be a finite $\tau$-structure as in \Cref{def:exincol}. Then the following holds:
	\begin{enumerate}
		\item $\Phi^\mathrm{ext}_\mathfrak{A}$ and $\Phi^\mathrm{int}_\mathfrak{A}$ are of polynomial size wrt.\ $|\mathfrak{A}|$ and the size of $\Phi$. 
		\item $\Omega_\mathfrak{A}$ projectively characterizes the class of (finite) $\tau$-structures having a homomorphism into $\mathfrak{A}$.
		\item\label{item:projecht} $\Phi^\mathrm{ext}_\mathfrak{A}$ projectively characterizes the class of (finite) models of $\Phi$  having a homomorphism into $\mathfrak{A}$.
	\end{enumerate}
\end{lemma}

\begin{proof}The first statement follows directly from the definitions of  $\Phi^\mathrm{ext}_\mathfrak{A}$ and $\Phi^\mathrm{int}_\mathfrak{A}$.
	
	For (2), let $\mathfrak{B}$ be a $\tau{\cup}\rho$-structure that satisfies $\Omega_\mathfrak{A}$. Let $\lambda$ be the $n$-labeling of $\mathfrak{B}$ given by the $\rho$-reduct of $\mathfrak{B}$. The map $\lambda$ is a homomorphism from the $\tau$-reduct  $\mathfrak{B}'$ of $\mathfrak{B}$ to $\mathfrak{A}$.
	 To see this assume that $\bold{b}\in  \sig{P}^{\mathfrak{B}'}$ holds for a tuple $\bold{b}$ from $B$. By the definition of $\Omega_\mathfrak{A}$	we get that  $\bold{b} \in (\chi^{\sigP}_\mathfrak{A})^\mathfrak{B}$ holds.
	 This implies that $\lambda(\bold{b}) \in  \sig{P}^\mathfrak{A}$ holds and therefore $\lambda$ is a homomorphism. 
	For the other inclusion assume the $h\colon\mathfrak{B} \rightarrow \mathfrak{A}$ is a homomorphism. It is easy to check that the implicit representation $\mathfrak{B}_h$ of the $n$-labeled structure $(\mathfrak{B},h)$ satisfies $\Omega_\mathfrak{A}$ and is an expansion of $\mathfrak{B}$. This proves the claim.
	
	Claim (3) follows immediately from the definition of $\Phi^\mathrm{ext}_\mathfrak{A}$ and the previous claim.\end{proof}

\begin{lemma}\label{lemma:phi int}
Let $\mathfrak{A}$ and $\rho$ as in \Cref{def:exincol} and let $\mathfrak{B}$ be an arbitrary $\tau$-structure such that $h\colon\mathfrak{B}\rightarrow \mathfrak{A}$ is a homomorphism.  Then it holds for every $\tau$-formula $\Phi$ and all assignments $\beta$ that 
$$ (\mathfrak{B},\beta) \models   \Phi(\bold{t}) \text{~~if and only if~~} (\mathfrak{B}_h,\beta)\models \Phi^\mathrm{int}_\mathfrak{A}(\bold{t}),$$
where $\mathfrak{B}_h$ is the implicit representation of the $n$-labeled structure $(\mathfrak{B},h)$.\end{lemma}

\begin{proof}We prove the statement by structural induction over the set of formulae.  For the induction beginning let $\Phi(\bold{t})$ be an atomic formula $\sig{P}(\bold{t})$ and let $\beta$ be a term assignment such that $ (\mathfrak{B},\beta) \models   \Phi(\bold{t}) $ holds. Since $h$ is a homomorphism it follows that $(\mathfrak{A},h\circ \beta)\models \sig{P}(\bold{t})$ and therefore also $(\mathfrak{B}_h,\beta)\models \chi^{\sigP}_\mathfrak{A}(\bold{t})$ holds. The second direction of the equivalence follows directly from the definition of $\chi^{\sigP}_\mathfrak{A}$ and since $\mathfrak{B}_h$ is an expansion of $\mathfrak{B}$.
	
	For the induction step let $\Phi(\bold{t})$ be of the form $\neg \Psi(\bold{t})$ such that the statement already holds for $\Psi$. Note that then clearly $\Phi^\mathrm{int}_\mathfrak{A}(\bold{t})$ is equal to $\neg \Psi^\mathrm{int}_\mathfrak{A}(\bold{t})$ and we get:
	\begin{align*}
(\mathfrak{B},\beta) \models   \Phi(\bold{t}) &\hspace{4ex}\text{iff} \hspace{-5ex}&& (\mathfrak{B},\beta) \not\models   \Psi(\bold{t}) \\
&\hspace{4ex}\text{iff} \hspace{-5ex}&& (\mathfrak{B}_h,\beta) \not\models   \Psi^\mathrm{int}_\mathfrak{A}(\bold{t})\\
&\hspace{4ex}\text{iff} \hspace{-5ex}&& (\mathfrak{B}_h,\beta)\models \Phi^\mathrm{int}_\mathfrak{A}(\bold{t}).
	\end{align*}
Assume that $\Phi(\bold{t})$ is of the from $\Psi(\bold{t}) \vee \Gamma(\bold{t})$ for formulae $\Psi(\bold{t})$ and $\Gamma(\bold{t})$ for which the statement holds. Since $\Phi^\mathrm{int}_\mathfrak{A}(\bold{t})$ is by its definition of the form $\Psi^\mathrm{int}_\mathfrak{A}(\bold{t}) \vee \Gamma^\mathrm{int}_\mathfrak{A}(\bold{t})$ we get
	\begin{align*}
	(\mathfrak{B},\beta) \models   \Phi(\bold{t}) &&&
	\\
	&\hspace{-10ex}\text{iff} \hspace{0ex}&& \hspace{-10ex}(\mathfrak{B},\beta) \models  \Psi(\bold{t}) \vee \Gamma(\bold{t}) \\
&\hspace{-10ex}\text{iff} \hspace{0ex}&& \hspace{-10ex}(\mathfrak{B},\beta) \models   \Psi(\bold{t})  \text{~~or~~}  (\mathfrak{B},\beta) \models \Gamma(\bold{t}) \\
&\hspace{-10ex}\text{iff} \hspace{0ex}&& \hspace{-10ex}(\mathfrak{B}_h,\beta) \models   \Psi^\mathrm{int}_\mathfrak{A}(\bold{t})\text{~~or~~}(\mathfrak{B}_h,\beta) \models \Gamma^\mathrm{int}_\mathfrak{A}(\bold{t}) \\
	&\hspace{-10ex}\text{iff} \hspace{0ex}&& \hspace{-10ex} (\mathfrak{B}_h,\beta)\models \Psi^\mathrm{int}_\mathfrak{A}(\bold{t}) \vee \Gamma^\mathrm{int}_\mathfrak{A}(\bold{t})\\
		&\hspace{-10ex}\text{iff} \hspace{0ex}&& \hspace{-10ex} (\mathfrak{B}_h,\beta)\models \Phi^\mathrm{int}_\mathfrak{A}(\bold{t}) .
\end{align*}
As in the last case assume that $\Phi(\bold{t})$ is of the form $\exists t_1 \Psi(\bold{t})$ for formulae $\Psi(\bold{t})$ for which the statement holds.  Then $\Phi^\mathrm{int}_\mathfrak{A}(\bold{t})$ is by its definition of the form $\exists t_1\Psi^\mathrm{int}_\mathfrak{A}(\bold{t})$. 
We get the following:
\begin{align*}
	(\mathfrak{B},\beta) \models   \Phi(\bold{t}) &&&\\
&\hspace{-12ex}\text{iff} \hspace{0ex}&&\hspace{-12ex} \exists b\in B \text{~such that~}  (\mathfrak{B},\beta\cup\{t_1\mapsto b\}) \models   \Psi(\bold{t})\\
&\hspace{-12ex}\text{iff} \hspace{0ex}&&\hspace{-12ex} \exists b\in B \text{~such that~}  (\mathfrak{B}_h,\beta\cup\{t_1\mapsto b\}) \models  \Psi^\mathrm{int}_\mathfrak{A}(\bold{t})\\
	&\hspace{-12ex}\text{iff} \hspace{0ex}&&\hspace{-12ex} (\mathfrak{B}_h,\beta)\models \Phi^\mathrm{int}_\mathfrak{A}(\bold{t}).
\end{align*}
This concludes the proof of the lemma.
\end{proof}

\lemmacoloring*

\begin{proof}
	By \Cref{lemma:projchar PHIintext} Item \ref{item:projecht}, the sentence $\Phi^\mathrm{ext}_\mathfrak{A}$ projectively characterizes the class of (finite) models of $\Phi$  having a homomorphism into $\mathfrak{A}$. Therefore the equivalence of 1) and 3) follows immediately.
	
The implication from 1) to 2) follows from \Cref{lemma:phi int}, 
since for every (finite) model $\mathfrak{B}$ of $\Phi$  admitting a homomorphism $h$ into $\mathfrak{A}$, the structure $\mathfrak{B}_\lambda$  with $\lambda=h$ is a (finite) model of $\Phi^\mathrm{int}_\mathfrak{A}$.

For the implication from 2) to 1) assume that $\mathfrak{B}$ is a model of $\Phi^\mathrm{int}_\mathfrak{A}$ and let $\lambda$ be the $n$-labeling of $\mathfrak{B}$ given by the $\rho$-reduct of $\mathfrak{B}$.
Let $\mathfrak{C}$ be the $\tau$-structure
obtained by setting $C=B$, $\sigc{}^\mathfrak{C}=\sigc{}^\mathfrak{B}$ for all constants, and $\sigP{}^\mathfrak{C} = (\chi^{\sigP}_\mathfrak{A})^\mathfrak{B}$.
It follows from the definition of  $\chi^{\sigP}_\mathfrak{A}(\bold{t})$  that $\lambda$ is a homomorphism from $\mathfrak{C}$ to $\mathfrak{A}$. Furthermore the definition of $\mathfrak{C}$ implies that for all $\sig{P}\in \tau$ and all $\bold{b}$ from $B$ it holds 
$$ \bold{b} \in {(\chi^{\sigP}_\mathfrak{A})}^{\mathfrak{C}_\lambda}
 \text{~~~if and only if~~~} \bold{b} \in{(\chi^{\sigP}_\mathfrak{A})}^\mathfrak{B}. $$
Since $\mathfrak{B}$ satisfies $\Phi^\mathrm{int}_\mathfrak{A}$ it follows that
$\mathfrak{C}_\lambda$ does as well. Now we apply \Cref{lemma:phi int} and get that $\mathfrak{C}$ satisfies $\Phi$ which proves that 1) holds.\end{proof}

\thearbsameasfin*
\begin{proof}
In order to show $\homclfin{\getmodels{\Psi}}\subseteq\homclfin{\getfinmodels{\Psi}}$ let $\mathfrak{B}\in \homclfin{\getmodels{\Psi}}$.  Then
  $\Psi^\mathrm{ext}_\mathfrak{B}$ and  $\Psi^\mathrm{int}_\mathfrak{B}$ have by \Cref{lemma:coloring} models. By the assumption $\Logic^\mathrm{ext}$ or $\Logic^\mathrm{int}$ have the FMP which means that $\Phi^\mathrm{ext}_\mathfrak{B}$ or $\Phi^\mathrm{int}_\mathfrak{B}$ are finitely satisfiable. Again by \Cref{lemma:coloring}, we get that $\Phi$ has a  finite model admitting a homomorphism into $\mathfrak{B}$ which shows that 
$ \mathfrak{B}\in \homclfin{\getfinmodels{\Psi}}$ holds.
\end{proof}

\firstmain*

\begin{proof}
	First consider \EpFO{}, noting that every such sentence is already homclosed. Thus $\mathrm{InHomCl}_\mathrm{(fin)}$ is equivalent to checking modelhood for the given $\mathfrak{A}$. This can be done in nondeterministic time by guessing an instantiation of the existentially quantified variables and checking satisfaction. \textsc{NP}-hardness follows from a reduction from the graph homomorphism problem (even for disjunction-free sentences). 	
	We now consider the other fragments. Extrinsic coloring for \GNFO{} and \TGF{} 
	and intrinsic coloring for the others establishes a polynomial reduction to the corresponding satisfiability problems via \Cref{lemma:coloring}. Hardness follows by a reduction from satisfiability via \Cref{prop:SatToInHomCl} in all cases (see \Cref{fig:main} for references for the complexity of the satisfiability problem in each case). 
\end{proof}

\defgridformula*

\begin{lemma}[restate=gridformula, label=gridformula, name=]
	\begin{enumerate}[itemindent=0ex, leftmargin=4ex, itemsep=-0.5ex]
		\item \label{grid1}
		Every model $\mathfrak{B}$  of  $\Phi_\mathrm{grid}$ is a  $\{\sigH{},\sigV{}\}$-structures such that for every domain element
		$b \in B$ there is a homomorphism $h: \mathfrak{A}_{\mathbb{N} \times \mathbb{N}} \to \mathfrak{B}$ with $h((0,0))=b$. 
		\item \label{grid2}
		Every model of $\Phi_\mathrm{grid}$ has a homomorphism into a structure that is not a model of $\Phi_\mathrm{grid}$.
		\item \label{grid3}
		Every homomorphism from the $\mathbb{N} \times \mathbb{N}$ grid into $\mathfrak{A}_\mathcal{D}$ corresponds to a $\mathcal{D}$-tiling.
		\item \label{grid4}
		Every homomorphism from some model of $\Phi_\mathrm{grid}$ into $\mathfrak{A}_\mathcal{D}$ gives rise to a $\mathcal{D}$-tiling.
		\item \label{grid5}
		Every homomorphism from some finite model of $\Phi_\mathrm{grid}$ into $\mathfrak{A}_\mathcal{D}$ gives rise to an ultimately periodic $\mathcal{D}$-tiling.
		\item \label{grid6}
		Both membership in $\homcl{\getmodels{\Phi_\mathrm{grid}}}$ and $\homcl{\getfinmodels{\Phi_\mathrm{grid}}}$ are undecidable problems.   	
	\end{enumerate}
\end{lemma}

\begin{proof} For (\ref{grid1}), note that an easy inductive argument under the usage of the extension properties defined in $\Phi_\mathrm{grid}$ proves the claim.
	
	For (\ref{grid2}), simply choose $\mathfrak{F} \uplus \mathfrak{I}$ as target structure.
	
	(\ref{grid3}) follows immediately from the definition of a tiling and from the definition of the structure $\mathfrak{A}_\mathcal{D}$ .
	
	For (\ref{grid4}), let $\mathfrak{B}$ be a model of  $\Phi_\mathrm{grid}$  that has a homomorphism to 
	$\mathfrak{A}_\mathcal{D}$. By (\ref{grid1}) we get that  $\mathfrak{A}_{\mathbb{N} \times \mathbb{N}} $ also has a homomorphism to $\mathfrak{B}$ and therefore also to $\mathfrak{A}_\mathcal{D}$. (\ref{grid3}) implies that this gives rise to a  $\mathcal{D}$-tiling.
	
	For (\ref{grid5}), assume that $\mathfrak{B}$ is a finite model of  $\Phi_\mathrm{grid}$ with a homomorphism $t$ into $\mathfrak{A}_\mathcal{D}$. Let $v\colon B\rightarrow B$ and $h\colon B\rightarrow B$ to functions such that $(x,h(x))\in\sigH{}^\mathfrak{B}$ and $(x,v(x))\in\sigV{}^\mathfrak{B}$ hold. It is clear that such functions exist since $\mathfrak{B}$ satisfies $\Phi_\mathrm{grid}$.
	We define a map $f$ from $\mathbb{N} \times \mathbb{N} $ to $B$ as follows.   For an arbitrary element $o\in B$ we set $f(0,0)=o$. In the next step we define inductively the map on the  elements $W_0=\{ (x,0) \mid x\in \mathbb{N} \}$. This is done by $f(i+1,0)=h(f(i,0))$. Since $h$ is a  function this is well defined.
	Now we can assume that the map is defined on the elements $W_i=\{ (x,i) \mid x\in \mathbb{N} \}$ and we extend it to $W_{i+1}=\{ (x,i+1) \mid x\in \mathbb{N} \}$ by $f(j,i+1)= v(f(j,i+1))$.
	Note that  by the use of the function $h$ it is clear that after a finite initial sequence of length at most $|B|$, the mapping $f$ is periodic on the first row $W_0$  with respect to the assigned elements from $B$. Since the definition of $f$ at the $n$th column of the grid only depends on the value $f(n,0)$ the columns
	are in the same way periodic as $W_0$ is. 
	Furthermore the graph of the map $v$ contains finitely many cycles of finite length. Let $q$ be the greatest common multiple of these lengths. It is clear that for every element $x\in B$ that lies in a cycle applying $q$ times the function $v$ on $x$ results in the element $x$ itself. This shows that $f$ also has a repetition on the rows in the grid in its definition after a possible initial sequence of at most $|B|$ elements.
	
	What we observed about the function $f$ implies immediately that if $t\circ f$ defines a tiling of the grid by tiles of $\mathcal{D}$ then this tiling is ultimately periodic.
	
	In order to prove that $t\circ f$  defines a tiling it is enough according to (\ref{grid3}) to show that $f$ is a homomorphism from $\mathfrak{A}_{\mathbb{N} \times \mathbb{N}}$ to $ \mathfrak{B}$.

	Consider the substructure  $\mathfrak{W}_0$ of  $\mathfrak{A}_{\mathbb{N} \times \mathbb{N}}$  induced by $W_0$.

	The map from $\mathfrak{W}_0$ to $\mathfrak{B}$  induced by $f$ is a homomorphism, since 
	$(f(i,0),f(i+1,0) )= (f(i,0), h(f(i,0))) \in  \sigH{}^{\mathfrak{B}} $ holds for all
	$((i,0),(i+1,0)) \in \sigH{}^{\mathfrak{W}_0}$.
	
	Now we prove by induction that $f$ is a homomorphism on all structure $\mathfrak{W}_{i+1}$ which are induced by the sets $W_{i+1}$. It follows immediately by the choice of $v$ that all tuples of the form $ ((j,i),(j,i+1)) \in \sigV{}^{\mathfrak{W}_{i+1}}$ for arbitrary $j\in\mathbb{N}$ are preserved by $f$. 
	It remains to show the preservation of tuples of the form $((j,i+1),(j+1,i+1)) \in \sigH{}^{\mathfrak{W}_{i+1}}$. Consider the elements $f(j,i+1),f(j+1,i+1), f(j,i)$ and $f(j+1,j)$ 
	in $B$. By the induction hypothesis and the observation from before we get that $\sigH(f(j,i), f(j+1,i)) $,  $\sigV(f(j,i), f(j,i+1)   )$ and $\sigV(f(j+1,i), f(j+1,i+1))$ hold in $\mathfrak{B}$. Since $\mathfrak{B}$ satisfies the sentence $\Phi_\mathrm{grid}$ it follows that also $\sigH(f(j,i+1),f(j+1,i+1))$ holds in $\mathfrak{B}$. This implies that $f$ induces a homomorphism from $\mathfrak{W}_{i+1}$ to $\mathfrak{B}$ and therefore concludes the proof that $f$ is a homomorphism.

	For (\ref{grid6}), we prove the statements by reductions from the  tiling problem and the ultimately periodic tiling problem.
    The undecidability of  these problems was shown in \cite{berger1966undecidability} and in	\cite{remarksonberger}, respectively.

	The reductions map an instance of the tiling problem, that is a  set of tiles $\mathcal{D}$ to the structure $\mathfrak{A}_\mathcal{D}$.
	
	To see that this is a correct reduction to the membership problem of $\homcl{\getmodels{\Phi_\mathrm{grid}}}$ note that $\mathcal{D}$ tiles by (\ref{grid3}) the grid if and only if   $\mathfrak{A}_{\mathbb{N} \times \mathbb{N}} $ has a homomorphism to $\mathfrak{A}_\mathcal{D}$. By (\ref{grid1})  and since  $\mathfrak{A}_{\mathbb{N} \times \mathbb{N}} $ is a model of  $\Phi_\mathrm{grid}$ this holds if and only if $\mathfrak{A}_\mathcal{D} \in \homcl{\getmodels{\Phi_\mathrm{grid}}}$. This proves the undecidability of the membership problem for $ \homcl{\getmodels{\Phi_\mathrm{grid}}}$.
	
	For the undecidability of the membership problem for $\homcl{\getfinmodels{\Phi_\mathrm{grid}}}$ we use the same mapping from  sets of tiles to structures.
	Assume that a set of tiles $\mathcal{D}$ admits an ultimately periodic tiling  $t$ of the grid. Let
	$k_\mathrm{init},k_\mathrm{period},\ell_\mathrm{init},\ell_\mathrm{period}$ from the definition of ultimately periodic tiling. Let $\mathfrak{D}'$ be the finite substructure of $\mathfrak{A}_{\mathbb{N} \times \mathbb{N}} $ that is induced by elements of the from $(x,y)$ where $x\leq  k_\mathrm{init}+k_\mathrm{period}$ and $y\leq  l_\mathrm{init}+l_\mathrm{period}$ hold.
	We add new tuples to the relations $\sigH{}^{\mathfrak{D}'}$ and $\sigV{}^{\mathfrak{D}'}$  from $\mathfrak{D}'$ and define in this way a new structure $\mathfrak{D}$.
	For all elements of the form $(k_\mathrm{init}+k_\mathrm{period},y)  $ we add the tuple 
	$((k_\mathrm{init}+k_\mathrm{period},y)  ,(k_\mathrm{init}+1,y))$ to the relation $\sigH{}^{\mathfrak{D}'}$. 
	Analogous we add for all elements of the form $(x,l_\mathrm{init}+l_\mathrm{period})  $ the tuple
	$((x,l_\mathrm{init}+l_\mathrm{period})  ,(x,l_\mathrm{init}+1))$ to the relation $\sigV{}^{\mathfrak{D}'}$. 
	It is easy to see that the so defined structure $\mathfrak{D}$ satisfies $\Phi_\mathrm{grid}$. Since the tiling $t$ was ultimately periodic the restriction of the tiling to the elements of $\mathfrak{D}$ defines a homomorphism from $\mathfrak{D}$ to $\mathfrak{A}_\mathcal{D} $. Since $\mathfrak{D}$ is clearly finite this implies that the structure $\mathfrak{A}_\mathcal{D}$ is in $\homcl{\getfinmodels{\Phi_\mathrm{grid}}}$ .
	
	For the other direction of the reduction proof assume that $\mathfrak{A}_\mathcal{D}$ is in $\homcl{\getfinmodels{\Phi_\mathrm{grid}}}$ . This implies by (\ref{grid5}) immediately that $\mathcal{D}$ admits an ultimately periodic tiling of the grid. \end{proof}

\undecTGD*

\begin{proof}
Immediate from Item~\ref{grid6} of \Cref{gridformula}. 
\end{proof}

\subsection{Proofs of Homclosedness Section}

\SattToHomClosed*

\begin{proof}
	We $\textsc{C}'$-reduce (finite) unsatisfiability in $\Logic'$ to $\mathrm{HomClosed}_\mathrm{(fin)}$ for $\Logic$ as follows: given $\Phi$, compute $f(\Phi)$. 

If the set of (finite) models of $f(\Phi)$ is closed under (finite-target) homomorphisms then 
$\getffinmodels{f(\Phi)}|_{\tau \cup \{\sigU{}\}}$ is closed under (finite-target) homomorphisms 
and therefore an empty set. This implies that $\Phi$ is (finitely) unsatisfiable.
On the other hand if  $\Phi$ is (finitely) unsatisfiable then $\getffinmodels{f(\Phi)}|_{\tau \cup \{\sigU{}\}}$  as well as $\getffinmodels{f(\Phi)}$ are closed under (finite-target) homomorphisms since they are empty.\end{proof}

\SattToHomClosedCompanion*

\begin{proof}
%
%
%
	We $\textsc{C}'$-reduce (finite) unsatisfiability in $\Logic'$ to $\mathrm{HomClosed}_\mathrm{(fin)}$ for $\Logic$ as follows: given $\Phi$, compute $f(\Phi)$. 

If $\Phi$ is (finitely) unsatisfiable, then $f(\Phi)$ is (finitely) valid and therefore its (finite) models are closed under (finite target) homomorphisms.

If $\Phi$ is (finitely) satisfiable witnessed by some (finite) model $\mathfrak{A}$, this means that $(\mathfrak{A},A)$ is not a model of $f(\Phi)$. Let now $\mathfrak{B}= \mathfrak{A} \cdot \{\widecheck{\sigU{}}\}$, which clearly is a (finite) model of $f(\Phi)$. Moreover, the identity function is a homomorphism from $\mathfrak{B}$ to $(\mathfrak{A},A)$, which shows that $f(\Phi)$ is not (finitely) homomorphism closed. 

It follows that checking (finite) homomorphism-closedness of $f(\Phi)$ is equivalent to checking (finite) unsatisfiability of $\Phi$.
\end{proof}

\HomSeparation*

\begin{proof}
	Given $\mathfrak{A}$, $\mathfrak{B}$, and $h$, we obtain $\mathfrak{C}$, $h_1$, and $h_2$ as follows:
	\begin{itemize}[itemindent=0ex, leftmargin=3ex, itemsep=-0.5ex]  
		\item $C = A \uplus B$
		\item $h_1\colon A \injhom C$ is set to $id_A$
		\item $h_2\colon C \sshom B$ is obtained via $h_2 = h \uplus id_B$
		
		\item $\sigc{}^\mathfrak{C} = \sigc{}^\mathfrak{A}$ for every constant $\sigc{} \in \tau$
		\item $\sigP{}^\mathfrak{C} = \{(c_1,\ldots,c_k) \mid (h_2(c_1),\ldots,h_2(c_k)) \in \sigP{}^\mathfrak{B}\}$ for every $k$-ary predicate $\sigP\in\tau$.\qedhere
	\end{itemize}
\end{proof}

\InjFormula*

\begin{proof}
	Let $\Phi$ be some $\tau$-sentence and let $\tau'$ consist of copies $\sig{P}'$ for each predicate $\sig{P}\in\tau$.

	For the first direction, assume that $\Phi^\spoil_\cl{\injhom}$ is (finitely) satisfiable. Then there is a (finite) model $\mathfrak{D} \vDash \Phi^\spoil_\cl{\injhom}$ over the signature $\tau\uplus\tau'\uplus \{ \sig{U} \}$. We now define two $\tau$-structures, $\mathfrak{B}$ as $\mathfrak{D}|_{\tau}$ and $\mathfrak{A}$ as the following one:
	\begin{itemize}
		\item The domain of $\mathfrak{A}$ is $\sig{U}^{\mathfrak{D}}$.
	 	\item For each constant $\sigc{}\in\tau$ we put $\sig{c}^{\mathfrak{A}} = \sig{c}^{\mathfrak{D}}$. Note that, by definition of the relativization of a formula, $\sig{c}^{\mathfrak{D}}\in\sig{U}^{\mathfrak{D}}$ for all constant symbols $\sig{c}$. So this is well defined.
	 	\item For each predicate $\sig{P}\in\tau$ let $\sig{P}^{\mathfrak{A}} = (\sig{U}^{\mathfrak{D}})^{\mathrm{ar}(\sig{P})} \cap \sig{P}^{\mathfrak{D}}\cap\sig{P}'^{\mathfrak{D}}$. Note that the first conjunct $(\sig{U}^{\mathfrak{D}})^{\mathrm{ar}(\sig{P})}$ ensures well definedness.
	\end{itemize}
	Define $s\colon\mathfrak{A}\to\mathfrak{B}, a\mapsto a$. $s$ is an injective homomorphism. As the injectivity is plain to see, we just verify that $s$ is a homomorphism. For each constant $\sig{c}\in\tau$ we have $s(\sig{c}^{\mathfrak{A}}) = \sig{c}^{\mathfrak{D}}$. Since $\sig{c}^{\mathfrak{D}} = \sig{c}^{\mathfrak{B}}$ for $\sig{c}\in\tau$ by means of $\mathfrak{B} = \mathfrak{D}|_{\tau}$ the map $s$ preserves constants. Now let $\bold{a}\in\sig{P}^{\mathfrak{A}}$ for some predicate $\sig{P}\in\tau$, which additionally means that $\bold{a}$ is a tuple from $\mathfrak{A}$ by the $(\sig{U}^{\mathfrak{D}})^{\mathrm{ar}(\sig{P})}$ part of the definition. By construction $\bold{a}\in\sig{P}^{\mathfrak{D}}$. Since $\sig{P} \in\tau$ it holds that $\sig{P}^{\mathfrak{D}}  =\sig{P}^{\mathfrak{B}}$. Thus we obtain by componentwise application $s(\bold{a}) \in\sig{P}^{\mathfrak{B}}$. Hence we conclude that $s$ is a homomorphism.
	\newline
	Now we show that $\mathfrak{A}\stackrel{s}{\injhom}\mathfrak{B}$ constitutes a (finite) injective spoiler, of which $\mathfrak{D}$ is a witness. For showing the spoiler property (as we already established that $s$ is an injective homomorphism this is the only thing left to prove) we note that by construction $\mathfrak{A}\vDash\Phi$ as $\mathfrak{D}\vDash(\Phi^-)^{\mathrm{rel}(\sigU{})}$ and the universe of $\mathfrak{A}$ is $\sig{U}^{\mathfrak{D}}$. Furthermore, as $\mathfrak{B}=\mathfrak{D}|_{\tau}$ we obtain $\mathfrak{B}\vDash\lnot\Phi$ since $\mathfrak{D}\vDash\lnot\Phi$ and $\Phi$ is a $\tau$-sentence.
	\newline
	Finally, we constructed $\mathfrak{A}$ and $\mathfrak{B}$ such that $\mathfrak{D}|_{\tau}=\mathfrak{B}$, $\sig{U}^{\mathfrak{D}}=s(A)$, and $\sig{P}^{\mathfrak{D}}\cap\sig{P}'^{\mathfrak{D}}=\sig{P}^{\mathfrak{A}}$ for every predicate $\sig{P}\in\tau$. By definition of $s$ this establishes that $\mathfrak{D}$ is a witness of the (finite) injective spoiler $\mathfrak{A}\stackrel{s}{\injhom}\mathfrak{B}$.
	
	For the other direction, let $\mathfrak{A}\stackrel{s}{\injhom}\mathfrak{B}$ be a (finite) injective spoiler for the $\tau$-sentence $\Phi$. We define $\mathfrak{D}$ as follows:
	\begin{itemize}
		\item The domain of $\mathfrak{D}$ is $B$.
		\item For each constant $\sig{c}\in\tau$ we put $\sig{c}^{\mathfrak{D}} = \sig{c}^{\mathfrak{B}}$.
		\item For each predicate $\sig{P}\in\tau$ we put $\sig{P}^{\mathfrak{D}} = \sig{P}^{\mathfrak{B}}$. 
		\item For each predicate $\sig{P}\in\tau$ we put \newline $\sig{P}'^{\mathfrak{D}}=\{(s(a_1),\ldots,s(a_k)) \mid   (a_1,\ldots,a_k) \in \sigP{}^\mathfrak{A}\}$.
		\item $\sig{U}^{\mathfrak{D}}=s(A)$.
	\end{itemize}
	It is obvious that $\mathfrak{D}|_{\tau}=\mathfrak{B}$, $\sig{U}^{\mathfrak{D}}=s(A)$ and $\sigP{}^\mathfrak{D} {\cap}\,{\sigP{}'}^\mathfrak{D}\! = \{(s(a_1),\ldots,s(a_k)) \mid   (a_1,\ldots,a_k) \in \sigP{}^\mathfrak{A}\}$ for every predicate $\sig{P}$ from $\tau$. Hence $\mathfrak{D}$ is a (finite) witness of the injective spoiler $\mathfrak{A}\stackrel{s}{\injhom}\mathfrak{B}$. By $\mathfrak{A}\stackrel{s}{\injhom}\mathfrak{B}$ being a spoiler we obtain that $\mathfrak{D}\vDash \lnot\Phi$ as $\Phi$ is a $\tau$-sentence and $\mathfrak{D}|_{\tau}=\mathfrak{B}$. Furthermore, as $\mathfrak{A}\vDash \Phi$ and $s$ being injective, $\mathfrak{D}\vDash (\Phi^{-})^{\mathrm{rel}(\sigU{})}$ also holds. Hence we obtain $\mathfrak{D}\vDash \Phi^\spoil_\cl{\injhom}$.
\end{proof}

\squeeze*

\begin{proof}
	Let $\lambda$ be defined by $a \mapsto \min(n,|h^{-1}(a)|)$ where $h^{-1}(a)$, as usual, denotes the set $\{b \in B \mid h(b)= a\}$. 
	Now define $\hslash \colon B \rightarrow A^\lambda $ such that the elements of $h^{-1}(a)$ are mapped to $\{(a,1), \ldots, (a,\lambda(a))\}$ in a surjective way, which is possible by construction. 
	By definition of $\mathfrak{A}^\lambda$, the map $\hslash$ is  a strong surjective homomorphism and clearly satisfies $h = \pi \circ \hslash$.
	The equivalence of $\mathfrak{B}$ and $\mathfrak{A}^\lambda$ regarding $\FOe$ sentences of quantifier rank $n$ can be shown as usual via an $n$-step Ehrenfeucht-Fra\"{i}ssé game, where, upon spoiler picking $e$, duplicator picks a fresh element from 
{$h^{-1}(\pi(e))$ or $\pi^{-1}(h(e))$}, respectively.    
\end{proof}

\mainsecond*

\begin{proof}
	For all cases, membership can be shown for the finite cases directly via \Cref{FinHomClFormula} by a polynomial reduction to a fragment with the appropriate unsatisfiability problem. 
	The general cases can be reduced to the finite ones via \Cref{thm:FSP}. 
	
	Hardness follows for \GNFO, \TGF, $\FOte{}$, and \AsFO{} as well as \AAEE{} via \Cref{prop:SattToHomClosed} from the corresponding known unsatisfiability complexities.
	For $\EsFO$ and $\EEAA$, it follows via \Cref{prop:SattToHomClosedCompanion} from the unsatisfiability complexities of  \AsFO{} and \AAEE{}, respectively.
\end{proof}

\TGDNPtheorem*

\begin{proof}
	For \textsc{NP}-hardness, we reduce the 3-colorability problem to $\mathrm{HomClosed}_\mathrm{(fin)}$ as follows: given a finite graph expressed as $\{E\}$-structure $\mathfrak{G}$, three-colorability of $\mathfrak{G}$ is equivalent to the existence of a homomorphism from $\mathfrak{G}$ to $\mathfrak{K_3}$ (the complete but loop-free 3-element graph). However, this is the case exactly if the \TGD{} sentence $\Phi \ = \ \cq{\mathfrak{K_3}} \Rightarrow \cq{\mathfrak{G}}$ is a tautology. This, in turn coincides with $\getmodels{\Phi}$ being closed under homomorphisms (as, for the negative case, we find that $\mathfrak{K_3} \notin \getmodels{\Phi}$ receives a homomorphism from $\mathfrak{I} \in \getmodels{\Phi}$).  
	
	In order to show \textsc{NP} membership, we provide a corresponding nondeterministic polytime algorithm that answers the question. From now on, we consider a \TGD{} sentence $\Phi$ as the finite set of conjuncts, each representing a single TGD (simply called ``rule'' in the following). We first apply a \textsc{PTime} preprocessing as follows: One by one, go through all the rules $\forall \bold{x}.\varphi(\bold{x}) {\impl} \exists\bold{y}.\psi(\bold{x},\bold{y})$ of $\Phi$ and do the following: 
	Let $\sim$ be the smallest equivalence relation over all the atoms from $\psi$ such that $\alpha \sim \alpha'$ whenever $\alpha$ and $\alpha'$ share an $\exists$-quantified variable from $\bold{y}$. For every $\sim$-equivalence class, form the conjunction over its atoms and $\exists$-quantify over all variables from $\boldsymbol{y}$ occurring therein. Let $\exists\bold{y_1}.\psi_1,\ldots, \exists\bold{y_n}.\psi_n$ be the formulae thus obtained. Now, replace $\forall \bold{x}.\varphi(\bold{x}) {\impl} \exists\bold{y}.\psi(\bold{x},\bold{y})$
	by the rules $\forall \bold{x}.\varphi(\bold{x}) {\impl} \exists\bold{y_1}.\psi_1(\bold{x},\bold{y_1}), \ldots, \forall \bold{x}.\varphi(\bold{x}) {\impl} \exists\bold{y_n}.\psi_n(\bold{x},\bold{y_n})$. Finally, the new rule set is deterministically partitioned into two sets: those where all variables present in $\psi$'s atoms are $\exists$-quantified will be called \emph{disconnected}, all other rules will be called \emph{connected}. Obviously all this preprocessing can be done in deterministic polynomial time and produces a partitioned rule set of polynomial size in the original one. Moreover, the obtained rule set is semantically equivalent to the original one.
	
	We proceed to describe a polysize certificate witnessing homomorphism-closedness of a preprocessed rule set. 
	Our certificate contains the following:
	\begin{itemize}[itemindent=0ex, leftmargin=3ex, itemsep=-0.5ex, topsep=0ex]
	\item a \emph{relevance partitioning} of the rule set into \emph{self-redundant} and \emph{self-irredundant} rules, where all connected rules (but possibly also some disconnected ones) are categorized as self-redundant. It is straight forward to check validity of such a partition in \textsc{PTime}.
    \item a \emph{redundancy witness} for each of the self-redundant rules: a substitution $h\colon \bold{y} \rightarrow \bold{x}$ sending every atom inside $\psi$ to an atom inside $\varphi$. Obviously redundancy witnesses can be verified in \textsc{PTime}.
	\item a \emph{universal model derivation}: a forward-chaining derivation sequence starting from $\mathfrak{I}$ using only self-irredundant rules and applying each rule at most once,  including the homomorphisms used for the rule applications. Let $\mathfrak{M}$ denote the universal model constructed in this derivation.  Again such a guess is of polynomial size and allows for \textsc{PTime} verification.
 \item a \emph{discharge witness} for each self-irredundant rule: a mapping $h\colon \bold{y} \rightarrow M$ sending every atom $\sig{R}(y_1,\ldots,y_n)$ inside $\psi$ to a ``semantic atom'' inside $\mathfrak{M}$, i.e., $(h(y_1),\ldots, h(y_n))  \in \sig{R}^\mathfrak{M}$.
 Discharge witnesses can be verified in \textsc{PTime}. 
	\end{itemize}
	
	Summing up, the algorithm's guess consists of one relevance partitioning, one redundancy witness for every self-redundant rule, one universal model derivation, and one discharge witness per self-irredundant rule. Verification for each can be done in \textsc{PTime}. 
	
	We now argue that the algorithm is correct (i.e., that the existence of a certificate as above coincides with the TGD set's homomorphism-closedness):
	
	\begin{enumerate}[label={\roman*}), itemindent=0ex, leftmargin=3ex, itemsep=-0.5ex]
		\item First, consider the case where one of the connected rules is not self-redundant (i.e., no appropriate witness can be provided). For such a rule $\forall \bold{x}.\varphi(\bold{x}) {\impl} \exists\bold{y}.\psi(\bold{x},\bold{y})$, let $\mathfrak{A}_\varphi$ be such that $\cq{\mathfrak{A}_\varphi} = \exists \bold{x}.\varphi(\bold{x})$. 
		Then, $\mathfrak{F}$ is a model, but $\mathfrak{F} \uplus \mathfrak{A}_\varphi$ is not -- due to the special form of $\psi$ that we achieved in the preprocessing. This is despite the existence of a homomorphism from $\mathfrak{F}$. Thus the rule set is not homclosed, as required. In the following, we consider the case where all connected rules are indeed (self-)redundant.
		
		\item Obviously, self-redundant rules are tautologies and hence have no influence on the set of models, so the question if the rule set is homomorphism-closed reduces to the question if its subset of self-irredundant rules is. By now we know that, in the considered rule set, every such rule $\forall \bold{x}.\varphi(\bold{x}) {\impl} \exists\bold{y}.\psi(\bold{x},\bold{y})$
		must be disconnected, i.e., be such that no variables from $\bold{x}$ occur in $\psi$, therefore any such rule can be written as $(\exists \bold{x}.\varphi(\bold{x})) \impl (\exists\bold{y}.\psi(\bold{y}))$. Henceforth, we assume our considered rule set contains the self-irredundant in this syntactic form. From the structure of the rules follows that there exists a finite universal model (i.e. one which has a homomorphism into every other model) $\mathfrak{M}$ which can be obtained in a forward-chaining derivation where every rule is applied at most once. This information is provided by the certificate's universal model derivation.

		\item We proceed to consider two cases.
		
		\begin{enumerate}
			\item Assume the non-existence of a discharge witness of $(\exists \bold{x}.\varphi(\bold{x})) \impl (\exists\bold{y}.\psi(\bold{y}))$. 
			This means that among the remaining rules, there is one $(\exists \bold{x}.\varphi(\bold{x})) \impl (\exists\bold{y}.\psi(\bold{y}))$ for which $\mathfrak{M} \not\models \exists\bold{y}.\psi(\bold{y})$. Yet in that case $\mathfrak{M} \uplus \mathfrak{A}_\varphi$ is not a model (mark that the rule is not self-redundant by assumption), despite the existence of the trivial homomorphism from $\mathfrak{M}$ -- consequently the rule set is not homclosed. 
			\item The existence of all discharge witnesses proves that each remaining rule's head is satisfied in $\mathfrak{M}$. Then, the same holds for every structure $\mathfrak{M}'$ into which a homomorphism from $\mathfrak{M}$ exists, which trivially entails that all remaining rules are satisfied in any such $\mathfrak{M}'$ and hence the rule set is homclosed. \qedhere
		\end{enumerate}
	\end{enumerate} 
\end{proof}

\subsection{Proofs of Characterizing Homclosures Section}

\notFO*

\begin{proof}
	We note that $\mathfrak{A} \in \homcl{\getmodels{\Phi_{\infty}}}$ exactly if $\mathfrak{A}$ contains an infinite $\sigP{}$-path.
	Hence, if $\mathfrak{A}$ is finite, $\mathfrak{A} \in \homcl{\getmodels{\Phi_{\infty}}}$ exactly if $\mathfrak{A}$ contains a directed $\sigP{}$-cycle.
	Checking for the existence of directed cycles is \textsc{NLogSpace}-complete \cite{DBLP:journals/jcss/Jones75} whereas
	$\FOe$ sentences can only describe problems in \textsc{AC}$^0$ \cite{DBLP:books/daglib/0095988}, which is strictly weaker than \textsc{NLogSpace}~\cite{FurstSS84}.
\end{proof}

\ExCharHomA*

\begin{proof}
	We show the claim by a reduction from (finite) satisfiability.
	Let $\Psi$ be an arbitrary $\FOe$ formula.
	
	Consider $\Psi'$ defined as   
	$\Psi^{\mathrm{rel}(\sigU{})} \wedge \Phi_{\text{grid}}$ chosen in a way that the signature $\tau_1 \uplus \{\sigU{}\}$ of $\Psi^{\mathrm{rel}(\sigU{})}$ and the signature $\tau_2$ of $\Phi_{\text{grid}}$ are disjoint.
	If $\Psi$ is (finitely) unsatisfiable, then so is $\Psi'$. In that case $\homclffin{\getffinmodels{\Psi'}}$ is empty and can be characterized by $\Phi_\bold{false}$.

	If $\Psi$ is (finitely) satisfiable, then any sentence $\tilde\Psi$ (finitely) characterizing $\homclffin{\getffinmodels{\Psi'}}$ can be turned into a procedure for deciding if for a domino system $\mathcal{D}$ a(n ultimately periodic) $\mathcal{D}$-tiling exists: this is the case exactly if the structure $\mathfrak{A}_\mathcal{D} \cdot\widehat{\tau_1} \cdot \widehat{\{\sigU{}\}}$ is a model of $\tilde\Psi$. 
	
	It follows that the (finite) homomorphism closure of $\getffinmodels{\Psi'}$ is characterizable (in any non-degenerate logic $\Logic$) exactly if $\Psi$ is unsatisfiable. Checking the latter is undecidable.  
\end{proof}

\ExCharHomB*

\begin{proof}
	We show the claim by a reduction from the existence of a(n ultimately periodic) tilling by the deterministic domino set $\mathcal{D}$.
	
Let the TGD sentence $\Phi_{\mathcal{D}\text{-}\mathrm{tiling}}$ with signature $\tau$ be defined as in \Cref{def:TGDdettiling}. Let
\begin{equation}
\Phi'_{\mathcal{D}\text{-}\mathrm{tiling}} = \Phi_{\mathcal{D}\text{-}\mathrm{tiling}} \wedge \Big( \big(\exists x.  \sigT{}_\emptyset(x) \big) \impl \cq{\mathfrak{F}} \Big).
\end{equation}	
	If no (ultimately periodic) $\mathcal{D}$-tiling exists, then $\cq{\mathfrak{F}}$ characterizes $\homclffin{\getffinmodels{\Phi'_{\mathcal{D}\text{-}\mathrm{tiling}}}}$.

	If, on the other hand, some (ultimately periodic) $\mathcal{D}$-tiling exists, then any sentence $\Psi$ (finitely) characterizing $\homclffin{\getffinmodels{\Phi'_{\mathcal{D}\text{-}\mathrm{tiling}}}}$ can be turned into a procedure for deciding if for some other domino system $\mathcal{D}'$ a(n ultimately periodic) $\mathcal{D}'$-tiling exists: this is the case exactly if the structure $\mathfrak{A}_{\mathcal{D}'} \cdot\widehat{\sigma}$ with $\sigma=\tau \setminus \{\sigH{},\sigV{}\}$  is a model of $\Psi$.
This holds since all $\{\sig{H},\sig{V}\}$-reducts of models of  $\Phi'_{\mathcal{D}\text{-}\mathrm{tiling}}$ satisfy  $\Phi_{\text{grid}}$ and furthermore if an ultimately periodic $\mathcal{D}'$-tiling exists we can find a  model of  $\Phi'_{\mathcal{D}\text{-}\mathrm{tiling}}$ such that the period of the grid induced by its $\{\sig{H},\sig{V}\}$-reduct is a multiple of the period of the $\mathcal{D}'$-tiling. This is possible since a domino set that tiles the grid with an ultimately periodic tiling with period length $k$  also tiles it with an ultimately periodic tiling with period length $k\cdot s$  for all $s\in \mathbb{N}$.
	
	It follows that the (finite) homomorphism closure of $\getffinmodels{\Phi'_{\mathcal{D}\text{-}\mathrm{tiling}}}$ is characterizable (in any decent logic $\Logic$ in the above sense) exactly if some (ultimately periodic) $\mathcal{D}$-tiling exists. Checking the latter is undecidable.  
\end{proof}

\EAFOinEPO*

\begin{proof}
	Assume $\Phi$ is over the signature $\tau$.
	Let $\Phi^\mathrm{sk}$ be the $\AsFO$ sentence obtained from $\Phi$ via skolemization of all the existentially quantified variables into constants. Clearly then  $\getffinmodels{\Phi} = \getffinmodels{\Phi^\mathrm{sk}}|_\tau$. Since $\Phi^\mathrm{sk}$ is a $\AsFO$ sentence, any model $\mathfrak{A}$ gives rise to a ``small'' submodel by taking the substructure $\mathfrak{A}'$ induced by the set $A' \subseteq A$ containing all elements denoted by constants (if the signature has no constants, pick one arbitrary element). Moreover, the identity on $A'$ is a homomorphism from $\mathfrak{A}'$ to $\mathfrak{A}$. Let $\mathcal{A}$ be the set of small submodels thus obtained, for which therefore holds $\homcl{\mathcal{A}}=\homcl{\getmodels{\Phi^\mathrm{sk}}}$ and hence also $\homcl{\mathcal{A}|_\tau}=\homcl{\getmodels{\Phi}}$. Due to the bound on the domain size (number of constants), $\mathcal{A}$ (and thus $\mathcal{A}|_\tau$) can contain only finitely many different structures (up to isomorphism) and is effectively computable via enumeration and model checking. Hence $\homcl{\mathcal{A}|_\tau}$ can be characterized by the  $\EsFO$ sentence $\Psi = \bigvee_{\mathfrak{A} \in \mathcal{A}|_\tau} \cq{\mathfrak{A}}$.
\end{proof}

\allhomcaptures*

We split the proof into separate statements and prove them individually:
\begin{itemize}
\item \Cref{prop:charFOte} deals with homcapturing of $\FOte$
\item \Cref{prop:charGFO} deals with homcapturing of $\GFO$
\item \Cref{prop:charGNFO} deals with homcapturing of $\GNFO$
\item \Cref{prop:charTGF} deals with homcapturing of $\TGF$
\item \Cref{prop:charEAAE} deals with homcapturing of $\EAAE$
\end{itemize}

Before, however, we will show a useful result, which allows us to -- without loss of generality -- restrict our attention to (projective) normal forms of the considered logics (as they are typically used in results related to satisfiability).

\begin{lemma}[projective normal form sufficiency]\label{lem:normalformOK}
For some $\tau \subseteq \tau' \subseteq \tau''$, let 
\begin{itemize}[itemindent=0ex, leftmargin=3ex, itemsep=-0.7ex, topsep=-0.5ex] 
\item 	
$\Phi$ be a $\tau$-sentence, 
\item 
$\Phi'$ be a $\tau'$-sentence projectively characterizing $\Phi$, i.e., $\getmodels{\Phi'}|_\tau = \getmodels{\Phi}$,
\item 
$\Psi$ be a $\tau''$-sentence projectively characterizing the homclosure of $\Phi'$, i.e., $\homcl{\getmodels{\Phi'}} = \getmodels{\Psi}|_{\tau'}$.   
\end{itemize} 
Then $\Psi$ also projectively characterizes the homclosure of $\Phi$, i.e., $\homcl{\getmodels{\Phi}} = \getmodels{\Psi}|_{\tau}$ holds.
\end{lemma}

\begin{proof}
Since $\tau \subseteq \tau'$, we obtain 
\begin{equation}
\getmodels{\Psi}|_\tau = \getmodels{\Psi}|_{\tau'}|_\tau = \homcl{\getmodels{\Phi'}}|_\tau \stackrel{*}{=} \homcl{\getmodels{\Phi'}|_\tau} = \homcl{\getmodels{\Phi}}.
\end{equation}
For the equality denoted by $\stackrel{*}{=}$, we show the two inclusions separately:

``$\subseteq$'': Let $\mathfrak{B} \in \homcl{\getmodels{\Phi'}}|_\tau$. Then there exists a $\tau'$-structure $\mathfrak{B}' \in \homcl{\getmodels{\Phi'}}$ with $\mathfrak{B}'|_\tau = \mathfrak{B}$. Hence there exist a $\tau'$-structure $\mathfrak{A}' \in \getmodels{\Phi'}$ and a homomorphism $h \colon \mathfrak{A}' \to \mathfrak{B}'$. But then $h$ also serves as a homomorphism from $\mathfrak{A} = \mathfrak{A}'|_\tau$ to $\mathfrak{B}'|_\tau = \mathfrak{B}$. By $\mathfrak{A} \in  \getmodels{\Phi'}|_\tau$ we then obtain $\mathfrak{B} \in \homcl{\getmodels{\Phi'}|_\tau}$.

``$\supseteq$'': Let $\mathfrak{B} \in \homcl{\getmodels{\Phi'}|_\tau}$. Then there exist a $\tau'$-structure $\mathfrak{A}' \in \getmodels{\Phi'}$ and a homomorphism $h \colon \mathfrak{A}'|_{\tau} \to \mathfrak{B}$. Let now $\mathfrak{B}' =  \mathfrak{B} \cdot {\raisebox{1pt}{$\widehat{\ \ \ \ \ \ }$} \hspace{-4.5ex}  { \tau' {\setminus} \tau}}$. Clearly,
then $h$ serves as homomorphism from $\mathfrak{A}'$ to $\mathfrak{B}'$, thus $\mathfrak{B}' \in \homcl{\getmodels{\Phi'}}$. Yet, since $\mathfrak{B}'|_\tau = \mathfrak{B}$, we obtain $\mathfrak{B} \in \homcl{\getmodels{\Phi'}}|_\tau$.
\end{proof}

As the next three results (namely \Cref{prop:charFOte}, \Cref{prop:charGFO}, and \Cref{prop:charTGF}) will make use of some essentially equal parts in the homcapturing sentence we introduce these ``generic'' parts beforehand.

\begin{definition}
	Let $\Phi$ be a $\FO{}$ sentence over some fixed signature $\tau$ and let $C$ denote the set of all constant symbols from $\tau$. Additionally, we fix some map $\mathrm{width}$ mapping each $\tau$-formula to some natural number, its so-called \emph{width}. We specify this map when the need arises.
	
	A \emph{type $\tp$} is a maximal consistent set of $\tau$-literals (i.e., negated and unnegated atoms where the relation symbol is from $\tau$) over terms from $\bold{v} \cup C$, where $\bold{v}$ is a dedicated, finite, linearly ordered set of \emph{type variables} of cardinality $\geq \mathrm{width}(\Phi)$. When necessary we might explicitly mention the signature at consideration, i.e. we say \emph{$\tau$-type}.
	To each type $\tp$ we associate a $\bold{v}$-subsequence $\bold{v}_{\tp}$ enumerating the variables of $\tp$.
	We say $\tp$ is an \emph{$n$-type}, or equivalently is a type of \emph{order $n$}, if $\bold{v}_{\tp}$ is a sequence of exactly $n$ variables. We define $\mathrm{order}$ to be a map from the set of types to the natural numbers, associating to each type its order, i.e., an $n$-type $\tp$ has $\mathrm{order}(\tp) = n$.
	For a structure $\mathfrak{A}$ we say that an $n$-type $\tp$ is \emph{realized} by an $n$-tuple $\bold{a} = (a_{1}, \ldots, a_{n}) \in A^{n}$, if for the componentwise evaluation map $\nu \colon \bold{v}_{t} \mapsto \bold{a}, (\bold{v}_{\tp})_{i} \mapsto a_{i}$ the relation $(\mathfrak{A}, \nu) \vDash \tp$ holds.
	We say a type $\tp$ of order $n$ is \emph{realized} ($k$ times) in $\mathfrak{A}$ if it is realized by ($k$ distinct) $n$-tuples from $\mathfrak{A}$.

	Given an atom $\delta$ using exactly the variables from $\bold{x}$ (and possibly additional constants), we say a type $\tp$ is \emph{guarded by $\delta$} if there exists a variable substitution $\nu\colon \bold{x} \to \bold{v}_{\tp}\cup C$ with $\bold{v}_{\tp} \subseteq \{\nu(x) \mid x \in \bold{x}\}$ such that $\nu(\delta) \in \tp$, where the application of a substitution $\nu$ to an expression replaces all occurrences of every variable $x$ in $\nu$'s domain by $\nu(x)$. Then $\nu$ will be called \emph{guard witness}.
	A type $\tp$ is \emph{guarded}, if it is guarded by some atom.
	A type $\tp$ is called \emph{rigid} if it does not contain literals of the form $v=v'$ (for distinct variables $v,v'$ from $\bold{v}_{\tp}$) or $\sigc = v$ or $v = \sigc$ for $\sig{c}\in C$. 
	For $\bold{v}' \subseteq \bold{v}_{\tp}$, we obtain the $\bold{v}'$-\emph{subtype} of $\tp$, denoted $\tp|_{\bold{v}'}$, by removing from $\tp$ all literals mentioning any variable from $\bold{v}_\tp \setminus \bold{v}'$.
	
	By $\mathcal{E}$ we will denote a fixed finite set of rigid types and call it the set of \emph{eligible types}. For a $\tau$-structure $\mathfrak{A}$ we denote by $\mathcal{E}(\mathfrak{A})$ its \emph{type summary}, which is a pair $(\mathcal{E}_{+}, \mathcal{E}_{!})$ of subsets of $\mathcal{E}$ such that $\mathcal{E}_{+}$ ($\mathcal{E}_{!}$) contains every type realized (exactly once) in $\mathfrak{A}$. Furthermore, the \emph{model summary} $\getsummary{\Phi}$ of a sentence $\Phi$ is defined by $\{ \mathcal{E}(\mathfrak{A}) \ | \ \mathfrak{A}\in \getmodels{\Phi} \}$.
	
	For a set $\mathcal{E}$ of eligible types we extend the signature $\tau$ by a set $\sigma$ of \emph{type predicates}: $\sigma$ contains a new $n$-ary predicate $\sig{Tp}_{\tp}$ for every eligible $n$-type $\tp\in\mathcal{E}$.
	Let $\mathfrak{A}$ be a $\tau$-structure. The \emph{$\mathcal{E}$-adornment} of $\mathfrak{A}$, denoted by $\mathfrak{A}\cdot\mathcal{E}$ is the $\tau\uplus\sigma$-structure expanding $\mathfrak{A}$ such that $\sig{Tp}_{\tp}^{\mathfrak{A}\cdot\mathcal{E}} = \{ \bold{a} \ | \ \bold{a} \text{ realizes } \tp \text{ in } \mathfrak{A} \}$.
	
	If $\mathfrak{A}, \mathfrak{B}$ are $\tau$-structures such that there is a homomorphism $h \colon \mathfrak{A}\to\mathfrak{B}$ and we have a set of eligible types $\mathcal{E}$, then we denote by $\mathfrak{B}\cdot h(\mathfrak{A}\cdot\mathcal{E}|_{\sigma})$ the $\tau\uplus\sigma$-structure $\mathfrak{B}'$ expanding $\mathfrak{B}$ with $\sig{Tp}_{\tp}^{\mathfrak{B}'} = \{ h(\bold{a}) \ | \ \bold{a} \in \sig{Tp}_{\tp}^{\mathfrak{A}\cdot\mathcal{E}} \}$.
	
	We define
	\begin{equation}
		\Psi_\mathrm{hom} = \!\bigwedge_{\tp \in \mathcal{E}} \!\forall \bold{v}_{\tp}.\big( \Tp_\tp(\bold{v}_{\tp}) \impl \smallbigwedge \{\alpha \tight{\in} \tp \mid \alpha \text{ positive lit.}\} \big).
	\end{equation}
	
	For every $(\mathcal{E}_{+}, \mathcal{E}_{!}) \in \getsummary{\Phi}$ let $\Psi_{\mathrm{gen}, (\mathcal{E}_{+}, \mathcal{E}_{!})}$ be the conjunction of the following formulae:

	\newcommand{\less}{\hspace{-1.5ex}}
	\noindent
	\begin{eqnarray}
		\less\exists \bold{v}_\tp. \sig{Tp}_\tp (\bold{v}_\tp) \less& & \less \hspace{25ex}\mbox{for } \tp \in \mathcal{E}_{+}  \label{eq: gen 1}\\
		\less\forall \bold{v}_\tp. \sig{Tp}_\tp (\bold{v}_\tp) \less& \Rightarrow &\less \bot \hspace{20ex}\mbox{for } \tp \in \mathcal{E}\setminus \mathcal{E}_{+}  \label{eq: gen 2}\\
		\less\forall \bold{v}_\tp. \sig{Tp}_\tp (\bold{v}_\tp) \less& \Rightarrow &\less 	\sig{Tp}_{\tp'}(\bold{v}_{\tp'}) \hspace{5.25ex}\mbox{for } \tp,\tp' \in \mathcal{E}_{+}, \tp'=\tp|_{\bold{v}_{\tp'}}  \label{eq: gen 3} \\
		\less\forall \bold{x}. \sig{Tp}_\tp (\bold{x}) \less& \Rightarrow &\less 	\sig{Tp}_{\kappa\tp}( \eta^{\kappa\bold{v}_\tp}_\bold{v} (\bold{x}) ) \hspace{1.75ex}\mbox{for } \tp \in \mathcal{E}, \kappa \colon \bold{v}\injsur \bold{v} \label{eq: gen 3plus}
	\end{eqnarray}
	\let\less\undefined
	where 
	\begin{itemize}
		\item
		the permutation $\kappa$ on the set $\bold{v}$ is extended to variable tuples in the usual way and also used as substitution on logical expressions, 
		\item for a (repetition-free) tuple $\bold{v}'$ of variables from $\bold{v}$, we let $\eta^{\bold{v}'}_\bold{v}$ denote the sorting function rearranging the entries of $|\bold{v}'|$-tuples in a way that $\eta^{\bold{v}'}_\bold{v}(\bold{v}')$ becomes a subsequence of $\bold{v}$.
	\end{itemize}
\end{definition}
Several remarks to this definition:
\begin{enumerate}
	\item Note that each signature $\tau$ is finite (hence also the set of constant symbols $C$ from $\tau$) and the set of dedicated type variables $\bold{v}$ is finite as well. Hence there are only finitely many distinct literals with terms from $\bold{v}\cup C$. This implies the finiteness of every type $\tp$. For convenience we may identify $\tp$ (which is a set of literals) with the conjunction over all the literals contained in it, i.e. $\bigwedge \tp$ (which then is a formula).
	\item Given a $\tau$-structure $\mathfrak{A}$ and a tuple $\bold{a} = (a_{1}, \ldots, a_{n})$ (such that $\mathrm{length}(\bold{a}) \le \mathrm{length}(\bold{v})$), $\bold{a}$ realizes one, up to permutation of type variables, unique type $\tp$ characterizing the induced substructure of $\mathfrak{A}$ by the elements from $\bold{a}$ and the interpretations of every constant symbol. To be more precise, let $\bold{v}_{\tp} = (v_{1}, \ldots, v_{n})$ be the initial segment of $\bold{v}$ of length $\mathrm{length}(\bold{a})$ and $\mu \colon \bold{v}_{\tp}\to\bold{a}, v_i\mapsto a_i$ the componentwise evaluation of $\bold{v}_{\tp}$ by $\bold{a}$. We extend $\mu$ to $\bold{v}_{\tp}\cup C$ by defining $\mu(\sig{c}) = \sig{c}^{\mathfrak{A}}$ for every $\sig{c}\in C$. For each $\sig{R}$ and $\bold{t} = (t_{1}, \ldots, t_{\mathrm{ar}(\sig{R})})\subseteq \bold{v}_{\tp}\cup C$ the literal
	\begin{itemize}
		\item $\sig{R}(\bold{t})$ is in $\tp$, iff $\mu(\bold{t}) \in \sig{R}^{\mathfrak{A}}$ and
		\item $\lnot\sig{R}(\bold{t})$ is in $\tp$, iff $\mu(\bold{t}) \notin \sig{R}^{\mathfrak{A}}$.
	\end{itemize}
	This set is consistent as (by $\mathfrak{A}$ being a structure) for $\bold{t}$ we cannot have $\mu(\bold{t})\in \sig{R}^{\mathfrak{A}}$ and $\mu(\bold{t})\notin \sig{R}^{\mathfrak{A}}$ simultaneously. Furthermore it is maximal since either $\mu(\bold{t})\in \sig{R}^{\mathfrak{A}}$ or $\mu(\bold{t})\notin \sig{R}^{\mathfrak{A}}$. To show uniqueness (up to variable permutations), assume there were a second type $\tp^{\ast}$ also realized by $\bold{a}$ such that for all $\kappa \colon \bold{v}\injsur\bold{v}$ the inequality $\kappa\tp^{\ast}\neq\tp$ holds. Without loss of generality we may assume that $\bold{v}_{\tp} = \bold{v}_{\tp^{\ast}} = (v_{1}, \ldots, v_{n})$ and $\kappa$ mapping each $v\in\bold{v}_{\tp}$ to itself (and hence $\kappa\tp^{\ast}=\tp^{\ast}$). As $\tp \neq \tp^{\ast}$ there is some $\sig{R} \in \tau$ and $\bold{t}(t_{1}, \ldots, t_{\mathrm{ar}(\sig{R})})\subseteq \bold{v_{\tp}} \cup C$ such that $\sig{R}(\bold{t})\in \tp$ and $\sig{R}(\bold{t})\notin\tp^{\ast}$. By maximal consistency of types this means $\lnot\sig{R}(\bold{t})\in \tp^{\ast}$. But as $\bold{a}$ realizes both types we obtain $\mu(\bold{t})\in\sig{R}^{\mathfrak{A}}$ and $\mu(\bold{t})\notin\sig{R}^{\mathfrak{A}}$, a contradiction.
	\item A note on eligible types: Depending on the logical fragment under scrutiny, we will make the notion of eligibility more specific. In fact, for $\Phi$ a sentence from some fragment of $\FO{}$, the finiteness will be ensured by considering just types whose order is bounded by $\mathrm{width}(\Phi)$. 
	\item As a corollary to sets $\mathcal{E}$ being finite, we immediately obtain that for an $\FO{}$ sentence $\Phi$ (and a set of eligible types $\mathcal{E}$) its model summary $\getsummary{\Phi}$ is finite.
\end{enumerate}
\begin{lemma}\label{lem:eligclosure}
	Let $\tau$ be a signature and $\Phi$ a $\FO{}$ $\tau$-sentence. Let $\mathcal{E}$ be a finite set of rigid types.
	\begin{enumerate}
		\item\label{item eligclosure:1} For every $\tp, \tp' \in \mathcal{E}$ and every $\mathfrak{A}\in\getmodels{\Phi}$ with $\mathcal{E}(\mathfrak{A}) = (\mathcal{E}_{+}, \mathcal{E}_{!}) \in \getsummary{\Phi}$ with $\tp' = \tp|_{\bold{v}_{\tp'}}$ and $\tp\in \mathcal{E}_{+}$, then also $\tp' \in \mathcal{E}_{+}$. More specifically: Every componentwise assignment $\mu\colon\bold{v}_{\tp}\to\bold{a}$ of the variables from $\tp$ to a tuple $\bold{a}$ from $\mathfrak{A}$ of length $\mathrm{order}(\tp)$ witnessing the realization of $\tp$ by $\bold{a}$ in $\mathfrak{A}$ gives rise to a subtuple $\bold{a}'$ of $\bold{a}$, namely $\mu(\bold{v}_{\tp'})$, such that $\bold{a}' = \mu|_{\bold{v}_{\tp'}}$ witnesses the realization of $\tp'$ by $\bold{a}'$ in $\mathfrak{A}$.
		\item\label{item eligclosure:2} If $\mathcal{E}$ is closed under taking variable permutations, i.e., for every $\tp \in \mathcal{E}$ and $\kappa\colon \bold{v}\injsur\bold{v}$ also $\kappa\tp \in \mathcal{E}$, 
		then so is $\mathcal{E}_{+}$ for every $(\mathcal{E}_{+}, \mathcal{E}_{!})\in\getsummary{\Phi}$.
	\end{enumerate}
\end{lemma}
Note that Item \ref{item eligclosure:1} essentially asserts for every $(\mathcal{E}_{+}, \mathcal{E}_{!}) \in \getsummary{\Phi}$ a kind of limited closure under taking subtypes in $\mathcal{E}_{+}$ .

\begin{proof}
	Let $\Phi$ be a $\FO{}$ $\tau$-sentence and $\mathcal{E}$ a set of eligible types.

		We first show Item \ref{item eligclosure:1}. Let $\tp, \tp'\in \mathcal{E}$ with $\tp'=\tp|_{\bold{v}_{\tp'}}$, $\mathfrak{A}\in\getmodels{\Phi}$, $\mathcal{E}(\mathfrak{A})=(\mathcal{E}_{+}, \mathcal{E}_{!})\in\getsummary{\Phi}$, and $\tp \in\mathcal{E}_{+}$. Hence $\tp$ is realized in $\mathfrak{A}$ by some tuple $\bold{a} = (a_{1}, \ldots, a_{n})\in A^{n}$ where $n = \mathrm{order}(\tp)$. Let $\mu\colon \bold{v}_{\tp} \to \bold{a}$ be the componentwise variable evaluation witnessing this realization, i.e. $(\mathfrak{A}, \mu) \vDash \tp$. Now we set $\iota\colon \bold{v}_{\tp'}\to\bold{v}_{\tp}, v \mapsto v$. Then the map $\mu\circ\iota$ witnesses the realization of $\tp' = \tp|_{\bold{v}_{\tp'}}$ in $\mathfrak{A}$ by the subtuple $\mu(\bold{v}_{\tp'})$ (obtained by componentwise application of $\mu$) of $\bold{a}$. As $\tp' \in \mathcal{E}$ we conclude $\tp'\in \mathcal{E}_{+}$.
	
	Now take $\mathcal{E}$ to be closed under taking permutations of variables, let $\mathfrak{A}\in\getmodels{\Phi}$, and $\mathcal{E}(\mathfrak{A}) = (\mathcal{E}_{+}, \mathcal{E}_{!})\in\getsummary{\Phi}$. Furthermore take some type $\tp\in \mathcal{E}_{+}$ and some permutation $\kappa\colon\bold{v}\injsur\bold{v}$. Since $\tp\in\mathcal{E}_{+}$, it is realized by some tuple $\bold{a}=(a_{1}, \ldots, a_{n})\in A^{n}$ where $n = \mathrm{order}(\tp)$. We denote by $\mu \colon \bold{v}_{\tp}\to \bold{a}$ the componentwise variable assignment witnessing the realization of $\tp$ by $\bold{a}$ in $\mathfrak{A}$, i.e. $(\mathfrak{A}, \mu)\vDash \tp$. Then the tuple $\eta_{\bold{v}}^{\kappa\bold{v}_{\tp}}(\bold{a})$ realizes $\kappa\tp$, witnessed by $\mu' \colon \kappa\bold{v}_{\tp}\to \eta_{\bold{v}}^{\kappa\bold{v}_{\tp}}(\bold{a}), (\eta_{\bold{v}}^{\kappa\bold{v}_{\tp}}(\bold{v}_{\tp}))_{i}\mapsto (\eta_{\bold{v}}^{\kappa\bold{v}_{\tp}}(\bold{a}))_{i}$. As $\mathcal{E}$ is closed under taking permutations of variables, we conclude that $\kappa\tp\in\mathcal{E}_{+}$. This establishes Item \ref{item eligclosure:2}.
\end{proof}
\begin{lemma}\label{lem:PsiGen}
	Let $\Phi$ be a $\FO{}$ $\tau$-sentence, $\mathcal{E}$ be a set of eligible types (i.e., a finite set of rigid type) closed under variable permutations, and $\sigma$ be the set of fresh type predicates $\sig{Tp}_{\tp}$ for all $\tp\in \mathcal{E}$. Then for $\mathfrak{A} \in\getmodels{\Phi}$ the $\mathcal{E}$-adornment $\mathfrak{A}\cdot \mathcal{E}$ satisfies $\Psi_{\mathrm{gen}, \mathcal{E}(\mathfrak{A})}$. Furthermore, for every $\mathfrak{B}$ being a homomorphic codomain of $\mathfrak{A}$ via a homomorphism $h \colon \mathfrak{A}\to\mathfrak{B}$ the $\tau\uplus\sigma$-structure $\mathfrak{B} \cdot h(\mathfrak{A}\cdot \mathcal{E}|_{\sigma})$ equally satisfies $\Psi_{\mathrm{gen}, \mathcal{E}(\mathfrak{A})}$.
\end{lemma}
\begin{proof}
	Let $\Phi$ be a $\FO{}$ $\tau$-sentence,  $\mathfrak{A}\in\getmodels{\Phi}$ and let $\mathcal{E}$ be the set of eligible types and $\mathcal{E}(\mathfrak{A}) = (\mathcal{E}_{+}, \mathcal{E}_{!})$ denote the type summary of $\mathfrak{A}$. 
	We show that the $\mathcal{E}$-adornment $\mathfrak{A}\cdot\mathcal{E}$ of $\mathfrak{A}$ satisfies $\Psi_{\mathrm{gen}, \mathcal{E}(\mathfrak{A})}$. More specifically:
	\newline
	For all $\tp \in \mathcal{E}$ the interpretation $\sig{Tp}_{\tp}^{\mathfrak{A}\cdot \mathcal{E}}$ contains by definition some tuple $\bold{a}$ from $\mathfrak{A}\cdot\mathcal{E}$ if and only if $\bold{a}$ realizes the associated type $\tp$ in $\mathfrak{A}$. As every $\tp \in \mathcal{E}_{+}$ is realized in $\mathfrak{A}$ we conclude that $\mathfrak{A}\cdot\mathcal{E}$ satisfies Formulae \ref{eq: gen 1}. As none of the types $\tp \in \mathcal{E}\setminus \mathcal{E}_{+}$ is being realized in $\mathfrak{A}$, we see by the same token that the interpretations $\sig{Tp}_{\tp}^{\mathfrak{A}\cdot\mathcal{E}}$ have to be empty. Hence we immediately obtain that $\mathfrak{A}\cdot \mathcal{E}$ satisfies Formulae \ref{eq: gen 2}.
	\newline
	Now take $\tp, \tp' \in \mathcal{E}$ such that $\tp' = \tp|_{\bold{v}_{\tp'}}$ and let $\bold{a}$ be some tuple from $\mathfrak{A}\cdot\mathcal{E}$ such that $\bold{a}\in \sig{Tp}_{\tp}^{\mathfrak{A}\cdot\mathcal{E}}$. By definition $\bold{a}$ realizes $\tp$ in $\mathfrak{A}$, let $\mu\colon \bold{v}_{\tp} \to\bold{a}$ be the componentwise evaluation witnessing this.
	By \Cref{lem:eligclosure} $\tp'$ is realized in $\mathfrak{A}$ (i.e. $\tp' \in\mathcal{E}_{+}$) by the subtuple $\mu(\bold{v}_{\tp'})$ ($\mu$ applied componentwise) of $\bold{a}$ realizing $\tp'$. Hence it is in $\sig{Tp}_{\tp'}^{\mathfrak{A}\cdot\mathcal{E}}$ by definition. Finally, observing that $\mu$ assigns the tuple $\bold{a}$ componentwise to $\bold{v}_{\tp}$, yields the satisfaction of Formulae \ref{eq: gen 3} by $\mathfrak{A}\cdot\mathcal{E}$.
	\newline
	To show the satisfaction of the last kind of Formulae (i.e. Formulae \ref{eq: gen 3plus}), let $\tp \in \mathcal{E}$, $\kappa \colon \bold{v} \injsur \bold{v}$ a permutation of the variables and $\bold{a}$ a tuple from $\mathfrak{A}\cdot\mathcal{E}$ such that $\bold{a}\in \sig{Tp}_{\tp}^{\mathfrak{A}\cdot\mathcal{E}}$. Note that, since $\mathcal{E}$ is closed under taking permutations of (type) variables, $\kappa\tp\in\mathcal{E}$. As for every $\tp \in\mathcal{E}\setminus\mathcal{E}_{+}$ the interpretation $\sig{Tp}_{\tp}^{\mathfrak{A}\cdot\mathcal{E}}$ is empty, it suffices to discuss the case when $\tp\in \mathcal{E}_{+}$. Hence $\bold{a}$ realizes $\tp$ in $\mathfrak{A}$. It is plain that $\eta_{\bold{v}}^{\kappa\bold{v}_{\tp}}(\bold{a})$ realizes $\kappa\tp$ in $\mathfrak{A}$, hence $\kappa\tp\in\mathcal{E}_{+}$. Thus we obtain $\eta_{\bold{v}}^{\kappa\bold{v}_{\tp}}(\bold{a}) \in \sig{Tp}_{\kappa\tp}^{\mathfrak{A}\cdot\mathcal{E}}$, showing that $\mathfrak{A}\cdot\mathcal{E}$ satisfies Formulae \ref{eq: gen 3plus}.

	We now turn to the second part of the lemma. Let $\mathfrak{B}$ be a $\tau$-structure such that there is a homomorphism $h\colon \mathfrak{A} \to \mathfrak{B}$. We show that the structure $\mathfrak{B}\cdot h(\mathfrak{A}\cdot\mathcal{E}|_{\sigma})$ (which we will denote by $\mathfrak{B}'$) also satisfies $\Psi_{\mathrm{gen}, \mathcal{E}(\mathfrak{A})}$. Additionally, $h\colon \mathfrak{A}\to\mathfrak{B}$ is also a homomorphism if lifted to $h' \colon \mathfrak{A}\cdot\mathcal{E} \to \mathfrak{B}'$. As $h$ is already a homomorphism from $\mathfrak{A}$ to $\mathfrak{B}$ we only have to check whether for every $\sig{Tp}_{\tp} \in\sigma$ and $\bold{a}\in\sig{Tp}_{\tp}^{\mathfrak{A}\cdot\mathcal{E}}$ the image of $\bold{a}$ under $h'$ satisfies $h'(\bold{a})\in\sig{Tp}_{\tp}^{\mathfrak{B}'}$. But by definition of $\mathfrak{B}'$ the interpretation of these type predicates is already defined to be the homomorphic image of the respective interpretations in $\mathfrak{A}\cdot\mathcal{E}$. Hence $h'(\bold{a})\in\sig{Tp}_{\tp}^{\mathfrak{B}'}$ yielding that $h'$ is also a homomorphism.
	\newline
	Applying this knowledge we immediately obtain that $\mathfrak{B}'$ satisfies Formulae \ref{eq: gen 1}: As $\mathfrak{A}\cdot\mathcal{E}$ satisfies \ref{eq: gen 1} we find for every $\tp \in \mathcal{E}_{+}$ a tuple $\bold{a}$ such that $\bold{a}\in \sig{Tp}_{\tp}^{\mathfrak{A}\cdot\mathcal{E}}$. As $h'$ is a homomorphism we get $h'(\bold{a}) \in \sig{Tp}_{\tp}^{\mathfrak{B}'}$.
	To show satisfaction of Formulae \ref{eq: gen 2} we note first, that by definition of $\mathfrak{B}'$ we interpreted the type predicates exactly as the homomorphic image of their interpretation in $\mathfrak{A}\cdot\mathcal{E}$. Hence for every $\tp\in \mathcal{E}\setminus\mathcal{E}_{+}$ we obtain $\sig{Tp}_{\tp}^{\mathfrak{B}'} = \emptyset$ as $\sig{Tp}_{\tp}^{\mathfrak{A}\cdot\mathcal{E}}$ was already empty. This yields Formulae \ref{eq: gen 2}.
	\newline
	Now let $\tp, \tp' \in \mathcal{E}$ such that $\tp' = \tp|_{\bold{v}_{\tp'}}$ and let $\bold{b}$ be a tuple from $\mathfrak{B}'$ such that $\bold{b}\in \sig{Tp}_{\tp}^{\mathfrak{B}'}$. As $\bold{b}$ is in $\sig{Tp}_{\tp}^{\mathfrak{B}'}$ we know by definition of $\mathfrak{B}'$ that $\bold{b}$ is the homomorphic image of some tuple from $\mathfrak{A}\cdot\mathcal{E}$, more specifically we find a tuple $\bold{a}$ in $\mathfrak{A}\cdot\mathcal{E}$ such that $\bold{a}\in \sig{Tp}_{\tp}^{\mathfrak{A}\cdot\mathcal{E}}$ and $h'(\bold{a}) = \bold{b}$. Let $\mu\colon \bold{v}_{\tp} \to\bold{a}$ be the componentwise evaluation witnessing the satisfaction of $\sig{Tp}_{\tp}$ by $\bold{a}$ in $\mathfrak{A}\cdot\mathcal{E}$. By \Cref{lem:eligclosure} the subtuple $\bold{a}^{\ast} = \mu(\bold{v}_{\tp'})$ (where $\mu$ is applied componentwise) of $\bold{a}$ satisfies $\sig{Tp}_{\tp'}$ in $\mathfrak{A}\cdot\mathcal{E}$. By applying $h'$ to $\bold{a}^{\ast}$ we obtain a subtuple $\bold{b}^{\ast} = h'(\bold{a}^{\ast})$ of $\bold{b}$ which satisfies $\bold{b}^{\ast}\in \sig{Tp}_{\tp'}^{\mathfrak{B}'}$. This ensures the satisfaction of Formulae \ref{eq: gen 3} by $\mathfrak{B}'$. 
	\newline
	Finally, let $\tp \in \mathcal{E}$, $\kappa \colon \bold{v} \injsur \bold{v}$ a permutation of the variables and $\bold{b}$ a tuple from $\mathfrak{B}'$ such that $\bold{b}\in \sig{Tp}_{\tp}^{\mathfrak{B}'}$. Since by definition of $\mathfrak{B}'$ the interpretation $\sig{Tp}_{\tp}^{\mathfrak{B}'}$ is non-empty if and only if $\tp \in \mathcal{E}_{+}$ it 	suffices to assume just that. From $\bold{b}\in \sig{Tp}_{\tp}^{\mathfrak{B}'}$ we obtain via the definition of $\mathfrak{B}'$ that there is a tuple $\bold{a}$ in $\mathfrak{A}\cdot\mathcal{E}$ such that $\bold{a}\in \sig{Tp}_{\tp}^{\mathfrak{A}\cdot\mathcal{E}}$ and $h'(\bold{a}) = \bold{b}$. As Formulae \ref{eq: gen 3plus} hold in $\mathfrak{A}\cdot\mathcal{E}$ we can conclude that $\eta_{\bold{v}}^{\kappa\bold{v}_{\tp}}(\bold{a})\in \sig{Tp}_{\kappa\tp}^{\mathfrak{A}\cdot\mathcal{E}}$. As the application of maps that are defined on elements	extends to tuples componentwise, $\eta_{\bold{v}}^{\kappa\bold{v}_{\tp}}$ and $h'$ commute, hence 
	\begin{equation}
		h'(\eta_{\bold{v}}^{\kappa\bold{v}_{\tp}}(\bold{a})) = \eta_{\bold{v}}^{\kappa\bold{v}_{\tp}}(h'(\bold{a})) = \eta_{\bold{v}}^{\kappa\bold{v}_{\tp}}(\bold{b})
	\end{equation}
	and since $h'$ is a homomorphism we obtain immediately $\eta_{\bold{v}}^{\kappa\bold{v}_{\tp}}(\bold{b})\in \sig{Tp}_{\kappa\tp}^{\mathfrak{B}'}$. This proves that $\mathfrak{B}'$ satisfies Formulae \ref{eq: gen 3plus} and concludes the whole proof of the lemma.
\end{proof}

\begin{proposition}[restate=FOTEinESO, name=]\label{prop:charFOte}
	For every \FOte{} $\tau$-sentence $\Phi$, 
	there exists
	a \FOte{} 
	sen\-tence $\Psi$ 
	such that $\homcl{\getmodels{\Phi}} = \getmodels{\Psi}|_\tau$. 
\end{proposition}

\begin{proof}
Thanks to \Cref{lem:normalformOK}, we can w.l.o.g. assume $\Phi$ to be in Scott's normal form, i.e. 
\begin{equation}
\Phi = \forall xy. \varphi(x,y) \wedge \bigwedge_i \forall x \exists y.\psi_i(x,y),
\end{equation}
 where $\varphi$ and all $\psi_i$ are quantifier-free $\FOte$ $\tau$-formulae.
   	
Let $\bold{v} = \{ v_{1}, v_{2} \}$ be the set of dedicated type variables we will use in this proof.
We now define the set $\mathcal{E}$ of eligible types as all rigid $1$-types and $2$-types (using variables from $\bold{v}$) wrt. $\tau$. Note that $\mathcal{E}$ is closed under variable permutations.
%
	As there are only two permutations on $\{ v_{1}, v_{2} \}$ (the identity and the permutation $\kappa$ mapping $v_1$ to $v_2$ and $v_2$ to $v_1$) we will denote for every $\tp\in\mathcal{E}$ by $\tp^-$ its flipped version $\kappa\tp$. Note that this leaves constants untouched.
	
Let now 
\begin{equation} 
\Psi = \Psi_\mathrm{hom}\ \wedge \hspace{-2ex}\bigvee_{(\mathcal{E}_\text{+},\mathcal{E}_!)\in \getsummary{\Phi}} \hspace{-2ex} \Psi_{(\mathcal{E}_\text{+},\mathcal{E}_!)},
\end{equation}  
where $\Psi_{(\mathcal{E}_\text{+},\mathcal{E}_!)}$ is obtained as the conjunction over $\Psi_{\mathrm{gen}, (\mathcal{E}_{+}, \mathcal{E}_{!})}$ and the following sentences (letting $\mathcal{E}_{\text{+},1}$ and $\mathcal{E}_{\text{+},2}$ denote the $1$- and $2$-types in $\mathcal{E}_{\text{+}}$, respectively, and letting $\mathcal{E}_{!,1}$ denote the $1$-types in $\mathcal{E}_{!}$):
	\noindent
	\begin{align}
%
%
%
		\!\!\!\!\!\forall xy. \sig{Tp}_\tp(x) \wedge \sig{Tp}_\tp(y) & \Rightarrow x=y & & \hspace{-1ex}\mbox{for }\tp \in \mathcal{E}_{\text{!},1} \label{proj0}\\
%
%
%
%
%
		\!\!\!\!\forall xy. \sig{Tp}_{\tp'} (x) \wedge \sig{Tp}_{\tp''} (y) & \Rightarrow \hspace{-8.5ex} \bigvee_{{\scriptstyle \tp \in \mathcal{E}_{\text{+},2}} \atop {\scriptstyle \quad\quad\quad\  \tp|_{v_1} = \tp',\,  \tp|_{v_2} = \tp''}} \hspace{-8.5ex} \sig{Tp}_{\tp} (x,y) & &\hspace{-1ex}\mbox{for } \tp'\!,\!\tp''\! \in\hspace{-1pt} \mathcal{E}_{\text{+},1}\hspace{-2ex} \label{proj5}\\
		\forall x \exists y. \Big(\hspace{-2.2ex}\bigvee_{\scriptstyle\quad\tp \in \mathcal{E}_{\text{+},1}} \hspace{-2ex} \sig{Tp}_\tp (x) & \Rightarrow \hspace{-7.5ex} \bigvee_{{\scriptstyle\quad\tp' \in \mathcal{E}_{\text{+},2}}\atop{\scriptstyle\quad\quad\quad \tp' \models \psi_i(v_1,v_2)}} \hspace{-7.2ex} \sig{Tp}_{\tp'} (x,y)\Big) & &\hspace{-1ex}\hspace{0.5ex} \psi_i \mbox{ from } \Phi \label{proj6}
	\end{align} 
	
	Clearly, $\Psi$ is a \FOte{} $\tau{\uplus}\sigma$-sentence.
	Moreover, we obtain $\homcl{\getmodels{\Phi}} = \getmodels{\Psi}|_\tau$:
	
	``$\subseteq$'': Let $\mathfrak{A} \in \getmodels{\Phi}$ and let $h\colon\mathfrak{A} \to \mathfrak{B}$. Obtain $\mathfrak{A}'=\mathfrak{A}\cdot\mathcal{E}$ and $\mathfrak{B}'=\mathfrak{B}\cdot h(\mathfrak{A}\cdot\mathcal{E}|_\sigma)$ as described before. Pick $(\mathcal{E}_\text{+},\mathcal{E}_!) = \mathcal{E}(\mathfrak{A})$. Then
	$\mathfrak{A}' \models \Psi_{(\mathcal{E}_\text{+},\mathcal{E}_!)}$ because of $\Phi$-modelhood and $\mathfrak{A}' \models \Psi_\mathrm{hom}$ by construction. 
	Furthermore, by applying \Cref{lem:PsiGen} $\mathfrak{A}'$ as well as $\mathfrak{B}'$ satisfy $\Psi_{\mathrm{gen}, (\mathcal{E}_{+}, \mathcal{E}_{!})}$.
	Next, Formulae \ref{proj0} holds by definition of $\mathfrak{A}'$ and the fact that $\mathcal{E}_{!, 1}$ contains the $1$-types realized exactly once. Hence every type predicate associated to such a $1$-type can only be satisfied by exactly one element from $\mathfrak{A}'$. By definition of $\mathfrak{B}'$ Formulae \ref{proj0} also hold there: Since every type predicate from $\sigma$ is interpreted in $\mathfrak{B}'$ as the homomorphic image of the interpretation of the very same type predicate in $\mathfrak{A}'$, and those are singleton sets, we obtain the claimed satisfaction.
	Now let $b_{1} , b_{2}$ be elements from $\mathfrak{B}'$ such that there are types $\tp', \tp'' \in\mathcal{E}_{+,1}$ with $b_{1} \in\sig{Tp}_{\tp'}^{\mathfrak{B}'}$ and $b_{1} \in\sig{Tp}_{\tp'}^{\mathfrak{B}'}$. By definition there are $a_{1}$ and $a_{2}$ in $\mathfrak{A}'$ with $a_1 \in\sig{Tp}_{\tp'}^{\mathfrak{A}'}$ and $a_1 \in\sig{Tp}_{\tp'}^{\mathfrak{A}'}$ such that $h(a_1) = b_1$ and $h(a_2) = b_2$.  Without loss of generality we now assume $\bold{v}_{\tp'} = \{v_1 \}$ and $\bold{v}_{\tp''} = \{ v_2 \}$. If this were not the case, using the closure under variable permutation we could enforce this condition. Note that it will be only important for the next step in a technical way, to extract a $2$-type that is rigid. Because now we obtain by definition of $\mathfrak{A}'$ (and the uniqueness of the realized types of elements/tuples) that $a_1$ realizes $\tp'$, $a_2$ realizes $\tp''$ and $(a_{1}, a_{2})$ realizes some type $\tp$ that naturally satisfies $\tp|_{v_1} = \tp'$ and $\tp|_{v_2} = \tp''$. Hence $(a_1, a_2) \in\sig{Tp}_{\tp}^{\mathfrak{A}'}$ and applying $h$ we obtain $(b_1, b_2) = h(a_1, a_2) \in \sig{Tp}_{\tp}^{\mathfrak{B}'}$.
	This argument yields \ref{proj5}.
	Finally we discuss that $\mathfrak{B}'$ satisfies \ref{proj6}. Suppose $b$ be in $\mathfrak{B}'$ and let $\tp\in\mathcal{E}_{+, 1}$ be some type such that $b \in\sig{Tp}_{\tp}^{\mathfrak{B}'}$. Hence there is $a$ in $\mathfrak{A}'$ such that $h(a)=b_1$ and $a \in\sig{Tp}_{\tp}^{\mathfrak{A}'}$. As $\mathfrak{A}'$ as an expansion has the same domain as $\mathfrak{A}$ and $\mathfrak{A}\vDash \Phi$ we obtain an element $a_{i}$ such that $(a, a_i)$ satisfies $\psi_{i}(x,y)$ for all $\forall\exists$-conjuncts in $\Phi$. The type $\tp(v_1,v_2)$ realized by $(a,a_i)$ is in $\mathcal{E}_{+, 2}$ and we have $\tp'\vDash\psi_{i}(v_1, v_2)$. Hence $(a,a_i) \in\sig{Tp}_{\tp'}^{\mathfrak{A}'}$. Applying $h$ we first obtain a $b_i = h(a_i)$ such that also $(b,b_i) \in\sig{Tp}_{\tp'}^{\mathfrak{B}'}$ holds. By this argument \ref{proj6} holds in $\mathfrak{B}'$.
	Consequently $\mathfrak{B}'$ must be a model of $\Psi$, hence $\mathfrak{B} \in \getmodels{\Psi}|_\tau$.


	``$\supseteq$'': 	Let $\mathfrak{B}'$ be a model of $\Psi$ and let $\Psi_{(\mathcal{E}_\text{+},\mathcal{E}_!)}$ be the disjunct made true in the rear part of $\Psi$. 
		The crucial observation about the structure $\mathfrak{B}' $ is that in general, elements $a$ or a tuple $(a,b)$ from $B'$ may satisfy more than one type predicate. Note that if this would not be the case, we'd easily find a weak substructure $ \mathfrak{A}'$ of $\mathfrak{B}$ which is projectively in $ \getmodels{\Phi}$. This holds since the type predicates inherit the satisfiability of $\Phi$. In order to overcome this difficulty of ``ambiguity of types'' we 
	construct a structure $\mathfrak{A} \in \getmodels{\Phi}$ and a homomorphism $\mathfrak{A} \to \mathfrak{B}'|_\tau$ as follows: first construct the $\sigma$-structure $\mathfrak{A}'$ by letting 
	\begin{equation} 
	A' = \{ (b,\tp) \mid b\in B, b\in \sig{Tp}_\tp^\mathfrak{B}, \tp\in \mathcal{E}_{\text{+},1}\}.\label{projdef1}
    \end{equation}  
	For the unary type predicates, let 
	$\sig{Tp}_\tp^{\mathfrak{A}'} = A' \cap (B \times \{\tp\})$ 
	and for binary type predicates, let
	\begin{align*}
		& \sig{Tp}_{\tp}^\mathfrak{A'} = \Big\{ \big((b,\tp'),\!(b'\!,\tp'')\big)\ \Big|\ 
		\tp|_{v_1}{=}\tp', \tp|_{v_2}{=}\tp'', (b,b') \in \sig{Tp}_{\tp}^\mathfrak{B} \Big\} \label{eq75}
	\end{align*} 
	In a next step, we define $\mathfrak{A}$ from $\mathfrak{A}'$ as follows: let $k$ be the maximum number of distinct $\tp'' \in \mathcal{E}_{\text{+},2}$ with $(a,a') \in \sig{Tp}_{\tp''}^\mathfrak{A'}$ across all pairs $(a,a')$ from $A'$. Define a function
	$\gettp \colon A' \times A' \times \{0,\ldots,k-1\} \to \mathcal{E}_{\text{+},2}$ such that for every $a,a' \in A'$, the set $\{ \gettp(a,a',i) \mid 0\leq i < k \}$ contains all and only the  distinct $\tp''$ with $(a,a') \in \sig{Tp}_{\tp''}^\mathfrak{A'}$ (with possible repetitions) and $\gettp(a,a',i) = \gettp(a',a,i)^-$. Note that the last identity is possible by (\ref{eq: gen 3plus}).
	Now, let 
	\begin{equation} 
	A_! = \{(b,\tp) \in A' \mid \tp \in \mathcal{E}_{\text{!},1} \}
	\end{equation}  
	and let
	\begin{equation} 
	A = A_! \times \{0\} \cup (A' \setminus A_!) \times \{0,\ldots, k-1\}. 
	\end{equation}  

Note that $A_!$ contains for every $\tp \in \mathcal{E}_{\text{!},1}$ exactly one element of the form $(b,\tp) $ since $\mathfrak{B}$ satisfies (\ref{proj0}).
	We assign relations to the type predicates as follows (for all $\tp \in \mathcal{E}_{\text{+},1}$ and $\tp'' \in \mathcal{E}_{\text{+},2}$):
\begin{eqnarray}
	\sig{Tp}_\tp^\mathfrak{A} & \hspace{-2ex} = \hspace{-2ex} & \{ (a,i) \mid a \in \sig{Tp}_\tp^\mathfrak{A'} \} \text{ for } \tp \in \mathcal{E}_{\text{+},1} \label{projdef2}\\
    \sig{Tp}_{\tp''}^\mathfrak{A} & \hspace{-2ex} = \hspace{-2ex} & \{ ((a,i),(a',j)) \mid \gettp(a,a',i\,{+}\,j\!\! \mod k) = \tp''\}. \quad\quad 	     
\end{eqnarray}

It follows from our construction that the structure $\mathfrak{A}$ satisfies formulae (\ref{eq: gen 1}), (\ref{eq: gen 2}), and (\ref{proj0}). To see that (\ref{eq: gen 3}) is satisfied consider $((a,i),(a',j))\in A^2$ such that  $((a,i),(a',j)) \in \sig{Tp}_\tp^\mathfrak{A}$ holds. By the definition of $\sig{Tp}_\tp^\mathfrak{A}$ we get $\gettp(a,a',i\,{+}\,j\!\! \mod k) = \tp$. 
	Let 
	 $(a,a')=((b,\tp^a),(b',\tp^{a'})   )$.
	Since $(a,a')\in \sig{Tp}_{\tp}^{{\mathfrak{A}'}}$ holds  by definition of $\gettp$ we get by the definition of  $\sig{Tp}_{\tp}^{\mathfrak{A}'}$ that $\tp|_{v_1}=\tp^a$ holds.
	By the definition of  $\sig{Tp}_{\tp}^{\mathfrak{A}'}$ and since $\mathfrak{B}$ satisfies (\ref{eq: gen 3}) we have that $(b,\tp^a)$ is in $\sig{Tp}_{\tp^a}^{\mathfrak{A}'}$.
	This implies that $((b,\tp^a),i)$ is in $\sig{Tp}_{\tp^a}^{\mathfrak{A}}$. Since  $\tp|_{v_1}=\tp^a$ holds and the same argument works for $v_{2}$ and $\tp^{a'}$, we showed that (\ref{eq: gen 3}) is satisfied in $\mathfrak{A}$.

Formulae (\ref{eq: gen 3plus})  clearly hold in $\mathfrak{A}$. In consequence, $\Psi_{(\mathcal{E}_{+}, \mathcal{E}_{!})}$ as well as (\ref{proj0}) holds in $\mathfrak{A}$.

	For signature elements from $\tau$, choose $\cdot^\mathfrak{A}$ such that all (pairs of) domain elements realize exactly the types as indicated by the type predicates. This is possible since, by construction, there is exactly one $2$-type-predicate holding for every pair of domain elements.
   Furthermore, every domain element has been assigned exactly two $1$-type-predicates (whose associated types only differ in the used variable name but are otherwise identical by (\ref{eq: gen 3plus})) by the definitions (\ref{projdef1}) and  (\ref{projdef2}).
		The $1$-type-predicates  also fit with all the two-type predicates 
	by what we have shown before.
	
	We now claim that $\mathfrak{A}|_\tau$ is a model of $\Phi$. To see this, consider an arbitrary element $a=((b,\tp),i)\in A$. Since $\mathfrak{B}$ satisfies (\ref{proj6}), there exists a $b'\in B$ such that that the inner implication in (\ref{proj6}) is satisfied (for $x=b$ and $y=b'$). 
	By the construction of $\mathfrak{A}$, we know  that $b\in \sig{Tp}_{\tp}^{\mathfrak{B}}$ holds and therefore the premise of the implication is satisfied. This implies the existence of a type $\tp'$ such that $(b,b')\in \sig{Tp}_{\tp'}^{\mathfrak{B}}$ and
	$	\tp' \models \psi_i(v_1,v_2)$ hold. 
	By the definition of $\gettp$, we can find a $j$ such that $\gettp( (b,\tp),(b',\tp'|_{v_2}),i\,{+}\,j\!\! \mod k) = \tp'$. Therefore, we get for $a'=((b',\tp'|_{v_2}),j) $ that $(a,a') \in \sig{Tp}_{\tp'}^{\mathfrak{A}} $ holds. By the definition of the $\tau$-relations from $\mathfrak{A}$, it follows that $\psi_{j}(a,a')$ holds in $\mathfrak{A}$. 
		The proof for the conjunct $\forall xy. \varphi(x,y) $ is by an analogous (but easier, since ``type-internal'') use of formula (\ref{proj5}).
		Consider arbitrary elements $a=((b, \tp),i)\in A$ and $a'=((b', \tp'), i')\in A$. By construction of $\mathfrak{A}$ we know that $b\in\sig{Tp}_{\tp}^{\mathfrak{B}}$ and $b'\in\sig{Tp}_{\tp'}^{\mathfrak{B}}$. By \ref{proj5} there is a at least one $\tp''\in\mathcal{E}_{+, 2}$ with $\tp''|_{v_1}=\tp$ and $\tp''|_{v_2}=\tp'$ such that $(b,b')\in\sig{Tp}_{\tp''}^{\mathfrak{B}}$. Hence by definition of $\gettp$ one of those types $\tp''$ satisfies $\gettp((b, \tp), (b', \tp'), i\,{+}\,i'\!\! \mod k) = \tp''$. Consequently $(a,a')\in\sig{Tp}_{\tp''}^{\mathfrak{A}}$ by construction. By the definition of the $\tau$-relations from $\mathfrak{A}$ we immediately obtain that $(a,a')$ realizes $\tp''$. Now note that, as $\tp''\in\mathcal{E}_{+,2}$, $\tp''\vDash\varphi(v_1,v_2)$ since $\mathcal{E}_{+}$ collects all eligible types realized by some $\mathfrak{A}^{\ast}\in\getmodels{\Phi}$. Hence $(a,a')$ satisfies $\varphi$ in $\mathfrak{A}$. Thus the conjunct $\forall xy. \varphi(x,y)$ holds for $\mathfrak{A}$.
	
By our construction, $h\colon \mathfrak{A}|_\tau \to \mathfrak{B}|_\tau$ with $((b,\tp),i) \mapsto b$ is a homomorphism, which proves the statement.
\end{proof}


Before characterizing the homclosures of a \GFO{} or a \TGF{} sentence, we first discuss normal forms for each fragment, which we will facilitate in the upcoming proofs.

\begin{lemma}\label{lem:GFONF}
	Let $\Phi$ be a \GFO{} $\tau$-sentence. Then there exists a \emph{normal form} $\Phi'$ over an extended signature $\sigma \uplus \tau \uplus \{ \sig{D} \}$ (with $\sig{D}$ being a unary relational symbol) such that $\Phi'$ has the form $\Phi' = \Phi_{\sig{D}}\land\Phi'_{\forall} \land \Phi'_{\forall\exists}$ where
	\begin{align}
		\Phi_{\sig{D}} &= \forall x. \sig{D}(x) ,\label{GFONF phi_D} \\
		\Phi'_{\forall} &= \bigwedge_{i} \forall \bold{x}.\Big(\alpha_{i}(\bold{x})\Rightarrow\vartheta_{i}[\bold{x}]\Big) \label{eq: GFONF phi_A}, \\
		\Phi'_{\forall\exists} &= \bigwedge_{j}\forall \bold{z}.\Big(\beta_{j}(\bold{z})\Rightarrow\exists\bold{y}.\gamma_{j}(\bold{yz})\Big) \label{eq: GFONF phi_AE},
	\end{align}
	with $\vartheta_{i}$ is a disjunction of literals, $\alpha_{i}, \beta_{j}$ are guard atoms, and $\gamma_{j}$ is an atom,	such that every model of $\Phi$ can be extended to a model of $\Phi'$ and any model of $\Phi'$ restricted to the signature $\tau$ is a model of $\Phi$.
\end{lemma}
\begin{proof}
	Let $\Phi$ be a \GFO{} $\tau$-sentence.
	The construction will be done iteratively. As a note on notation we will collect for every iteration step $i$ in the first big transformation some sentences in sets $\bold{A}^{i}_{\Phi}$ and $\bold{E}^{i}_{\Phi}$. The collected sentences will be of specific forms: Sentences from $\bold{A}^{i}_{\Phi}$ have the form 
	\begin{equation}
		\forall\bold{x}.\alpha(\bold{x}) \Rightarrow \vartheta[\bold{x}]
	\end{equation}
	(where $\alpha$ is a guard atom and $\vartheta$ is quantifier-free) and those from $\bold{E}^{i}_{\Phi}$ have the form 
	\begin{equation}
		\forall \bold{z}.\Big(\beta(\bold{z})\Rightarrow\exists\bold{y}.\gamma(\bold{yz}) \land\varphi[\bold{yz}]\Big)
	\end{equation}
	(where $\beta$ and $\gamma$ are guard atoms and $\varphi$ is quantifier-free). Finally, $\bold{A}^{i}_{\Phi}$ and $\bold{E}^{i}_{\Phi}$ will constitute two finite monotonically increasing sequences of sets, i.e. $\bold{A}^{i}_{\Phi} \subseteq \bold{A}^{i+1}_{\Phi}$ and $\bold{E}^{i}_{\Phi} \subseteq \bold{E}^{i+1}_{\Phi}$ for every $i\ge 0$ such that $\bold{A}^{i+1}_{\Phi}$ and $\bold{E}^{i+1}_{\Phi}$ are defined.
	
	If $\chi$ is a (proper) subformula in $\Phi$ and $\psi$ a formula with the same free variables as $\chi$, we will write $\Phi[\chi / \psi]$ for the formula we obtain by replacing every occurrence of $\chi$ in $\Phi$ by $\psi$.
	
	In a preliminary step we will transform $\Phi$ to a sentence $\Phi^{0}$ to ensure that
	\begin{enumerate}
		\item\label{item: GFONF 1 1} every proper subformula of $\Phi^{0}$ has at least one free variable,
		\item\label{item: GFONF 1 2} $\Phi^{0}$ has an outermost universal and guarded quantification (and every other quantification is in the scope of this universal quantifier), and
		\item\label{item: GFONF 1 3} every quantification properly inside the scope of the outermost universal one is existential. 
	\end{enumerate}
	Let $u$ be a fresh variable symbol completely unused. To establish Item \ref{item: GFONF 1 1} we introduce for every relational symbol  $\sig{R}\in \tau$ a fresh symbol $\sig{R}'$ (which we add to $\sigma$) such that $\mathrm{ar}(\sig{R}') = \mathrm{ar}(\sig{R}) + 1$. We then substitute every atom $\sig{R}(\bold{t})$ in $\Phi$ by the corresponding $\sig{R}'(\bold{t}, u)$. The so obtained formula will be denoted by $\psi(u)$. Note that by uniformly using $u$ the formula $\psi(u)$ stays guarded. We furthermore let $\bold{A}^{0}_{\Phi}$ consist of just the sentences $\forall \bold{x}\forall u. \sig{R}'(\bold{x}, u) \Rightarrow \sig{R}(\bold{x})$ for each relational symbol $\sig{R}\in\tau$. Also we set $\bold{E}_{\Phi}^{0}=\emptyset$.
	\newline
	To ensure Item \ref{item: GFONF 1 2} we introduce the fresh relation symbol $\sig{D}$ of arity $1$. We now let $\Psi^{0}$ be the sentence $\forall u. \sig{D}(u) \Rightarrow \psi(u)$.
	\newline
	Ensuring Item \ref{item: GFONF 1 3} is a simple exercise of double negation: Let every occurrence of a subformula from $\psi(u)$ of the form $\forall \bold{x}.\alpha(\bold{xy}) \Rightarrow \varphi[\bold{xy}]$ be substituted by $\lnot \exists \bold{x}. \alpha(\bold{xy}) \land \lnot\varphi[\bold{xy}]$. The so obtained formula will be denoted by $\psi'(u)$.
	
	Hence we obtain a guarded sentence $\Phi^{0} = \forall u. \sig{D}(u) \Rightarrow \psi'(u)$ that satisfies Items \ref{item: GFONF 1 1} - \ref{item: GFONF 1 3}. Note that every model of $\mathfrak{A}\in \getmodels{\Phi}$ can be expanded to one of $\Phi_{\sig{D}} \land \Phi^{0} \land \bigwedge_{\Theta \in \bold{A}^{0}_{\Phi}}\Theta$. To see this, let $\mathfrak{A}'$ (with domain $A$) be the expansion of $\mathfrak{A}$, where $\sig{D}^{\mathfrak{A}'} = A$ and for every $\sig{R}\in \tau$ we set $(\sig{R}')^{\mathfrak{A}'} = \sig{R}^{\mathfrak{A}}\times A$. $\mathfrak{A}'$ satisfies $\Phi_{\sig{D}}\land\Phi^{0} \land \bigwedge_{\Theta \in \bold{A}^{0}_{\Phi}}\Theta$.
	\newline
	On the other hand, if $\mathfrak{A}$ is a model of $\Phi_{\sig{D}}\land\Phi^{0} \land \bigwedge_{\Theta \in \bold{A}^{0}_{\Phi}}\Theta$, then $\mathfrak{A}|_{\tau}$ is one of $\Phi$. To ensure this we added to $\bold{A}^{0}_{\Phi}$ the sentences $\forall \bold{x}\forall u. \sig{R}'(\bold{x}, u) \Rightarrow \sig{R}(\bold{x})$. Note that, since we do not allow empty models, $\Phi_{\sig{D}}$ also ensures the non-emptiness of interpretations of $\sig{D}$.

	We now begin to iteratively ``unnest'' the quantifiers in $\Phi^{0}$.
	Let $\chi_{0}(\bold{z})$ be a proper subformula of $\Phi^{0}$ of the form
	\begin{equation}
		\chi_{0}(\bold{z}) = \exists\bold{y}. \beta(\bold{yz}) \land \varphi[\bold{yz}]
	\end{equation}
	with $\varphi$ being quantifier-free. We introduce a fresh relation symbol $\sig{R}_{\chi_{0}}$ (which we add to $\sigma$) of the same arity as length of $\bold{z}$. We obtain a sentence $\Phi^{1}$ by substituting $\sig{R}_{\chi_{0}}(\bold{z})$ for each occurrence of $\chi_{0}(\bold{z})$ in $\Phi^{0}$, i.e.
\begin{equation}
		\Phi^{1} = \Phi^{0}[\chi_{0}(\bold{z}) / \sig{R}_{\chi_{0}}(\bold{z})] .
	\end{equation}
	Let $\bold{E}^{1}_{\Phi}$ be $\bold{E}^{0}_{\Phi}$ with the added sentence
	\begin{equation}
		\forall\bold{z}. \sig{R}_{\chi_{0}}(\bold{z}) \Rightarrow \exists\bold{y}. \beta(\bold{yz})\land \varphi[\bold{yz}]
	\end{equation}
	and let $\bold{A}^{1}_{\Phi}$ be $\bold{A}^{0}_{\Phi}$ with the added sentence
	\begin{equation}
		\forall\bold{yz}. \beta(\bold{yz}) \Rightarrow (\varphi[\bold{yz}] \Rightarrow R_{\chi_{0}}(\bold{z})) . 
	\end{equation}
	Note that the latter sentence is of the form $\forall\bold{yz}. \beta(\bold{yz}) \Rightarrow \vartheta[\bold{yz}]$ with $\vartheta$ quantifier-free, as $(\varphi[\bold{yz}] \Rightarrow R_{\chi_{0}}(\bold{z}))$ is quantifier-free.
	\newline
	By iterating this process, picking from $\Phi^{i}$ a proper subformula $\chi_{i}(\bold{z})$ of the form \[ \chi_{i}(\bold{z}) = \exists\bold{y}. \beta(\bold{yz}) \land \varphi[\bold{yz}] \] with $\varphi$ being quantifier-free and introducing a fresh relation symbol $\sig{R}_{\chi_{i}}$ obtaining $\Phi^{i+1}$ and enlarging $\bold{A}^{i}_{\Phi}$ and $\bold{E}^{i}_{\Phi}$ to $\bold{A}^{i+1}_{\Phi}$ and $\bold{E}^{i+1}_{\Phi}$ respectively, we will terminate at some $n$ with the sentence
	\begin{equation}
		\Phi^{n} = \forall u. \sig{D}(u) \Rightarrow \sig{R}_{\chi_{n-1}}
	\end{equation}
	and sets $\bold{A}^{n}_{\Phi}$, $\bold{E}^{n}_{\Phi}$. We set $\bold{A}_{\Phi}$ to be $\bold{A}^{n}_{\Phi}$ with $\Phi^{n}$ added and $\bold{E}_{\Phi}= \bold{E}^{n}_{\Phi}$.
	\newline
	Note that in every iteration, every model $\mathfrak{A} \in \getmodels{\Phi_{\sig{D}}\land\Phi^{i-1}}$ can be expanded to one of $\Phi_{\sig{D}}\land\Phi^{i}\land \Theta_{\bold{A},  i}\land \Theta_{\bold{E}, i}$ with $\Theta_{\bold{A}, i}$ being the unique element from the singleton set $\bold{A}^{i}_{\Phi} \setminus \bold{A}^{i-1}_{\Phi}$ and $\Theta_{\bold{E}, i}$ being the unique element from the singleton set $\bold{E}^{i}_{\Phi} \setminus \bold{E}^{i-1}_{\Phi}$. The trick is to interpret $\sig{R}_{\chi_{i-1}}$ to be the set of tuples satisfying $\chi_{i-1}$ in $\mathfrak{A}$.  On the other hand, by ``forgetting'' the relation symbols introduced in the $i$-th iteration we obtain from a model $\mathfrak{A}\in\getmodels{\Phi_{\sig{D}}\land\Phi^{i}\land \Theta_{\bold{A}, i}\land \Theta_{\bold{E}, i}}$ (where $\Theta_{\bold{A}, i}$ is the unique element from the singleton set $\bold{A}^{i}_{\Phi} \setminus \bold{A}^{i-1}_{\Phi}$ and $\Theta_{\bold{E}, i}$ is the unique element from the singleton set $\bold{E}^{i}_{\Phi} \setminus \bold{E}^{i-1}_{\Phi}$) a model of $\Phi_{\sig{D}}\land\Phi^{i-1}$. The choice of $\Theta_{\bold{A}, i}$ and $\Theta_{\bold{E}, i}$ in the construction ensures that this works.
	
	As now every information of our initial $\Phi^{0}$ is stored in the sets $\bold{A}_{\Phi}$ and $\bold{E}_{\Phi}$, the next manipulations will take place inside those sets. First we look at $\bold{E}_{\Phi}$. We know that every sentence from $\bold{E}_{\Phi}$ is of the form
	\begin{equation}
		\forall\bold{xz}. \beta(\bold{xz}) \Rightarrow \exists\bold{y}. \gamma(\bold{yz})\land \varphi[\bold{yz}], 
	\end{equation}
	where $\beta$ and $\gamma$ are respective guard atoms and $\varphi$ is quantifier free. We go through every sentence of $\bold{E}_{\Phi}$ and do the following steps:
	\begin{itemize}
		\item For every $\Theta \in \bold{E}_{\Phi}$ such that
		\begin{equation}
			\forall\bold{xz}. \beta(\bold{xz}) \Rightarrow \exists\bold{y}. \gamma(\bold{yz})\land \varphi[\bold{yz}]
		\end{equation}
		we introduce fresh relation symbols $\sig{P}_{\Theta}$ (arity being equal to the length of $\bold{z}$) and $\sig{Q}_{\Theta}$ (arity being the length of $\bold{yz}$) to $\sigma$.
		\item Then we decompose each $\Theta$ into the three sentences 
			\begin{align}
				\forall\bold{xz}. \beta(\bold{xz}) &\Rightarrow  \sig{P}_{\Theta}(\bold{z}) \label{eq: GFONF 1} \\
				\forall\bold{z}. \sig{P}_{\Theta}(\bold{z}) &\Rightarrow \exists\bold{y}. \sig{Q}_{\Theta}(\bold{yz}) \label{eq: GFONF 2} \\
				\forall\bold{yz}. \sig{Q}_{\Theta}(\bold{yz}) &\Rightarrow \gamma(\bold{yz})\land \varphi[\bold{yz}] \label{eq: GFONF 3} .
			\end{align}
		\item Finally, we substitute in $\bold{E}_{\Phi}$ each sentence $\Theta$ by the corresponding sentence \ref{eq: GFONF 2} obtained in the previous step and furthermore we add to $\bold{A}_{\Phi}$ all the sentences \ref{eq: GFONF 1} and \ref{eq: GFONF 3} yielded by decomposing each $\Theta$.
	\end{itemize}
	This step leaves us with particularly simple sentences in $\bold{E}_{\Phi}$: They are just of the form
	\begin{equation}
		\forall\bold{z}.\beta(\bold{z}) \Rightarrow \exists\bold{y}. \gamma(\bold{yz})
	\end{equation}
	where $\beta$ and $\gamma$ are atoms. Furthermore, $\gamma$ contains all the variables contained in $\beta$.

	Finally we simplify the sentences in the set $\bold{A}_{\Phi}$. First we note, that every sentence $\Theta \in \bold{A}_{\Phi}$ is of the form
	\begin{equation}
		\forall\bold{x}. \alpha(\bold{x}) \Rightarrow \vartheta[\bold{x}]
	\end{equation}
	with $\alpha$ a guard atom and $\vartheta$ being a quantifier-free formula. We now transform $\vartheta$ into conjunctive normal form, i.e. $\vartheta = \bigwedge_{i}\vartheta_{i}$, where the $\vartheta_{i}$ are disjunctions over literals. As the conjunction commutes with the universal quantification, $\varphi$ is equivalent to the sentence
	\begin{equation}
		\bigwedge_{i}\forall\bold{x}. \alpha(\bold{x}) \Rightarrow \vartheta_{i}[\bold{x}] .
	\end{equation}
	Replacing $\Theta$ in $\bold{A}_{\Phi}$ by all the sentences 
	\begin{equation}
		\forall\bold{x}. \alpha(\bold{x}) \Rightarrow \vartheta_{i}[\bold{x}]
	\end{equation}
	of that conjunction and doing this for every $\Theta\in\bold{A}_{\Phi}$ yields us a simpler $\bold{A}_{\Phi}$. To summarize, every sentence of $\bold{A}_{\Phi}$ now has the form 
	\begin{equation}
		\forall\bold{x}. \alpha(\bold{x}) \Rightarrow \vartheta[\bold{x}]
	\end{equation}
	with $\alpha$ a guard atom and $\vartheta$ a disjunction of literals.

	In consequence we obtain the sentence $\Phi' = \Phi_{\sig{D}}\land \Phi'_{\forall}\land\Phi'_{\forall\exists}$ where
	\begin{align}
		\Phi_{\sig{D}} & =  \forall x. \sig{D}(x),\\[2ex]
		\Phi'_{\forall} & =  \bigwedge_{\Theta \in \bold{A}_{\Phi}}\Theta, \text{ and} \\	
		\Phi'_{\forall\exists} & =  \bigwedge_{\Theta\in\bold{E}_{\Phi}}\Theta .
	\end{align}
	Note that $\Phi_{\sig{D}}$ is guarded, hence $\Phi'$ is, too. Finally, as every step ensured the models could be expanded or restricted to satisfy the respective sentences after or before the step, every model of $\Phi$ can be expanded to one of $\Phi'$ and every model of $\Phi'$ can be restricted to get one of $\Phi$.
\end{proof}

As a corollary we obtain a quite  similar normal form for \TGF{}.
\begin{corollary}\label{lem:TGFNF}
	Let $\Phi$ be a \TGF{} $\tau$-sentence. Then there exists a \emph{normal form} $\Phi'$ over an extended signature $\sigma\uplus\tau\uplus\{ \sig{Univ} \}$ (where $\sig{Univ}$ is a binary auxiliary predicate) such that $\Phi'$ has the form $\Phi' = \Phi'_{\sig{Univ}}\land \Phi'_{\forall} \land \Phi'_{\forall\exists} $	where 
	\begin{eqnarray}
		\Phi_{\sig{Univ}} & = & \forall xy. \sig{Univ}(x,y),	\label{eq: TGF Univ}\\[2ex]
		\Phi'_{\forall} & = & \bigwedge_{i}\forall\bold{x}. (\alpha_{i}(\bold{x}) \Rightarrow \vartheta_{i}[\bold{x}]), \text{ and} \label{eq: TGF A}\\	
		\Phi'_{\forall\exists} & = & \bigwedge_{j}\forall\bold{z}. (\beta_{j}(\bold{z}) \Rightarrow \exists\bold{y}.\gamma_{j}(\bold{y}\bold{z})), \label{eq: TGF AE}
	\end{eqnarray}
	with $\vartheta_{i}$ a disjunction of literals, $\alpha_{i}$, $\beta_{j}$ guard atoms, and $\gamma_{j}$ an atom, such that every model of $\Phi$ can be extended to a model of $\Phi'$ and any model of $\Phi'$ restricted to the signature $\tau$ is a model of $\Phi$.
\end{corollary} 
\begin{proof}
	Let $\Phi$ be a \TGF{} $\tau$-sentence. In a preliminary step we guard every unguarded quantification (which can only happen for quantification over one or two variables) inside of $\Phi$ by $\sig{Univ}$. What we mean by this is the following: Substitute every subformula
	\begin{itemize}
		\item of the form $\exists x. \varphi[x]$ \quad \, by $\exists x. \sig{Univ}(x,x) \land \varphi[x]$,
		\item of the form $\exists x. \varphi[x,y]$ \, by $\exists x. \sig{Univ}(x,y) \land \varphi[x,y]$,
		\item of the form $\exists x y. \varphi[x,y]$ by $\exists x y. \sig{Univ}(x,y) \land \varphi[x,y]$,
		\item of the form $\forall x. \varphi[x]$ \quad \, by $\forall x. \sig{Univ}(x,x) \Rightarrow \varphi[x]$,
		\item of the form $\forall x. \varphi[x,y]$ \,  by $\forall x. \sig{Univ}(x,y) \Rightarrow \varphi[x,y]$,
		\item of the form $\forall x y. \varphi[x,y]$ by $\forall x y. \sig{Univ}(x,y) \Rightarrow \varphi[x,y]$.
	\end{itemize}
	We denote the so obtained sentence by $\Psi$. First note that every model of $\Phi_{\sig{Univ}}\land \Psi$ is a model of $\Phi$ after ``forgetting'' $\sig{Univ}$. On the other hand, by interpreting $\sig{Univ}$ as the full binary relation every model of $\Phi$ can be expanded to one of $\Phi_{\sig{Univ}}\land \Psi$.
	More importantly, $\Psi$ itself is a \GFO{} $\tau\uplus\{ \sig{Univ} \}$-sentence. By applying \Cref{lem:GFONF} we hence obtain a normal form over an extended signature $\tau \uplus\{ \sig{Univ} \} \uplus \sigma \uplus\{ \sig{D} \}$ as the sentence $\Psi' = \Phi_{\sig{D}} \land \Psi'_{\forall}\land\Psi'_{\forall\exists}$ where
	\begin{align}
		\Phi_{\sig{D}} &= \forall x. \sig{D}(x) , \\[2ex]
		\Psi'_{\forall} &= \bigwedge_{i} \forall \bold{x}.\Big(\alpha_{i}(\bold{x})\Rightarrow\vartheta_{i}[\bold{x}]\Big) , \\
		\Psi'_{\forall\exists} &= \bigwedge_{j}\forall \bold{z}.\Big(\beta_{j}(\bold{z})\Rightarrow\exists\bold{y}.\gamma_{j}(\bold{yz})\Big) ,
	\end{align}
	with $\vartheta_{i}$ a disjunction of literals, $\alpha_{i}$, $\beta_{j}$ guard atoms, and $\gamma_{j}$ an atom. In a final step we will eliminate every subformula in $\Psi'_{\forall}$ and $\Psi'_{\forall\exists}$ of the form $\sig{D}(x)$ by $\sig{Univ}(x,x)$. We will denote the so obtained $\tau \uplus\{ \sig{Univ} \} \uplus \sigma$-sentences by $\Phi'_{\forall}$ and $\Phi'_{\forall\exists}$ respectively. By dropping $\Phi_{\sig{D}}$ this yields us the desired normal form $\Phi' = \Phi_{\sig{Univ}}\land \Phi'_{\forall} \land \Phi'_{\forall\exists}$.
\end{proof}

Having established these normal forms, we are in a position of establishing further preliminary results. Though we will  first need some notations to actually talk about those.

\begin{definition}\label{def:TBPfitting}
	Let $\Phi$ be a \GFO{} or \TGF{} $\tau$-sentence in normal form, i.e. $\Phi = \Phi_{\sig{D}}\land\Phi_{\forall}\land\Phi_{\forall\exists}$ or $\Phi_{\sig{Univ}}\land\Phi_{\forall}\land\Phi_{\forall\exists}$ respectively, where
	\begin{align}
		\Phi_{\sig{D}} & =  \forall x. \sig{D}(x),	\\[2ex]
		\Phi_{\sig{Univ}} & =  \forall xy. \sig{Univ}(x,y),	\\[2ex]
		\Phi_{\forall} & =  \bigwedge_{i}\forall\bold{x}. (\alpha_{i}(\bold{x}) \Rightarrow \vartheta_{i}[\bold{x}]), \text{ and} \\	
		\Phi_{\forall\exists} & =  \bigwedge_{j}\forall\bold{z}. (\beta_{j}(\bold{z}) \Rightarrow \exists\bold{y}.\gamma_{j}(\bold{y}\bold{z})),
	\end{align}
	with $\vartheta_{i}$ a disjunction of literals, $\alpha_{i}$, $\beta_{j}$ guard atoms, and $\gamma_{j}$ an atom.  Additionally let
	$\mathrm{width}(\Phi)$ be the maximal arity of a relation symbol appearing in $\Phi$ and $\mathcal{E}$ be the set of all guarded types that are also rigid and of order $\le \mathrm{width}(\Phi)$. To every $(\mathcal{E}_{+}, \mathcal{E}_{!}) \in \getsummary{\Phi}$, we associate for each conjunct in
	\begin{itemize}
		\item $\Phi_{\forall}$ a set $T_{i}^{(\mathcal{E}_{+}, \mathcal{E}_{!})}$ containing all pairs $(\tp, \nu)$ where $\tp$ is a type from $\mathcal{E}_{+}$ guarded by $\alpha_{i}$ as witnessed by $\nu$,
		\item $\Phi_{\forall\exists}$ a set $B_{j}^{(\mathcal{E}_{+}, \mathcal{E}_{!})}$ containing all pairs $(\tp,\nu)$ where $\tp$ is a type from $\mathcal{E}_{+}$ guarded by $\beta_{j}$ as witnessed by $\nu$, and
		\item $\Phi_{\forall\exists}$ a set $P_{j}^{(\mathcal{E}_{+}, \mathcal{E}_{!})}$ containing all pairs $(\tp,\nu)$ where $\tp$ is a type from $\mathcal{E}_{+}$ guarded by $\gamma_{j}$  
		as witnessed by $\nu$.
	\end{itemize}
	For every $(\mathcal{E}_{+}, \mathcal{E}_{!}) \in \getsummary{\Phi}$ we call two elements $(\tp,\nu)\in B_{j}^{(\mathcal{E}_{+}, \mathcal{E}_{!})}$ and $(\tp',\nu')\in P_{j}^{(\mathcal{E}_{+}, \mathcal{E}_{!})}$ \emph{fitting} and write $(\tp,\nu) \Subset_j (\tp',\nu')$ if
	$\nu'$ extends $\nu$, and
	the subtype of $\tp'$ w.r.t. the variables $\bold{v}^{\ast} = \bold{v} \cap \nu\hspace{2pt}'^{\scriptscriptstyle\hspace{-4.5pt}(\hspace{1.5pt})}\hspace{-0.5pt}(\bold{z})$ coincides with $\tp$,   
	that is, $\tp = \tp'|_{\bold{v}^{\ast}}$.
\end{definition}
Note that since our types are maximally consistent, each $\tp\in\mathcal{E}_{+}$ (for $(\mathcal{E}_{+}, \mathcal{E}_{!}) \in \getsummary{\Phi}$) with $\Phi$ a sentence from \GFO{} or \TGF{} in normal form contains for every term (i.e. type variable or constant symbol) $t, t_{1}, t_{2}$ from $\tp$ the atoms $\sig{D}(t)$ or $\sig{Univ}(t_{1}, t_{2})$ depending on whether $\Phi$ is \GFO{} or \TGF{}. This presupposes that $\Phi$ is in the respective normal form (see \Cref{lem:GFONF} or \Cref{lem:TGFNF} respectively).

Note also, that by our types being maximally consistent each $(\tp, \nu) \in T_{i}^{(\mathcal{E}_{+}, \mathcal{E}_{!})}$ contains automatically some disjunct from $\nu(\vartheta_{i})$.

Lastly, if $\Phi$ is from \TGF{} it does not really matter, whether we state that $\mathcal{E}$ consists of rigid types or not, since \TGF{} expressly disallows equality atoms.

\begin{definition}\label{def:20}
	Let $\Phi$ be given as a \GFO{} or \TGF{} $\tau$-sentence. Additionally let $\mathrm{width}(\Phi)$ be the maximal arity of a relation symbol appearing in $\Phi$ and $\mathcal{E}$ be the set of all guarded types that are also rigid and of order $\le \mathrm{width}(\Phi)$. We will denote by $\sigma$ the set of fresh relation symbols $\sig{Tp}_{\tp}$ of arity $\mathrm{order}(\tp)$ for every $\tp\in\mathcal{E}$. Furthermore let $(\mathcal{E}_{+}, \mathcal{E}_{!}) \in \getsummary{\Phi}$.
	


	Let $C_{/\sim}$ denote the factorization of the set $C$ of constant symbols from $\tau$ by the equivalence relation $\sim$ that is defined by letting $\sig{c} \sim \sig{d}$ hold whenever the literal $\sig{c}\!=\!\sig{d}$ is contained in some $\tp \in \mathcal{E}_{+}$. We obtain the $C$-structure $\mathfrak{C}$ by taking $C_{/\sim}$ as domain and letting 
	$\sigc^{\mathfrak{C}} = [\sigc]_\sim$ for any $\sigc \in C$. 
	\newline
	Furthermore, the $\sigma\uplus C$-structure $\mathfrak{C}\cdot\widecheck{\sigma}$ (the extension of $\mathfrak{C}$ to the signature $\sigma\uplus C$ by interpreting the relation symbols from $\sigma$ as empty) will sometimes also be denoted with $\mathfrak{C}$. In this case we will make clear, or it is clear from the context, that we talk about the $\sigma\uplus C$-structure and not just the $C$-structure.
	
	For $\tp \in \mathcal{E}_{+}$, we denote by $\mathfrak{D}_\tp$ the relation-minimal $\sigma$-structure  
	with domain $D_\tp = \bold{v}_\tp$ satisfying $\bold{v}_\tp \in \sig{Tp}_{\tp}^{\mathfrak{D}_\tp}$
	as well as all further $\sigma$-relations implied by the conjuncts from \ref{eq: gen 3} and \ref{eq: gen 3plus}.

 We define	for every $\tp\in\mathcal{E}_{+}$  a $\tau\uplus \sigma$-structure $(\mathfrak{D}_{\tp})^{\tau}$ with domain $(D_{\tp})^{\tau} = D_{\tp}\uplus C_{/\sim}$ as follows: 
	\begin{itemize}
		\item\label{item Dtauexpansion:1} $\sig{c}^{(\mathfrak{D}_{\tp})^{\tau}} = [\sig{c}]_{\sim}$ for $\sig{c}\in C$,
		\item\label{item Dtauexpansion:2} $\sig{T}^{(\mathfrak{D}_{\tp})^{\tau}} = \sig{T}^{\mathfrak{D}_{\tp}}$ for every $\sig{T}\in\sigma$, and
		\item\label{item Dtauexpansion:3} $\sig{R}^{(\mathfrak{D}_{\tp})^{\tau}} = \{ \bold{b} \in ((D_{\tp})^{\tau})^{\mathrm{ar}(\sig{R})} \ | \ \sig{R}(\bold{b})\in\tp \}$ for every $\sig{R}\in\tau\setminus C$.
	\end{itemize}
	
	Let $\mathfrak{K}_{\mathcal{E}_{+}}$ be the $\sigma\uplus C$-structure that is the disjoint union of $\mathfrak{C}\cdot\widecheck{\sigma}$ and all $\mathfrak{D}_\tp$ for $\tp \in \mathcal{E}_{+}$. We assume that taking the disjoint union includes the renaming of all domain elements of every $\mathfrak{D}_\tp$ with fresh, nowhere-else occurring element names, prior to taking the union. We denote by $f_{\tp}^{\mathfrak{K}_{\mathcal{E}_{+}}} \colon\mathfrak{D}_{\tp} \to \mathfrak{K}_{\mathcal{E}_{+}}$ the injection that maps $\mathfrak{D}_{\tp}$ to its isomorphic copy. Additionally, we extend $f_{\tp}^{\mathfrak{K}_{\mathcal{E}_{+}}}$ to a map $f_{\tp, \mathfrak{C}}^{\mathfrak{K}_{\mathcal{E}_{+}}}\colon \mathfrak{D}_{\tp}\uplus (\mathfrak{C}\cdot\widecheck{\sigma})\to \mathfrak{K}_{\mathcal{E}_{+}}$ with $f_{\tp, \mathfrak{C}}^{\mathfrak{K}_{\mathcal{E}_{+}}}(\sig{c}^{\mathfrak{C}\cdot\widecheck{\sigma}})=\sig{c}^{\mathfrak{K}_{\mathcal{E}_{+}}}$ for all $c\in C$. 
\end{definition}

We will make some notes on \Cref{def:20}.
\begin{enumerate}
	\item\label{item remdef20:1} For each $\mathcal{E}_{+}$ we note that if some $\tp\in\mathcal{E}_{+}$ contains a literal whose terms are just from $C$ (so no variables are mentioned), i.e. a closed literal, then this literal is contained in every $\tp\in\mathcal{E}_{+}$. This follows from the fact that first, our types are maximal consistent and second that $\mathcal{E}_{+}$ arises as a component of the type summary of a given model $\mathfrak{A}\in\getmodels{\Phi}$. More specifically, let $\Theta$ be a literal having as terms just constants (hence being a sentence) and let $\tp\in\mathcal{E}_{+}$ with $\Theta\in\tp$. Let $\tp' \in \mathcal{E}_{+}$ some arbitrary, different type, and take $\mathfrak{A}\in\getmodels{\Phi}$ such that $\mathcal{E}(\mathfrak{A})=(\mathcal{E}_{+}, \mathcal{E}_{!})$. By $\tp\in\mathcal{E}_{+}$ we conclude that $\mathfrak{A}\vDash\Theta$. Suppose, by maximal consistency of types, $\lnot\Theta\in\tp'$. As also $\tp' \in \mathcal{E}_{+}$ we obtain $\mathfrak{A}\vDash\lnot\Theta$, a contradiction. Hence every closed literal that is contained in some $\tp\in\mathcal{E}_{+}$ automatically is contained in every $\tp\in\mathcal{E}_{+}$. An example for this is the closed literal $\sig{c}\!=\!\sig{d}$ (which assumes $\Phi$ to be a \GFO{} sentence).
	\item\label{item remdef20:2} If $\Phi$ is from \TGF{}, rigidity always holds for types. And especially, in the definition of $\sim$ we obtain that $C_{/ \sim}$ is, up to renaming, equal to $C$ (note that the former is a set of equivalence classes, whereas the latter is not).
	\item\label{item remdef20:3} Note that when we defined $\mathfrak{D}_{\tp}$ for some $\tp\in\mathcal{E}_{+}$ we made use of (type) variables (which are \emph{syntactic} objects) as domain elements (which are \emph{semantic} objects). This shall not cause any confusion, especially as it is a convenient and intuitive choice for our domain elements. Even more important, the usage of such $\mathfrak{D}_{\tp}$ will be only temporary and always followed by a renaming of its domain elements.
	\item\label{item remdef20:4} In the definition of $\mathfrak{D}_{\tp}$ we closed $\mathfrak{D}_{\tp}$ under taking variable permutations and (valid) subtypes. Note that this plays out on the level of the type predicates and not the types themselves. 
	\item\label{item remdef20:5} For $\tp\in\mathcal{E}_{+}$ the structure $(\mathfrak{D}_{\tp})^{\tau}$ satisfies  for every $\tp'\in\mathcal{E}_{+}$ an equivalence between the set of tuples $\bold{d}$ from $(\mathfrak{D}_{\tp})^{\tau}$ satisfying $\sig{Tp}_{\tp}$ and the set of tuples $\bold{d}$ realizing $\tp'$ in $(\mathfrak{D}_{\tp})^{\tau}$. As this is an important fact, we will prove it separately in \Cref{lem:Dtauexpansion}.
	\item\label{item remdef20:6} For $\tp\in\mathcal{E}_{+}$ the structure $(\mathfrak{D}_{\tp})^{\tau}$ satisfies the following properties: every tuple $\bold{d}$ satisfying some $\sig{Tp}_{\tp'}$ in $(\mathfrak{D}_{\tp})^{\tau}$ consists of pairwise distinct elements from $\bold{v}_{\tp}$.
	\item\label{item remdef20:7} When we defined $\mathfrak{K}_{\mathcal{E}_{+}}$, note that by including $\mathfrak{C}$ in the disjoint union, we add all domain elements denoted by constants. This part of the model does not participate in any $\Tp_{\tp}$ relations, due to definition. For our construction later on this will be no issue: 
	The rigidity of the types from $\mathcal{E}$ enforces that the elements from $\mathfrak{C}$ do not participate in any $\Tp_{\tp}$ relations. 
	\item\label{item remdef20:8} As $f_{\tp}^{\mathfrak{K}_{\mathcal{E}_{+}}}$ is injective, so is $f_{\tp, \mathfrak{C}}^{\mathfrak{K}_{\mathcal{E}_{+}}}$ by definition (as we only additionally map bijectively $\mathfrak{C}$ to its copy in $\mathfrak{K}_{\mathcal{E}_{+}}$). Additionally, $f_{\tp, \mathfrak{C}}^{\mathfrak{K}_{\mathcal{E}_{+}}}$ is an embedding. That it is a homomorphism is clear by definition. {So let $\sig{T}\in\sigma$ and $\bold{d}\in D_{\tp}\uplus C_{/\sim}$ such that $f_{\tp, \mathfrak{C}}^{\mathfrak{K}_{\mathcal{E}_{+}}}(\bold{d})\in\sig{T}^{\mathfrak{K}_{\mathcal{E}_{+}}}$. As the $\mathfrak{C}$ part of $\mathfrak{K}_{\mathcal{E}_{+}}$ does not participate in any $\sigma$-relations we obtain by definition of $f_{\tp, \mathfrak{C}}^{\mathfrak{K}_{\mathcal{E}_{+}}}$ that $\bold{d}\in D_{\tp}$. And by definition of $\mathfrak{K}_{\mathcal{E}_{+}}$ we conclude $\bold{d}\in\sig{T}^{\mathfrak{D}_{\tp}}$} and therefore $\bold{d}\in\sig{T}^{\mathfrak{D}_{\tp}\uplus ({\mathfrak{C}\cdot\widecheck{\sigma}})}$.	
\end{enumerate}

In the following lemmata (\Cref{lem:Dtauexpansion} and \Cref{lem:A0tauexpansion}) we will introduce results that essentially allow us to switch between the notion of realizing a type $\tp$ from some set of eligible types $\mathcal{E}$ and satisfying the associated type predicates $\sig{Tp}_{\tp}$. As we will start from structures whose relations are just interpretations of type predicates, the general idea is the following: For each tuple $\bold{a}$ satisfying a type predicate $\sig{Tp}_{\tp}$, we will look into the associated type $\tp$ and, according to the information therein, we will introduce $\tau$-relations such that $\bold{a}$ also realizes $\tp$.
At this stage this will be restricted to the structures $(\mathfrak{D}_{\tp})^{\tau}$ and an expanded version of $\mathfrak{K}_{\mathcal{E}_{+}}$. In the first case, we actually only have to check that the general idea already holds. In the second case we expand the structure according to the general idea. Those might be simple cases but they lie at the heart of our construction we will employ to prove the characterizations of the homclosure for \GFO{} or \TGF{} sentences.

\begin{lemma}\label{lem:Dtauexpansion}
	Let $\Phi$ be given as a \GFO{} or \TGF{} $\tau$-sentence. Additionally let $\mathrm{width}(\Phi)$ be the maximal arity of a relation symbol appearing in $\Phi$ and $\mathcal{E}$ be the set of all guarded types that are also rigid and of order $\le \mathrm{width}(\Phi)$. We will denote by $\sigma$ the set of fresh relation symbols $\sig{Tp}_{\tp}$ of arity $\mathrm{order}(\tp)$ for every $\tp\in\mathcal{E}$. Furthermore let $(\mathcal{E}_{+}, \mathcal{E}_{!}) \in \getsummary{\Phi}$. 
	
	Then for all $\tp'\in\mathcal{E}$ it holds that a tuple $\bold{d}$ from $D_{\tp}$ satisfies $\sig{Tp}_{\tp'}\in\sigma$ in $(\mathfrak{D}_{\tp})^{\tau}$ if and only if $\bold{d}$ realizes $\tp'$ in $(\mathfrak{D}_{\tp})^{\tau}$.
\end{lemma}
\begin{proof}
	Let $\bold{d}$ be a tuple in $(\mathfrak{D}_{\tp})^{\tau}$ such that $\bold{d}\in\sig{Tp}_{\tp'}^{(\mathfrak{D}_{\tp})^{\tau}}$ for some $\tp'\in\mathcal{E}$. 
	There is a permutation $\kappa\colon\bold{v}\injsur\bold{v}$ with $\kappa\tp'$ being a subtype of $\tp$. 
	By definition of $(\mathfrak{D}_{\tp})^{\tau}$ every (positive) literal, i.e. atom $\delta[\bold{v}_{\tp'}]\in\tp'$ is satisfied by $\bold{d}$ in $(\mathfrak{D}_{\tp})^{\tau}$. Note that the atom $\delta$ might contain constant symbols that we do not mention explicitly.  On the other hand for negative literals, suppose $\lnot\delta[\bold{v}_{\tp'}]\in\tp'$. Hence $\lnot\delta[\bold{v}_{\tp'}]\in\tp$ and by $\tp$ being a type, thus maximally consistent, it cannot contain the non-negated atom $\delta[\bold{v}_{\tp'}]$. Thus by definition, $\bold{d}$ does not satisfy $\delta[\bold{v}_{\tp'}]$ and consequently satisfying the negation of it, i.e. $\lnot\delta[\bold{v}_{\tp'}]$. Hence we obtain that $\bold{d}$ realizes $\tp'$ in $(\mathfrak{D}_{\tp})^{\tau}$.
	
	Now let $\bold{d}$ from $(\mathfrak{D}_{\tp})^{\tau}$ satisfy $\tp'\in\mathcal{E}_{+}$. 
	By definition, every literal $\lambda[\bold{v}_{\tp'}]\in\tp'$ is satisfied by $\bold{d}$. Trivially by assumption $\lambda[\bold{v}_{\tp'}]\in\tp$. Hence there is a permutation $\kappa\colon\bold{v}\injsur\bold{v}$ such that $\kappa\tp|_{\bold{v}_{\tp'}}=\tp'$. Without loss of generality we may assume $\kappa$ being the identity and $\bold{v}_{\tp'}\subseteq\bold{v}_{\tp}$. As the tuple $\bold{v}_{\tp}$ satisfies $\sig{Tp}_{\tp}$ in $\mathfrak{D}_{\tp}$ and the structure satisfies both \ref{eq: gen 3} and \ref{eq: gen 3plus} we thus obtain via $\tp' = \tp|_{\bold{v}_{\tp'}}$ that $\bold{d}\in \sig{Tp}_{\tp'}^{(\mathfrak{D}_{\tp})^{\tau}}$.
\end{proof}

\begin{lemma}\label{lem:A0tauexpansion}
	Let $\Phi$ be given as a \GFO{} or \TGF{} $\tau$-sentence. Additionally let $\mathrm{width}(\Phi)$ be the maximal arity of a relation symbol appearing in $\Phi$ and $\mathcal{E}$ be the set of all guarded types that are also rigid and of order $\le \mathrm{width}(\Phi)$. We will denote by $\sigma$ the set of fresh relation symbols $\sig{Tp}_{\tp}$ of arity $\mathrm{order}(\tp)$ for every $\tp\in\mathcal{E}$. Furthermore let $(\mathcal{E}_{+}, \mathcal{E}_{!}) \in \getsummary{\Phi}$.
	
	Then there is a canonical $\sigma\uplus\tau$-expansion $(\mathfrak{K}_{\mathcal{E}_{+}})^{\tau}$ of $\mathfrak{K}_{\mathcal{E}_{+}}$,
	such that for all $\tp\in\mathcal{E}$ it holds that a tuple $\bold{a}$ from $(\mathfrak{K}_{\mathcal{E}_{+}})^{\tau}$ satisfies $\sig{Tp}_{\tp}\in\sigma$ in $(\mathfrak{K}_{\mathcal{E}_{+}})^{\tau}$ if and only if $\bold{a}$ realizes $\tp$ in $(\mathfrak{K}_{\mathcal{E}_{+}})^{\tau}$.

	Additionally, for every $\tp\in\mathcal{E}$ the map $f_{\tp, \mathfrak{C}}^{\mathfrak{K}_{\mathcal{E}_{+}}}$ gives rise to an embedding $(f')_{\tp, \mathfrak{C}}^{\mathfrak{K}_{\mathcal{E}_{+}}}\colon(\mathfrak{D}_{\tp})^{\tau}\to(\mathfrak{K}_{\mathcal{E}_{+}})^{\tau}$, being the same map as $f_{\tp, \mathfrak{C}}^{\mathfrak{K}_{\mathcal{E}_{+}}}$.
\end{lemma}
\begin{proof}
	For every $\tp\in\mathcal{E}$ we use the maps $f_{\tp, \mathfrak{C}}^{\mathfrak{K}_{\mathcal{E}_{+}}}$ to expand $\mathfrak{K}_{\mathcal{E}_{+}}$ in the following way: Note first, that $(\mathfrak{D}_{\tp})^{\tau}$ has the same domain as $\mathfrak{D}_{\tp} \uplus (\mathfrak{C}\cdot\widecheck{\sigma})$. For $\sig{R}\in\tau\setminus C$ we let 
	\begin{equation}\label{eq: defsigR 1}
		\sig{R}^{(\mathfrak{K}_{\mathcal{E}_{+}})^{\tau}} = \{ f_{\tp}^{\mathfrak{K}_{\mathcal{E}_{+}}}(\bold{d}) \ | \ \tp\in\mathcal{E}_{+} \text{ and } \bold{d}\in\sig{R}^{(\mathfrak{D}_{\tp})^{\tau}} \}.
	\end{equation}
	By definition for different $\tp, \tp'\in\mathcal{E}_{+}$ the sets $f_{\tp, \mathfrak{C}}^{\mathfrak{K}_{\mathcal{E}_{+}}}(D_{\tp})$ and $f_{\tp', \mathfrak{C}}^{\mathfrak{K}_{\mathcal{E}_{+}}}(D_{\tp'})$ are disjoint. Hence for every tuple $\bold{a}\in\sig{R}^{(\mathfrak{K}_{\mathcal{E}_{+}})^{\tau}}$ {containing at least one element not named by a constant} there is a unique $\tp\in\mathcal{E}_{+}$ such that $\bold{a} = f_{\tp, \mathfrak{C}}^{\mathfrak{K}_{\mathcal{E}_{+}}}(\bold{d})$ for a tuple $\bold{d}$ consisting exclusively of elements from $\mathfrak{C}$ and $\mathfrak{D}_{\tp}$. {If $\bold{a}$ consists only of  elements named by constants, we remember Remark \ref{item remdef20:1} after \Cref{def:20} from which we immediately conclude, that it does not matter which $\tp\in\mathcal{E}_{+}$ we choose in the discussion. Hence in this case we fix one.} By $f_{\tp, \mathfrak{C}}^{\mathfrak{K}_{\mathcal{E}_{+}}}$ being injective, this $\bold{d}$ is also unique.
	Additionally, $f_{\tp, \mathfrak{C}}^{\mathfrak{K}_{\mathcal{E}_{+}}}$ gives rise to an embedding $(f')_{\tp, \mathfrak{C}}^{\mathfrak{K}_{\mathcal{E}_{+}}}\colon(\mathfrak{D}_{\tp})^{\tau}\to(\mathfrak{K}_{\mathcal{E}_{+}})^{\tau}$, being the same map as $f_{\tp, \mathfrak{C}}^{\mathfrak{K}_{\mathcal{E}_{+}}}$ considering both as maps between the underlying domains of the structures. For simplicity of exposition we will denote it by just $f'_{\tp}$. By the interpretation of the $\sig{R}\in\tau\setminus C$ in $(\mathfrak{K}_{\mathcal{E}_{+}})^{\tau}$ as defined in \ref{eq: defsigR 1} the map $f_{\tp}'$ is a homomorphism. $f'_{\tp}$ is injective since $f_{\tp, \mathfrak{C}}^{\mathfrak{K}_{\mathcal{E}_{+}}}$ is. 
	So let $\bold{d}$ be a tuple from $(\mathfrak{D}_{\tp})^{\tau}$ such that $f'_{\tp}(\bold{d})\in \sig{R}^{(\mathfrak{K}_{\mathcal{E}_{+}})^{\tau}}$. Then by definition $\bold{d}\in\sig{R}^{(\mathfrak{D}_{\tp})^{\tau}}$, establishing that $f'_{\tp}$ is an embedding.
	
	We now show that for every type $\tp\in\mathcal{E}$ it holds that a tuple $\bold{a}$ from $(\mathfrak{K}_{\mathcal{E}_{+}})^{\tau}$ satisfies $\sig{Tp}_{\tp'}^{(\mathfrak{K}_{\mathcal{E}_{+}})^{\tau}}$ in $(\mathfrak{K}_{\mathcal{E}_{+}})^{\tau}$ if and only if $\bold{a}$ realizes $\tp'$ in $(\mathfrak{K}_{\mathcal{E}_{+}})^{\tau}$. 
	Note that by rigidity of types on the one hand and Remark \ref{item remdef20:7} after \Cref{def:20} on the other, only tuples that solely consists of elements not named by constants can satisfy type predicates or realize types.
	
	So let $\bold{a}\in\sig{Tp}_{\tp'}^{(\mathfrak{K}_{\mathcal{E}_{+}})^{\tau}}$ for some $\tp'\in\mathcal{E}$. As by definition the domain of $\mathfrak{K}_{\mathcal{E}_{+}}$ is the disjoint union of $\mathfrak{C}\cdot\widecheck{\sigma}$ and isomorphic copies of $\mathfrak{D}_{\tp}$ there exists a unique $\tp\in\mathcal{E}_{+}$ such that $\bold{a}$ is from $f'_{\tp}(D_{\tp})$. Since $f'_{\tp}$ is an embedding, we find $\bold{d}$ in $D_{\tp}$ with $f'_{\tp}(\bold{d}) = \bold{a}$ and $\bold{d}\in\sig{Tp}_{\tp'}^{(\mathfrak{D}_{\tp})^{\tau}}$. \Cref{lem:Dtauexpansion} yields that $\bold{d}$ realizes $\tp'$ in $(\mathfrak{D}_{\tp})^{\tau}$. By invoking $f'_{\tp}$ being an embedding once again, we obtain that $\bold{a}=f'_{\tp}(\bold{d})$ realizes $\tp'$ in $(\mathfrak{K}_{\mathcal{E}_{+}})^{\tau}$.
	
	Suppose now that $\bold{a}$ realizes $\tp'$ for some $\tp'\in\mathcal{E}$. As $\tp'$ is guarded, there exists a guard atom $\delta$ (witnessed by $\nu$) such that $\bold{a}$ satisfies $\nu(\delta)$ in $(\mathfrak{K}_{\mathcal{E}_{+}})^{\tau}$. By \ref{eq: defsigR 1} this necessitates, that there is a $\tp''\in\mathcal{E}$ such that $\bold{a}$ consists of elements exclusively from $f'_{\tp''}(D_{\tp''})$. More importantly, $\bold{a}$ cannot consist of images from different $D_{\tp''}$ via $f'_{\tp''}$.
	But since $\bold{a}$ realizes $\tp'$ we find some $\tp\in\mathcal{E}_{+}$ and a unique $\bold{d}$ from $D_{\tp}$ with $f'_{\tp}(\bold{d}) = \bold{a}$ such that $\bold{d}$ also realizes $\tp'$ (using that $f'_{\tp}$ is an embedding). By \Cref{lem:Dtauexpansion} we obtain that $\bold{d}\in\sig{Tp}_{\tp'}^{(\mathfrak{D}_{\tp})^{\tau}}$ and by $f'_{\tp}$ being an embedding we conclude $\bold{a}\in\sig{Tp}_{\tp'}^{(\mathfrak{K}_{\mathcal{E}_{+}})^{\tau}}$.
\end{proof}

As we are concerned with the homclosure of \GFO{} or \TGF{} sentences, we need to start shedding light on the behaviour of our (``enriched'') structures under homomorphisms into some other structures. It is natural to start with the simple case of $\mathfrak{K}_{\mathcal{E}_{+}}$ (or, for the ``enriched'' structure $(\mathfrak{K}_{\mathcal{E}_{+}})^{\tau}$). The structures to which we will homomorphically map $\mathfrak{K}_{\mathcal{E}_{+}}$ (or $(\mathfrak{K}_{\mathcal{E}_{+}})^{\tau}$) cannot be arbitrary ones but have to already satisfy some sentence, which will be part of the sentence characterizing the homclosure. This is what the following lemma (\Cref{lem:A0}) does. 

\begin{lemma}\label{lem:A0}
	Let $\Phi$ be a \GFO{} or \TGF{} $\tau$-sentence with $C$ being the set of constants in $\tau$, $\mathcal{E}$ a set of eligible types (with the usual provision of being finite and containing rigid types) closed under taking variable permutations and $(\mathcal{E}_{+}, \mathcal{E}_{!})\in \getsummary{\Phi}$. Let $\sigma$ be the signature containing for every $\tp\in\mathcal{E}$ the predicate $\sig{Tp}_{\tp}$ having as its arity the order of $\tp$. Then $\mathfrak{K}_{\mathcal{E}_{+}}$ satisfies $\Psi_{\mathrm{gen}, (\mathcal{E}_{+}, \mathcal{E}_{!})}$ 
	with the additional property that for every model $\mathfrak{B}$ (over the signature $\tau \uplus\sigma$) of $\Psi_{\mathrm{hom}}\land\Psi_{\mathrm{gen}, (\mathcal{E}_{+}, \mathcal{E}_{!})}$ there exists a homomorphism $h\colon \mathfrak{K}_{\mathcal{E}_{+}} \to \mathfrak{B}|_{\sigma\uplus C}$.
	\newline
{Additionally, $h$ gives rise to a homomorphism $h'\colon (\mathfrak{K}_{\mathcal{E}_{+}})^{\tau}\to \mathfrak{B}$ and $(\mathfrak{K}_{\mathcal{E}_{+}})^{\tau}$ satisfies $\Psi_{\mathrm{hom}}\land\Psi_{\mathrm{gen}, (\mathcal{E}_{+}, \mathcal{E}_{!})}$.}
\end{lemma}
\begin{proof}
	We show that $\mathfrak{K}_{\mathcal{E}_{+}}$ satisfies $\Psi_{\mathrm{gen}, (\mathcal{E}_{+}, \mathcal{E}_{!})}$. 
	Since by definition $\mathfrak{K}_{\mathcal{E}_{+}}$ contains for every $\tp\in\mathcal{E}_{+}$ a substructure isomorphic to $\mathfrak{D}_{\tp}$, it contains some tuple $\bold{a}$ such that $\bold{a} \in \sig{Tp}_{\tp}^{\mathfrak{K}_{\mathcal{E}_{+}}}$. As this holds for every $\tp\in \mathcal{E}_{+}$ Formulae \ref{eq: gen 1} is satisfied.
	
	To show the satisfaction of Formulae \ref{eq: gen 2} observe that in the construction we either added for $\tp\in \mathcal{E}_{+}$ a $\mathfrak{D}_{\tp}$ (hence ensuring the satisfaction of some $\sig{Tp}_{\tp}$) or we closed the disjunction ultimately yielding $\mathfrak{K}_{\mathcal{E}_{+}}$ under \ref{eq: gen 3} and \ref{eq: gen 3plus}. Consequently, only by these closure operations could we unwittingly made a $\sig{Tp}_{\tp}$ with $\tp\in\mathcal{E}\setminus\mathcal{E}_{+}$ true for some tuple in $\mathfrak{K}_{\mathcal{E}_{+}}$. But this cannot have happened, as by the fact that $\mathcal{E}$ is closed under taking variable permutations and applying \Cref{lem:eligclosure} (both Item \ref{item eligclosure:1} and \ref{item eligclosure:2}) we only could have introduced $\sig{Tp}_{\tp}$ for $\tp\in\mathcal{E}_{+}$. Consequently, for every $\tp \in \mathcal{E}\setminus\mathcal{E}_{+}$ we obtain $\sig{Tp}_{\tp}^{\mathfrak{K}_{\mathcal{E}_{+}}} = \emptyset$, implying the validity of Formulae \ref{eq: gen 2}.
	
	Now take $\tp, \tp' \in \mathcal{E}_{+}$ such that $\tp' = \tp|_{\bold{v}_{\tp'}}$ and let $\bold{a}$ be some tuple from $\mathfrak{K}_{\mathcal{E}_{+}}$ such that $\bold{a}\in \sig{Tp}_{\tp}^{\mathfrak{K}_{\mathcal{E}_{+}}}$, let $\mu\colon \bold{v}_{\tp} \to\bold{a}$ be the componentwise evaluation witnessing this. By definition of $\mathfrak{K}_{\mathcal{E}_{+}}$ there exists some type $\tp^{\ast}\in\mathcal{E}_{+}$ with $\tp^{\ast}|_{\bold{v}_{\tp}}=\tp$, such that the induced substructure by $\bold{a}$ in $\mathfrak{K}_{\mathcal{E}_{+}}$ is an induced one of an isomorphic copy of $\mathfrak{D}_{\tp^{\ast}}$. Hence we know by definition (closure with respect to eligible subtypes in $\mathfrak{D}_{\tp}$) that the subtuple $\mu(\bold{v}_{\tp'})$ of $\bold{a}$ satisfies $\mu(\bold{v}_{\tp'})\in \sig{Tp}_{\tp'}^{\mathfrak{K}_{\mathcal{E}_{+}}}$. Finally, observing that $\mu$ assigns the tuple $\bold{a}$ componentwise to $\bold{v}_{\tp}$, this yields the satisfaction of Formulae \ref{eq: gen 3} by $\mathfrak{K}_{\mathcal{E}_{+}}$.
	
	To show the satisfaction of the last kind of Formulae (i.e., Formulae \ref{eq: gen 3plus}), let $\tp \in \mathcal{E}$, $\kappa \colon \bold{v} \injsur \bold{v}$ a permutation of the variables and $\bold{a}$ a tuple from $\mathfrak{K}_{\mathcal{E}_{+}}$ such that $\bold{a}\in \sig{Tp}_{\tp}^{\mathfrak{K}_{\mathcal{E}_{+}}}$. Note that, since $\mathcal{E}$ is closed under taking permutations of (type) variables, $\kappa\tp\in\mathcal{E}$. And as for every $\tp \in\mathcal{E}\setminus\mathcal{E}_{+}$ the interpretation $\sig{Tp}_{\tp}^{\mathfrak{K}_{\mathcal{E}_{+}}}$ is empty, it suffices to discuss the case when $\tp\in \mathcal{E}_{+}$. 
	By definition of $\mathfrak{K}_{\mathcal{E}_{+}}$ there exists some type $\tp^{\ast}\in\mathcal{E}_{+}$ and some permutation $\kappa^{\ast}\colon\bold{v}\injsur \bold{v}$ with $\kappa^{\ast}\tp^{\ast}=\tp$ such that the induced substructure by $\bold{a}$ in $\mathfrak{K}_{\mathcal{E}_{+}}$ is an induced one of an isomorphic copy of $\mathfrak{D}_{\tp^{\ast}}$. Since every $\mathfrak{D}_{\tp}$ used in the definition of $\mathfrak{K}_{\mathcal{E}_{+}}$ is itself closed under \ref{eq: gen 3plus} (i.e. closed under permutations, which is a fact holding by definition) we obtain that $\eta_{\bold{v}}^{\kappa\bold{v}_{\tp}}(\bold{a}) \in \sig{Tp}_{\kappa\tp}^{\mathfrak{K}_{\mathcal{E}_{+}}}$ holds. As this shows the validity of the Formulae \ref{eq: gen 3plus} we finally have shown that $\mathfrak{K}_{\mathcal{E}_{+}}$ satisfies $\Psi_{\mathrm{gen}, (\mathcal{E}_{+}, \mathcal{E}_{!})}$.

	Furthermore, for every $\mathfrak{B}\in\getmodels{\Psi_{\mathrm{hom}}\land\Psi_{\mathrm{gen}, (\mathcal{E}_{+}, \mathcal{E}_{!})}}$ the structure $\mathfrak{K}_{\mathcal{E}_{+}}$ has a homomorphism $h$ into the structure $\mathfrak{B}|_{\sigma\uplus C}$, which we will denote by $\mathfrak{B}^\ast$. To see this, note that $\mathfrak{K}_{\mathcal{E}_{+}}$ is the disjoint union of $\mathfrak{C}\cdot\widecheck\sigma$ and structures isomorphic to $\mathfrak{D}_{\tp}$ for $\tp \in \mathcal{E}_{+}$, each of which in separation homomorphically maps to $\mathfrak{B}^\ast$.
	As for $\mathfrak{C}$, the mapping $[\sigc]_\sim \mapsto \sigc^{\mathfrak{B}^\ast}$ does the trick since for any $\sigc,\sig{d} \in C$ with $\sigc \sim \sig{d}$, the fact that $\mathfrak{B}$ is a model of $\Psi_{\mathrm{hom}}$ and realizes some $\Tp_\tp$ containing the literal $\sig{c}\!=\!\sig{d}$ ensures $\sigc^\mathfrak{B} = \sig{d}^\mathfrak{B}$. For $\Phi$ from \TGF{} this comment is not interesting since $\sig{c}\sim\sig{d}$ if and only if $\sig{c}$ and $\sig{d}$ are the same constant symbols. We can now extend this map stepwise for every component isomorphic to some $\mathfrak{D}_{\tp}$ of the disjoint union. Because of isomorphy it is enough to discuss $\mathfrak{D}_{\tp}$ instead of the actual component. By definition each $\mathfrak{D}_\tp$ satisfies a conjunct from \ref{eq: gen 1} saturated under  \ref{eq: gen 3} and \ref{eq: gen 3plus}. As $\mathfrak{B}\in\getmodels{\Psi_{\mathrm{gen}, (\mathcal{E}_{+}, \mathcal{E}_{!})}}$ we can map the elements $\bold{v}_{\tp}$ of $\mathfrak{D}_{\tp}$ componentwise to some arbitrary $\bold{b}\in \sig{Tp}_{\tp}^{\mathfrak{B}}$. Then the extension of the homomorphism is obtained by mapping the tuple $f_{\tp}^{\mathfrak{K}_{\mathcal{E}_{+}}}(\bold{v}_{\tp})$ componentwise to $\bold{b}$. As we took the disjoint union, we obtain a homomorphism $h\colon\mathfrak{K}_{\mathcal{E}_{+}} \to\mathfrak{B}^{\ast}.$
	
	Finally we show that $h$ (as a map) constitutes a homomorphism $h'\colon (\mathfrak{K}_{\mathcal{E}_{+}})^{\tau}\to \mathfrak{B}$. As $h$ is a homomorphism when restricting to the signature $\sigma\uplus C$, it suffices to discuss $\sig{R}\in\tau\setminus C$. 
	Furthermore, we will make use of the embedding $(f')_{\tp, \mathfrak{C}}^{\mathfrak{K}_{\mathcal{E}_{+}}}$ obtained from \Cref{lem:A0tauexpansion} (for $\tp\in\mathcal{E}_{+}$). For brevity, we will denote those maps by just $f'_{\tp}$. Let $\bold{a}\in\sig{R}^{(\mathfrak{K}_{\mathcal{E}_{+}})^{\tau}}$ for some $\sig{R}\in\tau \setminus C$. 
	If $\bold{a}$ {contains at least one element not named by a constant} there is a unique $\tp\in\mathcal{E}_{+}$ and a unique $\bold{d}$ from $\bold{v}_{\tp}$ such that $\bold{d}\in\sig{R}^{(\mathfrak{D}_{\tp})^{\tau}}$ and $f'_{\tp}(\bold{d}) = \bold{a}$.
	{If $\bold{a}$ consists of only elements named by constants, we remember Remark \ref{item remdef20:1} after \Cref{def:20} from which we immediately conclude, that it does not matter which $\tp\in\mathcal{E}_{+}$ we choose in the discussion. Hence in this case we fix one.} 
	Let $\bold{x}$ be a tuple of pairwise different variables of length $\mathrm{length}(\bold{d})$ and $\mu\colon \bold{x}\to \bold{d}$ componentwise.
	Then $\sig{R}(\bold{x})$ is satisfied by $\bold{d}$ as witnessed by $\mu$. Furthermore let $\bold{v}_{\bold{d}}$ be the subtuple of $\bold{d}$ consisting of elements from $\bold{v}_{\tp}$, i.e. $\bold{v}_{\bold{d}} = \bold{d}\cap\bold{v}_{\tp}$. 
	Take $\delta(\bold{v}_{\bold{d}})$ to be the atom obtained from $\sig{R}(\bold{d})$ where each element of the form $[\sigc]_{\sim}$ for $\sig{c}\in C$ appearing in $\bold{d}$ is substituted by $\sig{c}$. The choice of the representative does only matter in so far as we will not obtain a unique $\delta$. 
	Nevertheless this non-uniqueness will have no impact on the proof. Note for this that $\sim$ is defined on $C$ in terms of equality atoms $\sig{c}=\sig{d}$ in the types $\tp\in\mathcal{E}_{+}$. 
	And by the note after \Cref{def:20}, realization of the types from $\mathcal{E}_{+}$ automatically enforces this equivalence relation. As $\bold{v}_{\tp}$ realizes $\tp$ in $(\mathfrak{D}_{\tp})^{\tau}$, $\bold{v}_{\tp}\in\sig{Tp}_{\tp}^{(\mathfrak{D}_{\tp})^{\tau}}$ by \Cref{lem:Dtauexpansion}. As $\delta\in\tp$ the subtype $\tp' = \tp|_{\bold{v}_{\bold{d}}}$ is guarded by $\delta$ (as witnessed by $\nu\colon\bold{v}_{\bold{d}}\to\bold{v}_{\tp'}\uplus C, v\mapsto v$). Hence $\tp'\in\mathcal{E}_{+}$ by \Cref{lem:eligclosure} and by $\mathfrak{D}_{\tp}$ (hence $(\mathfrak{D}_{\tp})^{\tau}$) satisfying \ref{eq: gen 3} (by definition) we obtain that $\bold{v}_{\bold{d}}\in\sig{Tp}_{\tp'}^{(\mathfrak{D}_{\tp})^{\tau}}$. Applying $f'_{\tp}$ yields $\bold{a}^{\ast}\in\sig{Tp}_{\tp'}^{(\mathfrak{K}_{\mathcal{E}_{+}})^{\tau}}$ with $f'_{\tp}(\bold{v}_{\bold{d}}) = \bold{a}^{\ast}$, where $\bold{a}^{\ast}$ is the subtuple of $\bold{a}$ consisting of all elements that are no interpretations of constant symbols. Hence $h'(\bold{a}^{\ast}) \in \sig{Tp}_{\tp'}^{\mathfrak{B}}$.
	 By $\mathfrak{B}$ satisfying $\Psi_{\mathrm{hom}}$, $h'(\bold{a}^{\ast})$ satisfies the atom $\delta$ in $\mathfrak{B}$. But this means that $h'(\bold{a}) \in\sig{R}^{\mathfrak{B}}$ showing that $h'$ is also a homomorphism for the $\tau\uplus \sigma$-structures.
\end{proof}
We will now start to make more precise the sentence $\Psi$ that will characterize the homclosure of some \GFO{} or \TGF{} sentence $\Phi$. For this we will utilize our normal forms (see \Cref{lem:GFONF} and \Cref{lem:TGFNF}). 
Thanks to those normal forms, the \GFO{} and the \TGF{} case will differ only slightly. 
This comes as little surprise after observing that our normal form for \TGF{} sentences is in and of itself ``essentially guarded'' with the exception of the $\Phi_{\sig{Univ}}$ part. 
Hence the difference in $\Psi$ depending on whether $\Phi$ is from \GFO{} or \TGF{} is connected to $\Phi_{\sig{Univ}}$.

We will first define the, more or less, ``last puzzle pieces'' for $\Psi$. Additionally we introduce the notion of a violation tuple which will be put to good use in a stepwise construction procedure later on.

\begin{definition}\label{def:guardvioltgf}
	Let $\Phi$ be a \GFO{} or \TGF{} $\tau$-sentence in normal form, i.e. $\Phi = \Phi_{\sig{D}}\land\Phi_{\forall}\land\Phi_{\forall\exists}$ or $\Phi =\Phi_{\sig{Univ}}\land\Phi_{\forall}\land\Phi_{\forall\exists}$ respectively, where
	\begin{align}
		\Phi_{\sig{D}} & =  \forall x. \sig{D}(x),	\\[2ex]
		\Phi_{\sig{Univ}} & =  \forall xy. \sig{Univ}(x,y),	\\[2ex]
		\Phi_{\forall} & =  \bigwedge_{i}\forall\bold{x}. (\alpha_{i}(\bold{x}) \Rightarrow \vartheta_{i}[\bold{x}]), \text{ and} \\	
		\Phi_{\forall\exists} & =  \bigwedge_{j}\forall\bold{z}. (\beta_{j}(\bold{z}) \Rightarrow \exists\bold{y}.\gamma_{j}(\bold{y}\bold{z})),
	\end{align}
	with $\vartheta_{i}$ a disjunction of literals, $\alpha_{i}$, $\beta_{j}$ guard atoms, and $\gamma_{j}$ an atom. Additionally let $\mathcal{E}$ be the set of all guarded rigid types of order $\le \mathrm{width}(\Phi)$, and $\sigma$ be the set of all type predicates associated to the types from $\mathcal{E}$. Then for every $(\mathcal{E}_{+}, \mathcal{E}_{!})$ we define $\Psi_{\mathrm{guard}, (\mathcal{E}_{+}, \mathcal{E}_{!})}$ to be the conjunction of the $\sigma$-sentences
	\newcommand{\less}{\hspace{-1.5ex}}
	\noindent
	\begin{eqnarray}
		\forall \bold{v}_\tp. \sig{Tp}_\tp (\bold{v}_\tp) \Rightarrow  \hspace{-6ex} \bigvee_{\quad\quad(\tp,\nu) \Subset_j (\tp',\nu')} \hspace{-6ex} \exists \bold{v}_{\tp'}{\setminus}\nu(\bold{z}) .\sig{Tp}_{\tp'}(\bold{v}_{\tp'})\label{eq: Psi_guarded}
	\end{eqnarray}
	where $(\tp, \nu)$ ranges over $B_{j}^{(\mathcal{E}_{+}, \mathcal{E}_{!})}$ and $j$ ranges over the number of conjuncts in $\Phi_{\forall\exists}$
	
	We call $(\bold{a},\mu,\tp, \nu)$ a \emph{violation tuple} of some structure $\mathfrak{A}\in\getmodels{\Phi}$ with $\mathcal{E}(\mathfrak{A}) = (\mathcal{E}_{+}, \mathcal{E}_{!})$ if $\bold{a}$ is a tuple from $\mathfrak{A}$ and $(\tp, \nu )\in B_j^{(\mathcal{E}_{+}, \mathcal{E}_{!})}$ such that the componentwise variable assignment $\mu\colon \bold{v}_\tp \to\bold{a}$ witnesses the violation of the $(\tp,\nu )$-conjunct from (\ref{eq: Psi_guarded}) 
	in the structure $\mathfrak{A}$, 
	i.e.,  $(\mathfrak{A}, \mu)\vDash \sig{Tp}_{\tp}$ and for every $(\tp', \nu')\in P_{j}^{(\mathcal{E}_{+}, \mathcal{E}_{!})}$ with $(\tp, \nu) \Subset_j (\tp', \nu')$ and every $\mu' \colon \bold{v}_{\tp'} \to A$ with $\mu'|_{\bold{v}_{\tp}} = \mu$ the condition $(\mathfrak{A}, \mu')\not\vDash \sig{Tp}_{\tp'}$ holds.
	
	Additionally, if $\Phi$ is from \TGF{}, we define for every $(\mathcal{E}_{+}, \mathcal{E}_{!})$ the $\sigma$ sentence $\Psi_{\text{\TGF{}}, (\mathcal{E}_{+}, \mathcal{E}_{!})}$ to be the conjunction of the following sentences for every $\tp,\tp' \in \mathcal{E}_{\text{+}}$ with $\bold{v}_\tp=v$ and $\bold{v}_{\tp'}=v'$:
	\begin{equation}\label{eq: Psi_TGF}
		\forall vv'. \sig{Tp}_\tp (v) \wedge \sig{Tp}_{\tp'} (v') \Rightarrow \hspace{-3ex} \bigvee_{{\scriptstyle\tp'' \in \mathcal{E}_{\text{+}},\ \bold{v}_{\tp''}=vv'} \atop {\scriptstyle\tp''|_{v} = \tp,\  \tp''|_{v'} = \tp'}} \hspace{-3ex} \sig{Tp}_{\tp''} (v,v')
	\end{equation}
\end{definition}

For the upcoming lemmata (\Cref{lem:PsiGuardTGF} to \Cref{lem:ConstrA}), which record some technical machinery needed, we will fix the following:

Let $\Phi$ be a \GFO{} or \TGF{} $\tau$-sentence in normal form, i.e. $\Phi = \Phi_{\sig{D}}\land\Phi_{\forall}\land\Phi_{\forall\exists}$ or $\Phi =\Phi_{\sig{Univ}}\land\Phi_{\forall}\land\Phi_{\forall\exists}$ respectively, where
\begin{align}
	\Phi_{\sig{D}} & =  \forall x. \sig{D}(x),	\\[2ex]
	\Phi_{\sig{Univ}} & =  \forall xy. \sig{Univ}(x,y),	\\[2ex]
	\Phi_{\forall} & =  \bigwedge_{i}\forall\bold{x}. (\alpha_{i}(\bold{x}) \Rightarrow \vartheta_{i}[\bold{x}]), \text{ and} \\	
	\Phi_{\forall\exists} & =  \bigwedge_{j}\forall\bold{z}. (\beta_{j}(\bold{z}) \Rightarrow \exists\bold{y}.\gamma_{j}(\bold{y}\bold{z})),
\end{align}
with $\vartheta_{i}$ a disjunction of literals, $\alpha_{i}$, $\beta_{j}$ guard atoms, and $\gamma_{j}$ an atom. Additionally let $\mathcal{E}$ be the set of all guarded rigid types of order $\le \mathrm{width}(\Phi)$, and $\sigma$ be the set of all type predicates associated to the types from $\mathcal{E}$. Let $C$ be the set of all constants from $\tau$ and $\sigma$ the signature containing for every $\tp$ a fresh predicate symbol $\sig{Tp}_{\tp}$ of arity $\mathrm{order}(\tp)$. We furthermore fix some $(\mathcal{E}_{+}, \mathcal{E}_{!})\in\getsummary{\Phi}$.

We will prove an additional ``preservation'' lemma akin to \Cref{lem:PsiGen}.
\begin{lemma}\label{lem:PsiGuardTGF}
	Then for $\mathfrak{A} \in\getmodels{\Phi}$ the $\mathcal{E}$-adornment $\mathfrak{A}\cdot \mathcal{E}$ satisfies $\Psi_{\mathrm{hom}}$ and $\Psi_{\mathrm{guard}, \mathcal{E}(\mathfrak{A})}$. If $\Phi$ is from \TGF{}, $\mathfrak{A}\cdot\mathcal{E}$ additionally satisfies $\Psi_{\text{\TGF{}}, \mathcal{E}(\mathfrak{A})}$. 
	
	Furthermore, for every $\mathfrak{B}$ a homomorphic codomain of $\mathfrak{A}$ via a homomorphism $h \colon \mathfrak{A}\to\mathfrak{B}$ the $\tau\uplus\sigma$-structure $\mathfrak{B} \cdot h(\mathfrak{A}\cdot \mathcal{E}|_{\sigma})$ equally satisfies $\Psi_{\mathrm{hom}}$ and $\Psi_{\mathrm{guard}, \mathcal{E}(\mathfrak{A})}$. If $\Phi$ is from \TGF{}, $\mathfrak{B} \cdot h(\mathfrak{A}\cdot \mathcal{E}|_{\sigma})$ additionally satisfies $\Psi_{\text{\TGF{}}, \mathcal{E}(\mathfrak{A})}$. 
\end{lemma}
\begin{proof}
	Let $\mathfrak{A}\in\getmodels{\Phi}$. We set $\mathcal{E}(\mathfrak{A}) = (\mathcal{E}_{+}, \mathcal{E}_{!})$. We first prove everything for $\Phi$ from \GFO{} or \TGF{}. We then conclude with the special case that $\Phi$ is just from \TGF{}.
	
	Note that $\mathfrak{A}\cdot\mathcal{E}$ trivially satisfies $\Psi_{\mathrm{hom}}$ by definition, as $\bold{a}\in\sig{Tp}_{\tp}^{\mathfrak{A}\cdot\mathcal{E}}$ if and only if $\bold{a}$ realizes the type $\tp$. Hence it realizes every positive literal in $\tp$.
	
	We show that $\mathfrak{A}\cdot\mathcal{E}$ satisfies $\Psi_{\mathrm{guard}, \mathcal{E}(\mathfrak{A})}$. For this we will make use of the fact that we assumed $\Phi$ to be in normal form (see also \Cref{lem:GFONF}). Let $\bold{a}$ be a tuple from $\mathfrak{A}\cdot\mathcal{E}$ (and hence from $\mathfrak{A}$) and $(\tp, \nu)\in B_{j}^{\mathcal{E}(\mathfrak{A})}$ for some $j$ such that $\bold{a}\in \sig{Tp}_{\tp}^{\mathfrak{A}\cdot\mathcal{E}}$. Let furthermore $\mu_{\bold{a}}\colon\bold{v}_{\tp}\to \bold{a}$ be the componentwise variable assignment witnessing this fact, i.e. $(\mathfrak{A}, \mu_{\bold{a}})\vDash \sig{Tp}_{\tp}$. We extend $\mu_{\bold{a}}$ to $\mu'_{\bold{a}}\colon\bold{v}_{\tp}\uplus C \to A$ where for every $\sig{c}\in C$ we set $\mu'_{\bold{a}}(\sig{c}) = \sig{c}^{\mathfrak{A}}$. By definition $\bold{a}$ realizes $\tp$ and thus (by virtue of $(\tp, \nu)$ being in $B_{j}^{\mathcal{E}(\mathfrak{A})}$) we obtain by componentwise application that $(\mathfrak{A}, \mu'_{\bold{a}}\circ \nu) \vDash \beta_j(\bold{z})$.
	By $\Phi_{\forall\exists}$ there is a supertuple $\bold{a}^{\ast}$ of $\bold{a}$ with $\nu'\colon\bold{yz}\to\bold{v}^{\ast}\uplus C$ extending $\nu$ (i.e. $\nu'|_{\bold{z}}= \nu$) and $\mu_{\bold{a}^{\ast}}\colon \bold{v}^{\ast}\uplus C\to A$ (where $\mu_{\bold{a}^{\ast}}(\bold{v}^{\ast}) = \bold{a}^{\ast}$ holds componentwise) extending $\mu'_{\bold{a}}$ such that $(\mathfrak{A}, \mu_{\bold{a}^{\ast}}\circ\nu') \vDash \gamma_j(\bold{yz})$. Then $\bold{a}^{\ast}$ has a unique type $\tp'$ with $\bold{v}_{\tp'}=\bold{v}^{\ast}$ componentwise and $\tp'$ being guarded by $\gamma_j$ as witnessed by $\nu'$. So $\tp'\in\mathcal{E}$ and hence $(\tp', \nu')\in P_{j}^{\mathcal{E}(\mathfrak{A})}$. Since $\bold{v}_{\tp}\subseteq \bold{v}_{\tp'}$, $\nu'$ extends $\nu$, $\mu_{\bold{a}^{\ast}}$ extends $\mu_{\bold{a}}$ we obtain that $\tp'|_{\bold{v}_{\tp}}=\tp$, hence $(\tp, \nu)\Subset_j (\tp', \nu')$. By definition of $\mathfrak{A}\cdot\mathcal{E}$ the tuple $\bold{a}^{\ast}$ is in $\sig{Tp}_{\tp'}^{\mathfrak{A}\cdot\mathcal{E}}$.
	In conclusion, this argument shows that $\mathfrak{A}\cdot\mathcal{E}$ satisfies $\Psi_{\mathrm{guard}, \mathcal{E}(\mathfrak{A})}$.
	
 	Now let $\mathfrak{B}$ be a $\tau$-structure with $h\colon \mathfrak{A}\to\mathfrak{B}$ homomorphism. We let $\mathfrak{B}^{\ast}$ denote the $\tau\uplus\sigma$-structure $\mathfrak{B} \cdot h(\mathfrak{A}\cdot \mathcal{E}|_{\sigma})$. Note that by definition $h$ is also a homomorphism $\mathfrak{A}\cdot\mathcal{E}\to\mathfrak{B} \cdot h(\mathfrak{A}\cdot \mathcal{E}|_{\sigma})$.
	Let $\bold{b}$ be a tuple from $\mathfrak{B}^{\ast}$ satisfying $\sig{Tp}_{\tp}$ for some $\tp\in \mathcal{E}_{+}$. By definition there is $\bold{a}$ in $\mathfrak{A}^{\ast}$ (and hence in $\mathfrak{A}$) such that $\bold{a}\in\sig{Tp}_{\tp}^{\mathfrak{A}^{\ast}}$ and $h(\bold{a})=\bold{b}$. As $\mathfrak{A}^{\ast}$ satisfies $\Psi_{\mathrm{hom}}$ by definition we obtain that $\bold{a}$ realizes $\tp$. Especially it satisfies the positive literals from $\tp$. We extend the componentwise variable assignment $\mu\colon\bold{v}_{\tp}\to\bold{a}$ to a map $\mu'\colon \bold{v}_{\tp}\uplus C \to A$ such that $\mu'|_{\bold{v}_{\tp}}=\mu$ and for every $\sig{c}\in C$ we have $\mu'(\sig{c})=\sig{c}^{\mathfrak{A}}$. Hence for every atom $\sig{R}(\bold{t})\in\tp$ with $\bold{t}\subseteq\bold{v}_{\tp}\uplus C$ we have (applying $\mu'$ componentwise) $\mu'(\bold{t})\in\sig{R}^{\mathfrak{A}}$. Consequently, the (componentwise) image $h(\mu'(\bold{t}))$ is a subtuple from $\bold{b}$ such that $h(\mu'(\bold{t}))\in\sig{R}^{\mathfrak{B}^{\ast}}$. We thus obtain that $\mathfrak{B}^{\ast}\vDash \Psi_{\mathrm{hom}}$.
	
	Next we will show that $\mathfrak{B}^{\ast}$ also satisfies $\Psi_{\mathrm{guard}, (\mathcal{E}_{+}, \mathcal{E}_{!})}$. So let $\bold{b}$ be some tuple from $\mathfrak{B}^{\ast}$ and $(\tp, \nu)\in B_{j}^{(\mathcal{E}_{+}, \mathcal{E}_{!})}$ for some $j$ such that $\bold{b}\in \sig{Tp}_{\tp}^{\mathfrak{B}^{\ast}}$. Additionally let $\mu_{\bold{b}} \colon\bold{v}_{\tp}\to \bold{b}$ be the componentwise variable assignment. By definition of $\mathfrak{B}^{\ast}$ there is $\bold{a}$ in $\mathfrak{A}^{\ast}$ such that $\bold{a}\in\sig{Tp}_{\tp}^{\mathfrak{A}^{\ast}}$, $\mu_{\bold{a}}\colon\bold{v}_{\tp}\to\bold{a}$ componentwise variable assignment and $h(\bold{a})= \bold{b}$ (again, componentwise). Since $\mathfrak{A}^{\ast}$ satisfies $\Psi_{\mathrm{guard}, (\mathcal{E}_{+}, \mathcal{E}_{!})}$ we find $(\tp', \nu')\in P_{j}^{(\mathcal{E}_{+}, \mathcal{E}_{!})}$ with $(\tp, \nu)\Subset_j (\tp', \nu')$ and a tuple $\bold{a}^{\ast}$ such that the componentwise variable assignment $\mu_{\bold{a}^{\ast}}\colon\bold{v}_{\tp'}\to\bold{a}^{\ast}$ extends $\mu_{\bold{a}}$, i.e. $\mu_{\bold{a}^{\ast}}|_{\bold{v}_{\tp}} = \mu_{\bold{a}}$. Applying $h$ to $\bold{a}^{\ast}$ we obtain $\bold{b}^{\ast}=h(\bold{a}^{\ast})$ a supertuple of $\bold{b}$. The assignment $\mu_{\bold{b}^{\ast}} = h\circ \mu_{\bold{a}^{\ast}}$ extends $\mu_{\bold{b}}$:
	\begin{equation}
		\mu_{\bold{b}^{\ast}}|_{\bold{v}_{\tp}} = h\circ (\mu_{\bold{a}^{\ast}}|_{\bold{v}_{\tp}}) = h\circ \mu_{\bold{a}} = \mu_{\bold{b}}.
	\end{equation}
	As $h$ is a homomorphism we obtain $\bold{b}^{\ast}\in\sig{Tp}_{\tp'}^{\mathfrak{B}^{\ast}}$.
	This argument shows that $\mathfrak{B}^{\ast}\vDash \Psi_{\mathrm{guard}, \mathcal{E}(\mathfrak{A})}$.
	
	Now suppose in addition, that $\Phi$ is from \TGF{}.
	We first show that then $\mathfrak{A}\cdot\mathcal{E}$ satisfies $\Psi_{\text{\TGF{}}, \mathcal{E}(\mathfrak{A})}$. Take $a_1,a_2$ from $\mathfrak{A}\cdot\mathcal{E}$ (hence from $\mathfrak{A}$) satisfying for some types $\tp_1, \tp_2\in\mathcal{E}_{+}$ the following: $a_1 \in\sig{Tp}_{\tp_1}^{\mathfrak{A}\cdot\mathcal{E}}$ and $a_2 \in\sig{Tp}_{\tp_2}^{\mathfrak{A}\cdot\mathcal{E}}$. Without loss of generality we may assume that $\bold{v}_{\tp_1} = \{ v_1 \}$ and $\bold{v}_{\tp_2} = \{ v_2 \}$ for $v_1\neq v_2$. By definition $a_1$ realizes $\tp_1$ and $a_2$ realizes $\tp_2$ in $\mathfrak{A}$ (and $\mathfrak{A}\cdot\mathcal{E}$). For $(a_1,a_2)$ 
	there is a unique $2$-type $\tp'$ over $\tau$ with $\bold{v}_{\tp'} = v_1 v_2$ (as a subsequence from $\bold{v}$) such that $(a_1, a_2)$ realizes $\tp'$ in $\mathfrak{A}$ and $\tp'|_{v_1}=\tp_1$ and $\tp'|_{v_2}=\tp_2$. Because of $\mathfrak{A}\vDash\Phi_{\sig{Univ}}$ the type $\tp'$ contains automatically the atoms $\sig{Univ}(v_1,v_2)$ and $\sig{Univ}(v_2,v_1)$. Hence $\tp'$ is guarded and so we see that $\tp'\in\mathcal{E}$ (and of course also in $\mathcal{E}_{+}$). By definition of $\mathfrak{A}\cdot\mathcal{E}$ we obtain $(a_1, a_2)\in\sig{Tp}_{\tp'}^{\mathfrak{A}\cdot\mathcal{E}}$. Hence also $\mathfrak{A}\cdot\mathcal{E} \vDash\Psi_{\text{\TGF{}}, \mathcal{E}(\mathfrak{A})}$.
	
	Finally, also $\mathfrak{B}^{\ast}$ satisfies $\Psi_{\text{\TGF{}}, \mathcal{E}(\mathfrak{A})}$. To see this, let $b_1, b_2$ be from $\mathfrak{B}^{\ast}$ such that there are $\tp_1, \tp_2\in\mathcal{E}_{+}$ with $b_1\in\sig{Tp}_{\tp_1}^{\mathfrak{B}^{\ast}}$ and $b_2\in\sig{Tp}_{\tp_2}^{\mathfrak{B}^{\ast}}$.  Without loss of generality we may assume that $\bold{v}_{\tp_1} = \{ v_1 \}$ and $\bold{v}_{\tp_2} = \{ v_2 \}$ for $v_1\neq v_2$. By definition of $\mathfrak{B}^{\ast}$ there are $a_1,a_2$ in $\mathfrak{A}\cdot\mathcal{E}$ such that $a_1\in\sig{Tp}_{\tp_1}^{\mathfrak{A}\cdot\mathcal{E}}$ and $a_2\in\sig{Tp}_{\tp_2}^{\mathfrak{A}\cdot\mathcal{E}}$ with $h(a_1) = b_1$ and $h(a_2)=b_2$. Since we already established that $\mathfrak{A}\cdot\mathcal{E}$ satisfies $\Psi_{\text{\TGF{}}, \mathcal{E}(\mathfrak{A})}$,there is a unique $2$-type $\tp'\in\mathcal{E}_{+}$ with $\bold{v}_{\tp'} = v_1 v_2$ (as a subsequence from $\bold{v}$) with $\tp'|_{v_1}=\tp_1$ and $\tp'|_{v_2}=\tp_2$ such that $(a_1, a_2)\in\sig{Tp}_{\tp'}^{\mathfrak{A}\cdot\mathcal{E}}$. As previously mentioned, $h$ as a map is also a homomorphism $\mathcal{A}\cdot\mathcal{E}\to\mathfrak{B}^{\ast}$. Hence, by $h((a_1, a_2)) = (b_1, b_2)$, we obtain $(b_1, b_2)\in\sig{Tp}_{\tp'}^{\mathfrak{B}'}$. This argument shows $\mathfrak{B}^{\ast} \vDash\Psi_{\text{\TGF{}}, \mathcal{E}(\mathfrak{A})}$ and concludes our proof.
\end{proof}

We will now prove a lemma that will allow us to extend a structure not satisfying $\Psi_{\mathrm{guard}, (\mathcal{E}_{+}, \mathcal{E}_{!})}$ to one that does. Of course, this will not be so simple. Hence we first show that, if there is a ``problem'' for satisfaction of $\Psi_{\mathrm{guard}, (\mathcal{E}_{+}, \mathcal{E}_{!})}$ then we are able to ``repair'' this problem without introducing additional ones with respect to the original structure. This will not necessarily preclude obtaining new ``problems'' in the extended structure, though those new ones will be located in the new additions. Dealing with this will be part of \Cref{lem:inindstep} and \Cref{lem:GFOindstep}. The notion of ``problem'' in this context has been made formal by the definition of violation tuples in \Cref{def:guardvioltgf}. The main idea is the following: If we have a ``problem'' (formalized via a violation tuple), we can take some type ``fitting'' for the position where the ``problem'' occurs (consistent with $\Psi_{\mathrm{guard}, (\mathcal{E}_{+}, \mathcal{E}_{!})}$) and glue $\mathfrak{D}_{\tp}$ (or $(\mathfrak{D}_{\tp})^{\tau}$ in the enriched setting) on accordingly. And as we are concerned with homclosures we will also discuss the cases when we have homomorphisms from a structure ``having problems'' to one that has none.

\begin{lemma}\label{lem:addtype}
	Let $\mathfrak{A}$ be a finite $\sigma\uplus C$-structure satisfying $\Psi_{\mathrm{gen}, (\mathcal{E}_{+}, \mathcal{E}_{!})}$. 
	Let $(\tp, \nu)\in B_{j}^{(\mathcal{E}_{+}, \mathcal{E}_{!})}$ and $(\tp', \nu')\in P_{j}^{(\mathcal{E}_{+}, \mathcal{E}_{!})}$ be fitting, i.e. $(\tp, \nu)\Subset_{j}(\tp', \nu')$. Furthermore let $\bold{a}$ be a tuple of length $\mathrm{order}(\tp)$ in $\mathfrak{A}$ and $\mu\colon\bold{v}_{\tp}\to\bold{a}$ be the componentwise variable assignment such that $(\mathfrak{A}, \mu) \vDash \sig{Tp}_{\tp}(\bold{v}_{\tp})$, i.e. $\bold{a}\in \sig{Tp}_{\tp}^{\mathfrak{A}}$. Then we can extend $\mathfrak{A}$ to a $\sigma\uplus C$-structure $\mathfrak{A}^{\ast}$ and the componentwise assignment $\mu$ to a componentwise one $\mu^{\ast} \colon\bold{v}_{\tp'}\to \mathfrak{A}^{\ast}$ in a canonical way such that $(\mathfrak{A}^{\ast}, \mu^{\ast})\vDash \sig{Tp}_{\tp'}(\bold{v}_{\tp'})$, and $\mu^{\ast}(\bold{v}_{\tp'}\setminus\bold{v}_{\tp})$ has no element in common with the domain $A$ of $\mathfrak{A}$.
	Furthermore $\mathfrak{A}$ is an induced substructure of $\mathfrak{A}^{\ast}$ and $\mathfrak{A}^{\ast}$ satisfies $\Psi_{\mathrm{gen}, (\mathcal{E}_{+}, \mathcal{E}_{!})}$.
	
	If there is a $\sigma\uplus C$-structure $\mathfrak{B}$ satisfying
	\begin{equation}
		\forall\bold{v}_{\tp} .\sig{Tp}_{\tp}(\bold{v}_{\tp}) \Rightarrow \exists\bold{v}_{\tp'}\setminus \nu(\bold{z}). \sig{Tp}_{\tp'}(\bold{v}_{\tp'})
	\end{equation}
	 and a homomorphism $h\colon \mathfrak{A}\to\mathfrak{B}$ we can extend it to a homomorphism $h^{\ast}\colon \mathfrak{A}^{\ast}\to\mathfrak{B}$.
	 
	 Additionally, if there are $\tau\uplus\sigma$-structures $\mathfrak{A}'$ and $\mathfrak{B}'$ such that $\mathfrak{A}=\mathfrak{A}'|_{\sigma\uplus C}$ and $\mathfrak{B}=\mathfrak{B}'|_{\sigma\uplus C}$, $\mathfrak{B}'\vDash\Psi_{\mathrm{hom}}$, 
	 for all $\tp\in\mathcal{E}$ it holds that a tuple $\bold{a}$ without constants from $\mathfrak{A}'$ satsifies $\sig{Tp}_{\tp}$ in $\mathfrak{A}'$ if and only if $\bold{a}$ realizes $\tp$ in $\mathfrak{A}'$, 
	 and there is a homomorphism $h'\colon\mathfrak{A}'\to\mathfrak{B}'$, then we can extend $\mathfrak{A}'$ to a structure $(\mathfrak{A}')^{\ast}$ with $(\mathfrak{A}')^{\ast}|_{\sigma\uplus C}=\mathfrak{A}^{\ast}$ and $h'$ to a homomorphism $(h')^{\ast}\colon (\mathfrak{A}')^{\ast}\to \mathfrak{B}'$. Furthermore, $\mathfrak{A}'$ is an induced substructure of $(\mathfrak{A}')^{\ast}$, and $(\mathfrak{A}')^{\ast}$ satisfies $\Psi_{\mathrm{gen}, (\mathcal{E}_{+}, \mathcal{E}_{!})}$. 
	 Additionally, for all $\tp\in\mathcal{E}$ it holds that a tuple $\bold{a}$ without constants from $(\mathfrak{A}')^{\ast}$ satsifies $\sig{Tp}_{\tp}$ in $(\mathfrak{A}')^{\ast}$ if and only if $\bold{a}$ realizes $\tp$ in $(\mathfrak{A}')^{\ast}$.
	 
	 Let $Q_{\bold{a}}$ be the set of all violation tuples of $\mathfrak{A}$ (or $\mathfrak{A}'$) of the form $q_i = (\bold{c}_i, \mu_i, \tp_i, \nu_i)$ where the $\bold{c}_i$ consist only of elements from $\bold{a}$. Not only is $Q_{\bold{a}}$ finite but it contains already all violation tuples of the form $(\bold{d}, \mu_{\ast}, \tp_{\ast}, \nu_{\ast})$ of $\mathfrak{A}^{\ast}$ (or $(\mathfrak{A}')^{\ast}$) where $\bold{d}$ consists of elements from $\bold{a}$.
\end{lemma}

\begin{proof}
	Let $\mathfrak{A}$ be a finite $\sigma\uplus C$ structure satisfying $\Psi_{\mathrm{gen}, (\mathcal{E}_{+}, \mathcal{E}_{!})}$. 
	Let $(\tp, \nu)\in B_{j}^{(\mathcal{E}_{+}, \mathcal{E}_{!})}$ and $(\tp', \nu')\in P_{j}^{(\mathcal{E}_{+}, \mathcal{E}_{!})}$ be fitting, i.e. $(\tp, \nu)\Subset_{j}(\tp', \nu')$. Furthermore let $\bold{a}$ be a tuple of length $\mathrm{order}(\tp)$ in $\mathfrak{A}$ and $\mu\colon\bold{v}_{\tp}\to\bold{a}$ be the componentwise variable assignment such that $(\mathfrak{A}, \mu) \vDash \sig{Tp}_{\tp}(\bold{v}_{\tp})$, i.e. $\bold{a}\in \sig{Tp}_{\tp}^{\mathfrak{A}}$. 
	
	Now, consider the $\sigma$-structure $\mathfrak{D}_{\tp'}$ and be reminded that its domain are the elements from $\bold{v}_{\tp'}$.
	We obtain $\mathfrak{D}'_{\tp'}$ from $\mathfrak{D}_{\tp'}$ by renaming all elements $\nu'(z) \in \bold{v}_{\tp'} \cap \nu'(\bold{z})$ into $\mu(z)$ (i.e., into elements already present in $\mathfrak{A}$) and assigning a fresh name to any other element $v \in \bold{v}_{\tp'} \setminus \nu'(\bold{z})$. 
	Now we obtain $\mathfrak{A}^{\ast}$ by taking the (non-disjoint) union of $\mathfrak{A}$ with $\mathfrak{D}'_{\tp'}$.
	
	Note that by $(\tp, \nu)$ and $(\tp', \nu')$ being fitting we obtain that $\mathfrak{A}$ is an induced substructure of $\mathfrak{A}^{\ast}$ by construction. 
	As $(\tp', \nu')\in P_{j}^{(\mathcal{E}_{+}, \mathcal{E}_{!})}$, clearly $\tp' \in\mathcal{E}_{+}$. By definition and applying \Cref{lem:eligclosure} this means that $\mathfrak{D}_{\tp'}$ and hence $\mathfrak{D}'_{\tp'}$ satisfies $\Psi_{\mathrm{gen}, (\mathcal{E}_{+}, \mathcal{E}_{!})}$. More specifically, the validity of Formulae \ref{eq: gen 3} and \ref{eq: gen 3plus} is clear from the closure under those formulae in the definition, whereas \ref{eq: gen 1} and \ref{eq: gen 2} follow from $\tp\in\mathcal{E}_{+}$ and \Cref{lem:eligclosure} and the definition. 
	Note that the saturation with respect to Formulae \ref{eq: gen 3} and \ref{eq: gen 3plus} is closely tied to the relation of the types to which the type predicates are associated. 
	And, by again applying that $(\tp, \nu)$ and $(\tp', \nu')$ are fitting, we conclude that the union of $\mathfrak{A}$ and $\mathfrak{D}'_{\tp'}$, which is $\mathfrak{A}^{\ast}$ inherits the satisfaction of $\Psi_{\mathrm{gen}, (\mathcal{E}_{+}, \mathcal{E}_{!})}$. 
	This is clear while observing that the substructure induced by $\bold{a}$ in $\mathfrak{A}$ and the one in $\mathfrak{D}'_{\tp'}$ are the same. Hence we do not create any new relations with respect to relation symbols from $\sigma$.
	By this token, $\mathfrak{A}$ is an induced substructure of $\mathfrak{A}^{\ast}$, as no new relations are introduced on tuples purely from $\mathfrak{A}$.
	
	If there is a $\sigma\uplus C$-structure $\mathfrak{B}$ satisfying
	\begin{equation}
		\forall\bold{v}_{\tp} .\sig{Tp}_{\tp} \Rightarrow \exists\bold{v}_{\tp'}\setminus \nu(\bold{z}). \sig{Tp}_{\tp'}(\bold{v}_{\tp'}) \label{eq: addtype 1}
	\end{equation}
	and a homomorphism $h\colon \mathfrak{A}\to\mathfrak{B}$ we adjust the construction in the following way:
	Since $h\colon \mathfrak{A}\to\mathfrak{B}$ is a homomorphism the tuple $\bold{b}=h(\bold{a})$ satisfies $\bold{b}\in\sig{Tp}_{\tp}^{\mathfrak{B}}$. This is witnessed by the componentwise assignment $\mu' \colon\bold{v}_{\tp} \to\bold{b}$ defined as $h\circ\mu$. As $\mathfrak{B}$ satisfies Formula \ref{eq: addtype 1} we obtain a tuple $\bold{b}^{\ast}$ containing $\bold{b}$ as a subtuple and can extend $\mu'$ to the componentwise assignment $\mu^{\ast}\colon\bold{v}_{\tp'}\to\bold{b}^{\ast}$ (i.e. $\mu^{\ast}|_{\bold{v}_{\tp}}=\mu'$).
	Then, by what we have seen earlier, the componentwise map $g\colon\bold{v}_{\tp'} \to \bold{b}^{\ast}, v\mapsto \mu'(v)$ is a homomorphism from $\mathfrak{D}_{\tp'}$ to $\mathfrak{B}$. By moving from $\mathfrak{D}_{\tp'}$ to $\mathfrak{D}'_{\tp'}$ via renaming (which establishes an isomorphism $f\colon \mathfrak{D}_{\tp'}\to\mathfrak{D}'_{\tp'}$) we obtain via $g\circ f^{-1}$ a homomorphism $\dot{g}\colon \mathfrak{D}'_{\tp'}\to \mathfrak{B}$
	
	After obtaining $\mathfrak{A}^{\ast}$ via the union of $\mathfrak{A}$ with $\mathfrak{D}'_{\tp'}$, we let $h^{\ast}$ be the union of the maps $h$ and $\dot{g}$, i.e. $h\cup\dot{g}$. Note that this is well defined, since $h(a) = \dot{g}(a)$ for every $a\in\mathbf{a}$. To be more precise, by the variable assignments $\mu, \mu'$ being componentwise and $(\tp, \nu)$ and $(\tp', \nu')$ being fitting we obtain for $v\in\bold{v}_{\tp}$ with $\mu(v)=a$:
	\begin{equation}
		h(a) = h(\mu(v)) = \mu'(v) = g(v) = g\circ f^{-1}(a) = \dot{g}(a).
	\end{equation}
	Since the  substructures induced by $\bold{a}$  in $\mathfrak{D}'_{\tp'}$ and  in $\mathfrak{A}$ are equal, this implies together with  $h$ and $\dot{g}$ being local homomorphisms  that $h^{\ast}\colon\mathfrak{A}^{\ast}\to\mathfrak{B}$ is a homomorphism.
	
	We turn to the part about $\tau\uplus \sigma$-structures $\mathfrak{A}'$ and $\mathfrak{B}'$. As $\mathfrak{A}'|_{\sigma\uplus C}= \mathfrak{A}$ and $\mathfrak{B}'|_{\sigma\uplus C}=\mathfrak{B}$, $\mathfrak{A}'$ satisfies $\Psi_{\mathrm{gen}, (\mathcal{E}_{+}, \mathcal{E}_{!})}$ whereas $\mathfrak{B}'$ satisfies \ref{eq: addtype 1}. Additionally, $h'\colon\mathfrak{A}'\to\mathfrak{B}'$ is a homomorphism $\mathfrak{A}\to\mathfrak{B}$. Hence we apply the previous results from this lemma/proof and extend $\mathfrak{A}'$ to a structure $\mathfrak{A}'_{\mathrm{temp}}$ by adding $\mathfrak{D}_{\tp'}$ as previously described. Additionally we extend $h'$ to $(h')^{\ast}$ again as previously described. We hence obtain $\mathfrak{A}'_{\mathrm{temp}}|_{\sigma\uplus C} = \mathfrak{A}^{\ast}$ and $(h')^{\ast}\colon \mathfrak{A}'_{\mathrm{temp}}|_{\sigma\uplus C}\to\mathfrak{B}'|_{\sigma\uplus C}$ being a homomorphism. 
	Let $f'\colon \mathfrak{C}\uplus \mathfrak{D}_{\tp'}\to \mathfrak{A}'_{\mathrm{temp}}$ be the map that bijectively maps $\mathfrak{D}_{\tp'}$ to $\mathfrak{D}'_{\tp'}$ in $\mathfrak{A}'_{\mathrm{temp}}$ and $f'(\sig{c}^{\mathfrak{C}}) = \sig{c}^{(\mathfrak{A}')^{\ast}}$. For $\sig{R}\in\tau\setminus C$ we define
	\begin{equation}
		\sig{R}^{(\mathfrak{A}')^{\ast}} = \sig{R}^{\mathfrak{A}'_{\mathrm{temp}}} \cup \{ f'(\bold{d})\ | \ \bold{d}\in\sig{R}^{ (\mathfrak{D}_{\tp'})^{\tau}}  \}.
	\end{equation}
	By definition, $f'$ gives rise to an embedding $f'_{\tp'}\colon (\mathfrak{D}_{\tp'})^{\tau}\to (\mathfrak{A}')^{\ast}$: That $f'_{\tp'}$ is an injective homomorphism is plain. So assume we have $\bold{d}$ such that $f'_{\tp'}(\bold{d})\in\sig{R}^{(\mathfrak{A}')^{\ast}}$ for $\sig{R}\in\tau\setminus C$. By definition $\bold{d}\in\sig{R}^{(\mathfrak{D}_{\tp'})^{\tau}}$. Hence we have to discuss $f'_{\tp'}(\bold{d})\in\sig{T}^{(\mathfrak{A}')^{\ast}}$ for $\sig{T}\in\sigma$. As the $\sig{T}\in\sigma$ relations holding for elements from $f'_{\tp'}((\mathfrak{D}_{\tp'})^{\tau})$ have been taken over exactly from $\mathfrak{D}_{\tp'}$ we easily obtain $\bold{d}\in\sig{T}^{(\mathfrak{D}_{\tp'})^{\tau}}$ showing that $f'_{\tp'}$ is an embedding.
	Additionally, every tuple $\bold{a}$ from $(\mathfrak{A}')^{\ast}$ such that $\bold{a}\in\sig{Tp}_{\tp''}^{(\mathfrak{A}')^{\ast}}$ realizes $\tp''$ in $(\mathfrak{A}')^{\ast}$ for arbitrary $\tp''\in\mathcal{E}$. If $\bold{a}$ is purely from $\mathfrak{A}'$ this holds by definition. Hence we only need to discuss the case where $\bold{a}$ is from $f'_{\tp'}((\mathfrak{D}_{\tp'})^{\tau})$. Since $(\tp, \nu)$ and $(\tp', \nu')$ where fitting, we immediately obtain that as soon as $\bold{a}$ contains at least one element from $f'_{\tp'}((\mathfrak{D}_{\tp'})^{\tau})$ all of its elements are from this set.
	But by $f'_{\tp'}$ being an embedding there is a unique $\bold{d}$ in $(\mathfrak{D}_{\tp'})^{\tau}$ with $\bold{d}\in \sig{Tp}_{\tp''}^{(\mathfrak{D}_{\tp'})^{\tau}}$ such that $f'_{\tp'}(\bold{d})=\bold{a}$. By \Cref{lem:Dtauexpansion} $\bold{d}$ realizes $\tp''$ in $(\mathfrak{D}_{\tp'})^{\tau}$ and hence by $f'_{\tp'}$ being an embedding, $\bold{a}=f'_{\tp'}(\bold{d})$ realizes $\tp''$ in $(\mathfrak{A}')^{\ast}$.
	\newline
	For the other direction suppose $\bold{a}$ realizes some type $\tp''$ in $(\mathfrak{A}')^{\ast}$. Then, by construction, it is a tuple either from just $\mathfrak{A}'$, where by assumption $\bold{a}\in\sig{Tp}_{\tp''}^{(\mathfrak{A}')^{\ast}}$ holds, or it is from $f'_{\tp'}((\mathfrak{D}_{\tp'})^{\tau})$. As $f'_{\tp'}$ is an embedding there exists a unique $\bold{d}$ from $(\mathfrak{D}_{\tp'})^{\tau}$ such that $f'_{\tp'}(\bold{d})=\bold{a}$ and $\bold{d}$ realizes $\tp''$ in $(\mathfrak{D}_{\tp'})^{\tau}$. By \Cref{lem:Dtauexpansion} $\bold{d}\in\sig{Tp}_{\tp''}^{(\mathfrak{D}_{\tp'})^{\tau}}$ and hence via $f'_{\tp'}$ being an embedding we obtain $\bold{a}\in\sig{Tp}_{\tp''}^{(\mathfrak{A}')^{\ast}}$.
	\newline
	Lastly, $(h')^{\ast}$ is a homomorphism $(\mathfrak{A}')^{\ast}\to\mathfrak{B}'$. For this let $\sig{R}\in\tau\setminus C$ and $\bold{a}$ be a tuple from $(\mathfrak{A}')^{\ast}$ with $\bold{a}\in\sig{R}^{(\mathfrak{A}')^{\ast}}$. {We continue the discussion by using the atom $\sig{R}(\bold{x})$ containing only the (pairwise distinct) variables from $\bold{x}$.} Set $\mu\colon\bold{x}\to\bold{a}$ as the componentwise valuation of $\bold{x}$ by $\bold{a}$. Let $\bold{y}=(y_1, \ldots, y_m)$ be a tuple of (fresh) variables, where $m$ is the number of all components (allowing repetitions of elements to be counted) of $\bold{a}$ that are not the interpretation of constant symbols. We define $\nu_{\mu}\colon\bold{x}\to\bold{y}\uplus C$ as follows: $\nu_{\mu}(x)=\sig{c}\in C$ if and only if $\mu(x)=\sig{c}^{(\mathfrak{A}')^{\ast}}$, $\nu_{\mu}(x)\in\bold{y}$ otherwise. Additionally, $\nu_{\mu}$ shall map the subtuple $\bold{x}^{\ast}$ of $\bold{x}$ induced by $\{ x \in\bold{x} \ | \ \nu_{\mu}(x)\in\bold{y}  \}$ componentwise to the tuple $\bold{y}$. By definition $\mu_{\nu}$ is a bijection between $\bold{x}^{\ast}$ and $\bold{y}$. Finally note that $\sig{c}$ in the definition of $\nu_{\mu}$ might not be unique. In this case we choose one of the possible options. Still this map will be well defined for our purposes. We now define the atom $\delta(\bold{y})$ as $\nu_{\mu}(\sig{R}(\bold{x}))$. Additionally, we acknowledge that $\bold{a}^{\ast}$, as a tuple being the componentwise image of $\bold{x}^{\ast}$ of $\mu$, satisfies $\delta(\bold{y})$ in $(\mathfrak{A}')^{\ast}$. Let $\tp^{\ast}$ be the $\tau$-type realized by $\bold{a}^{\ast}$ in $(\mathfrak{A}')^{\ast}$. As $\bold{x}^{\ast}$ satisfies $\delta(\bold{y})$, $\tp^{\ast}$ is guarded by $\delta$ (witnessed by the componentwise map $\nu\colon\bold{y}\to\bold{v}_{\tp^{\ast}}$). Hence $\tp^{\ast}\in\mathcal{E}$ and $\bold{a}^{\ast}$ satisfies $\sig{Tp}_{\tp^{\ast}}$ in $(\mathfrak{A}')^{\ast}$ by construction. As $\mathfrak{A}'$ satisfies $\Psi_{\mathrm{gen}, (\mathcal{E}_{+}, \mathcal{E}_{!})}$ we obtain that $\tp^{\ast}\in\mathcal{E}_{+}$ (by applying Formula \ref{eq: gen 2}). As $(h')^{\ast}$ has also been shown to be a homomorphism $(\mathfrak{A}')^{\ast}|_{\sigma\uplus C}\to\mathfrak{B}'|_{\sigma\uplus C}$, we obtain $(h')^{\ast}(\bold{a}^{\ast})\in\sig{Tp}_{\tp^{\ast}}^{\mathfrak{B}'}$. Since $\mathfrak{B}'\vDash \Psi_{\mathrm{hom}}$, $(h')^{\ast}(\bold{a}^{\ast})$ satisfies $\delta(\bold{y})$ in $\mathfrak{B}'$. Again, as $(h')^{\ast}$ has also been shown to be a homomorphism $(\mathfrak{A}')^{\ast}|_{\sigma\uplus C}\to\mathfrak{B}'|_{\sigma\uplus C}$, we conclude further that since $(h')^{\ast}(\sig{c}^{(\mathfrak{A}')^{\ast}})=\sig{c}^{\mathfrak{B}'}$, application of $(h')^{\ast}$ to $\bold{a}$ yields $(h')^{\ast}(\bold{a})\in\sig{R}^{\mathfrak{B}'}$. Consequently, $(h')^{\ast}\colon(\mathfrak{A}')^{\ast}\to\mathfrak{B}'$ is a homomorphism.
	
	We show that $\mathfrak{A}'$ is an induced substructure of $(\mathfrak{A}')^{\ast}$. It is clear by construction, that the injective map $\iota\colon\mathfrak{A}'\to(\mathfrak{A}')^{\ast}, a\mapsto a$ is a homomorphism. As by construction no new type predicate has been introduced on tuples from purely $\mathfrak{A}'$, there also did not arise the need to introduce new $\tau$-relations on tuples from purely $\mathfrak{A}'$ in order to saturate for a correspondence between type predicates and $\tau$-types (as this was already given by assumption). Hence $\iota$ is already an embedding, as every $\tau$-relation $\sig{R}$ holding for some $\iota(\bold{a})$ in $(\mathfrak{A}')^{\ast}$ (where $\bold{a}$ is a tuple from $\mathfrak{A}'$) already held in $\mathfrak{A}'$.
	Hence, $\mathfrak{A}'$ is an induced substructure of $(\mathfrak{A}')^{\ast}$, and $(\mathfrak{A}')^{\ast}$ satisfies $\Psi_{\mathrm{gen}, (\mathcal{E}_{+}, \mathcal{E}_{!})}$.
	
	For the remainder it is enough to restrict our argument to the $\sigma\uplus C$-structures. Let $Q_{\bold{a}}$ be the set of all violation tuples of $\mathfrak{A}$ of the form $q_i = (\bold{c}_i, \mu_i, \tp_i, \nu_i)$ where the $\bold{c}_i$ consist only of elements from $\bold{a}$.
	As $\bold{a}$ is finite and so are the number if types and type variables, $Q_{\bold{a}}$ is finite. Hence we can assume $Q_{\bold{a}} = \{ q_1, \ldots, q_m \}$ with the $q_i$ as mentioned.

	It remains to show, that all violation tuples of the form $(\bold{d}, \mu_{\ast}, \tp_{\ast}, \nu_{\ast})$ of $\mathfrak{A}^{\ast}$ where $\bold{d}$ consists of elements from $\bold{a}$ are already in $Q_{\bold{a}}$. Suppose we have  such a violation tuple $(\bold{d}, \mu_{\ast}, \tp_{\ast}, \nu_{\ast})$ of $\mathfrak{A}^{\ast}$. Consequently $(\mathfrak{A}^{\ast}, \mu_{\ast})\vDash \sig{Tp}_{\tp_{\ast}}(\bold{v}_{\tp_{\ast}})$, i.e. $(\bold{d})\in \sig{Tp}_{\tp_{\ast}}^{\mathfrak{A}^{\ast}}$. 
	But since in the construction of $\mathfrak{A}^{\ast}$ we have that $(\tp, \nu)\Subset_j (\tp', \nu')$, i.e. they are fitting, the substructures of $\mathfrak{A}$ and $\mathfrak{D}'_{\tp'}$ are the same (hence isomorphic), every predicate $\sig{Tp}_{\tp''}\in \sigma$ that holds for tuples built from elements of $\bold{a}$ in $\mathfrak{A}^{\ast}$ already held in $\mathfrak{A}$ and vice versa. Consequently $\bold{d}\in\sig{Tp}_{\tp_{\ast}}^{\mathfrak{A}}$. But then we conclude by construction that $(\bold{d}, \mu_{\ast}, \tp_{\ast}, \nu_{\ast})$ was already a violation tuple of $\mathfrak{A}$. By definition this means it is in $Q_{\bold{a}}$.
\end{proof}

As we (essentially) just did for $\Psi_{\mathrm{guard}, (\mathcal{E}_{+}, \mathcal{E}_{!})}$ in \Cref{lem:addtype} we will show how we can enrich structures in order to make $\Psi_{\text{\TGF{}}, (\mathcal{E}_{+}, \mathcal{E}_{!})}$ true.

\begin{lemma}\label{lem:psiTGFclosure}
	We restrict $\Phi$ to be a \TGF{} sentence in normal form. Let $\mathfrak{A}, \mathfrak{B}$ be $\tau\uplus\sigma$-structures with $\mathfrak{A}$ being finite and a homomorphism $h\colon \mathfrak{A}\to\mathfrak{B}$ such that both $\mathfrak{A}$ and $\mathfrak{B}$ satisfy $\Psi_{\mathrm{gen}, (\mathcal{E}_{+}, \mathcal{E}_{!})}$. 
	Assume that $\mathfrak{A}$ satisfies the following: For all $\tp\in\mathcal{E}$ it holds that a tuple $\bold{a}$ from $\mathfrak{A}$ without constants  satisfies $\sig{Tp}_{\tp}\in\sigma$ in $\mathfrak{A}$ if and only if $\bold{a}$ realizes $\tp$ in $\mathfrak{A}$.
	Further let $\mathfrak{B}$ additionally satisfy $\Psi_{\text{\TGF{}}, (\mathcal{E}_{+}, \mathcal{E}_{!})}$ and $\Psi_{\mathrm{hom}}$. Then there is a $\tau\uplus\sigma$-structure $\mathfrak{A}'$ with domain $A$ satisfying $\Psi_{\text{\TGF{}}, (\mathcal{E}_{+}, \mathcal{E}_{!})}$ in addition to $\Psi_{\mathrm{gen}, (\mathcal{E}_{+}, \mathcal{E}_{!})}$, and for every type $\tp\in\mathcal{E}$ it holds that a tuple $\bold{a}$ {without constants} from $\mathfrak{A}'$ satisfies $\sig{Tp}_{\tp}\in\sigma$ in $\mathfrak{A}'$ if and only if $\bold{a}$ realizes $\tp$ in $\mathfrak{A}'$. Lastly the identity map $\iota\colon\mathfrak{A}\to\mathfrak{A}', a\mapsto a$ is a homomorphism.
	
	Furthermore, every induced substructure of $\mathfrak{A}$ satisfying $\Psi_{\text{\TGF{}}, (\mathcal{E}_{+}, \mathcal{E}_{!})}$ also is an induced substructure of $\mathfrak{A}'$.
\end{lemma}
\begin{proof}
	Let $\mathfrak{A}, \mathfrak{B}$ be $\tau\uplus\sigma$-structures with $\mathfrak{A}$ being finite, $h\colon \mathfrak{A}\to\mathfrak{B}$ a homomorphism, $\mathfrak{A}, \mathfrak{B}\vDash \Psi_{\mathrm{gen}, (\mathcal{E}_{+}, \mathcal{E}_{!})}$ and $\mathfrak{B}\vDash \Psi_{\text{\TGF{}}, (\mathcal{E}_{+}, \mathcal{E}_{!})}\land\Psi_{\mathrm{hom}}$. Furthermore, for every tuple $\bold{a}$ from $\mathfrak{A}$, containing no $\sig{c}^{\mathfrak{A}}$ for $\sig{c}\in C$ and every $\tp\in\mathcal{E}$ the following equivalence holds: $\bold{a}$ satisfies $\sig{Tp}_{\tp}\in\sigma$ in $\mathfrak{A}$ if and only if $\bold{a}$ realizes $\tp$ in $\mathfrak{A}$. If $\mathfrak{A}$ already satisfies $\Psi_{\text{\TGF{}}, (\mathcal{E}_{+}, \mathcal{E}_{!})}$ we are already done, so suppose this were not the case.
	We will start by successively saturating the $\sigma$-relations of $\mathfrak{A}$ in a minimal manner such that the so obtained structure (denoted by $\mathfrak{A}^{\ast}$) satisfies $\Psi_{\text{\TGF{}}, (\mathcal{E}_{+}, \mathcal{E}_{!})}$. Beside that, nothing more will be done in this first step. This also means, that the interpretations of $\sig{R}\in\tau$ relations will not change, equally the domain remains to be $A$. 
	Hence we do the following for every two elements $a_{1}, a_{2}$ from $\mathfrak{A}$ such that there are 1-types $\tp_{1}(v_{1}), \tp_{2}(v_{2})\in \mathcal{E}_{+}$ with $a_{1}\in\sig{Tp}_{\tp_{1}}^{\mathfrak{A}}, a_{2}\in\sig{Tp}_{\tp_{2}}^{\mathfrak{A}}$:
	If there is already a 2-type $\tp^{\ast} \in \mathcal{E}_{+}$ with $\sig{Tp}_{\tp^{\ast}}^{\mathfrak{A}^{\ast}}$ containing $(a_{1}, a_{2})$ then, by Formula \ref{eq: gen 3}, $\Psi_{\text{\TGF{}}, (\mathcal{E}_{+}, \mathcal{E}_{!})}$ cannot be violated by this pair.
	This is also the case, when $a_1=a_2$. We see this by the following argument: Since $\Phi$ is a normal form sentence from \TGF{} every type $\tp\in\mathcal{E}$ contains for each $v\in\bold{v}_{\tp}$ the atom $\sig{Univ}(v,v)$. Hence, by $a_{1}\in\sig{Tp}_{\tp_{1}}^{\mathfrak{A}}$ and the assumption that type predicate satisfaction and type realization of the associated types coincides we obtain $(a_1,a_1)\in\sig{Univ}^{\mathfrak{A}}$. By $a_1=a_2$ this is the same as $(a_1,a_2)\in\sig{Univ}^{\mathfrak{A}}$. Hence $(a_1,a_2)$ satisfies $\sig{Univ}(v_1,v_2)$. So the type $\tp'$ of $(a_1,a_2)$ with variables $(v_1, v_2)$ is guarded (by $\sig{Univ}(v_1,v_2)$) and consequently $\tp'\in\mathcal{E}$. By going back to type predicates, applying that $\mathfrak{A}\vDash\Psi_{\mathrm{gen}, (\mathcal{E}_{+}, \mathcal{E}_{!})}$ we obtain that $\tp'\in\mathcal{E}_{+}$. Hence $(a_1,a_2)\in\sig{Tp}_{\tp'}^{\mathfrak{A}}$.
	\newline
	So suppose the pair $(a_{1}, a_{2})$ is not yet in $\sig{Tp}_{\tp^{\ast}}^{\mathfrak{A}^{\ast}}$ for any 2-type $\tp^{\ast}$. We know by $h$ being a homomorphism that $h(a_{1})\in\sig{Tp}_{\tp_{1}}^{\mathfrak{B}}, h(a_{2})\in\sig{Tp}_{\tp_{2}}^{\mathfrak{B}}$. Since $\mathfrak{B}$ satisfies $\Psi_{\text{\TGF{}}, (\mathcal{E}_{+}, \mathcal{E}_{!})}$ we find some 2-type $\tp^{\ast} \in \mathcal{E}_{+}$ satisfying $\bold{v}_{\tp^{\ast}} = v_{1}v_{2}$, $\tp^{\ast}|_{v_{1}} = \tp_{1}$, $\tp^{\ast}|_{v_{2}} = \tp_{2}$, and $(h(a_{1}), h(a_{2}))\in\sig{Tp}_{\tp^{\ast}}^{\mathfrak{B}}$. We now have to pick exactly one such $\tp^{\ast}$ and then add $(a_{1}, a_{2})$ to the extension of $\sig{Tp}_{\tp^{\ast}}^{\mathfrak{A}^{\ast}}$.
	Moreover, for every permutation $\kappa \colon \bold{v} \injsur \bold{v}$ we also add $\eta_{\bold{v}}^{\kappa\bold{v}_{\tp^{\ast}}}((a_{1}, a_{2}))$ to $\sig{Tp}_{\kappa\tp^{\ast}}^{\mathfrak{A}^{\ast}}$.
	This ensures that even after augmenting, $\Psi_{\mathrm{gen}, (\mathcal{E}_{+}, \mathcal{E}_{!})}$ remains satisfied. Furthermore the construction enforces that now $\Psi_{\text{\TGF{}}, (\mathcal{E}_{+}, \mathcal{E}_{!})}$ also holds for $\mathfrak{A}^{\ast}$.
	Lastly, as we only ever added tuples to the relations in order to obtain $\mathfrak{A}^{\ast}$ the identity map is clearly a homomorphism $\iota \colon\mathfrak{A}\to\mathfrak{A}^{\ast}$.
	
{We will collect the tuples $((a_1,a_2), \tp^{\ast})$ consisting of pairs $(a_1, a_2)$ not contained in any type predicates interpretation of $\mathfrak{A}$ and the type $\tp^{\ast}$ chosen in the construction, in the set $\mathcal{N}$. Note that this means the following: For all $((a_1, a_2), \tp^{\ast})$ we have $(a_1,a_2)\in\sig{Tp}_{\tp^{\ast}}^{\mathfrak{A}^{\ast}}$ but $(a_1,a_2)\notin\sig{Tp}_{\tp^{\ast}}^{\mathfrak{A}}$.}
	
{Now let $((a_1,a_2), \tp^{\ast})\in\mathcal{N}$.} To prepare for the next step, note the following: We ensured, that there is a homomorphism $f_{\tp^{\ast}}^{(a_1, a_2)}\colon\mathfrak{D}_{\tp^{\ast}}\to\mathfrak{A}^{\ast}|_{\sigma}$ via $f_{\tp^{\ast}}^{(a_1, a_2)}(v_1)=a_1$ and $f_{\tp^{\ast}}^{(a_1, a_2)}(v_2)=a_2$. Furthermore, $f_{\tp^{\ast}}^{(a_1, a_2)}$ is even injective. We will extend $f_{\tp^{\ast}}^{(a_1, a_2)}$ to the map $f_{\tp^{\ast},\mathfrak{C}}^{(a_1, a_2)}\colon(\mathfrak{D}_{\tp^{\ast}}\uplus (\mathfrak{C}\cdot\widecheck{\sigma}))\to \mathfrak{A}^{\ast}|_{\sigma\uplus C}$, where $f_{\tp^{\ast},\mathfrak{C}}^{(a_1, a_2)}(\sig{c}^{\mathfrak{C}})=\sig{c}^{\mathfrak{A}^{\ast}}$ which is a homomorphism, too.

	In the next step, we will make use of $f_{\tp^{\ast},\mathfrak{C}}^{(a_1, a_2)}$ by applying it to $(\mathfrak{D}_{\tp^{\ast}})^{\tau}$ and homomorphically transferring the $\tau$ relations. The so obtained structure will be denoted by $\mathfrak{A}'$. Hence for $\sig{R}\in\tau\setminus C$ we obtain
	\begin{equation}
		\sig{R}^{\mathfrak{A}'} = \sig{R}^{\mathfrak{A}^{\ast}}\cup \bigcup_{((a_1,a_2), \tp^{\ast})} \{ f_{\tp^{\ast},\mathfrak{C}}^{(a_1, a_2)}(\bold{d}) \ | \ \bold{d}\in\sig{R}^{(\mathfrak{D}_{\tp^{\ast}})^{\tau}} \},
	\end{equation}
	{where the union extends over all $((a_1,a_2), \tp^{\ast})$ such that a $f^{(a_1, a_2)}_{\tp^{\ast}, \mathfrak{C}}$ exists, i.e. we introduced in the previous step a binary type predicate for $(a_1, a_2)$.}
	We thus lift $f_{\tp^{\ast},\mathfrak{C}}^{(a_1, a_2)}$ to a homomorphism $(f')_{\tp^{\ast},\mathfrak{C}}^{(a_1, a_2)}\colon (\mathfrak{D}_{\tp^{\ast}})^{\tau}\to \mathfrak{A}'$. This is, additionally, an embedding.
	By assumption, $(a_1, a_2)$ did not realize any $2$-type $\tp'\in\mathcal{E}_{+}$. If that were the case, by assumption on $\mathfrak{A}$, $(a_1,a_2)$ would have satisfied $\sig{Tp}_{\tp'}(v_1, v_2)$ in $\mathfrak{A}$ and as $\mathfrak{A}$ satisfies $\Psi_{\mathrm{gen}, (\mathcal{E}_{+}, \mathcal{E}_{!})}$ it holds that $a_1\in\sig{Tp}_{\tp'|_{v_1}}^{\mathfrak{A}}$ and $a_2\in\sig{Tp}_{\tp'|_{v_2}}^{\mathfrak{A}}$. Hence there would have been no addition of $(a_1, a_2)$ to some $\sig{Tp}_{\tp'}^{\mathfrak{A}^{\ast}}$.
	And again by assumption on $\mathfrak{A}$ there cannot be any binary $\tau$ relation holding for $(a_1,a_2)$ as this would give rise to a guarded $2$-type realized by $(a_1,a_2)$, leading to the same argument.
	{By assumption we exclude the case where one of $a_1$ or $a_2$ is an interpretation of a constant symbol.}
	
	{We now show that realization of types and satisfaction of the associated type predicate coincides for tuples without constants.} Let $\bold{a}$ be such a tuple from $\mathfrak{A}'$ satisfying $\sig{Tp}_{\tp'}$ for some $\tp'\in\mathcal{E}$, i.e. $\bold{a}\in\sig{Tp}_{\tp'}^{\mathfrak{A}'}$. By construction, $\bold{a}\in\sig{Tp}_{\tp'}^{\mathfrak{A}}$ holds for all $\tp'$ being not of order $2$. 
	Then by assumption, $\bold{a}$ realizes $\tp'$ in $\mathfrak{A}$ and hence in $\mathfrak{A}'$. So suppose $\tp'$ being of order $2$ and such that $\bold{a}\not\in\sig{Tp}_{\tp'}^{\mathfrak{A}}$ (for order $2$ and $\bold{a}\in\sig{Tp}_{\tp'}^{\mathfrak{A}}$ a repeat of the previous argument yields the realization of the type). Hence $\bold{a}=(a_1, a_2)$ and we have the embedding $(f')_{\tp',\mathfrak{C}}^{(a_1, a_2)}\colon (\mathfrak{D}_{\tp'})^{\tau}\to \mathfrak{A}'$ by construction. More specifically, we know that $(f')_{\tp',\mathfrak{C}}^{(a_1, a_2)}(\bold{v}_{\tp'})=(a_1, a_2)$ and $\bold{v}_{\tp'}$ satisfies $\sig{Tp}_{\tp'}$ in $(\mathfrak{D}_{\tp'})^{\tau}$. By \Cref{lem:Dtauexpansion} $\bold{v}_{\tp'}$ realizes $\tp'$ in $(\mathfrak{D}_{\tp'})^{\tau}$ and by $(f')_{\tp',\mathfrak{C}}^{(a_1, a_2)}$ being an embedding we obtain that $(a_1, a_2)$ realizes $\tp'$ in $\mathfrak{A}'$.
	On the other hand, let $\bold{a}$ realize $\tp'\in\mathcal{E}$ in $\mathfrak{A}'$. If $\bold{a}$ already realizes $\tp'$ in $\mathfrak{A}$ we immediately obtain $\bold{a}\in\sig{Tp}_{\tp'}^{\mathfrak{A}}$ and hence $\bold{a}\in\sig{Tp}_{\tp'}^{\mathfrak{A}'}$. By construction, the only non trivial cases are those for $\mathrm{order}(\tp')=2$ and $\bold{a}$ not already realizing $\tp'$ in $\mathfrak{A}$. The latter is equivalent to saying, that $\bold{a}$ does not satisfy a binary $\tau$ relation in $\mathfrak{A}$. Hence there is the embedding $(f')_{\tp',\mathfrak{C}}^{(a_1, a_2)}\colon (\mathfrak{D}_{\tp'})^{\tau}\to \mathfrak{A}'$ by construction. Using this embedding, we obtain $(f')_{\tp',\mathfrak{C}}^{(a_1, a_2)}(\bold{v}_{\tp})=(a_1, a_2)$ and $\bold{v}_{\tp'}$ realizing $\tp'$ in $(\mathfrak{D}_{\tp'})^{\tau}$. Hence by \Cref{lem:Dtauexpansion} $\bold{v}_{\tp'}\in\sig{Tp}_{\tp'}^{(\mathfrak{D}_{\tp'})^{\tau}}$. Applying $(f')_{\tp',\mathfrak{C}}^{(a_1, a_2)}$ we obtain $(f')_{\tp',\mathfrak{C}}^{(a_1, a_2)}(\bold{v}_{\tp'})=(a_1, a_2)\in\sig{Tp}_{\tp'}^{\mathfrak{A}'}$.
	
	If $\mathfrak{D}$ is an induced substructure of $\mathfrak{A}$ already satisfying $\Psi_{\text{\TGF{}}, (\mathcal{E}_{+}, \mathcal{E}_{!})}$, then in the construction we have for every pair of elements from $D$ already some type predicate in whose interpretation they are. Consequently the construction does not touch any elements from $D$, hence $\mathfrak{D}$ is also an induced substructure of $\mathfrak{A}'$.
\end{proof}

Now we will employ the local ``problem fixing'' of \Cref{lem:addtype} to ensure $\Psi_{\mathrm{guard}, (\mathcal{E}_{+}, \mathcal{E}_{!})}$. This is split into two lemmata (\Cref{lem:inindstep} and \Cref{lem:GFOindstep}).

\begin{lemma}\label{lem:inindstep}
	Given a $\tau\uplus \sigma$-structure $\mathfrak{B}$ satisfying $\Psi_{\mathrm{hom}}$ as well as $\Psi_{\mathrm{gen}, (\mathcal{E}_{+}, \mathcal{E}_{!})}$ and $\Psi_{\mathrm{guard}, (\mathcal{E}_{+}, \mathcal{E}_{!})}$.
	Suppose we have a $\tau\uplus\sigma$-structure $\mathfrak{A}$ with a homomorphism $h\colon \mathfrak{A} \to \mathfrak{B}$ that already satisfies $\Psi_{\mathrm{gen}, (\mathcal{E}_{+}, \mathcal{E}_{!})}$. 
	Additionally, for all $\tp\in\mathcal{E}$ it holds that a tuple $\bold{a}$ {without constants} from $\mathfrak{A}$ satisfies $\sig{Tp}_{\tp}\in\sigma$ in $\mathfrak{A}$ if and only if $\bold{a}$ realizes $\tp$ in $\mathfrak{A}$.
	Let $q = (\bold{a},\mu,\tp, \nu)$ be a violation tuple of $\mathfrak{A}$. Then there is a finite $\tau\uplus\sigma$-structure $\mathfrak{A}'$ and a homomorphism $h' \colon \mathfrak{A}' \to \mathfrak{B}$ with the following properties:
	\begin{enumerate}
		\item \label{item inindstep:i}$\mathfrak{A}'$ satisfies conjuncts $\Psi_{\mathrm{gen}, (\mathcal{E}_{+}, \mathcal{E}_{!})}$,
		\item\label{item inindstep:ii} $\mathfrak{A}$ is an induced substructure of $\mathfrak{A}'$,
		\item\label{item inindstep:iii} $ h'$ extends $h$, i.e., $h'|_{A} = h$,  
		\item\label{item inindstep:iv} $q$ is no violation tuple of $\mathfrak{A}'$,
		\item\label{item inindstep:v} no new violation tuple is being added,
		\item\label{item inindstep:vi} for all $\tp\in\mathcal{E}$ it holds that a tuple $\bold{a}$ without constants from $\mathfrak{A}'$ satisfies $\sig{Tp}_{\tp}\in\sigma$ in $\mathfrak{A}'$ if and only if $\bold{a}$ realizes $\tp$ in $\mathfrak{A'}$.
	\end{enumerate}
	If $\mathfrak{B}$ and $\mathfrak{A}$ additionally satisfy $\Psi_{\text{\TGF{}}, (\mathcal{E}_{+}, \mathcal{E}_{!})}$ we can ensure that $\mathfrak{A}'$ also satisfies $\Psi_{\text{\TGF{}}, (\mathcal{E}_{+}, \mathcal{E}_{!})}$.
\end{lemma}
\begin{proof}
	Let $q=(\bold{a},\mu,\tp, \nu)$ be a violation tuple of $\mathfrak{A}$ which implies  $\bold{a} \in \sig{Tp}_\tp^{\mathfrak{A}}$.
	Since $h$ is a homomorphism,  the tuple $h(\bold{a})$ in $\mathfrak{B}$ also satisfies the premise of the conjunct from $\Psi_{\mathrm{guard}, (\mathcal{E}_{+}, \mathcal{E}_{!})}$ (see \ref{eq: Psi_guarded}) indexed by $(\tp, \nu)$.
	By assumption, $\mathfrak{B}$ satisfies $\Psi_{\mathrm{guard}, (\mathcal{E}_{+}, \mathcal{E}_{!})}$ which implies the existence of $ (\tp',\nu')$ and a tuple $\bold{b}$ from $\mathfrak{B}$ (and hence $\mathfrak{B}^{\ast}$) such that $(\tp,\nu) \Subset_j (\tp',\nu')$ and $\bold{b}\in \sig{Tp}_{\tp'}^\mathfrak{B}$ hold.
	
	Applying \Cref{lem:addtype} we obtain a structure $\mathfrak{A}^{\ast}$ and a homomorphism $h^{\ast}\colon \mathfrak{A}^{\ast}\to \mathfrak{B}$ such that $q$ is no violation tuple anymore.
	We also note two things: For one, we do not add any violation tuples via \Cref{lem:addtype}. This is clear when observing the fact that analyzing only locally the violation tuples of $\mathfrak{A}$ at $\bold{a}$, i.e. via $Q_{\bold{a}}$, is enough since the rest of the structure (everything not inside the induced substructure by $\bold{a}$ of $\mathfrak{A}$) remains completely unmoved by the construction. Secondly, \Cref{lem:addtype} ensures that $\mathfrak{A}$ is an induced substructure of $\mathfrak{A}^{\ast}$. Finally, we obtain that $\mathfrak{A}^{\ast}$ satisfies $\Psi_{\mathrm{gen}, (\mathcal{E}_{+}, \mathcal{E}_{!})}$. 
	This is equally achieved by \Cref{lem:addtype}.
	Hence Items \ref{item inindstep:i} to \ref{item inindstep:v} are established. {Lastly Item \ref{item inindstep:vi} is also guaranteed \Cref{lem:addtype}.} 
	
	If either $\mathfrak{A}$ or $\mathfrak{B}$ does not satisfy $\Psi_{\text{\TGF{}}, (\mathcal{E}_{+}, \mathcal{E}_{!})}$, we rename $\mathfrak{A}^{\ast}$ to $\mathfrak{A}'$ and $h^{\ast}$ to $h'$ and are done.
	If on the other hand, additionally $\Psi_{\text{\TGF{}}, (\mathcal{E}_{+}, \mathcal{E}_{!})}$ holds for $\mathfrak{B}$ and $\mathfrak{A}$, we apply \Cref{lem:psiTGFclosure} to $\mathfrak{A}^{\ast}$ and $\mathfrak{B}$ and the homomorphism $h^\ast$ and get an structure $\mathfrak{A'}$. 
	Finally, as $\mathfrak{A}$ does satisfy $\Psi_{\text{\TGF{}}, (\mathcal{E}_{+}, \mathcal{E}_{!})}$, again \Cref{lem:psiTGFclosure} yields that $\mathfrak{A}^{\ast}$ is an induced substructure of $\mathfrak{A}'$. And by transitivity of being induced substructures we preserve Item \ref{item inindstep:ii} even after this step. {Preservation of the remaining items is ensured by \Cref{lem:psiTGFclosure}.} 
\end{proof}

\begin{lemma}\label{lem:GFOindstep}
	Given a $\tau\uplus \sigma$-structure $\mathfrak{B}$ satisfying $\Psi_{\mathrm{hom}}$ as well as $\Psi_{\mathrm{gen}, (\mathcal{E}_{+}, \mathcal{E}_{!})}$ and $\Psi_{\mathrm{guard}, (\mathcal{E}_{+}, \mathcal{E}_{!})}$.
	Suppose we have a finite $\tau\uplus\sigma$-structure $\mathfrak{A}$ with a homomorphism $h\colon \mathfrak{A} \to \mathfrak{B}$ that already satisfies $\Psi_{\mathrm{gen}, (\mathcal{E}_{+}, \mathcal{E}_{!})}$. 
	Additionally, for every tuple $\bold{a}$ {without constants} from $\mathfrak{A}$ it holds for all $\tp\in\mathcal{E}$ that $\bold{a}$ satisfies $\sig{Tp}_{\tp}\in\sigma$ in $\mathfrak{A}$ if and only if $\bold{a}$ realizes $\tp$ in $\mathfrak{A}$. Then there is a finite $\tau\uplus\sigma$-structure $\mathfrak{A}'$ and a homomorphism $h' \colon \mathfrak{A}' \to \mathfrak{B}$ with the following properties:
	\begin{enumerate}
		\item\label{item GFOindstep:1} $\mathfrak{A}'$ satisfies all $\Psi_{\mathrm{gen}, (\mathcal{E}_{+}, \mathcal{E}_{!})}$, 
		\item\label{item GFOindstep:2} $\mathfrak{A}$ is an induced substructure of $\mathfrak{A}'$,
		\item\label{item GFOindstep:3} $h'$ extends $h$, i.e., $h'|_{A} = h$,  
		\item\label{item GFOindstep:4} every conjunct  in (\ref{eq: Psi_guarded}) is satisfied in $\mathfrak{A}'$ if the universal quantification is restricted to elements from $\mathfrak{A}$,
		\item\label{item GFOindstep:5} for all $\tp\in\mathcal{E}$ it holds that a tuple $\bold{a}$ without constants from $\mathfrak{A}'$ satisfies $\sig{Tp}_{\tp}\in\sigma$ in $\mathfrak{A}'$ if and only if $\bold{a}$ realizes $\tp$ in $\mathfrak{A'}$.
	\end{enumerate}
	If $\mathfrak{A}$ additionally satisfies $\Psi_{\text{\TGF{}}, (\mathcal{E}_{+}, \mathcal{E}_{!})}$ we can ensure that $\mathfrak{A}'$ also satisfies $\Psi_{\text{\TGF{}}, (\mathcal{E}_{+}, \mathcal{E}_{!})}$.
\end{lemma}
\begin{proof}
	Let $\mathfrak{B}^{\ast}$ denote $\mathfrak{B}|_{\sigma\uplus C}$ and $Q$ be the set of violation tuples of $\mathfrak{A}$.
	Since $Q$ is clearly finite (as we have finitely many types to consider and $\mathfrak{A}$ is a finite structure) we can assume an enumeration $\{q_1,\ldots,q_m\}$ of its elements. 
	We define inductively homomorphisms $h^j\colon \mathfrak{A}^j\rightarrow \mathfrak{B}^{\ast}$ and
	structures $\mathfrak{A}^j$  for  $j\in \{0,\ldots,m\}$ 
	such that for every $j\in \{1,\ldots,m\}$ the following holds
	\begin{enumerate}[label={\roman*})]
		\item\label{item GFOindstep:i} $\mathfrak{A}$ is a substructure of $\mathfrak{A}^{j}$,
		\item\label{item GFOindstep:ii} $A^j$ is finite,
		\item \label{item GFOindstep:iii}$\mathfrak{A}^j$ satisfies conjuncts $\Psi_{\mathrm{gen}, (\mathcal{E}_{+}, \mathcal{E}_{!})}$,
		\item\label{item GFOindstep:iv} $\mathfrak{A}^{j-1}$ is an induced substructure of $\mathfrak{A}^{j}$,
		\item\label{item GFOindstep:v} $h^{j}$ extends $h^{j-1}$, i.e. $h^{j}|_{A^{j-1}}=h^{j-1}$,
		\item\label{item GFOindstep:vi} no element from $\{q_1,\ldots,q_{j-1}\}$ is a violation tuple of $\mathfrak{A}^j$,
		\item\label{item GFOindstep:vii} at most the elements from $\{ q_{j}, \ldots, q_{m} \}$ are violation tuples of $\mathfrak{A}^{j}$, i.e. we do not add new violation tuples in each step),
		\item\label{item GFOindstep:viii} {for all $\tp\in\mathcal{E}$ it holds that a tuple $\bold{a}$ without constants from $\mathfrak{A}^j$ satisfies $\sig{Tp}_{\tp}\in\sigma$ in $\mathfrak{A}^j$ if and only if $\bold{a}$ realizes $\tp$ in $\mathfrak{A}^j$.}
	\end{enumerate}

	We put $\mathfrak{A}^=\mathfrak{A}$ and 
	$ h^0=h$. Suppose $\mathfrak{A}^k$ and $h^k$ are defined for all $k\in\{1,\ldots,j\}$ and satisfy the requirements. 
	
	Now consider $q_{j}=(\bold{a},\mu,\tp, \nu)$ which implies  $\bold{a} \in \sig{Tp}_\tp^{\mathfrak{A}}$. Since $\mathfrak{A}$ is a substructure of $\mathfrak{A}^{j-1}$ we find the tuple $\bold{a}$ in $\mathfrak{A}^{j-1}$.
	Hence $q_{j}$ is also a violation tuple of $\mathfrak{A}^{j-1}$. Applying \Cref{lem:inindstep} we obtain $\mathfrak{A}^{j}$ and $h^{j}$ with Items \ref{item GFOindstep:ii}-\ref{item GFOindstep:v} holding. Obviously Item \ref{item GFOindstep:i} holds by transitivity of the substructure relation. Furthermore $q_{j}$ is not a violation tuple of $\mathfrak{A}^{j}$, hence Item \ref{item GFOindstep:vi} is true, too. Item \ref{item inindstep:v} of \Cref{lem:inindstep} yields \ref{item GFOindstep:vii} {and Item \ref{item inindstep:vi} of \Cref{lem:inindstep} yields \ref{item GFOindstep:viii}}.
	
	This process terminates after $m$ steps and we obtain a structure $\mathfrak{A}^{m}$ from it, which we will denote now by $\mathfrak{A}'$ and a homomorphism $h^{m}\colon \mathfrak{A}^{m}\to \mathfrak{B}^{\ast}$ renamed to $h'$. As every step ensured that Items \ref{item GFOindstep:i}-\ref{item GFOindstep:vii} hold, $\mathfrak{A}'$ is a finite structure,
	\begin{itemize}
		\item that satisfies $\Psi_{\mathrm{gen}, (\mathcal{E}_{+}, \mathcal{E}_{!})}$,
		\item has $\mathfrak{A}$ as induced substructure (by transitivity of this relation), and
		\item $h'$ extends $h$, i.e. $h'|_{A}=h$.
	\end{itemize}
	Hence $\mathfrak{A}'$ with $h'$ satisfy Items \ref{item GFOindstep:1}-\ref{item GFOindstep:3}. Since no violation tuple of $\mathfrak{A}$ is one of $\mathfrak{A}'$ and there never have been created new ones along the way, Item \ref{item GFOindstep:4} is true as well. {Item \ref{item GFOindstep:5} holds thanks to Item \ref{item GFOindstep:viii}.}
	
	For the second part of the lemma note that if $\Psi_{\text{\TGF{}}, (\mathcal{E}_{+}, \mathcal{E}_{!})}$ holds for $\mathfrak{B}$ and $\mathfrak{A}$ we can additionally apply in each step the second part of \Cref{lem:inindstep}. This then ensures that $\Psi_{\text{\TGF{}}, (\mathcal{E}_{+}, \mathcal{E}_{!})}$ holds for every $A^{j}$ with $j\in \{ 1, \ldots, m \}$ and hence also for $\mathfrak{A}'$. 
\end{proof}

We will now put those lemmata together in order to do the following: For a $\tau\uplus \sigma$-model $\mathfrak{B}$ satisfying $\Psi_{\mathrm{hom}}$, all of the $\Psi_{\mathrm{gen}, (\mathcal{E}_{+}, \mathcal{E}_{!})}$ and $\Psi_{\mathrm{guard}, (\mathcal{E}_{+}, \mathcal{E}_{!})}$ (and in the case of $\Phi$ being a \TGF{} sentence all the $\Psi_{\text{\TGF{}}, (\mathcal{E}_{+}, \mathcal{E}_{!})}$) we will construct a $\tau\uplus\sigma$-model $\mathfrak{A}$ of $\Phi$ and a homomorphism $h\colon\mathfrak{A}\to\mathfrak{B}$ .
Note that we do not fix a $(\mathcal{E}_{+}, \mathcal{E}_{!})\in\getsummary{\Phi}$ this time.
\begin{lemma}\label{lem:ConstrA}
	Let $\Phi$ be a \GFO{} sentence $\Phi = \Phi_{\sig{D}}\land \Phi_{\forall} \land \Phi_{\forall\exists}$ with $\sig{D} \in \tau$ unary or a \TGF{} sentence $\Phi = \Phi_{\sig{Univ}}\land \Phi_{\forall} \land \Phi_{\forall\exists}$ with $\sig{Univ}\in\tau$ binary, where
	\begin{align}
		\Phi_{\sig{D}} &= \forall x. \sig{D}(x) ,\label{eq: propGFOTGF phi_D} \\[2ex]
		\Phi_{\sig{Univ}} & = \forall xy. \sig{Univ}(x,y),	\label{eq: propGFOTGF Univ}\\[2ex]
		\Phi_{\forall} &= \bigwedge_{i} \forall \bold{x}.\Big(\alpha_{i}(\bold{x})\Rightarrow\vartheta_{i}[\bold{x}]\Big) \label{eq: propGFOTGF phi_A}\\
		\Phi_{\forall\exists} &= \bigwedge_{j}\forall \bold{z}.\Big(\beta_{j}(\bold{z})\Rightarrow\exists\bold{y}.\gamma_{j}(\bold{yz})\Big) \label{eq: propGFOTGF phi_AE}
	\end{align}
	Define the $\tau \uplus \sigma$-sentences
	\begin{equation} 
		\Psi = \Psi_\mathrm{hom} \wedge \bigvee_{(\mathcal{E}_{+}, \mathcal{E}_{!}) \in \getsummary{\Phi}} \Psi_{(\mathcal{E}_{+}, \mathcal{E}_{!})}
	\end{equation}  
	and
	\begin{equation} 
		\Psi^{\text{\TGF{}}} = \Psi_\mathrm{hom} \wedge \bigvee_{(\mathcal{E}_{+}, \mathcal{E}_{!}) \in \getsummary{\Phi}} \Psi^{\text{\TGF{}}}_{(\mathcal{E}_{+}, \mathcal{E}_{!})}
	\end{equation}  
	where $\Psi_{(\mathcal{E}_{+}, \mathcal{E}_{!})}$ is the following sentences:
	\begin{equation} 
		\Psi_{(\mathcal{E}_{+}, \mathcal{E}_{!})} = \Psi_{\mathrm{gen}, (\mathcal{E}_{+}, \mathcal{E}_{!})} \land \Psi_{\mathrm{guard}, (\mathcal{E}_{+}, \mathcal{E}_{!})}
	\end{equation}  
	and $\Psi^{\text{\TGF{}}}_{(\mathcal{E}_{+}, \mathcal{E}_{!})}$ is the following sentences:
	\begin{equation} 
		\Psi^{\text{\TGF{}}}_{(\mathcal{E}_{+}, \mathcal{E}_{!})} = \Psi_{\mathrm{gen}, (\mathcal{E}_{+}, \mathcal{E}_{!})} \land \Psi_{\mathrm{guard}, (\mathcal{E}_{+}, \mathcal{E}_{!})} \land \Psi_{\text{\TGF{}}, (\mathcal{E}_{+}, \mathcal{E}_{!})}
	\end{equation}  
	Depending on $\Phi$ being a sentence from \GFO{} or \TGF{} let $\mathfrak{B}$ be a model of $\Psi$ or $\Psi^{\text{\TGF{}}}$. Then there exists a $\tau\uplus\sigma$-structure $\mathfrak{A}$ satisfying $\Phi$ and a homomorphism $h\colon\mathfrak{A}\to\mathfrak{B}$.
\end{lemma}
\begin{proof}
	Let $\mathfrak{B}$ be a model of $\Psi$ or a model of $\Psi^{\text{\TGF{}}}$ (depending on whether $\Phi$ is from \GFO{} or \TGF{} respectively) and let $(\mathcal{E}_{+}, \mathcal{E}_{!}) \in \getsummary{\Phi}$ such that the associated disjunct $\Psi_{(\mathcal{E}_{+}, \mathcal{E}_{!})}$ (or $\Psi^{\text{\TGF{}}}_{(\mathcal{E}_{+}, \mathcal{E}_{!})}$) is satisfied in $\mathfrak{B}$. 
	
	Applying \Cref{lem:A0}, let $\mathfrak{A}_0$ be $(\mathfrak{K}_{\mathcal{E}_{+}})^{\tau}$, hence for every type $\tp\in\mathcal{E}$ it holds that a tuple $\bold{a}$ (without elements named by constants) from $\mathfrak{A}_0$ satisfies $\sig{Tp}_{\tp}$ in $\mathfrak{A}_0$ if and only if $\bold{a}$ realizes $\tp$ in $\mathfrak{A}_0$.
	Let $h_0$ be the in \Cref{lem:A0} constructed homomorphism $h'\colon (\mathfrak{K}_{\mathcal{E}_{+}})^{\tau} \to \mathfrak{B}$.
	If $\Phi$ is a \TGF{} sentence, \Cref{lem:psiTGFclosure} allows us to assume $\mathfrak{A}_{0}$ to satisfy $\Psi_{\text{\TGF{}}, (\mathcal{E}_{+}, \mathcal{E}_{!})}$.
	
	This will constitute the induction beginning for a proof of the following statement:
	
	\underline{Claim:} There exist a sequence $(\mathfrak{A}_n)_{n\in \mathbb{N} }$ of $\tau\uplus\sigma$-structures  and a sequence $(h_n)_{n \in \mathbb{N}}$ of homomorphisms $h_n\colon \mathfrak{A}_n\rightarrow \mathfrak{B}$
	such that for every $i\in \mathbb{N}$, the following holds:
	
	\begin{enumerate}
		\item\label{item GFOTGF:1} $A_i$ is finite,
		\item\label{item GFOTGF:2} $\mathfrak{A}_i$ satisfies all $\Psi_{\mathrm{gen}, (\mathcal{E}_{+}, \mathcal{E}_{!})}$ (and if $\Phi$ is from \TGF{} also $\Psi_{\text{\TGF{}}, (\mathcal{E}_{+}, \mathcal{E}_{!})}$), 
		\item\label{item GFOTGF:3} $\mathfrak{A}_{i}$ is an induced substructure of $\mathfrak{A}_{i+1}$,
		\item\label{item GFOTGF:4} $ h_{i+1}$ extends $h_i$, i.e., $h_{i+1}|_{A_i} = h_{i}$,  
		\item\label{item GFOTGF:5} every conjunct  $\Psi_{\mathrm{guard}, (\mathcal{E}_{+}, \mathcal{E}_{!})}$ is satisfied in $\mathfrak{A}_{i+1}$ if the universal quantification is restricted to elements from $\mathfrak{A}_i$,
		\item\label{item GFOTGF:6} for all $\tp\in\mathcal{E}$ it holds that a tuple $\bold{a}$ without constants from $\mathfrak{A}_i$ satisfies $\sig{Tp}_{\tp}\in\sigma$ in $\mathfrak{A}_i$ if and only if $\bold{a}$ realizes $\tp$ in $\mathfrak{A}_i$.
	\end{enumerate}

	As mentioned before, we prove the statement by induction on $n$ and assume therefore that the claim holds for  structures $(\mathfrak{A}_n)_{n\leq i}$.
	In each induction step we apply \Cref{lem:GFOindstep} to obtain $\mathfrak{A}^{i+1}$ from $\mathfrak{A}^{i}$.
	This immediately yields the claim.
	
	We continue by using the claim for the proof of $\homcl{\getmodels{\Phi}|_{\tau}} \supseteq \getmodels{\Psi}|_{\tau}$ by defining a limit structure for the sequence $(\mathfrak{A}_n)_{n\in \mathbb{N} }$.
	By Item \ref{item GFOTGF:3}, we know that $A_i\subseteq A_{i+1}$ as well as  $\sig{R}^{\mathfrak{A}_i}\subseteq \sig{R}^{\mathfrak{A}_{i+1}}$ for every $\sig{R}\in \sigma$ holds.
	We define the $\tau\uplus\sigma$-structure $\mathfrak{A}$ as follows: The domain is given by  $A=\bigcup_{n\in \mathbb{N}}A_n$ and for $\sig{R}\in \tau\setminus C$ we have $\sig{R}^\mathfrak{A}=\bigcup_{n\in \mathbb{N}}\sig{R}^{\mathfrak{A}_n}$. By the choice of $\mathfrak{A}_0$,  Item \ref{item GFOTGF:2} and Item \ref{item GFOTGF:5} imply that the structure $ \mathfrak{A}$ satisfies $\Psi_{(\mathcal{E}_{+}, \mathcal{E}_{!})}$. 
	If $\Phi$ is a \TGF{} sentence, we also obtain $\mathfrak{A}\vDash\Psi_{\text{\TGF{}}, (\mathcal{E}_{+}, \mathcal{E}_{!})}$.
	Additionally, Item \ref{item GFOTGF:5} guarantees that $\mathfrak{A}$ satisfies $\Psi_{\mathrm{guard}, (\mathcal{E}_{+}, \mathcal{E}_{!})}$. If this were not the case we would find some violation tuple $q=(\bold{a}, \mu, \tp, \nu)$ of $\mathfrak{A}$. As $\mathfrak{A}$ is the union of an increasing sequence of structures, there is an index $k$ such that $q$ is a violation tuple of $\mathfrak{A}_{k}$. This also means that $\bold{a}$ is a tuple from $\mathfrak{A}_{k}$. By construction (applying \Cref{lem:addtype}) $q$ is no violation tuple in $\mathfrak{A}_{k+1}$. Since in every application of \Cref{lem:addtype} we never add a fresh violation tuple  this implies, that $q$ could not have been one of $\mathfrak{A}$. As violation tuples are witnessing violations of $\Psi_{\mathrm{guard}, (\mathcal{E}_{+}, \mathcal{E}_{!})}$ and we just argued that there cannot be any violation tuples in $\mathfrak{A}$, this yields $\mathfrak{A}\vDash\Psi_{\mathrm{guard}, (\mathcal{E}_{+}, \mathcal{E}_{!})}$.
	Finally, Item \ref{item GFOTGF:6} guarantees that type predicates correspond to the associated types. To be more precise: {For every type $\tp\in\mathcal{E}$ it holds that a tuple $\bold{a}$ (without constants) from $\mathfrak{A}$ satisfies $\sig{Tp}_{\tp}$ in $\mathfrak{A}$ if and only if $\bold{a}$ realizes $\tp$ in $\mathfrak{A}$.}
	Note that, since $\mathfrak{A}\vDash\Psi_{\mathrm{gen}, (\mathcal{E}_{+}, \mathcal{E}_{!})}$ the only interesting case is $\tp\in\mathcal{E}_{+}$. This allows us to switch freely between $\tp\in\mathcal{E}$ and $\sig{Tp}_{\tp}\in\sigma$ for any further discussions.

	Hence we can establish the validity of the $\tau$-sentence $\Phi$  based on the type predicates in $\mathfrak{A}$ alone. To be precise, consider at first (\ref{eq: propGFOTGF phi_A}) and suppose $\bold{a}$ from $\mathfrak{A}$ satisfies the premise $\alpha_i(\bold{x})$ of some conjunct from $\Phi_{\forall}$. Then the $\tau$ type $\tp_{\bold{a}}$ realized by $\bold{a}$ in $\mathfrak{A}$ is guarded by $\alpha_i(\bold{x})$ (as witnessed by some $\nu$). Consequently $\tp_{\bold{a}}\in\mathcal{E}$ and hence $\bold{a}\in\sig{Tp}_{\tp_{\bold{a}}}^{\mathfrak{A}}$ by construction. With $\mathfrak{A}\vDash\Psi_{\mathrm{gen}, (\mathcal{E}_{+}, \mathcal{E}_{!})}$ we obtain $\tp_{\bold{a}}\in\mathcal{E}_{+}$. As $(\mathcal{E}_{+}, \mathcal{E}_{!})\in\getsummary{\Phi}$, $\tp_{\bold{a}}$ contains some disjunct from $\vartheta_i[\bold{v}_{\tp_{\bold{a}}}]$. Hence $\bold{a}$ satisfies some disjunct from $\vartheta_i[\bold{x}]$ in $\mathfrak{A}$.
	
	Now consider (\ref{eq: propGFOTGF phi_AE}) and assume that a premise $\beta_j(\bold{z})$ of a conjunct from $\Phi_{\forall\exists}$ is satisfied by a tuple $\bold{c}$ from $\mathfrak{A}$. Hence the $\tau$ type $\tp_{\bold{c}}$ of $\bold{c}$ in $\mathfrak{A}$ is guarded by $\beta_j(\bold{z})$ (as witnessed by some $\nu$), so $\tp_{\bold{c}}\in\mathcal{E}$. By construction this means that $\bold{c}$ satisfies $\sig{Tp}_{\tp_{\bold{c}}}$ in $\mathfrak{A}$. By $\Psi_{\mathrm{gen}, (\mathcal{E}_{+}, \mathcal{E}_{!})}$ we obtain $\tp_{\bold{c}}\in\mathcal{E}_{+}$ which in turn implies that $(\tp_{\bold{c}}, \nu)\in B_{j}$. Since $\mathfrak{A}$ satisfies $\Psi$ and especially $\Psi_{\mathrm{guard}, (\mathcal{E}_{+}, \mathcal{E}_{!})}$, it follows that there exist $\bold{b}$ in $\mathfrak{A}$ and a type $\tp'\in\mathcal{E}_{+}$ guarded by $\gamma_{j}(\bold{y}\bold{z})$ as witnessed by $\nu'$, i.e. $(\tp', \nu') \in P_j$ such that $(\tp_{\bold{c}}, \nu)$ and $(\tp', \nu')$ are fitting and $\sig{Tp}_{\tp'}(\bold{b}\bold{c})$ holds. And by construction we obtain $\bold{b}\bold{c}$ realizing $\tp'$ hence satisfying $\gamma_{j}(\bold{y}\bold{z})$ in $\mathfrak{A}$.
	
	For the case $\Phi$ being a \GFO{} sentence, $\Phi_{\sig{D}}$ holds since, for every $a$ in $\mathfrak{A}$ its type $\tp_{a}(v_1)$ is, as a $1$-type, trivially guarded (hence in $\mathcal{E}$). By construction $a\in\sig{Tp}_{\tp_{a}}^{\mathfrak{A}}$, hence by $\Psi_{\mathrm{gen}, (\mathcal{E}_{+}, \mathcal{E}_{!})}$ we obtain $\tp_{a}\in\mathcal{E}_{+}$. As $\sig{D}(v_1)\in\tp_{a}$, $a\in\sig{D}^{\mathfrak{A}}$.
	
	Lastly, if $\Phi$ is a \TGF{} sentence, $\Phi_{\sig{Univ}}$ holds since, for every $(a_1,a_2)$ in $\mathfrak{A}$ we obtain via $\mathfrak{A}\vDash\Psi_{\text{\TGF{}}, (\mathcal{E}_{+}, \mathcal{E}_{!})}$ its type $\tp_{a}(v_1)$ is, as a $1$-type, trivially guarded (hence in $\mathcal{E}$). By construction $a\in\sig{Tp}_{\tp_{a}}^{\mathfrak{A}}$, hence by $\Psi_{\mathrm{gen}, (\mathcal{E}_{+}, \mathcal{E}_{!})}$ we obtain $\tp_{a}\in\mathcal{E}_{+}$. As $\sig{D}(v_1)\in\tp_{a}$, $a\in\sig{D}^{\mathfrak{A}}$.
	
	This implies that $\mathfrak{A}$ indeed satisfies $\Phi$, as desired.
	
	By Item \ref{item GFOTGF:4}, the union $h=\bigcup_{n\in \mathbb{N}} h_n$ is a well-defined map. 
	Moreover, it is a homomorphism from $ \mathfrak{A}$ to $\mathfrak{B}$. 
	This proves the statement.
\end{proof}

We are now ready to prove several results characterizing the homclosures.

\begin{proposition}[restate=GFOinESO, name=]\label{prop:charGFO}
For every \GFO{} $\tau$-sentence $\Phi$, 
there exists
a \GFO{} 
sen\-tence $\Psi$ 
such that $\homcl{\getmodels{\Phi}} = \getmodels{\Psi}|_\tau$. 
\end{proposition}



\begin{proof}
	Let $\Phi$ be a \GFO{} $\tau$-sentence, which, by \Cref{lem:normalformOK}, can without loss of generality be assumed to be in the following normal form (as discussed above in \Cref{lem:GFONF}): $\Phi = \Phi_{\sig{D}}\land \Phi_{\forall} \land \Phi_{\forall\exists}$ with $\sig{D} \in \tau$ unary, where
	\begin{align}
		\Phi_{\sig{D}} &= \forall x. \sig{D}(x) ,\label{eq: prop48 phi_D} \\[2ex]
		\Phi_{\forall} &= \bigwedge_{i} \forall \bold{x}.\Big(\alpha_{i}(\bold{x})\Rightarrow\vartheta_{i}[\bold{x}]\Big) \label{eq: prop48 phi_A}\\
		\Phi_{\forall\exists} &= \bigwedge_{j}\forall \bold{z}.\Big(\beta_{j}(\bold{z})\Rightarrow\exists\bold{y}.\gamma_{j}(\bold{yz})\Big) \label{eq: prop48 phi_AE}
	\end{align}
	with $\vartheta_{i}$ a disjunction of literals, $\alpha_{i}, \beta_{j}$ guard atoms and $\gamma_j$ an atom. 
	We define the \emph{width} of $\Phi$ (denoted by $\mathrm{width}(\Phi)$) to be the maximal arity of a relation symbol appearing in $\Phi$. Furthermore, we denote by $C$ the set of all constant symbols from $\tau$.
	
	Now, let the set $\mathcal{E}$ of eligible types consist of all rigid guarded types of order $\le\mathrm{width}(\Phi)$.


    Let $\sigma$ contain for every $\tp\in \mathcal{E}$ a fresh relational symbol $\sig{Tp}_\tp$ of appropriate arity, and define the $\tau \uplus \sigma$-sentence
	\begin{equation} 
		\Psi = \Psi_\mathrm{hom} \wedge \bigvee_{(\mathcal{E}_{+}, \mathcal{E}_{!}) \in \getsummary{\Phi}} \Psi_{(\mathcal{E}_{+}, \mathcal{E}_{!})},
	\end{equation}  
	where $\Psi_{(\mathcal{E}_{+}, \mathcal{E}_{!})}$ is the following sentences:
	
	\begin{equation} 
		\Psi_{(\mathcal{E}_{+}, \mathcal{E}_{!})} = \Psi_{\mathrm{gen}, (\mathcal{E}_{+}, \mathcal{E}_{!})} \land \Psi_{\mathrm{guard}, (\mathcal{E}_{+}, \mathcal{E}_{!})}.
	\end{equation}  
    
    We obtain $\homcl{\getmodels{\Phi}} = \getmodels{\Psi}|_{\tau}$ as follows:
    
	``$\subseteq$'': Let $\mathfrak{A} \in \getmodels{\Phi}$ and let $h\colon\mathfrak{A} \to \mathfrak{B}$. Take $\mathcal{E}$ to be the set of all eligible types for $\Phi$ and $\sigma$ the set of fresh relation symbols $\sig{Tp}_{\tp}$ associated to each $\tp \in \mathcal{E}$ with the appropriate arity. Then the $\mathcal{E}$-adornment of $\mathfrak{A}$, $\mathfrak{A}\cdot\mathcal{E}$, is a $\tau\uplus\sigma$-structure which we will denote by $\mathfrak{A}^{\ast}$. Furthermore, let $\mathfrak{B}\cdot h(\mathfrak{A}^{\ast} |_{\sigma})$ be the $\tau \uplus\sigma$-structure described before, which we denote by $\mathfrak{B}^{\ast}$. Obviously $\mathfrak{B}^{\ast}|_{\tau} = \mathfrak{B}$ holds.
    
	By \Cref{lem:PsiGen} $\mathfrak{B}^{\ast}$ satisfies $\Psi_{\mathrm{gen}, (\mathcal{E}_{+}, \mathcal{E}_{!})}$. Additionally, after applying \Cref{lem:PsiGuardTGF} $\mathfrak{B}^{\ast}$ satisfies $\Psi_{\mathrm{guard}, (\mathcal{E}_{+}, \mathcal{E}_{!})}$ and $\Psi_{\mathrm{hom}}$. Consequently $\mathfrak{B}^{\ast}\vDash \Psi$.
	  
  	``$\supseteq$'': Let $\mathfrak{B}$ be a model of $\Psi$. Application of \Cref{lem:ConstrA} yields a $\tau\uplus\sigma$-structure $\mathfrak{A}$ satisfying $\Phi$ and a homomorphism $h\colon\mathfrak{A}\to\mathfrak{B}$. Restricting both $\mathfrak{A}$ and $\mathfrak{B}$ to $\tau$ yields the claim.
\end{proof}

\begin{lemma}[restate=GNFOredGFO, label=GNFOredGFO, name=]
For every \GNFO{} $\tau$-sentence $\Phi$, there exists some 
\GFO{} 
sentence $\Phi^*$ with
$\homcl{\getmodels{\Phi}} = \homcl{\getmodels{\Phi^*}|_\tau}$.
\end{lemma}

\begin{proof}By \cite[Proposition 3.2]{GuardedNegation}, there exist a \GFO{} sentence $\Phi'$ and an \EpFO{} sentence (more precisely: a union of conjunctive queries) $\Phi''$ such that $\getmodels{\Phi} = \getmodels{\Phi' \wedge \neg \Phi''}|_{\tau}$ ($\dag$).
	
	
	Then, by \cite[Lemma 2.4]{BaranyGO13} and the corresponding proof, there exists a \GFO{} sentence $\Phi'''$ such that 
	(i) $\getmodels{\Phi' \wedge \neg \Phi''} \subseteq \getmodels{\Phi' \wedge \neg \Phi'''}$ and (ii) for every model $\mathfrak{A}$ of $\Phi' \wedge \neg \Phi'''$ there is
	a model  $\mathfrak{B}$ of $\Phi' \wedge \neg \Phi''$ such that a homomorphism $\mathfrak{B} \to \mathfrak{A}$ exists.
	Therefore, $\homcl{\getmodels{\Phi' \wedge \neg \Phi''}} = \homcl{\getmodels{\Phi' \wedge \neg \Phi'''}}$ ($\ddag$). 

So we obtain: 
\begin{eqnarray}
	\homcl{\getmodels{\Phi}} & \stackrel{\dag}{=} & \homcl{\getmodels{\Phi' \wedge \neg \Phi''}|_{\tau}} \\
	& \stackrel{*}{=} & \homcl{\getmodels{\Phi' \wedge \neg \Phi''}}|_{\tau} \\
	& \stackrel{\ddag}{=} & \homcl{\getmodels{\Phi' \wedge \neg \Phi'''}}|_{\tau} \\
	& \stackrel{*}{=} & \homcl{\getmodels{\Phi' \wedge \neg \Phi'''}|_{\tau}}.
\end{eqnarray} 
The equivalent transformation used in the steps denoted by $\stackrel{*}{=}$ has been justified in the proof of \Cref{lem:normalformOK}.	
Note that $\Phi' \wedge \neg \Phi''' \in \GFO$ and let $\Phi^* = \Phi' \wedge \neg \Phi'''$.
\end{proof}

\begin{proposition}\label{prop:charGNFO}
For every \GNFO{} $\tau$-sentence $\Phi$, 
there exists
a \GFO{} 
sen\-tence $\Psi$ 
such that $\homcl{\getmodels{\Phi}} = \getmodels{\Psi}|_\tau$. 
\end{proposition}

\begin{proof}
Given $\Phi$, we apply \Cref{GNFOredGFO} to obtain a \GFO{} $\tau'$-sentence $\Phi^*$ for some $\tau' \supseteq \tau$ satisfying $\homcl{\getmodels{\Phi}} = \homcl{\getmodels{\Phi^*}|_\tau}$.
Applying \Cref{prop:charGFO} to $\Phi^*$, we obtain a \GFO{} $\tau''$-sentence $\Psi$ for some $\tau'' \supseteq \tau'$ satisfying $\homcl{\getmodels{\Phi^*}} = \getmodels{\Psi}|_{\tau'}$. Yet, since $\tau \subseteq \tau'$, we also obtain $$\getmodels{\Psi}|_\tau = \getmodels{\Psi}|_{\tau'}|_\tau = \homcl{\getmodels{\Phi^*}}|_\tau \stackrel{*}{=} \homcl{\getmodels{\Phi^*}|_\tau} = \homcl{\getmodels{\Phi}}$$ (the equality denoted by $\stackrel{*}{=}$ has been shown in the proof of \Cref{lem:normalformOK}), therefore $\Psi$ serves the desired purpose.
\end{proof}

\begin{proposition}[restate=TGFinESO, name=]\label{prop:charTGF}
For every \TGF{} $\tau$-sentence $\Phi$, 
there exists
a \TGF{} 
sen\-tence $\Psi$ 
such that $\homcl{\getmodels{\Phi}} = \getmodels{\Psi}|_\tau$. 
\end{proposition}

\begin{proof}
	Let $\Phi$ be a \TGF{} $\tau$-sentence which we assume to be in normal form (see \Cref{lem:TGFNF}), i.e. with $\sig{Univ}\in\tau$
	\begin{equation}	
		\Phi = \Phi_{\sig{Univ}}\land \Phi_{\forall} \land \Phi_{\forall\exists} 
	\end{equation}	
	where 
	\begin{eqnarray}
		\Phi_{\sig{Univ}} & = & \forall xy. \sig{Univ}(x,y),	\label{eq: prop52 Univ}\\[2ex]
		\Phi_{\forall} & = & \bigwedge_{i}\forall\bold{x}. (\alpha_{i}(\bold{x}) \Rightarrow \vartheta_{i}[\bold{x}]), \text{ and} \label{eq: prop52 A}\\	
		\Phi_{\forall\exists} & = & \bigwedge_{j}\forall\bold{z}. (\beta_{j}(\bold{z}) \Rightarrow \exists\bold{y}.\gamma_{j}(\bold{y}\bold{z})), \label{eq: prop52 AE}
	\end{eqnarray}
	with $\alpha_{i}$, $\beta_{j}$ guard atoms, $\gamma_{j}$ just atoms, and $\vartheta_{i}$ being a disjunction of literals. 
	We define the \emph{width} of $\Phi$ (denoted by $\mathrm{width}(\Phi)$) to be the maximal arity of a relation symbol appearing in $\Phi$. We denote by $C$ the set of all constant symbols from $\tau$.

	Let the set $\mathcal{E}$ of eligible types consist of all rigid guarded types of order $\le\mathrm{width}(\Phi)$.
	
	Let $\sigma$ contain for each $\tp \in \mathcal{E}$ a fresh relational symbols $\sig{Tp}_{\tp}$. Then we define the $\tau \uplus\sigma$-sentence
	\begin{equation}	
		\Psi = \Psi_{\mathrm{hom}} \land\bigvee_{(\mathcal{E}_{+}, \mathcal{E}_{!}) \in\getsummary{\Phi}} \Psi_{(\mathcal{E}_{+}, \mathcal{E}_{!})} ,
	\end{equation}	
	where $\Psi_{(\mathcal{E}_{+}, \mathcal{E}_{!})}$ is the sentence:
	\begin{equation}
		\Psi_{(\mathcal{E}_{+}, \mathcal{E}_{!})} = \Psi_{\mathrm{gen}, (\mathcal{E}_{+}, \mathcal{E}_{!})}\land \Psi_{\mathrm{guard}, (\mathcal{E}_{+}, \mathcal{E}_{!})}\land \Psi_{\text{\TGF{}}, (\mathcal{E}_{+}, \mathcal{E}_{!})}.
	\end{equation}
		We now show $\homcl{\getmodels{\Phi}} = \getmodels{\Psi}|_\tau$ by proving the two inclusions separately:
	
	``$\subseteq$'': Let $\mathfrak{A} \in \getmodels{\Phi}$ and let $h\colon\mathfrak{A} \to \mathfrak{B}$. Take $\mathcal{E}$ to be the set of all eligible types for $\Phi$ and $\sigma$ the set of fresh relation symbols $\sig{Tp}_{\tp}$ associated to each $\tp \in \mathcal{E}$ with the appropriate arity. Then the $\mathcal{E}$-adornment of $\mathfrak{A}$, $\mathfrak{A}\cdot\mathcal{E}$, is a $\tau\uplus\sigma$-structure which we will denote by $\mathfrak{A}^{\ast}$. Furthermore, let $\mathfrak{B}\cdot h(\mathfrak{A}^{\ast} |_{\sigma})$ be the $\tau \uplus\sigma$-structure described before, which we denote by $\mathfrak{B}^{\ast}$. Obviously $\mathfrak{B}^{\ast}|_{\tau} = \mathfrak{B}$ holds.

	By \Cref{lem:PsiGen} $\mathfrak{B}^{\ast}$ satisfies $\Psi_{\mathrm{gen}, (\mathcal{E}_{+}, \mathcal{E}_{!})}$. Additionally, after applying \Cref{lem:PsiGuardTGF} $\mathfrak{B}^{\ast}$ satisfies $\Psi_{\mathrm{guard}, (\mathcal{E}_{+}, \mathcal{E}_{!})}$, $\Psi_{\mathrm{hom}}$, and $\Psi_{\text{\TGF{}}, (\mathcal{E}_{+}, \mathcal{E}_{!})}$. Consequently $\mathfrak{B}^{\ast}\vDash \Psi$.
	
	``$\supseteq$'': Let $\mathfrak{B}$ be a model of $\Psi$. Application of \Cref{lem:ConstrA} yields a $\tau\uplus\sigma$ structure $\mathfrak{A}$ satisfying $\Phi$ and a homomorphism $h\colon\mathfrak{A}\to\mathfrak{B}$. Restricting both $\mathfrak{A}$ and $\mathfrak{B}$ to $\tau$ yields the claim.
\end{proof}

We can ``piggyback'' on this result to obtain the comparable statement for $\EAAE$.

\begin{proposition}[restate=EAAEinESO, name=]\label{prop:charEAAE}
For every \EAAE{} $\tau$-sentence $\Phi$, 
there exists
a \TGF{} 
sen\-tence $\Psi$ 
such that $\homcl{\getmodels{\Phi}} = \getmodels{\Psi}|_\tau$. 
\end{proposition}
\begin{proof}
	$\Phi$ can be projectively characterized by a \TGF{} sentence $\Phi'$ obtained as follows: skolemize the outer existential quantifiers into extra constants and add an atom using a fresh predicate of adequate arity as a guard following every inner existential quantifier. Clearly, $\Phi'$ is in \TGF{} and it projectively characterizes $\Phi$, so the former can be seen as the ``$\TGF$ normal form'' of the latter. 
	Let now $\Psi$ be obtained from $\Phi'$ following \Cref{prop:charTGF}. We then obtain $\homcl{\getmodels{\Phi}} = \getmodels{\Psi}|_{\tau}$ thanks to \Cref{lem:normalformOK}.
\end{proof}

\phithreesat*

\begin{lemma}\label{lem:NPhardnesFO2}
	\begin{enumerate}[itemindent=0ex, leftmargin=3ex, itemsep=-0.5ex]
		\item $\Phi_\mathrm{3SAT}$ is contained in constant-free, equality-free $\FOt$ (and therefore in $\FOte$ and \TGF).
\item There exists a sentence $\Psi$  in constant-free, equality-free $\Al\Al\Ex\Ex\FO$ (and therefore in 
$\EAAE$) such that $ \homcl{\getffinmodels{\Psi}}|_\tau\subseteq \homcl{\getffinmodels{\Phi_\mathrm{3SAT}}}$.
		\item The size of $\mathfrak{A}_\mathcal{S}$ is polynomial in the size of $\mathcal{S}$.
		\item $\bigwedge_{\{\ell,\ell',\ell''\}\in \mathcal{S}} \ell \vee \ell' \vee \ell''$ is satisfiable exactly if $\mathfrak{A}_\mathcal{S} \in \homcl{\getffinmodels{\Phi_\mathrm{3SAT}}}$.
		\item Checking membership in $\homcl{\getffinmodels{\Phi_\mathrm{3SAT}}}$ is \textsc{NP}-hard in the size of the structure. 
		\item Checking membership in $\homcl{\getffinmodels{\Psi}}$ is \textsc{NP}-hard in the size of the structure. 

	\end{enumerate}
\end{lemma}

\begin{proof}
	Item 1 follows immediately from the definition of $\Phi_\mathrm{3SAT}$.
	
	For Item 2: We let $\Psi = \forall x y \exists z v.\psi$ where $\psi$ is the conjunction over  
\begin{eqnarray}
	& & \sig{First}(z) \\
	& & \sig{Eq}(x,x) \\
		\sig{First}(x) & {\!\!\!\!\impl\!\!\!\!} &  \sig{CLst}(x) \\
	\sig{CLst}(x) \wedge \sig{Eq}(x,y) & {\!\!\!\!\impl\!\!\!\!} &  \sig{Lit}(x,v) \wedge \sig{Sel}(v)\\
	\sig{CLst}(x) \wedge \sig{Lit}(x,y) & {\!\!\!\!\impl\!\!\!\!} &  \sig{Next}(x,v) \ \ \ \ \ \nonumber\\[-0.6ex]	
	&  &  \ \wedge\, \big(\sig{CLst}(v) \vee \sig{Nil}(v)\big)\ \ \ \ \ \\	
	\sig{Sel}(x) \wedge \sig{Sel}(y) & {\!\!\!\!\impl\!\!\!\!} & \sig{Cmp}(x,y).
\end{eqnarray}

The containment $ \getffinmodels{\Psi}|_\tau \subseteq \getffinmodels{\Phi_\mathrm{3SAT}}$ follows immediately and therefore also $ \homcl{\getffinmodels{\Psi}}|_\tau\subseteq \homcl{\getffinmodels{\Phi_\mathrm{3SAT}}}$ holds.

	Item 3 follows directly from the definition of  $\mathfrak{A}_\mathcal{S}$.
	
	For Item 4:  Let $V=\{p_1,\ldots,p_m\}$ be the variable set of the propositional formula $\varphi =\bigwedge_{\{\ell,\ell',\ell''\}\in \mathcal{S}} \ell \vee \ell' \vee \ell''$. 
	
	Assume that there exists an assignment $\alpha\colon V\rightarrow \{\top,\bot\} $ that witnesses the satisfiability of $\varphi$. 
	Let $\tau$ be the signature of   $\mathfrak{A}_\mathcal{S}$ and let $\mathfrak{B}'$ be the $\tau {\setminus} \{\sig{Sel}\}$-reduct of   $\mathfrak{A}_\mathcal{S}$.
	The $\tau$-structure $\mathfrak{B}$ is the expansion of   $\mathfrak{B}'$ where the relation $\sig{Sel}^\mathfrak{B}$ is defined as follows:

\begin{enumerate}
	\item $\sig{Sel}^\mathfrak{B}\subset \{b_i,b_i'\mid 1\leq i\leq m\}$,
	\item 	$b_i \in \sig{Sel}^\mathfrak{B}$  if and only $\alpha(p_i)=\top$,
		\item 	$b_i' \in \sig{Sel}^\mathfrak{B}$  if and only $\alpha(p_i)=\bot$.
\end{enumerate}
One can easily check that $\mathfrak{B}$ satisfies $\Phi_\mathrm{3SAT}$. Furthermore, the identity map from  $\mathfrak{B}$  to $\mathfrak{A}_\mathcal{S}$  is a homomorphism which proves that $\mathfrak{A}_\mathcal{S} \in \homcl{\getffinmodels{\Phi_\mathrm{3SAT}}}$ holds.

For the opposite direction suppose that there exists a structure $\mathfrak{B}$ from $\getffinmodels{\Phi_\mathrm{3SAT}}$ that has a homomorphism $h$ to  $\mathfrak{A}_\mathcal{S} $. Note that Equation \ref{Def14: 1} of \Cref{def: 3SAT} ensures that the structure $\mathfrak{B}$ contains an element $x_f$ such that $x_f\in \sig{First}^\mathfrak{B}$ holds. Since $h$ is a homomorphism $h(x_f)=a_1$ holds. For every element $x$ in $\sig{CLst}^\mathfrak{B}$ exists by Equation \ref{Def14: 3} an element $y$ such that $(x,y)\in \sig{Next}^\mathfrak{B}$ holds. In the structure $\mathfrak{A}_\mathcal{S}$ the same holds but with unique elements $y$. This implies that there exists a set $W\subseteq \sig{CLst}^\mathfrak{B}$ with $x_f\in W$ such that $h$ restricted to $W$ is a bijection from $W$ to $\{a_i \mid 1\leq i\leq n+1\}$. We denote this restricted map by $h'$.
Let $f\colon W\rightarrow \sig{Sel}^\mathfrak{B}$ be a map such that $(x,f(x))\in \sig{Lit}^\mathfrak{B}$  holds for all $x\in W$. Such a map exists by Equation \ref{Def14: 4}.
Observe that $h(f(W)) \subset \{b_i,b_i' \mid 1\leq i\leq m\}$ holds. Moreover since we have $f(W)\subset \sig{Sel}^\mathfrak{B}$ by defintion, all pairs of distinct elements in $f(W)$ are in the relation $\sig{Cmp}^\mathfrak{B}$ by Equation \ref{Def14: 5}. This implies that for all $i\in \{1,\ldots,m\}$ at most one of $b_i$ and $b_i'$ is in  $h(f(W))$. We define a variable assignment $\alpha\colon V\rightarrow \{\top,\bot\} $ as follows:

\begin{enumerate}
	\item $\alpha(p_i) =\top$ if $b_i\in h(f(W))$,
		\item $\alpha(p_i) =\bot$ if $b_i'\in h(f(W))$,
		\item $\alpha(p_i) =\top$ if $b_i$ and $b_i'$ are not in $ h(f(W))$.
	
\end{enumerate}
This is a well defined map by what we saw before. Furthermore, to see that the clause $C_j$ is satisfied, consider the element $y =h'^{-1}(a_j) \in Q$. By the definition of $f$ we have that $(y,f(y))\in \sig{Lit}^\mathfrak{B} $  and therefore $(h(y),h(f(y)))\in \sig{Lit}^{\mathfrak{A}_\mathcal{S}} $  holds. Assume without loss of generality that $h(f(y)) =b_l$ for some $l\in \{1,\ldots,m\}$. Then we have $(a_j,b_l)\in \sig{Lit}^{\mathfrak{A}_\mathcal{S}} $ which means by the definition of $\mathfrak{A}_\mathcal{S}$ that $p_l$ is a literal of $C_j$. By the definition of $\alpha$ we have $\alpha(p_l)=\top$ which proves that $C_j$ is satisfied.

For Item 5: The 3SAT problem is a well known \textsc{NP}-hard problem \cite{Karp1972}. We have a reduction from the 3SAT problem to the membership problem for $\homcl{\getffinmodels{\Phi_\mathrm{3SAT}}}$ by mapping a 3SAT instance $\mathcal{S}$ to the structure $\mathfrak{A}_\mathcal{S} $. 
This is indeed a  reduction by Item 4 and it is polynomial time computable by Item 3.

For Item 6: We reduce the 3SAT problem to the membership problem of $\homcl{\getffinmodels{\Psi}}$. We map an instance $\mathcal{S}$ to the structure $\mathfrak{A}^*_\mathcal{S}$ that is the expansion of $\mathfrak{A}_\mathcal{S}$ by the relation $\sig{Eq}^{\mathfrak{A}^*_\mathcal{S}}=\{(x,x) \mid x\in A_\mathcal{S}\}$.

Assume that $\mathcal{S}$ is satisfiable. Let $\mathfrak{B}$ be the structure from the first direction of 5). It is easy to see that its expansion $\mathfrak{B}^*$  satisfies $\Psi$ and therefore 
$\mathfrak{A}^*_\mathcal{S}$ is in  $\homcl{\getffinmodels{\Psi}}$.

For the other direction assume that $\mathfrak{A}^*_\mathcal{S}\in \homcl{\getffinmodels{\Psi}}$ holds. By 2) we get that  then also $\mathfrak{A}_\mathcal{S} \in \homcl{\getffinmodels{\Phi_\mathrm{3SAT}}}$ holds which implies by 4) that $\mathcal{S}$ is satisfiable.\end{proof}

\NPhardFOprop*

\begin{proof} The proposition follows from Item~1, Item~2, Item~5 and Item~6 of \Cref{lem:NPhardnesFO2}.
\end{proof}

\GNFOinLFP*

\begin{proof}
Given a $\GNFO$ $\tau$-sentence $\Phi$, bring it in $\GFO$ ``homclosure normal form'' according to \Cref{GNFOredGFO},  
		 obtaining a $\GFO$ $\tau'$-sentence $\Phi'$. From $\Phi'$ compute $\Psi$ as in \Cref{prop:charGFO}, adjusting the definition of $\Psi_\mathrm{hom}$ as follows:
\begin{equation}
\bigwedge_{\tp \in \mathcal{E}} \!\forall \bold{v}_\tp.\big( \Tp_\tp(\bold{v}_\tp) \impl \smallbigwedge \tp|_\tau^+ \big),
\end{equation}
where
\begin{equation}
\tp|_\tau^+ = \{\alpha \tight{\in} \tp \mid \alpha = \sigP(\bold{t}) \text{ with }\sigP\in\tau\}
\end{equation}
so, compared to the original $\Psi_\mathrm{hom}$, the new one disregards auxiliary symbols from $\tau'\setminus \tau$. Clearly, we still obtain $\homcl{\getmodels{\Phi}} = \getmodels{\Psi}|_\tau$ and hence also $\homcl{\getmodels{\Phi}} = \getmodels{\exists \sigma.\Psi}$ where  
\newcommand{\PsihomEE}{\Psi_{\!(\mathcal{E}_\text{+},\mathcal{E}_!),\mathrm{hom}}}
\begin{align}
	\exists \sigma.\Psi & = \exists \sigma . \Big(\Psi_\mathrm{hom} \wedge \hspace{-1ex} \bigvee_{\scriptscriptstyle(\mathcal{E}_\text{+},\mathcal{E}_!)\in \getsummary{\Phi}} \hspace{-1ex} \Psi_{\!(\mathcal{E}_\text{+},\mathcal{E}_!)}\Big) \\
	& \equiv \hspace{-1ex} \bigvee_{\scriptscriptstyle(\mathcal{E}_\text{+},\mathcal{E}_!)\in \getsummary{\Phi}} \hspace{-1ex} \exists \sigma . \big(\Psi_{\!(\mathcal{E}_\text{+},\mathcal{E}_!)} \wedge \Psi_\mathrm{hom} \big)\\
	& =: \hspace{-1ex} \bigvee_{\scriptscriptstyle(\mathcal{E}_\text{+},\mathcal{E}_!)\in \getsummary{\Phi}} \hspace{-1ex} \exists \sigma . \big(\PsihomEE\big),
\end{align}
with $\PsihomEE$ being the conjunction over 
\noindent
\newcommand{\less}{\hspace{-1.5ex}}
\newcommand{\more}{\hspace{-1.5ex}}
\begin{eqnarray}
	\exists \bold{v}_\tp. \big(\sig{Tp}_\tp (\bold{v}_\tp)\big) \less\less\ & & \less \hspace{27ex}\more\mbox{for } \tp \in \mathcal{E}_{+}\quad\quad\ \\
	\forall \bold{v}_\tp. \big(\sig{Tp}_\tp (\bold{v}_\tp) \less& \Rightarrow &\less \bot \big)\hspace{21ex}\more\mbox{for } \tp \in \mathcal{E}\setminus\mathcal{E}_{+}  \\
	\forall \bold{v}_\tp. \big(\sig{Tp}_\tp (\bold{v}_\tp) \less& \Rightarrow &\less \sig{Tp}_{\tp'}(\bold{v}_{\tp'})\big) \hspace{7ex}\more\mbox{for } \tp,\tp' \in \mathcal{E}_{+}, \tp'=\tp|_{\bold{v}_{\tp'}}  \\
	\forall \bold{v}_\tp. \big(\sig{Tp}_\tp (\bold{v}_\tp) \less& \Rightarrow &\less \sig{Tp}_{\kappa\tp}( \eta^{\kappa\bold{v}_\tp}_\bold{v} (\bold{v}_\tp) )\big) \hspace{1.3ex}\mbox{for } \tp \in \mathcal{E}, \kappa : \bold{v}\injsur \bold{v} \\
	\forall \bold{v}_\tp. \big(\sig{Tp}_\tp (\bold{v}_\tp) \less& \Rightarrow &\less \hspace{-7ex} \bigvee_{\quad\quad(\tp,\nu) \Subset_j^\bold{z} (\tp',\nu')} \hspace{-6ex} \exists \bold{v}_{\tp'}{\!\setminus}\nu(\bold{z}) .\sig{Tp}_{\tp'}(\bold{v}_{\tp'})\big) \hspace{3ex}\more\mbox{for } (\tp,\nu)\! \in \!B_{j}\\
    \forall \bold{v}_\tp.\big( \Tp_\tp(\bold{v}_\tp) \less& \Rightarrow &\less \smallbigwedge \tp|_\tau^+ \big) \hspace{22ex}\more\mbox{for } \tp \in \mathcal{E}
\end{eqnarray} 
where $j$ ranges over the number of conjuncts in $\Phi_{\forall\exists}$ (cf. the proof of \Cref{prop:charGFO}).

In a next step, we equivalently rewrite  all constituents of $\PsihomEE$ as follows (all side conditions remain untouched):
\begin{eqnarray}
	\exists \bold{v}_\tp. \neg\big(\neg\sig{Tp}_\tp (\bold{v}_\tp)\big) \less\less\ & &  \\
	\forall \bold{v}_\tp. \big(\neg\sig{Tp}_\tp (\bold{v}_\tp) \less& \Leftarrow &\less \top\big) \\
	\forall \bold{v}_\tp. \big(\neg\sig{Tp}_\tp (\bold{v}_\tp) \less& \Leftarrow &\less \neg \sig{Tp}_{\tp'}(\bold{v}_{\tp'})\big) \\
	\forall \bold{v}_\tp. \big(\neg\sig{Tp}_\tp (\bold{v}_\tp) \less& \Leftarrow &\less \neg\sig{Tp}_{\kappa\tp}( \eta^{\kappa\bold{v}_\tp}_\bold{v} (\bold{v}_\tp) )\big) \\
	\forall \bold{v}_\tp. \big(\neg\sig{Tp}_\tp (\bold{v}_\tp) \less& \Leftarrow &\less \hspace{-6ex} \bigwedge_{\quad\quad(\tp,\nu) \Subset_j^\bold{z} (\tp',\nu')} \hspace{-6ex} \forall \bold{v}_{\tp'}{\!\setminus}\nu(\bold{z}) .\neg\sig{Tp}_{\tp'}(\bold{v}_{\tp'})\big) \\
	\forall \bold{v}_\tp.\big(\neg\Tp_\tp(\bold{v}_\tp) \less& \Leftarrow &\less \neg\smallbigwedge \tp|_\tau^+ \big)
\end{eqnarray} 

We next let for $\overline{\sigma} = \{ \overline{\Tp}_\tp \mid {\Tp}_\tp \in \sigma\}$ be a set of ``complement predicates'', with the intention that $\overline{\Tp}_\tp$ is meant to denote the complement relation of ${\Tp}_\tp$. With this in mind, we rewrite the above sentences in the straightforward way, obtaining 
\newcommand{\compPsi}{\hspace{-0pt}\smash{\overline{\Psi}}\hspace{-7pt}\phantom{\Psi}_{\!(\mathcal{E}_\text{+},\mathcal{E}_!),\mathrm{hom}}}
$\compPsi$  
as the conjunction over the following $\tau\uplus\overline{\sigma}$-sentences:
\begin{eqnarray}
	\exists \bold{v}_\tp. \neg\big(\overline{\sig{Tp}}_\tp (\bold{v}_\tp)\big) \ \!\less\less\less& &  \less \hspace{27ex}\more\mbox{for } \tp \in \mathcal{E}_{+}\quad\quad\ \\
	\forall \bold{v}_\tp. \big(\overline{\sig{Tp}}_\tp (\bold{v}_\tp) \less& \Leftarrow &\less \top \big) 
    \hspace{21ex}\more\mbox{for } \tp \in \mathcal{E}\setminus \mathcal{E}_{+} \label{fp0}\\
	\forall \bold{v}_\tp. \big(\overline{\sig{Tp}}_\tp (\bold{v}_\tp) \less& \Leftarrow &\less \overline{\sig{Tp}}_{\tp'}(\bold{v}_{\tp'})\big) 
	    \hspace{7ex}\more\mbox{for } \tp,\tp' \in \mathcal{E}_{+}, \tp'=\tp|_{\bold{v}_{\tp'}}  \label{fp1}\\
\forall \bold{v}_\tp. \big(\overline{\sig{Tp}}_\tp (\bold{v}_\tp) \less& \Leftarrow &\less \overline{\sig{Tp}}_{\kappa\tp}( \eta^{\kappa\bold{v}_\tp}_\bold{v} (\bold{v}_\tp) )\big) \hspace{1.3ex}\mbox{for } \tp \in \mathcal{E}, \kappa : \bold{v}\injsur \bold{v} \\
	\forall \bold{v}_\tp. \big(\overline{\sig{Tp}}_\tp (\bold{v}_\tp) \less& \Leftarrow &\less \hspace{-7.5ex} \bigwedge_{\quad\quad\quad(\tp,\nu) \Subset_j^\bold{z} (\tp',\nu')} \hspace{-8ex} \forall \bold{v}_{\tp'}{\!\setminus}\nu(\bold{z}) .\overline{\sig{Tp}}_{\tp'}(\bold{v}_{\tp'})\big)
	    \hspace{2.5ex}\more\mbox{for } (\tp,\nu) \in B_{j} \label{fp2}\\
	\forall \bold{v}_\tp.\big(\overline{\sig{Tp}}_\tp(\bold{v}_\tp) \less& \Leftarrow &\less \neg\smallbigwedge \tp|_\tau^+ \big)  \hspace{20ex}\more\mbox{for } \tp \in \mathcal{E} \label{fp3}
\end{eqnarray} 

Note that, for the mentioned reasons, 
\begin{equation}
\exists \overline{\sigma} . \big(\compPsi \big) \equiv \exists \sigma . \big(\PsihomEE \big).
\end{equation} 
Also note that the set of sentences of the forms (\ref{fp1}), (\ref{fp2}), and (\ref{fp3}) can be grouped such that, for every $\overline{\Tp}_\tp \in \overline{\sigma}$, there exists precisely one sentence of the form
\begin{equation} 
\forall \bold{v}_\tp. \big(\overline{\sig{Tp}}_\tp (\bold{v}_\tp) \Leftarrow \xi_{\overline{\Tp}_\tp}[\bold{v}_\tp] \big),
\end{equation} 
where $\xi_{\overline{\Tp}_\tp}[\bold{v}_\tp]$ is a $\FOe$ sentence with only positive occurrences of $\overline{\sigma}$-atoms (obtain $\xi_{\overline{\Tp}_\tp}[\bold{v}_\tp]$ as the disjunction of all the bodys of those implications whose head is $\overline{\sig{Tp}}_\tp (\bold{v}_\tp)$). Therefore, the obtained set of sentences gives rise to a monotonic operator over $\tau \uplus \overline{\sigma}$-structures simultaneously inflating all $\overline{\sig{Tp}}_\tp$ relations. Hence, $\exists \overline{\sigma} . \big(\compPsi \big)$ is equivalent to the simultaneous inductive $\FOe$ least fixed-point expression $\compPsi^\mathbf{sim\text{-}lfp}$ defined by
\begin{equation} 
\bigwedge_{\tp \in \mathcal{E}_\text{+}} \exists \bold{x}.\neg\Big[\textbf{lfp}_{\overline{\Tp}_\tp} \big\{\, \overline{\Tp}_{\tp'}(\bold{v}_\tp) \leftarrow \xi_{\overline{\Tp}_{\tp'}}[\bold{v}_\tp] \,\, \big|\,  \tp' \tight{\in} \mathcal{E}\,\big\}\Big](\bold{x}).
\end{equation} 
As simultaneous induction can be expressed using plain least-fixed-point logic \cite{GurevichS86}, there exists some plain $\FOe^\mathbf{lfp}$ sentence $\compPsi^\mathbf{lfp}$ equivalent to $\compPsi^\mathbf{sim\text{-}lfp}$. By defining
\begin{equation} 
\Psi^\mathbf{lfp}=\bigvee_{\scriptscriptstyle(\mathcal{E}_\text{+},\mathcal{E}_!)\in \getsummary{\Phi}} \hspace{-1ex} \compPsi^\mathbf{lfp}, 
\end{equation}  
we obtain $\exists \sigma. \Psi \equiv \Psi^\mathbf{lfp}$
and therefore $\getmodels{\Psi^\mathbf{lfp}} = \homcl{\getmodels{\Phi}}$ as desired.
Using the fact that evaluating $\FOe^\mathbf{lfp}$ sentences over structures can be done in polynomial time \cite{Immerman82,Vardi82}, we obtain polytime data complexity of homclosure membership.
\end{proof}

\subsection{Characterizing Finite-Model Homclosures}


\cornercaseone*

\begin{proof}
	$\homcl{\getfinmodels{\Phi_\infty}}$ contains precisely all -- finite or infinite -- $\{\sigP\}$-structures that have a directed $\sigP$-cycle. Toward a contradiction, let us assume that there were an \ESO{} formula $\Psi$ characterizing $\homcl{\getfinmodels{\Phi_\infty}}$. This means the first-order part of $\Psi$, denoted $\Psi'$, projectively characterizes $\homcl{\getfinmodels{\Phi_\infty}}$. Now, for a given number $k$, let $\Gamma^{\not\circlearrowright}_k$ be the $\FO$ sentence disallowing directed $\sigP$-cycles of length $k$ defined by
\begin{equation} 
\Gamma^{\not\circlearrowright}_k \ \ = \ \ \forall x_1\ldots x_k . \neg \sigP(x_k,x_1) \vee \bigvee_{i=1}^{k-1} \neg \sigP(x_i,x_{i+1}).
\end{equation}  
	Now define the infinite $\FO$ theory $\Xi = \{\Psi'\}\cup\{ \Gamma^{\not\circlearrowright}_k \mid k > 0 \}$.
	Obviously any finite subset of $\Xi$ is satisfiable, while $\Xi$ is not. However this contradicts the compactness theorem of $\FOe$ and therefore the presumed existence of $\Psi$ is refuted.
\end{proof}

\cornercasetwo*

\begin{proof}
	``$\subseteq$'': let $\mathfrak{B} \in \homcl{\getfinmodels{\Phi}}$. Then, by definition, there exists some finite $\mathfrak{A} \in  \getfinmodels{\Phi}$ and a homomorphism $h:\mathfrak{A} \arbhom \mathfrak{B}$. Let $\mathfrak{C}$ be the finite substructure of $\mathfrak{B}$ induced by $h(A)$. Then $h$ also serves as a homomorphism from $\mathfrak{A}$ to $\mathfrak{C}$ and therefore $\mathfrak{C} \in \homclfin{\getfinmodels{\Phi}}$ and, by assumption, $\mathfrak{C} \in \getfinmodels{\Psi}$. We now obtain $\mathfrak{D}$ from $\mathfrak{B}$ by adding $\sigU$ to its signature and letting $\sigU^\mathfrak{D} = h(A)$. Then, obviously, $\mathfrak{C} = [\mathfrak{D}]_\sigU$ and, thanks to  $\mathfrak{C} \in \getfinmodels{\Psi}$ we obtain $\mathfrak{D} \in \getmodels{\Psi^{\mathrm{rel}(\sigU)}}$ by the remarks from \Cref{sec:tools}.
	By construction, $\sigU^\mathfrak{D}$ is finite, therefore $\mathfrak{B} \in \getmodels{\exists^\fin \sig{U}. \Psi^{\mathrm{rel}(\sigU)}}$ as desired.
	
	``$\supseteq$'': let $\mathfrak{B} \in \getmodels{\exists^\fin \sigU. \Psi^{\mathrm{rel}(\sigU)}}$. Then there must be a structure $\mathfrak{D} \in \getmodels{\Psi^{\mathrm{rel}(\sigU)}}$ expanding $\mathfrak{B}$ by interpreting an additional signature element $\sigU$ such that $\sigU^\mathfrak{D}$ is finite. Let $\mathfrak{C} = [\mathfrak{D}]_\sigU$, which is finite by construction and an induced substructure of $\mathfrak{B}$. Then by the properties of relativization follows $\mathfrak{C} \in \getmodels{\Psi}$ and due to finiteness of $\mathfrak{C}$ even $\mathfrak{C} \in \getfinmodels{\Psi}$, therefore, by assumption $\mathfrak{C} \in \homclfin{\getfinmodels{\Phi}}$, meaning there exists some $\mathfrak{A} \in \getfinmodels{\Phi}$ and a  homomorphism $h:\mathfrak{A} \to \mathfrak{C}$. But then $id_{C} \circ h$ is a homomorphism from $\mathfrak{A}$ to $\mathfrak{B}$, witnessing $\mathfrak{B} \in \homcl{\getfinmodels{\Phi}}$ as desired.
\end{proof}

\subsection{Proofs for Normal Form Section}

\boundedex*

\begin{proof}
	Assume $\Phi \in \bex^n\bal^*\FOe$ is over the signature $\tau {\uplus} \sigma$ with predicates $\tau$ and constants $\sigma$. Let $\sigma^\mathrm{sk}$ be a set of $n$ Skolem constants.
	Let $\Phi^\mathrm{sk}$ be the $\AsFO$ $\tau {\uplus} \sigma {\uplus} \sigma^\mathrm{sk}$-sentence obtained from $\Phi$ via skolemization of all the existentially quantified variables into (additional) constants. Clearly then  $\getffinmodels{\Phi} = \getffinmodels{\Phi^\mathrm{sk}}|_\tau$. Since $\Phi^\mathrm{sk}$ is a $\AsFO$ sentence, any model $\mathfrak{A}$ gives rise to a small submodel by taking the substructure $\mathfrak{A}'$ induced by the set $A' \subseteq A$ containing all elements denoted by constants. The identity on $A'$ is a homomorphism from $\mathfrak{A}'$ to $\mathfrak{A}$. Let $\mathcal{A}$ be the set of small submodels thus obtained, for which therefore holds $\homcl{\mathcal{A}}=\homcl{\getmodels{\Phi^\mathrm{sk}}}$ and hence also $\homcl{\mathcal{A}|_\tau}=\homcl{\getmodels{\Phi}} = \getmodels{\Phi}$. 
	Now let us obtain $\Psi = \exists \sigma^\mathrm{sk}.\bigvee_{\mathfrak{A}\in \mathcal{A}} \mathrm{cq}'(\mathfrak{A})$ where  
\begin{equation} 
\mathrm{cq}'(\mathfrak{A}) \ = \ \hspace{-2ex} \bigwedge_{{\sigc,\sigc'\in \sigma \cup \sigma^\mathrm{sk}}\atop{\mathfrak{A}\models \sigc{=}\sigc'}}\hspace{-2.5ex}
\sigc \,{=}\, \sigc' \quad \wedge \quad  \hspace{-3ex} \bigwedge_{{\sigP \in \tau,\, k=\text{ar}(\sigP)}\atop{{\mathfrak{A}\models \sigP{} (\sigc_1,\ldots,\sigc_k) }}} \hspace{-3.5ex} 
\sigP{} (\sigc_1,\ldots,\sigc_k). 
\end{equation}  
	
	Then we get $\getmodels{\Psi} = \homcl{\mathcal{A}|_\tau} = \getmodels{\Phi}$.
	Moreover $\Psi \in \bex^n\FOe^+$.
\end{proof}

\classnormalform*

\begin{proof}
	1) and 2) follow immediately since ``being superstructure'' and ``having surjective homomorphisms'' are transitive relations on structures.
	
	For 3) note that a class is closed under (finite) homomorphism if and only if it is closed under (finite) superstructures and (finite) surjective homomorphisms. By 1) we get that $\supclffin{(\homsucl{\mathcal{C}})}$  is closed under  (finite) superstructures. 
	Therefore it remains to show the closure under (finite) surjective homomorphisms.
	
	Let $\mathfrak{A}\in \supclffin{(\homsucl{\mathcal{C}})}$ and suppose that there exists a (finite) surjective homomorphism $h\colon \mathfrak{A}\rightarrow  \mathfrak{B}$ where $\mathfrak{B}$ is a $\tau$-structure.
	By definition of $\supclffin{(\homsucl{\mathcal{C}})}$ there exists a substructure $\mathfrak{A}_0$ of $\mathfrak{A}$ such that $\mathfrak{A}_0 \in \homsucl{\mathcal{C}}$. 
	Let $\mathfrak{B}_0$ be the substructure of $\mathfrak{B}$ that is induced by $h(A_0)$. The structure $\mathfrak{A}_0$ has clearly a (finite) surjective homomorphism to $\mathfrak{B}_0$ and since $\homsucl{\mathcal{C}}$ is by 2) closed under (finite) surjective homomorphisms it follows that $\mathfrak{B}_0\in \homsucl{\mathcal{C}}$. 
	Since $\mathfrak{B}_0$ is a substructure of $\mathfrak{B}$ it follows by the definition of  $\supclffin{(\homsucl{\mathcal{C}})}$ that $\mathfrak{B}\in  \supclffin{(\homsucl{\mathcal{C}})}$ holds.\end{proof}

\supstructure*

\begin{proof}
	Let $\mathfrak{A} \in \supclffin{\getffinmodels{\Psi}}$ and let $\mathfrak{B}$ a substructure of $\mathfrak{A}$ such that $\mathfrak{B}\in \getffinmodels{\Psi}$. Such a substructure exists by the definition of  $\supclffin{\getffinmodels{\Psi}}$. We can therefore choose $B\subseteq A$ as a witness for the set variable $\dot{\sigU}$ in $\exists \dot{\sigU}. \Psi^{\mathrm{rel}(\dot{\sigU})}$ and get that $\mathfrak{A} \in \getffinmodels{\Psi^{\mathrm{sup}}}$.
	For the other direction suppose that $\mathfrak{A} \in \getffinmodels{\Psi^{\mathrm{sup}}}$ holds. This means by definition that there exists a subset $B$ of $A$ such that $\Psi^{\mathrm{rel}(B)}$ holds on $\mathfrak{A}$. Therefore the structure $\mathfrak{B}$ that is induced by  $\mathfrak{A}$ on $B$ satisfies $\Psi$ and we get $\mathfrak{A} \in \supclffin{\getffinmodels{\Psi}}$.
\end{proof}

\surjhom*

\begin{proof}
	Let $\mathfrak{A} \in \homsucl{\getffinmodels{\Psi}}$ and let $\mathfrak{A}'$ be a $\tau{\uplus} \tau' {\uplus} \{\sigU\}$-expansion of $\mathfrak{A}$ that satisfies $\Theta_\sigU^{\tau'} \wedge\bigwedge_ {\sigP'\in \tau'} \forall \bold{x}.\big(\sigP(\bold{x}) \,{\impl}\, \sigP'(\bold{x})\big)$. Let $\mathfrak{B}_1$ be the substructure of $\mathfrak{A}'$  that is induced by $\sigU^{\mathfrak{A}'}$. Let $\mathfrak{B}_2$ be the $\tau'$-reduct of   $\mathfrak{B}_1$. We define $\mathfrak{B}$ as the $\tau$-structure that arises from the signature substitution $\tau' \mapsto \tau$ in the structure $\mathfrak{B}_2$.
	We claim that there exists a surjective homomorphism $ h\colon \mathfrak{A}\surhom \mathfrak{B}$ with the following definition: The map is the identity on $\sigU^{\mathfrak{A}'}$. Each element $x\not \in \sigU^{\mathfrak{A}'}$ is mapped to some element $y \in \sigU^{\mathfrak{A}'}$ such that $\eta^{\tau'}(x,y)$ holds. Such a mapping exists (but is not unique) according to the sentence $\Theta_\sigU^{\tau'}$. 
	Since $\sigP^{\mathfrak{A}'}\subseteq {\sigP'}^{\mathfrak{A}'}$ holds and by the formula  $\eta^{\tau'}(x,y)$ it follows that $h$ is a homomorphism. Note that $h$ is surjective since it is the identity on $\sigU^{\mathfrak{A}'}$, which is the domain of $\mathfrak{B}$. By our assumption that $\mathfrak{A}$ is from $ \homsucl{\getffinmodels{\Psi}}$ it follows that $\mathfrak{B}\in  \getffinmodels{\Psi}$ holds. This is nothing else than saying that $\mathfrak{A} \models \Psi_{\tau \mapsto \tau'}^{\mathrm{rel}(\sigU) }$ holds, which proves the first inclusion.

	For the second inclusion, assume $\mathfrak{A}\in \getffinmodels{\Psi^{\mathrm{shom}}}$ and that there exists a surjective homomorphism $h\colon \mathfrak{A} \surhom\mathfrak{B}$. Define the $\tau{\uplus} \tau' {\uplus} \{\sigU\}$-expansion $\mathfrak{A}'$ of $\mathfrak{A}$  as follows. For all $\sigP\in \tau$ we define the relation ${\sigP'}^{\mathfrak{A}'}$ as the preimage of $\sigP^\mathfrak{A}$ under $h$. The set $\sigU^{\mathfrak{A}'}$ is defined as some set of representatives of the kernel classes of $h$. It is clear that $\mathfrak{A}'$ satisfies $\Theta_\sigU^{\tau'} \wedge\bigwedge_ {\sigP'\in \tau'} \forall \bold{x}.\big(\sigP(\bold{x}) \,{\impl}\, \sigP'(\bold{x})\big)$ and therefore it satisfies
	$\Psi_{\tau \mapsto \tau'}^{\mathrm{rel}(\sigU) }$ as well.	
	Let $\mathfrak{A}_1$ be the structure that is induced  by $\mathfrak{A}'$ on $\sigU^{\mathfrak{A}'}$. 
	By the definition of $\mathfrak{A}'$ the $\tau'$-reduct of $\mathfrak{A}_1$ is isomorphic to $\mathfrak{B}$ after the signature substitution $\tau' \mapsto \tau$. This implies that $\mathfrak{B}$ satisfies $\Psi$. Since the surjective homomorphism $h$ and the structure  $\mathfrak{B}$ were arbitrary, we get $\mathfrak{A} \in \homsucl{\getffinmodels{\Psi}}$. \end{proof}



\end{document}